\documentclass[10pt]{article}
\usepackage{customcommands}

\title{Policy Transforms and Learning Optimal Policies\\\footnotetext{I thank Jiaying Gu, Ismael Mourifie, Eduardo Souza-Rodrigues, Adam Rosen, Stanislav Volgushev and Yuanyuan Wan for their feedback and encouragement, and I am especially grateful to JoonHwan Cho for many hours of discussion that helped to improve this paper. A previous version of this paper appeared in my doctoral thesis at the University of Toronto. This research was supported by the Social Sciences and Humanities Research Council of Canada. All errors are my own.}
}
\date{\small \today}
\author{Thomas M. Russell\footnote{Thomas M. Russell, Assistant Professor, Department of Economics, Carleton University, 1125 Colonel By Drive, Ottawa, Ontario, K1S5B6, Canada. Email: thomas.russell3@carleton.ca. 
}\\\textit{Carleton University}}

\begin{document}
\maketitle
\vspace{-.4in}
\begin{abstract}
\noindent We study the problem of choosing optimal policy rules in uncertain environments using models that may be incomplete and/or partially identified. We consider a policymaker who wishes to choose a policy to maximize a particular counterfactual quantity called a \textit{policy transform}. We characterize \textit{learnability} of a set of policy options by the existence of a decision rule that closely approximates the maximin optimal value of the policy transform with high probability. Sufficient conditions are provided for the existence of such a rule. However, learnability of an optimal policy is an ex-ante notion (i.e. before observing a sample), and so ex-post (i.e. after observing a sample) theoretical guarantees for certain policy rules are also provided. Our entire approach is applicable when the distribution of unobservables is not parametrically specified, although we discuss how semiparametric restrictions can be used. Finally, we show possible applications of the procedure to a simultaneous discrete choice example and a program evaluation example.
\end{abstract}

\medskip



\thispagestyle{empty}

\clearpage

\section{Introduction}

One of the fundamental goals of econometrics is to credibly translate knowledge of underlying economic mechanisms into models that, when combined with sample data, can be used to understand the effects of counterfactual policy experiments and can help guide policy decisions. In this paper we consider the problem of making policy decisions in settings where the econometric model is partially identified and/or incomplete. The paper is motivated by the fact that credible models are needed to honestly inform policy makers on the impacts of counterfactual policies, even if credible models provide an incomplete description of the true data generating process. 

Our framework is general enough to accommodate many existing structural econometric models. Our description of the environment is similar to descriptions found in \cite{jovanovic1989observable} and \cite{chesher2017generalized}, which in turn are extensions of the classical foundations for econometric modelling set forth in \cite{koopmans1950measuring} and \cite{hurwicz1950generalization}, among others. We assume the economic system under consideration manifests as a collection of random variables which can be partitioned into those that are observable---including a vector of observed endogenous variables $Y$ and a vector of exogenous variables $Z$---and those that are latent or unobservable---denoted by the vector $U$. We refer colloquially to the variables contained in $Y$ and $Z$ as the ``observables,'' and refer to the variables contained in $U$ as the ``unobservables.'' Unlike most of the existing literature, we do not take the distribution of $U$ as a model primitive. This is in accordance with the perspective that the latent variable $U$ represents the gap between what can be explained by a theoretical model, and what must remain unexplained; that is, ``errors in equations'' rather than ``errors in variables.''\footnote{These two explanations of the error term are documented by \cite{morgan1990history} Chapter 6. We recommend \cite{qin2001error} for a review of how attitudes towards the latent variables have evolved over time. } As we will demonstrate, such a distinction becomes especially important when performing counterfactual analyses.\todo{This point could be pressed further in the text, when applicable, and in the conclusion.}

The policymaker is assumed to have access to data on the observables, as well as an econometric model that describes how the observables are related to the unobservables. The model may depend on a vector of parameters $\theta \in \Theta$; here $\Theta$ is required only to be a complete and separable metric space, which permits many function spaces used in nonparametric analyses. We then let $\Gamma$ represent an abstraction of the set of all possible policies under consideration by the policymaker, where $\gamma \in \Gamma$ denotes one such policy. Each hypothetical policy $\gamma \in \Gamma$ represents an intervention on the underlying existing economic system, which operates to generate the endogenous variables from the exogenous and unobserved variables. After the economic system is modified, the resulting system may now generate a new, or counterfactual distribution of the endogenous variables. Thus, by altering the underlying economic system, a policy intervention induces a change between the factual (or observed) and counterfactual (hypothetical and unobserved) distributions of the endogenous outcome variables. Latent variables are not affected by the policy, and instead serve as important links between the factual and counterfactual domains.\footnote{From \cite{pearl2009causality} p. 211: ``The background variables are the main carriers of information from the actual world to the hypothetical world; they serve as the ``guardians of invariance" (or persistence) in the dynamic process that transforms the former into the latter."} \todolt{might be nice to reference interventionist approaches to counterfactuals} A policymaker's problem is then formulated as the problem of choosing a policy intervention that induces a counterfactual distribution of endogenous outcome variables that is favourable according to some criterion. 

We denote the counterfactual endogenous outcome variables as $Y_{\gamma}^\star$, where the $\gamma$ index is to emphasize the fact that its distribution will depend on the counterfactual policy experiment $\gamma \in \Gamma$ under consideration. Under this setup, this paper focuses on a particular class of counterfactual quantities that can be written in the following form:
\begin{align}
I[\varphi](\gamma):= \int \varphi(v) \, dP_{V_{\gamma}}.\label{eq_integral_of_interest}
\end{align}
Here $\varphi$ is some function, $V_{\gamma} := (Y_{\gamma}^\star,Y,Z,U)$ is a vector of all the random variables that describe the factual and counterfactual domains, $P_{V_{\gamma}}$ denotes the distribution of $V_{\gamma}$, and $v$ denotes a realization of $V_{\gamma}$. In particular, the operator $I[\,\cdot\,](\gamma)$ takes a function $\varphi$ of the vector $v$ of endogenous, exogenous, unobserved and counterfactual variables, and maps it to a function $I[\varphi](\gamma)$ of the policy parameter $\gamma$. For this reason, we refer to $I[\,\cdot\,](\gamma)$ as a \textit{policy transform}. As we will show in our examples on simultaneous discrete choice and program evaluation, counterfactual objects that can be written as policy transforms include counterfactual choice probabilities, and counterfactual average effects. If a policymaker's counterfactual object of interest can be written as the policy transform of some function $\varphi$, then the resulting policy transform gives all the information the policymaker needs to compare various policies and make a policy choice. 

Throughout the paper we consider a policymaker who wishes to maximize the value of the policy transform, although our analysis is equally applicable to the case when the policymaker wishes to minimize the value of the policy transform. With perfect knowledge of the distribution of the vector $V_{\gamma}$, the policymaker faces a trivial decision problem and can simply choose the policy $\gamma$ that obtains the maximum of the policy transform $I[\varphi](\gamma)$. However, this idealized decision problem is rarely encountered in practice, and instead the policymaker may only have access to a finite sample of the observed random variables. Furthermore, even with an infinite sample the policy transform may not be identified under any credible assumptions. This will be especially true throughout our discussion, since we will not require that the distribution of the unobservables $U$ be parametrically specified. 

To make progress, we model the policy decision problem as a decision under ambiguity, where we assume that the ``true state of the world'' belongs to a state space $\mathcal{S}\times \mathcal{P}_{Y,Z}$. Here $\mathcal{P}_{Y,Z}$ is the set of all Borel probability measures on the observable space $\mathcal{Y}\times\mathcal{Z}$. Furthermore, each $s \in \mathcal{S}$ is associated with a pair of conditional distributions $(P_{U|Y,Z}, P_{Y_{\gamma}^\star|Y,Z,U})$. Taking a pair $(s,P_{Y,Z}) \in \mathcal{S}\times \mathcal{P}_{Y,Z}$ to be the true state, the policymaker can evaluate the policy transform in \eqref{eq_integral_of_interest} corresponding to that state. Keeping the dependence on $P_{Y,Z}$ implicit, we denote the policy transform in state $(s,P_{Y,Z})$ as $I[\varphi](\gamma,s)$, and refer to it as the \textit{state-dependent policy transform}. We then consider the policymaker's decision problem when she has access to a finite sample from the true distribution. Let $\Psi_{n}$ denote the space of all possible $n-$samples $\{(y_{i},z_{i})\}_{i=1}^{n}$, and let $d: \Psi_{n} \to \Gamma$ denote a (measurable) decision rule that maps from sample realizations to policies. Before a sample $\psi \in \Psi_{n}$ is observed $d(\psi)$ will be a random variable, and the policymaker's problem is then translated into the problem of selecting a decision rule according to some reasonable criteria. 

However, without knowledge of the true state, it is unclear how the policymaker should (in a prescriptive sense) choose among, or rank, various decision rules. One nearly self-evident requirement on any method of ranking decision rules is that the ranking should respect weak dominance; that is, if for every $P_{Y,Z} \in \mathcal{P}_{Y,Z}$ we have $I[\varphi](d'(\psi),s) \leq I[\varphi](d(\psi),s)$ a.s. for every $s \in \mathcal{S}$, then $d$ should be preferred to $d'$. However, it is clear that many decisions rules will not be comparable according to this partial ordering. 

To progress further, we introduce a preference relation over the space of all decision rules that is motivated from computational learning theory. In particular, fix any $\kappa \in (0,1)$ and let $c_{n}(d,\kappa)$ be the smallest value satisfying:
\begin{align}
\inf_{P_{Y,Z} \in \mathcal{P}_{Y,Z}} P_{Y,Z}^{\otimes n} \left(\inf_{s \in \mathcal{S}} I[\varphi](d(\psi),s) + c_{n}(d,\kappa) \geq \sup_{\gamma \in \Gamma} \inf_{s \in \mathcal{S}} I[\varphi](\gamma,s)  \right) \geq \kappa.
\end{align}
Then under our framework, a decision rule $d:\Psi_{n} \to \Gamma$ is weakly preferred to decision rule $d':\Psi_{n} \to \Gamma$ at level $\kappa$ and sample size $n$ if $c_{n}(d,\kappa) \leq c_{n}(d',\kappa)$.\footnote{See Definition \ref{definition_pac_preference_relation}.} This preference relation appears to be new, and diverges (to some extent) from the existing literature on frequentist decision theory. However, its close connection to the probably approximately correct (PAC) learning framework from computational learning theory allows us to use a rich set of results from statistical learning theory and empirical process theory to study its theoretical properties. In addition, this preference relation induces a total ordering, and our first result in Section \ref{section_methodology} demonstrates that, at a minimum, this preference relation respects weak dominance.

Given this preference relation, throughout the paper we will use the value $c_{n}(d,\kappa)$ to measure the ``performance'' or ``quality'' of a decision rule $d$ for a given sample size $n$ and confidence level $\kappa$. We then provide two sets of theoretical results for the policymaker's decision problem. 

In the first set of results, we provide conditions on the decision problem that guarantees the existence of a decision rule $d$ such that $c_{n}(d,\kappa)$ tends to zero as the sample size $n$ becomes large. The existence of such a decision rule characterizes the notion of policy space learnability. The definition of policy space learnability appears to be new in economics, although it is adapted from the widely popular PAC learning framework from computer science proposed by \cite{valiant1984theory}. Our particular analysis deals mostly with the decision theoretic generalization of the PAC learning model proposed by \cite{haussler1992decision}, which is referred to as the \textit{agnostic} PAC learning model. 

We show that even in simple environments the policy space may not be learnable. In this case the policymaker's decision problem is still well-defined, but there will be theoretical limitations on how well any given policy can perform, even in large samples.  We then provide sufficient conditions for learnability which are related to certain complexity measures of the class of functions in our problem; in particular, to the behaviour of covering/packing numbers and metric entropy. We define an ``entropy growth condition,'' and we show that if certain key classes of functions in our environment satisfy this condition, then the policy space $\Gamma$ is learnable. Primitive conditions for our entropy growth condition can be found in the literature on empirical processes and statistical learning. In addition to being sufficient for learnability, we also show how the condition can be used to establish rates of convergence. 

However, since learnability is an \textit{ex-ante} notion (i.e. before observing the sample), verifying learnability can be uninformative about the \textit{ex-post} performance (i.e. after observing the sample) of a given policy rule. Thus, our second set of results provides a means for the policymaker to perform an ex-post analysis of her selected policy rule. First we study the finite sample properties of a particular decision rule, called the $\varepsilon-$maximin empirical (eME) rule, which selects a $\varepsilon-$maximizer of the worst case (over $s \in \mathcal{S}$) empirical version of $I[\varphi](\gamma,s)$. Using concentration inequalities, we provide an upper bound on the quantity $c_{n}(d,\kappa)$ when $d$ is the eME rule, and we demonstrate how the upper bound is affected by various features of the decision problem. 

However, the eME rule is only one particular rule, and for many reasons it may not be the policy rule selected by the policymaker. We thus turn to the problem of approximating the set of all policies $\gamma \in \Gamma$ satisfying:
\begin{align}
\gamma\mapsto \sup_{\gamma \in \Gamma} \inf_{s \in \mathcal{S}} I[\varphi](\gamma,s) - \inf_{s \in \mathcal{S}} I[\varphi](\gamma,s) \leq \delta,  
\end{align}
with probability at least $\kappa$; note that any decision rule that selects a policy in this set will thus have $c_{n}(d,\kappa)\leq \delta$. We call this set of policies the ``$\delta$-level set,'' and we show how a procedure from the literature on excess risk bounds in statistical learning theory can be adapted to our environment to approximate the $\delta-$level set. Finally, we show that the eME decision rule selects a policy in the $\delta-$level with high probability for $\delta$ sufficiently large, providing further justification for its use. Unlike the first ex-ante analysis of learnability, all of the results comprising the ex-post analysis do not require the entropy growth condition---or any other sufficient condition for learnability---to be satisfied. Thus, they are applicable whether or not the policy space $\Gamma$ is learnable, although they are silent about rates of convergence. Taken altogether, we believe our two sets of theoretical results provide a comprehensive means of making and evaluating policy decisions.  

This paper also makes a contribution from an identification perspective. Perhaps unsurprisingly, an important theoretical object in our study of policy decisions are the following policy transform envelope functions:
\begin{align*}
I_{\ell b}[\varphi](\gamma):= \inf_{s \in \mathcal{S}} I[\varphi](\gamma,s), &&I_{u b}[\varphi](\gamma):= \sup_{s \in \mathcal{S}} I[\varphi](\gamma,s).
\end{align*}
Regardless of the true (sub-)state $s_{0} \in \mathcal{S}$, at the true distribution $P_{Y,Z}$ the policy transform in \eqref{eq_integral_of_interest} can be ``sandwiched'' between these upper and lower envelope functions. This idea is illustrated in Figure \ref{fig_policy_choice}. Our ability to provide a tractable characterization of these envelope functions thus turns out to be critical to our ability to provide sufficient conditions for policy learnability, and for our ex-post analysis of the eME rule and the $\delta-$level set.  

\begin{figure}[!t]
\centering
\includegraphics[scale=0.6]{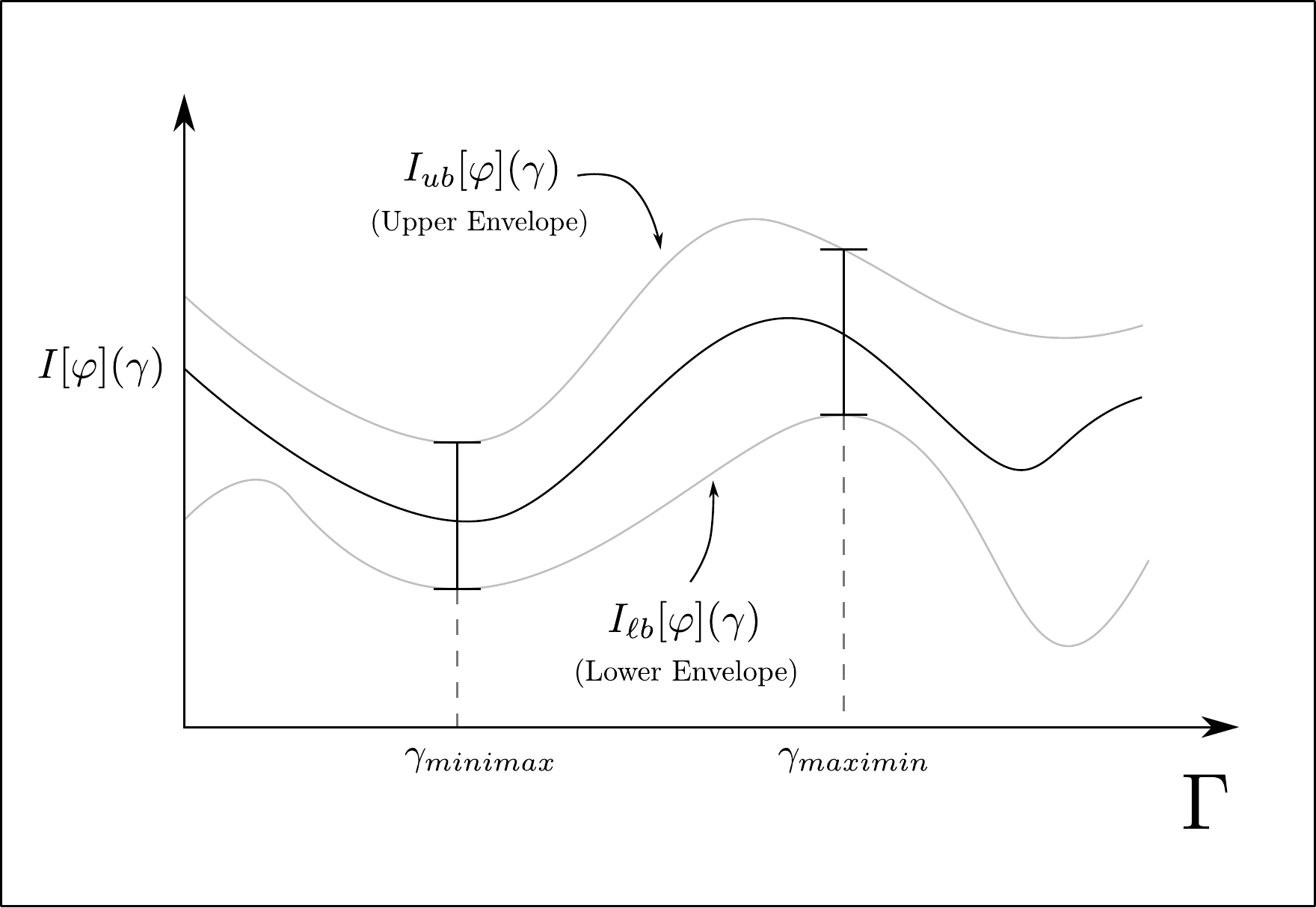}
\caption{This figure illustrates the policy transform of some function $\varphi$, as well as the upper and lower envelope functions $I_{u b}[\varphi](\gamma)$ and $I_{\ell b}[\varphi](\gamma)$ (resp.). The minimax (over (sub-)states $s \in\mathcal{S}$) policy is the policy that minimizes the upper envelope, and the maximin (over (sub-)states $s \in\mathcal{S}$) policy choice is the policy that maximizes the lower envelope. }\label{fig_policy_choice}
\end{figure}
The envelope functions may not be policy transforms themselves, but under some conditions they can be interpreted as sharp bounds on the policy transform $I[\varphi](\gamma)$, point-wise in the variable $\gamma$. It is here that we make a contribution in the identification literature by showing that the envelope functions can be expressed as the value functions of optimization problems parameterized by the policy variable $\gamma \in \Gamma$. The result is derived under assumptions found in the theory of error bounds and exact penalty functions from the literature on optimization, and the resulting optimization problems are closely related to \textit{mathematical programs with equilibrium constraints}, or MPECs.\footnote{See \cite{dolgopolik2016unifying} for a survey of exact penalty functions and their connection to error bounds, and see \cite{luo1996mathematical} for a textbook treatment MPECs.} A remarkable benefit of our optimization approach is that we show the bounds on the policy transform can be constructed without the need to first estimate the full identified set for $\theta$, the vector of model parameters. This is in contrast to typical approaches to bounding counterfactual quantities, which first estimate the identified set of structural parameters, and then perform a counterfactual for every possible value of the parameter vector in the identified set. \todo{give references?} A direct implication of our result is that, in either point- or partially-identified models, if the policymaker's counterfactual quantity of interest is the policy transform of some function $\varphi$, then all structural parameters can be treated as nuisance parameters when performing counterfactuals and making policy choices. These results on identification may be of substantial separate interest. 

Finally, throughout the text we discuss a simultaneous discrete choice and a program evaluation example in order to illustrate possible applications of the procedure. The simultaneous discrete choice example includes empirical entry games (e.g. \cite{tamer2003incomplete}) and empirical models of social interactions (e.g. \cite{brock2001discrete}) as special cases, and has become a canonical example of an incomplete model in the literature on partial identification. The second program evaluation example follows the setup in \cite{heckman1999local} and \cite{heckman2005structural}. This example has attracted recent attention in the literature on partial identification (e.g. \cite{mogstad2018using} and \cite{mourifie2020layered}) and is included to show the breadth of our procedure. 

\subsection{Related Literatures}   

This paper builds on results from a variety of different literatures, including recent work on counterfactuals in structural models, partial identification and random set theory, decision theory and optimal policy choice, and computational and statistical learning theory.

Our approach to modelling and counterfactuals in partially identified models extends the literature using random set theory in econometrics, including \cite{beresteanu2011sharp}, \cite{galichon2011set}, \cite{beresteanu2012partial} and \cite{chesher2017generalized}. As mentioned in the introduction, our general environment is similar to descriptions found in \cite{jovanovic1989observable} and more recently in \cite{chesher2017generalized}, which in turn are extensions of the classical foundations for econometric modelling set forth in \cite{koopmans1950measuring} and \cite{hurwicz1950generalization}, among others. The use of random set theory is convenient in order to permit application of the method to a wider class of models. In particular, our framework is applicable to models that may (or may not) be incomplete, which are an important class of models in the literature on partial identification. Incomplete models are now legion, and include entry games with multiple equilibria (\cite{bresnahan1990entry}, \cite{bresnahan1991empirical}, \cite{tamer2003incomplete}, \cite{jia2008happens}, \cite{ciliberto2018market}); english auctions (\cite{haile2003inference}, \cite{chesher2017incomplete}); discrete choice models with endogenous regressors or social interactions (\cite{chesher2012simultaneous}, \cite{chesher2013instrumental}, \cite{chesher2014instrumental}); matching models (\cite{uetake2019entry}); friendship networks (\cite{miyauchi2016structural}); and selection and treatment effect models (\cite{mourifie2018sharp}, \cite{russell2019sharp}).

From the perspective of policy choice, our general approach to the problem of policy decisions is new. However, there is now a large and growing literature on statistical treatment rules in econometrics, including papers by \cite{manski2004statistical}, \cite{hirano2009asymptotics}, \cite{stoye2009minimax}, \cite{stoye2012minimax}, \cite{chamberlain2011bayesian}, \cite{tetenov2012statistical}, \cite{kasy2016partial}, \cite{kitagawa2018should} and \cite{mbakop2019model}. In general these papers can be divided according to (i) whether they are frequentist/bayesian, (ii) whether they take a finite-sample or asymptotic approach, and (iii) whether they consider decision problems under uncertainty or ambiguity (or ``Knightian uncertainty''). In the current paper we take a frequentist, finite-sample approach to decision problems under ambiguity. However, unlike previous papers that belong to the same class, our method of evaluating statistical decision rules differs from the procedure proposed by \cite{wald1950statistical}. In the absence of ambiguity arising from the unknown sub-state $s\in \mathcal{S}$, our procedure is very similar to the PAC framework for inductive inference that has become enormously popular in the computer science literature. This model of learning was initially proposed in a seminal paper by \cite{valiant1984theory}, for which he won the prestigious Turing Award. The name ``probably approximately correct'' seems to have been first used by \cite{angluin1988learning}, who extended the model to the case of noisy data. The PAC model and its extensions have now become the dominant model of learning in the theoretical foundations of machine learning; influential textbook treatments that make this connection explicit include \cite{kearns1994introduction}, \cite{vapnik1995nature}, \cite{vapnik1998statistical}, \cite{vidyasagar2002theory}, \cite{shalev2014understanding} and \cite{mohri2018foundations}. Our particular analysis is most closely related to the decision theoretic generalization of the PAC learning model proposed by \cite{haussler1992decision}, as well as the general learning setting considered in \cite{vapnik1995nature}. Other important papers studying necessary and sufficient conditions for learnability in various machine learning settings include \cite{blumer1989learnability}, \cite{kearns1994efficient}, \cite{bartlett1996fat}, \cite{alon1997scale}, and \cite{shalev2010learnability}, among others. Our work here on providing sufficient conditions for learnability borrows heavily from this literature. However, the additional ambiguity that arises in relation to possible partial identification of the policy transform differentiates our setting from the statistical learning literature, and our incorporation of this notion of ambiguity into the PAC framework appears to be new. Many of our results are applicable to problems involving risk minimization subject to (stochastic) constraints, and thus may be of separate interest to researchers in machine learning.  

Surprisingly, we are unaware of any attempts to formally connect the literature on statistical decision theory with the literature on statistical learning theory.\footnote{\cite{kitagawa2018should} and \cite{mbakop2019model} make some connections with the statistical learning literature. However, their method of evaluating statistical treatment rules is different from that considered by the PAC model. Some discussion on the links with decision theory can be found in an influential paper by \cite{haussler1992decision}, although the discussion is very limited and no connection is made with Wald-style frequentist decision theory. As far as we are aware, this remains an open question.} On the one hand, the properties of a Wald-style analysis are (at this point) better understood; see, for example, \cite{stoye2011statistical} for an axiomatization of Wald's frequentist maximin procedure. On the other hand, we find the PAC style criterion to be much more amenable to informative ex-post analyses of particular decision rules, mostly due to its connection to the concentration of measure phenomenon, and thus its amenability to analysis using concentration inequalities.

The connections to the statistical learning literature permeate our theoretical results. There are connections of our work to the study of ratio-type empirical processes (e.g. \cite{gine2003ratio}, \cite{gine2006concentration}), and to the study of fixed-point equations and rates of convergence in risk minimization problems (e.g. \cite{massart2000some}, \cite{koltchinskii2000rademacher},  \cite{bousquet2002some}, \cite{bartlett2005local}, and \cite{koltchinskii2006local}). Overall our work is most closely related to the work of \cite{koltchinskii2006local}, and the subsequent textbook treatment \cite{koltchinskii2011oracle}. As we will see in the section on the ex-post analysis of certain decision rules, a key component of our approach is the use of Rademacher processes to construct data-dependent bounds on certain important empirical processes. This has the benefit of allowing the policymaker to avoid relying on any specific properties of the underlying function class, which are typically difficult to verify, and thus are applicable whether or not the associated policy space is learnable. Furthermore, the use of data-dependent complexity measures like the empirical Rademacher complexity ensures our finite sample guarantees are less conservative than otherwise. It appears this idea was independently offered by \cite{bartlett2002model} and \cite{koltchinskii2001rademacher}, and was developed further in \cite{koltchinskii2006local}. See also Section 4.2 in \cite{koltchinskii2011oracle}. A review of excess risk bounds and their application to classification problems in statistical learning theory can be found in \cite{boucheron2005theory} and \cite{koltchinskii2011oracle}.


Closely related to the identification component of this paper---which studies the envelope functions for the policy transform---is the work by \cite{ekeland2010optimal}, \cite{schennach2014entropic}, \cite{torgovitsky2019partial} and \cite{li2019general}. The paper of \cite{ekeland2010optimal} is focused on model specification testing, and allows for econometric models with only semiparametric restrictions on the distribution of unobservables in the form of moment conditions.\footnote{The paper of \cite{ekeland2010optimal} is related to a string of other papers by the same authors, namely \cite{galichon2006inference}, \cite{galichon2009test} and \cite{galichon2011set}.} \cite{schennach2014entropic} provides a general framework for models with moment conditions that depend on latent variables, and shows that the latent variables can be integrated out of the moment conditions without loss of information using a least-favourable entropy maximizing distribution. \cite{torgovitsky2019partial} shows that when restrictions on the distribution of the latent variables have a certain structure, sharp identified sets for functionals of partially-identified parameters can be characterized in terms of optimization problems. Finally, \cite{li2019general} shows that sharp identified sets for structural and counterfactual parameters can be constructed using a method that essentially profiles the latent variables out of the moment conditions. In the current paper, we use an idea related to \cite{li2019general} to eliminate unobservables from the counterfactual bounding problem. However, in contrast to \cite{li2019general} our focus on policy transforms means our formulation does not require replacing a finite number of moment conditions with a continuum of moment conditions. Furthermore, our approach does not require the policymaker to compute the full identified set of structural parameters. Our specific characterization of the bounds on the policy transform in terms of two parametric optimization problems was designed to be amenable to the theoretical analysis of policy space learnability, and the analysis of the eME rule and the $\delta-$level sets. Thus, our particular bounding approach is new. Finally, and perhaps most importantly, our focus is primarily on using the bounds to study the problem of policy choice, which is not considered in any of \cite{ekeland2010optimal}, \cite{schennach2014entropic}, \cite{torgovitsky2019partial} or \cite{li2019general}.

The idea that at least some structural parameters may be seen as nuisance parameters in the policy decision problem goes back at least as far as \cite{marshak1953measurements}. \cite{heckman2010building} refers to this idea as ``Marshak's Maxim." At a high level, the identification component of this paper is reminiscent of \cite{ichimura2000direct}, who discuss a method for performing ex-ante policy experiments in the treatment effect literature without estimating the structural parameters, and without specifying the error distribution. More recent examples of counterfactual analysis without first estimating the (identified set for the) structural parameters can be found in \cite{syrgkanis2018inference}, \cite{tebaldi2019nonparametric} and \cite{kalouptsidi2019picounterfactual}.\\~\\

\noindent The remainder of the paper will proceed as follows. Section \ref{section_methodology} introduces the notation and main definitions and assumptions, in addition to describing the decision environment and introducing the motivating examples. Importantly, Section \ref{section_methodology} introduces the policy transform, and defines the notion of learnability of a policy space. As described above, the theoretical results in this paper depend heavily on the nature of the upper and lower envelope functions for the policy transform. Thus, in Section \ref{section_envelope_functions} we define the identified set for the policy transform, and present our main identification result characterizing its upper and lower envelopes. Equipped with this result, Section \ref{sec_policy_analysis} then considers the problem of policy choice, providing sufficient conditions for learnability, and Section \ref{section_ex_post_analysis} provides an ex-post analysis of the performance of particular decision rules. Section \ref{section_conclusion} concludes. All proofs can be found in the Appendices. \\~\\

\noindent \textbf{Notation:} Given a subset $\mathcal{X}$ of a Polish space (a complete and separable metric space), we use $\mathfrak{B}(\mathcal{X})$ to denote the Borel $\sigma-$algebra on $\mathcal{X}$ (note the topology on $\mathcal{X}$ is the topology induced by the metric). We will often either leave the metric implicit, or will denote a generic metric by the function $d:\mathcal{X} \times \mathcal{X} \to \mathbb{R}$. For two measurable spaces $(\mathcal{X},\mathfrak{B}(\mathcal{X}))$ and $(\mathcal{X}',\mathfrak{B}(\mathcal{X}'))$, the product $\sigma-$algebra on $\mathcal{X}\times \mathcal{X}'$ is denoted by $\mathfrak{B}(\mathcal{X})\otimes \mathfrak{B}(\mathcal{X}')$. If $X: (\Omega,\mathfrak{A})\to (\mathcal{X},\mathfrak{B}(\mathcal{X}))$ is a random variable defined on the probability space $(\Omega,\mathfrak{A},P)$, then we use $P_{X}$ to denote the probability measure induced on $\mathcal{X}$ by $X$; that is, for any $A \in \mathfrak{B}(\mathcal{X})$, $P_{X}(A) := P(X^{-1}(A))$. We let $\sigma(X)\subseteq \mathfrak{A}$ denote the smallest sub $\sigma-$algebra making $X$ a measurable function. Furthermore, we interpret $P_{X|X'}(X \in A| X' =x)$ as a regular conditional probability measure. In many cases we do not explicitly differentiate between the true distribution of the random variable $X$, say $P_{X}$, or some other distribution of the random variable $X$, say $P_{X}'$, and instead leave the distinction to be resolved by context. To keep the notation clean, we will omit the transpose when combining column vectors; that is, if $v_{1}$ and $v_{2}$ are two column vectors, rather than write $v=(v_{1}^\top,v_{2}^\top)^\top$ we instead write $v=(v_{1},v_{2})$, where it is understood that $v$ is a column vector unless otherwise specified. Importantly, throughout the paper we use the convention that $\sup \emptyset = - \infty$ and $\inf \emptyset = + \infty$. Finally, we will largely ignore measurability issues in the main text, but we note that such issues are non-trivial in our framework, and are discussed and addressed in Appendix \ref{appendix_measurability}.

\section{Methodology}\label{section_methodology}

\subsection{Preliminaries}
As mentioned in the introduction, the description of the environment follows closely that of \cite{jovanovic1989observable} and \cite{chesher2017generalized}, which in turn are extensions of the classical foundations for econometric modelling set forth in \cite{koopmans1950measuring} and \cite{hurwicz1950generalization}, among others. However, there are some differences that will be pointed out as they occur. We will also make heavy use of random set theory in this paper. Random set theory has played a major role in the development of methods for partially identified models, for example in the contributions of \cite{beresteanu2011sharp}, \cite{galichon2011set}, \cite{beresteanu2012partial} and \cite{chesher2017generalized}, among others. We will also use random set theory in this paper, as it naturally generalizes many features of complete econometric models to incomplete models (see \cite{chesher2017generalized}). Since complete models can be seen as special cases of incomplete models, focusing on incomplete models will allow us to construct a method that applies to a broader class of econometric models. Some important definitions from random set theory---including the notion of Effros-measurability, the definition of a random set, the distribution of a random set, and the notion of a selection from a random set---have been moved to Appendix \ref{appendix_preliminaries} for brevity. The current section will presume some working knowledge of these concepts.

We begin by specifying the restrictions on the factual and counterfactual domains. First we will fix the probability space and define the unobserved random variables and parameters that are common to both domains.\todolt{In the end it seems there is actually no reason to require the underlying probability space to be complete.}  

\begin{assumption}\label{assump_preliminary}
There exists a fixed probability space $(\Omega,\mathfrak{A},P)$, and a random element $U: (\Omega,\mathfrak{A}) \to (\mathcal{U},\mathfrak{B}(\mathcal{U}))$ where $\mathcal{U}$ is a compact second-countable Hausdorff space. In addition, the parameter space $\Theta$ is a Polish space equipped with the $\sigma-$algebra $\mathfrak{B}(\Theta)$. 
\end{assumption}

Fixing the probability space throughout represents a departure from some of the existing literature on partial identification and random set theory in econometrics (e.g. \cite{galichon2011set}, \cite{chesher2017generalized}). Our reason for doing so is mostly conceptual. This paper is concerned with counterfactuals, and counterfactuals naturally involve some comparison of units between factual and counterfactual states. In any probabilistic framework, the underlying probability space naturally specifies the basic unit of observation (e.g. individuals, firms, types, etc.), so that it is necessary for the units of observation to be the same in both the factual and counterfactual states when performing a counterfactual analysis. The point may seem esoteric, but it will have a major impact on the statement and proofs of most of our results while also resolving some interpretative difficulties. \todo{This point could be pressed further in the text, when applicable, and in the conclusion.}

The restriction that $\mathcal{U}$ is a compact space in Assumption \ref{assump_preliminary} may seem overly restrictive; for example, the euclidean space $\mathbb{R}^{d}$ ($d<\infty$) with the usual topology is not a compact space. We might consider relaxing Assumption \ref{assump_preliminary} by allowing $\mathcal{U}$ to be a locally compact second-countable Hausdorff space, of which $\mathbb{R}^{d}$ (with the usual topology) is an example. However, any locally compact Hausdorff space has a one-point compactification; that is, assuming $\mathcal{U}$ is locally compact and Hausdorff, there exists a compact space $\widetilde{\mathcal{U}}$ with $\mathcal{U} \subset \widetilde{\mathcal{U}}$ such that $\widetilde{\mathcal{U}} \setminus \mathcal{U}$ consists of a single point.\footnote{See \cite{munkres2014topology} Theorem 29.1.} Furthermore, $\widetilde{\mathcal{U}}$ is unique up to a homeomorphism.\footnote{Recall a homeomorphism is a continuous invertible function with a continuous inverse.} A related argument has been presented in \cite{schennach2014entropic}. From this perspective, it is difficult to imagine an environment where a policymaker should have strong a priori reasons to model the unobservables using a locally compact Hausdorff space $\mathcal{U}$ versus its one-point compactification $\widetilde{\mathcal{U}}$, despite the fact that this is often done in practice. On the other hand, the theoretical benefits of taking $\mathcal{U}$ to be compact (or to be the one-point compactification of some locally compact Hausdorff space) are numerous. We will highlight these benefits as they arise. 

Note that we will not require that the distribution $U$ belong to a parametric class. This is in keeping with our desire to avoid treating the distribution of $U$ as a model primitive. This perspective is consistent with the idea that the latent variables represent components of the underlying economic system that remain unmodelled, due primarily to the policymaker's ignorance of the process determining $U$, and thus her inability to construct a complete mathematical description of the economic system under investigation. This interpretation becomes especially meaningful given the role the latent variables play in determining counterfactual outcomes. Instead, as we will see, the distribution of $U$ can be implicitly constrained by the remaining primitives of the model.  

Finally we note that equipping the parameter space with the Borel $\sigma-$algebra $\mathfrak{B}(\Theta)$ may seem odd. However, to make policy decisions in our framework will require measurability of certain functions to be introduced later on. Primitive conditions for the required measurability will make use of the measure space $(\Theta,\mathfrak{B}(\Theta))$. We return to similar points throughout the paper, and refer readers to Appendix \ref{appendix_measurability} for our results on measurability.   

We will now summarize the restrictions on the factual and counterfactual domains, beginning with the factual domain.\todolt{In the end it seems there is actually no reason to require the random set $\bm G^{-}$ to be closed (other than for compactness of $\mathcal{P}_{U|Y,Z}(\theta)$).}
\begin{assumption}[Factual Domain]\label{assumption_factual_domain}
The factual domain is represented by random vectors $Y: (\Omega,\mathfrak{A}) \to (\mathcal{Y},\mathfrak{B}(\mathcal{Y}))$ and $Z: (\Omega,\mathfrak{A}) \to (\mathcal{Z},\mathfrak{B}(\mathcal{Z}))$, where $\mathcal{Y}$ and $\mathcal{Z}$ are Polish spaces. There exists a (possibly multi-valued) map $\bm G^{-} : \mathcal{Y} \times\mathcal{Z}\times\Theta \to \mathcal{U}$ which is closed and Effros-measurable, and satisfies:\todolt{does conditionally matter?}
\begin{align}
P\left(U \in \bm G^{-}(Y,Z,\theta_{0})|Y=y,Z=z\right)=1,\label{eq_puyz}
\end{align}
$(y,z)-$a.s. for some $\theta_{0} \in \Theta$. Furthermore,
\begin{align}
\E_{P_{U|Y,Z} \times P_{Y,Z}} [m_{j}(Y,Z,U,\theta_{0})] \leq 0, \qquad j=1,\ldots,J,\label{eq_moments}
\end{align}
for some measurable functions $m_{j}: \mathcal{Y}\times \mathcal{Z}\times \mathcal{U}\times \Theta \to \mathbb{R}$, for $j=1,\ldots,J$, bounded in absolute value for each $\theta \in \Theta$. 
\end{assumption}
The first part of the assumption states that the unobserved random vector is a selection from the random set $\bm G^{-}(Y,Z,\theta_{0})$ (see Appendix \ref{appendix_preliminaries} for the definition of a selection).\footnote{A similar argument to the one presented in Appendix B of \cite{chesher2015characterizations} can be used to show that this characterization of selectionability conditional on $(y,z)$ a.s. is equivalent to using an analogous selectionability criteria for the joint distributions of $(Y,Z,U)$. A similar point will apply later on when we introduce Assumption \ref{assumption_counterfactual_domain}.} Note the assumption requires only that $\bm G^{-}(\cdot,\theta)$ admits a selection when $\theta=\theta_{0}$. The first part of the assumption can thus be interpreted as a support restriction for the vector of unobservables conditional on the observed data. These support restrictions are derived from the policymaker's econometric model, as we will see in the examples ahead. We also note that the random set $\bm G^{-}$ contains the $U-$level sets presented in \cite{chesher2017generalized} as a special case, and thus our framework will be applicable to the \textit{generalized instrumental variable} (GIV) models considered in their work. 

In the second part of the assumption we suppose that the factual domain satisfies the moment inequalities in \eqref{eq_moments}, which are allowed to depend on the unobserved random variable $U$. This differs from moment conditions in the generalized method of moments (GMM), as well as typical definition of moment inequalities (c.f. \cite{chernozhukov2007}). This places our paper in the narrow literature in partial identification that allows for moments to depend on unobserved random variables with a possibly unknown distribution (c.f. \cite{ekeland2010optimal}, \cite{schennach2014entropic}, \cite{torgovitsky2019partial} and \cite{li2019general}). The assumption of boundedness of the moment functions may appear to be restrictive. This assumption might be replaced by the weaker assumption that the moment functions are uniformly integrable with respect to the set of probability measures $P_{U|Y,Z} \times P_{Y,Z}$ satisfying the other components of Assumption \ref{assumption_factual_domain}.\footnote{See for example alternative assumptions given in \cite{ekeland2010optimal} and \cite{li2019general}.} However, regardless of how it is weakened, we contend that boundedness of the moment functions remains the most primitive assumption for our purposes. Finally, the fact that there are only a finite number of moment functions may also be restrictive; for example, this prohibits the use of conditional moment inequalities when the conditioning variable is continuous. Our identification result in Section \ref{section_envelope_functions} can be extended---under a suitable modification of our assumptions---to handle the case of an infinite number of moment inequalities. However, the same statement is not true of the results in Sections \ref{sec_policy_analysis} and \ref{section_ex_post_analysis} on policy decisions, which rely more crucially on the fact that the number of moment conditions is finite. We also note that both the Effros measurability of $\bm G^{-}$ and Borel measurability of each moment function $m_{j}$ with respect to $\mathfrak{B}(\mathcal{Y})\otimes \mathfrak{B}(\mathcal{Z})\otimes\mathfrak{B}(\Theta)$ (rather than only with respect to $\mathfrak{B}(\mathcal{Y})\otimes \mathfrak{B}(\mathcal{Z})$) will be required later on to ensure measurability of certain key classes of functions.  
  
Similar to the factual domain, we must specify restrictions on the counterfactual domain, and when specifying the counterfactual domain we must specify which counterfactuals are under consideration by the policymaker. We index various counterfactuals by an abstract parameter $\gamma$, where a fixed value of $\gamma$ represents a single counterfactual, and different values of $\gamma$ correspond to different counterfactuals. The interpretation of the parameter $\gamma$ that will be used throughout is that it is an abstraction of a policy tool under the control of the policymaker. The parameter $\gamma$ will play an important role in our policy decision procedure presented later in the paper. \todolt{In the end it seems there is actually no reason to require the random set $\bm G^{\star}$ to be closed.}

\begin{assumption}[$\Gamma$-Counterfactual Domains]\label{assumption_counterfactual_domain}
The $\Gamma$-counterfactual domains are represented by a stochastic process $\{Y^\star(\omega,\gamma) : \gamma \in \Gamma\}$ where $(\Gamma,\mathfrak{B}(\Gamma))$ is a measurable space with $\Gamma$ a Polish space, and where $Y_{\gamma}^\star:=Y^\star(\cdot,\gamma)$ is such that $Y^\star: (\Omega\times\Gamma,\mathfrak{A}\otimes \mathfrak{B}(\Gamma)) \to (\mathcal{Y}^\star,\mathfrak{B}(\mathcal{Y}^\star)$ is measurable, with $\mathcal{Y}^\star$ a Polish space. Furthermore, there exists a (possibly multi-valued) map $\bm G^{\star}: \mathcal{Y}\times\mathcal{Z}\times\mathcal{U}\times \Theta \times \Gamma \to \mathcal{Y}^{\star}$ which is closed and Effros measurable, and satisfies: 
\begin{align}
P\left(Y_{\gamma}^\star \in \bm G^\star(Y,Z,U,\theta_{0},\gamma)|Y=y,Z=z,U=u\right)=1,  \label{eq_pygamma}
\end{align}
$(y,z,u)-$a.s. for the same $\theta_{0} \in \Theta$ from Assumption \ref{assumption_factual_domain}, and for all $\gamma \in \Gamma$ . 
\end{assumption}
Compared to the existing literature, Assumption \ref{assumption_counterfactual_domain} appears to be new. It restricts the set of counterfactuals considered in this paper to be those that can be written as modifications of support-like restrictions on the random variables in the model. We contend that this assumption is able to accommodate most counterfactuals of interest in economics, although it rules out, for example, consideration of counterfactuals that modify the distributions of the latent variables. 
Under this assumption we have that $Y_{\gamma}^\star:=Y^\star(\cdot,\gamma)$ is a \textit{selection process} from the \textit{set-valued process} $\bm G^\star(Y,Z,U,\theta_{0},\gamma)$, where $\bm G^\star$ is required to be Effros-measurable with respect to the product $\sigma-$algebra. Again, the measurability requirement with respect to both $\Theta$ and $\Gamma$ may seem odd, but will be required in Section \ref{sec_policy_analysis} and \ref{section_ex_post_analysis} when we consider the question of policy choice. Note that---consistent with the remark following Assumption \ref{assump_preliminary}---the probability space in Assumptions \ref{assumption_factual_domain} and \ref{assumption_counterfactual_domain} are assumed to be the same. 

\begin{remark}[The ``No Back-Tracking'' Principle]\label{remark_no_backtracking}
From a purely mathematical standpoint there is no reason that the moment functions in Assumption \ref{assumption_factual_domain} cannot also be functions of $Y_{\gamma}^\star$ and/or $\gamma \in \Gamma$. However, we omit this extension for interpretive reasons and caution researchers interested in this approach. In particular, if the researcher is not judicious in her formulation of such moment functions, then it is possible to have environments where the counterfactual $\gamma \in \Gamma$ of interest has ``identifying power'' for the structural parameters $\theta \in \Theta$. Such environments are extremely puzzling since, intuitively, in these cases the counterfactual domain $\gamma \in \Gamma$ under consideration contains ``information'' on the values of the structural parameters $\theta \in \Theta$ existing in the factual domain. Environments that avoid such difficulties will be said to satisfy the ``no back-tracking principle.''\footnote{This principle is named in honour of the philosopher David Lewis who argued against similar ``back-tracking counterfactuals'' in \cite{lewis1979counterfactual}.} We will return to this idea at some point in our example on simultaneous discrete choice models.
\end{remark}
\begin{figure}[!t]
\centering
\includegraphics[scale=0.7]{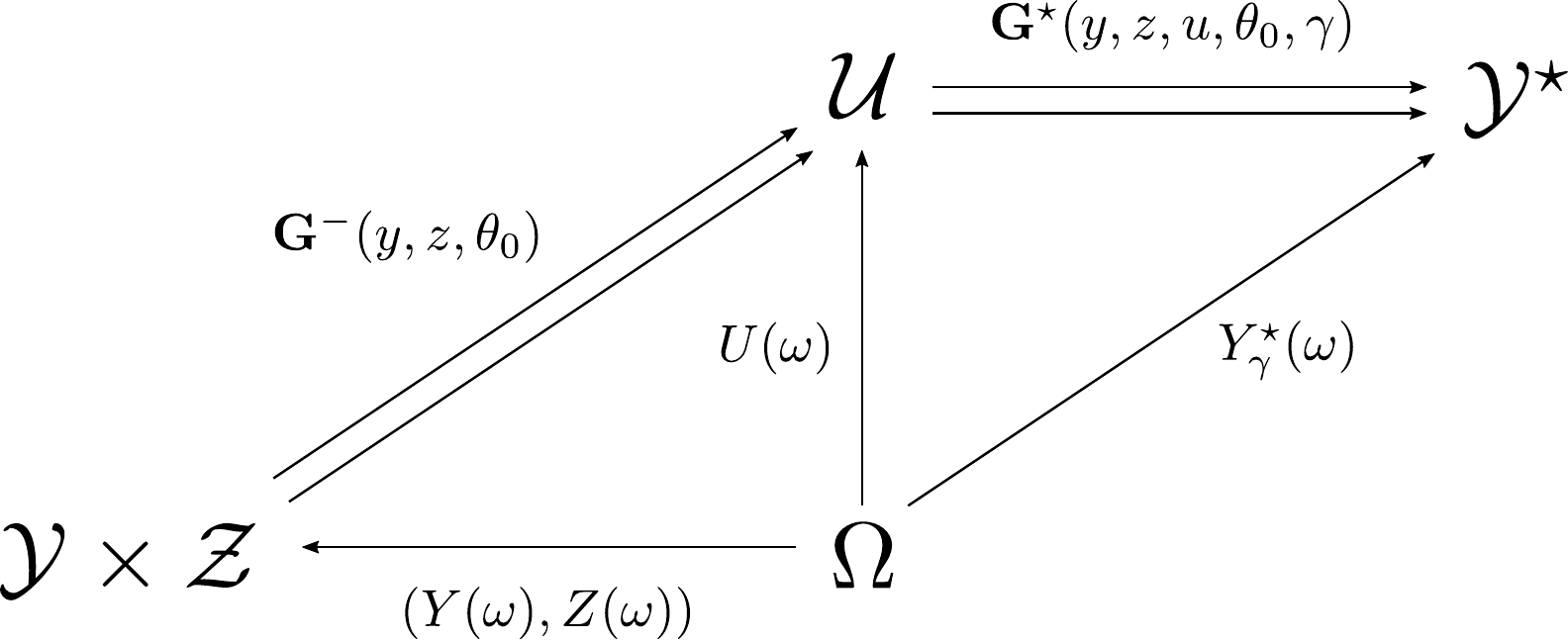}
\caption{Displayed above is an illustration of the setup implied by Assumptions \ref{assump_preliminary}, \ref{assumption_factual_domain} and \ref{assumption_counterfactual_domain}. In particular, note that all random variables are assumed to be defined on the same probability space. Furthermore, note the direction of the arrows from the factual domain $\mathcal{Y} \times \mathcal{Z}$ to the latent $\mathcal{U}$ to the counterfactual domain $\mathcal{Y}^{\star}$, intended to illustrate the process by which information from the factual domain informs on the counterfactual domain. }\label{fig_map_of_spaces}
\end{figure}

The setup implied by Assumptions \ref{assump_preliminary}, \ref{assumption_factual_domain} and \ref{assumption_counterfactual_domain} is illustrated in Figure \ref{fig_map_of_spaces}. Throughout the remainder of the paper, we let $V_{\gamma}:=(Y_{\gamma}^{\star},Y,Z,U)$ denote a random vector with realizations $v \in \mathcal{V}$, where $\mathcal{V}$ is a product space with the product $\sigma-$algebra. 

\subsection{Examples}

We will now turn to two examples to help illustrate the nature of the assumptions just introduced. The examples will be revisited throughout the remainder of the text. The introduction of the examples is lengthy, and readers may skip to Subsection \ref{sec_decision_problem} without loss of continuity. 

The first example we consider is a simultaneous discrete choice model. Simultaneous discrete choice models have seen a wide number of applications, including empirical entry games (e.g. \cite{tamer2003incomplete}), and discrete choice models with social interactions (e.g. \cite{brock2001discrete}). It is already known from the work of \cite{chesher2020structural} that this example falls into the class of GIV models considered by \cite{chesher2017generalized}. For readers familiar with these works, the model will serve as a natural point of comparison. The second example is a program evaluation example which closely mirrors the environment in \cite{heckman2005structural}. The example shows a model where the structural parameter is point-identified, but the counterfactual object of interest is partially-identified.\footnote{In our setting, this is due to the fact the instrument will be assumed to have finite support.} 
\setcounter{example}{0}
\begin{example}[Simultaneous Discrete Choice]\label{example_simultaneous_discrete_choice}
Consider a simultaneous discrete choice problem. In particular, assume that a binary outcome vector $Y:=(Y_{1},\ldots,Y_{K})$ has generic element $Y_{k} \in \mathcal{Y}$ determined by the equation:
\begin{align}
Y_{k}&=\mathbbm{1}\{ \pi_{k} \left(Z_{k},Y_{-k};\theta \right)  \geq U_{k}\}.
\end{align}
Here $Z_{k}$ is a vector of covariates, $U_{k}$ is an unobserved random variable, and $\theta$ is a vector of model parameters. We will define the vector $Z:=(Z_{1},\ldots,Z_{K})$ and $U:=(U_{1},\ldots,U_{K})$ where each variable $Z_{k}$ has support $\mathcal{Z}=\{z_{1},\ldots,z_{L}\}$, a finite subset of euclidean space, and each $U_{k}$ has support $\mathcal{U}= [-1,1]^{d_{u}}$.\footnote{Note that we could instead define $\mathcal{U}:=\overline{\mathbb{R}}^{d_{u}}$, but then:
\begin{align*}
\mathbbm{1}\{ \pi_{k} \left(Z_{k},Y_{-k};\theta \right)  \geq U_{k}\} = \mathbbm{1}\{ \tilde{\pi}_{k} \left(Z_{k},Y_{-k};\theta \right)  \geq \tilde{U}_{k}\},
\end{align*}
where $\tilde{\pi}_{k}\left(Z_{k},Y_{-k};\theta \right) = \text{tanh}\left(\pi_{k} \left(Z_{k},Y_{-k};\theta \right)\right)$ and $\tilde{U}_{k} = \text{tanh}(U_{k})$. In other words, the case with $\mathcal{U}:=\overline{\mathbb{R}}^{d_{u}}$ is homeomorphic to the case $\mathcal{U}:=[-1,1]^{d_{u}}$. } For each $k$, we assume that $\pi_{k}$ is a known measurable function of $(Z_{k},Y_{-k},\theta)$, mapping to $[-1,1]$ that is linear in parameters $\theta$ and has a gradient (with respect to $\theta$) bounded away from zero for each $(z,y_{-k})$. We also assume that $\theta=(\theta_{1},\ldots,\theta_{K})$, and that each $\pi_{k}$ depends only on the subvector $\theta_{k}$.  For simplicity we will assume that the parameter space $\Theta$ is a compact subset of $\mathbb{R}^{d_{\theta}}$, and that $U$ is continuously distributed. \todo{How will you enforce this?} To illustrate the use of semi-parametric restrictions, we will also assume that each coordinate of the vector $U$ is (i) median zero, and (ii) median independent of $(Z_{k},Y_{-k})$. Finally, we assume all random variables are supported on the same probability space $(\Omega,\mathfrak{A},P)$. Verification of Assumption \ref{assump_preliminary} under these conditions is presented in Appendix \ref{appendix_additional_details_sdc_A1_to_A3}. 

For the factual domain, we have the following multifunction:
\begin{align}
\bm G^{-} \left( Y,Z,\theta \right):= \text{cl}\left\{ u \in \mathcal{U} :  Y_{k}=\mathbbm{1}\{ \pi_{k} \left(Z_{k},Y_{-k};\theta \right)  \geq u_{k}\},\,\, k=1,\ldots,K. \right\}.\label{eq_DC_reverse_correspondence}
\end{align}
Note the closure is taken to ensure that $\bm G^{-}(\cdot,\theta)$ is a closed set for each $\theta$. However, this introduces no additional structure and serves merely as a technical simplification, since $\bm G^{-}(\cdot,\theta_{0})$ as defined above will be almost surely equal to the right hand side of \eqref{eq_DC_reverse_correspondence} without taking the closure, which follows from the assumption that $U$ is continuously distributed. To complete the description of the factual domain, we will impose the median zero and median independence assumptions for each coordinate of the vector $U$ as a sequence of moment conditions. In particular, for $k=1\ldots,K$, we will impose the moment conditions:
\begin{align}
\E[\left(\mathbbm{1}\{U_{k} \geq 0 \} - \mathbbm{1}\{ U_{k} \leq 0\}\right)\mathbbm{1}\{Z_{k}=z, Y_{-k}=y_{-k}\}] \leq 0,\quad \forall z \in \mathcal{Z}, \,\, y_{-k}\in \mathcal{Y}^{K-1},\label{eq_sdc_mom1}\\
\E[\left(\mathbbm{1}\{ U_{k} \leq 0\} - \mathbbm{1}\{U_{k} \geq 0 \}\right)\mathbbm{1}\{Z_{k}=z, Y_{-k}=y_{-k}\}] \leq 0,\quad \forall z \in \mathcal{Z}, \,\, y_{-k}\in \mathcal{Y}^{K-1}.\label{eq_sdc_mom2}
\end{align}
Taken together, \eqref{eq_sdc_mom1} and \eqref{eq_sdc_mom2} imply that the latent variables $U_{k}$ are both median zero and median independent of covariates $Z_{k}$ and the outcomes $Y_{-k}$.\footnote{Note that this restriction implies constraints on the joint distribution of the vector $(U_{1},\ldots,U_{K})$. Alternatively, we might instead impose only median independence of $U_{k}$ with $Z_{k}$, which restricts only the marginal distribution of $U_{k}$.} Verification of Assumption \ref{assumption_factual_domain}, including Effros-measurability of the multifunction \eqref{eq_DC_reverse_correspondence}, is provided in Appendix \ref{appendix_additional_details_sdc_A1_to_A3}. 

Turning to the counterfactual domain, there are many possible counterfactuals that may be of interest. For the sake of illustration, we will consider counterfactuals of the following form. Let $\gamma_{k}: \mathcal{Z} \times \mathcal{Y}^{K-1} \to \mathcal{Z} \times \mathcal{Y}^{K-1}$, $\gamma = (\gamma_{k})_{k=1}^{K}$, and $Y_{\gamma}^\star:=(Y_{1,\gamma}^\star,\ldots,Y_{K,\gamma}^\star)$ with typical element:
\begin{align}
Y_{k,\gamma}^\star&=\mathbbm{1}\{ \pi_{k} (\gamma (Z_{k},Y_{-k,\gamma}^\star);\theta )  \geq U_{k}\}.\label{eq_sdc_counterfactual_equations}
\end{align}
For example, our interest may be in the properties of the counterfactual random variable $Y_{k,\gamma}^\star$, such as its mean or its conditional mean. The multifunction for the counterfactual domain is then given by:
\begin{align}
\bm G^\star(Z,U,\theta,\gamma):=\left\{ y^\star \in \mathcal{Y}  : y_{k}^\star=\mathbbm{1}\{ \pi_{k} \left(\gamma(Z_{k},y_{-k}^\star);\theta \right)  \geq U_{k}\},\,\, k=1,\ldots,K. \right\}.\label{eq_DC_counterfactual_correspondence}
\end{align}
Note here we take $\mathcal{Y}^\star = \mathcal{Y}$. Verification of Assumption \ref{assumption_counterfactual_domain}, including Effros-measurability of the multifunction in \eqref{eq_DC_counterfactual_correspondence}, is provided in Appendix \ref{appendix_additional_details_sdc_A1_to_A3}.
\end{example}

\begin{example}[Program Evaluation]\label{example_program_evaluation}
Consider the problem of program evaluation. In this example, a binary variable $D \in \{0,1\}$ indicates participation in the treatment or control group for some program, and the observed real-valued outcome is given by:
\begin{align}
Y=U_{0}(1-D)+U_{1}D,\label{eq_POM}
\end{align}
where $U_{0}$ and $U_{1}$ are \textit{potential outcomes} that are never jointly observed. We will assume throughout that $U_{0}, U_{1} \in \mathcal{U} = [\underline{Y},\overline{Y}]$, and thus we also assume the outcome $Y$ takes values in the bounded interval $\mathcal{Y}:= [ \underline{Y},\overline{Y}]$. In the absence of a selection equation determining the values of $D$, the potential outcome model is \textit{incomplete}. This case is considered in \cite{russell2019sharp}, and the framework in this paper applies to this case as well. Alternatively, we will consider the more popular approach of \cite{heckman1999local} and \cite{heckman2005structural}, and will suppose that the treatment is determined by the equation:
\begin{align}
D = \mathbbm{1}\{g_{0}(Z) \geq U \},\label{eq_selection}
\end{align}
where $U$ is continuous, and $g_{0}(\cdot)$ is an unknown measurable function of the observable covariates $Z \in \mathcal{Z}\subset \mathbb{R}^{d_{z}}$, where $d_{z}$ is the dimension of the vector $Z$. We will assume that $\mathcal{Z}$ is finite, and will allow for the case when the vector $Z$ can be decomposed as $Z=(X,Z_{0})$ with (i) $U\independent Z_{0}|X$ (conditional independence) and (ii) $\E[U_{d}|Z] = \E[U_{d}|X]$ for $d \in\{0,1\}$ (mean independence). We will thus decompose $\mathcal{Z}$ as $\mathcal{Z} = \mathcal{Z}_{0} \times \mathcal{X}$, where $\mathcal{Z}_{0}$ is the support of $Z_{0}$ and $\mathcal{X}$ is the support of $X$. Under these assumptions, it is without loss of generality that $U$ be taken to be uniformly distributed on $[0,1]$ conditional on $Z$. As shown in \cite{vytlacil2002independence}, these assumptions, combined with the additive separability of the selection equation in \eqref{eq_selection}, are equivalent to the assumptions required to estimate the local average treatment effect (LATE) of \cite{imbens1994}. This connects this model with a large body of empirical work that focuses on obtaining estimates of the LATE. 

Set the parameter space as $\Theta= \mathcal{G}\times \mathcal{T}$. Here $\mathcal{G}$ can be taken equal to the space of all positive measurable functions on $\mathcal{Z}$ which is a metric space with the sup norm (for example); under finiteness of $\mathcal{Z}$, this space is Polish. Furthermore, we will take the component $\mathcal{T}$ of the parameter space to be the space of all possible measurable functions on $\mathcal{Z}$. This component of the parameter space will be used in the moment functions below. Finally, we shall denote a generic pair $(g,t) \in \Theta$ as $\theta$. 

We will denote the support of $(U_{0},U_{1},U)$ as $\mathcal{U}:= [\underline{Y},\overline{Y}]^{2} \times [0,1]$. We also assume that the random variables in the vector $(Y,D,Z,U_{0},U_{1},U)$ are all supported on the same probability space $(\Omega,\mathfrak{A},P)$. Under these conditions, Assumption \ref{assump_preliminary} is verified in Appendix \ref{appendix_additional_details_te_A1_to_A3}.

For the factual domain we have the multifunction:
\begin{align}
\bm G^{-} \left( Y,D,Z,\theta \right):= \text{cl}\left\{ (U_{0},U_{1},U) \in \mathcal{U} : \begin{array}{l}
    Y = U_{0}(1-D) + U_{1}D,  \\
    D = \mathbbm{1}\{g(Z) \geq U\}.
  \end{array} \right\}.\label{eq_POM_reverse_correspondence}
\end{align}
Note the closure is taken to ensure that $\bm G^{-}(\cdot,\theta)$ is a closed set for each $\theta$. However, this introduces no additional structure and serves merely as a technical simplification, since $\bm G^{-}(\cdot,\theta)$ as defined above will be almost surely equal to the right hand side of \eqref{eq_POM_reverse_correspondence} without taking the closure, which follows from the assumption that $U$ is continuously distributed. \todo{How will you enforce this?} Close inspection of this multifunction provides some simplification:
\begin{align}
\bm G^{-} \left( Y,D,Z,\theta \right)=\begin{cases}
	\{ Y \} \times [\underline{Y},\overline{Y}] \times [g(Z),1], &\text{ if } D=0,\\
	[\underline{Y},\overline{Y}] \times \{ Y \} \times [0,g(Z)], &\text{ if } D=1.
\end{cases}\label{eq_POM_reverse_correspondence2}
\end{align}
To complete the description of the factual domain, we will impose the independence condition $U \independent Z_{0}|X$ and the mean independence condition $\E[U_{d}|Z] = \E[U_{d}|X]$, for $d \in\{0,1\}$, as a sequence of moment conditions. In particular, since $\mathcal{Z}$ is assumed to be finite, let us partition $\mathcal{Z}$ into the product $\mathcal{Z}=\mathcal{Z}_{0}\times \mathcal{X}$, where $\mathcal{Z}_{0}:=\left\{z_{01}, \ldots, z_{0K} \right\}$ and $\mathcal{X}:= \left\{x_{1},\ldots,x_{L} \right\}$. Now consider the following sequence of moment inequalities:
\begin{align}
\E[\left(D - g(z_{0},x)\right)\mathbbm{1}\{Z_{0}=z_{0},X=x\}] \leq 0,\qquad \forall z_{0} \in \mathcal{Z}_{0}, \, x \in \mathcal{X},\label{eq_pom_mom1}\\
\E[\left(g(z_{0},x) - D\right)\mathbbm{1}\{Z_{0}=z_{0},X=x\}] \leq 0,\qquad \forall z_{0} \in \mathcal{Z}_{0}, \, x \in \mathcal{X},\label{eq_pom_mom2}
\end{align}
and:
\begin{align}
\E[\left(\mathbbm{1}\{U \leq g(z_{0},x)\} - g(z_{0},x)\right)\mathbbm{1}\{X=x\}] \leq 0,\qquad \forall z_{0} \in \mathcal{Z}_{0}, \, x \in \mathcal{X},\, \label{eq_pom_mom3}\\
\E[\left(g(z_{0},x) - \mathbbm{1}\{U \leq g(z_{0},x)\}\right)\mathbbm{1}\{X=x\}] \leq 0, \qquad\forall z_{0} \in \mathcal{Z}_{0},\, x \in \mathcal{X}.\label{eq_pom_mom4}
\end{align}
Together \eqref{eq_pom_mom1} and \eqref{eq_pom_mom2} imply $P(D=1|Z=z) = g(z)$ for all $z \in \mathcal{Z}$, and \eqref{eq_pom_mom3} and \eqref{eq_pom_mom4} imply $P(U\leq g(z)|Z=z) = P(U \leq g(z)|X=x) = g(z)$ for all $z_{0} \in \mathcal{Z}_{0}$ and $x \in \mathcal{X}$. Under finiteness of the support $\mathcal{Z}$, these moment inequalities represent the only observable implications of the independence condition $U \independent Z_{0}|X$. In addition, we will impose the following moment conditions:
\begin{align}
\E\left[t(z_{0},x)-\mathbbm{1}\{Z=z_{0},X=x\}\right] &\leq 0,\qquad\forall z_{0} \in \mathcal{Z}_{0},\, \forall x \in \mathcal{X}, \label{eq_pom_mom5}\\
\E\left[\mathbbm{1}\{Z=z_{0},X=x\}-t(z_{0},x)\right] &\leq 0,\qquad\forall z_{0} \in \mathcal{Z}_{0},\, \forall x \in \mathcal{X}, \label{eq_pom_mom6}
\end{align}
and:
\begin{align}
\E\left[U_{d}\left(\mathbbm{1}\{Z=z_{0},X=x\}\sum_{z_{0} \in \mathcal{Z}_{0}} t(z_{0},x)- \mathbbm{1}\{X=x\}t(z_{0},x) \right)\right] &\leq 0, \,\,\forall z_{0} \in \mathcal{Z}_{0},\, x \in \mathcal{X},\, d\in \{0,1\},\label{eq_pom_mom7}\\
\E\left[U_{d}\left(\mathbbm{1}\{X=x\}t(z_{0},x)-\mathbbm{1}\{Z=z_{0},X=x\}\sum_{z_{0} \in \mathcal{Z}_{0}} t(z_{0},x) \right) \right] &\leq 0, \,\,\forall z_{0} \in \mathcal{Z}_{0},\, x \in \mathcal{X},\, d\in \{0,1\}.\label{eq_pom_mom8}
\end{align}
Together \eqref{eq_pom_mom5} - \eqref{eq_pom_mom8} imply the mean independence condition: $\E[U_{d}|Z] = \E[U_{d}|X]$ for $d \in\{0,1\}$. In particular, \eqref{eq_pom_mom5} and \eqref{eq_pom_mom6} ensure $t(z_{0},x) = P(Z_{0}=z_{0},X=x)$, so that the moment conditions in \eqref{eq_pom_mom7} and \eqref{eq_pom_mom8} imply:
\begin{align*}
\E\left[U_{d}\left(\mathbbm{1}\{Z=z_{0},X=x\}P(X=x)- \mathbbm{1}\{X=x\}P(Z_{0}=z_{0},X=x)\right) \right] &= 0,\,\,\forall z_{0} \in \mathcal{Z}_{0},\, x \in \mathcal{X},\, d\in \{0,1\},
\end{align*}
or equivalently:
\begin{align*}
\E\left[U_{d}\left(\frac{\mathbbm{1}\{Z=z_{0},X=x\}}{P(Z_{0}=z_{0},X=x)}- \frac{\mathbbm{1}\{X=x\}}{P(X=x)}\right) \right] &= 0,\,\,\forall z_{0} \in \mathcal{Z}_{0},\, x \in \mathcal{X},\, d\in \{0,1\}.
\end{align*}
From here, a full verification of Assumption \ref{assumption_factual_domain} for the factual domain, including Effros measurability of the multifunction \eqref{eq_POM_reverse_correspondence2}, is provided in Appendix \ref{appendix_additional_details_te_A1_to_A3}.

With this setup, we might be interested in how the outcome variable changes when the factors $Z$ that determine an individual's treatment decision are modified. For example, let $\Gamma$ denote the set of all measurable functions $\gamma: \mathcal{Z} \to \mathcal{Z}$ (note that there are at most finitely many).\footnote{See \cite{carneiro2011estimating} for a discussion of other possible parameters under this setting.} We can then define:
\begin{align}
Y_{\gamma}^\star=U_{0}(1-D_{\gamma}^\star)+U_{1}D_{\gamma}^\star,\label{eq_POM_star}
\end{align}
where the random variable $D_{\gamma}^\star$ is then given by:
\begin{align*}
D_{\gamma}^{\star} = \mathbbm{1}\{g_{0}(\gamma(Z)) \geq U\}.   
\end{align*}
Note that as in \cite{heckman1999local} and \cite{heckman2005structural}, our counterfactual $\gamma \in \Gamma$ has no direct effect on $(U_{0},U_{1})$.  Our interest is in the properties of the random variable $Y_{\gamma}^\star$, such as its mean or its conditional mean. The multifunction for the counterfactual domain is given by:
\begin{align}
\bm G^\star(Z,U_{0},U_{1},U,\theta,\gamma):=\left\{ (Y_{\gamma}^\star,D_{\gamma}^\star) \in \mathcal{Y} \times \{0,1\} :   \begin{array}{l}
    Y_{\gamma}^\star = U_{0}(1-D_{\gamma}^\star) + U_{1}D_{\gamma}^\star,  \\
    D_{\gamma}^\star = \mathbbm{1}\{g(\gamma(Z)) \geq U\}.
  \end{array} \right\}.\label{eq_POM_counterfactual_correspondence}
\end{align}
Note here we take $\mathcal{Y}^\star = \mathcal{Y}$. Again, close inspection of this multifunction provides some simplification:
\begin{align}
\bm G^\star(Z,U_{0},U_{1},U,\theta,\gamma)=\begin{cases}
	(U_{1},1), &\text{ if } U \leq g(\gamma(Z)),\\
	(U_{0},0), &\text{ if } g(\gamma(Z)) < U.
\end{cases}\label{eq_POM_counterfactual_correspondence2}
\end{align} 
A full verification of Assumption \ref{assumption_counterfactual_domain} for the counterfactual domain, including Effros measurability of the multifunction \eqref{eq_POM_counterfactual_correspondence2}, is provided in Appendix \ref{appendix_additional_details_te_A1_to_A3}.

\end{example}

\subsection{The Policy Transform and Decision Problem}\label{sec_decision_problem}

Throughout the paper we will build on the environment established in the previous section to present a framework for making policy decisions based on the value of any counterfactual object of interest that can be written as an integral of some function of the vector $V_{\gamma}$. In particular, if $\varphi: \Omega \times \Gamma \to \mathbb{R}$ is some measurable function, then we will restrict attention to environments where policymakers are interested in either the \textit{policy transform} or the \textit{conditional policy transform} of $\varphi$, which are defined next. 

\begin{definition}[Policy Transform and Conditional Policy Transform]\label{definition_policy_transform}
Let $\varphi: \Omega \times \Gamma \to \mathbb{R}$ be a bounded and measurable function. The \textit{policy transform} of $\varphi$ is a function $I[\varphi](\gamma): \Gamma \to \mathbb{R}$ given by:
\begin{align}
I[\varphi](\gamma):= \int \varphi(\omega,\gamma) \,dP. \label{eq_integral}
\end{align}
Furthermore, if $\mathfrak{A}' \subset \mathfrak{A}$ is a $\sigma-$algebra, then a conditional policy transform of $\varphi$ given $\mathfrak{A}'$ is a function $\tilde{I}[\varphi]: \Omega \times \Gamma \to \mathbb{R}$ such that (i) $\tilde{I}[\varphi]:\Omega\times\Gamma \to \mathbb{R}$ is $\mathfrak{A}'\otimes\Gamma-$measurable, and (ii) $I[\tilde{I}[\varphi](\cdot,\gamma)\mathbbm{1}_{A}](\gamma) = I[\varphi \mathbbm{1}_{A}](\gamma)$ for every $A \in \mathfrak{A}'$. 
\end{definition}

We will focus on the unconditional policy transform throughout the remainder of the paper, since analogous results hold for the conditional policy transform. In addition, since the relevant random variables in our environment are given in the vector $V_{\gamma}$, we will abuse notation throughout the paper and instead focus on policy transforms of the form:
\begin{align}
I[\varphi](\gamma):= \int_{\Omega} \varphi(V_{\gamma}(\omega)) \,dP = \int_{\mathcal{V}} \varphi(v) \,dP_{V_{\gamma}}, \label{eq_integral2}
\end{align}
which are clearly a special case of the general policy transforms in Definition \ref{definition_policy_transform}.

In the remainder of the paper we take as primitive that the policymaker would like to choose $\gamma$ to maximize the value of the policy transform for some known function $\varphi: \mathcal{V} \to \mathbb{R}$, although all results apply equally to the case where the policymaker wishes to minimize the policy transform.\footnote{After we describe the decision problem, it will be apparent that desire of the policymaker to maximize or minimize the policy transform might be deduced using an axiomatic approach from a preference relation over the space of Borel probability measures on $\mathcal{V}$. We find this idea interesting, but will not pursue it here.} For pedagogical purposes, it is useful to first consider an idealized decision problem. In particular, when (i) the true distribution $P_{Y,Z}$ is known, (ii) the conditional distribution $P_{U|Y,Z}$ is known, and (iii) the counterfactual conditional distribution $P_{Y_{\gamma}^\star|Y,Z,U}$ is known, the policymaker's problem becomes trivial: she can simply compute the policy transform of $\varphi$ and choose the maximizing value of $\gamma$. However, clearly such idealized environments will be rare. Instead, we will consider the more realistic case when the policymaker only has access to an i.i.d. sample of size $n$ from the true distribution $P_{Y,Z}$, and knows only that Assumptions \ref{assump_preliminary}, \ref{assumption_factual_domain}, and \ref{assumption_counterfactual_domain} are satisfied. In such an environment, the policymaker may be unable to compute the policy transform due to (i) lack of perfect knowledge of $P_{Y,Z}$, (ii) lack of knowledge of $P_{U|Y,Z}$ and (iii) lack of knowledge of $P_{Y_{\gamma}^\star|Y,Z,U}$. All three cases can occur when the structural parameters are point- or partially-identified.  

We are now ready to define the decision problem under consideration.

\begin{definition}[The Decision Problem]\label{definition_decision_problem}
The policymaker's decision problem is characterized by:
\begin{enumerate}[label=(\roman*)]
	\item The population, represented by the probability space $(\Omega,\mathfrak{A},P)$.\label{decision_population}
	\item The action (or policy) space, given by $(\Gamma,\mathfrak{B}(\Gamma))$.\label{decision_action_space}
	\item The sample space, given by $(\Psi_{n},\Sigma_{\Psi_{n}}, P_{Y,Z}^{\otimes n})$, where $\Psi_{n} := (\mathcal{Y}\times \mathcal{Z})^{n}$, with typical element $\psi=\{(y_{i},z_{i})\}_{i=1}^{n}$, equipped with the product Borel $\sigma-$algebra $\Sigma_{\Psi_{n}} := (\mathfrak{B}(\mathcal{Y})\otimes\mathfrak{B}(\mathcal{Z}))^{\otimes n}$ and the product measure $P_{Y,Z}^{\otimes n}$. \label{decision_sample_space}
	\item The state space, given by $\mathcal{S}\times \mathcal{P}_{Y,Z}$, where $\mathcal{P}_{Y,Z}$ is the set of all Borel probability measures on $\mathcal{Y}\times \mathcal{Z}$, and $\mathcal{S}$ is the set of all triples $s=(\theta,P_{U|Y,Z},P_{Y_{\gamma}^\star|Y,Z,U})$ such that the pair $(s,P_{Y,Z})$ satisfies:
	\begin{enumerate}
	 	\item $\theta \in \Theta$,
	 	\item $P_{U|Y,Z}(U \in \bm G^{-}\left(Y,Z,\theta\right) | Y=y,Z=z)=1$, $(y,z)$-a.s.,
	 	\item $P_{Y_{\gamma}^\star|Y,Z,U}(Y_{\gamma}^\star \in \bm G^{\star}(Y,Z,U,\theta,\gamma) | Y=y,Z=z,U=u )=1$, $(y,z,u)$-a.s., and
	 	\item the elements $\theta \in \Theta$ and $P_{U|Y,Z}$ satisfy: 
	 	\begin{align}
		\max_{j=1,\ldots,J} \E_{P_{U|Y,Z} \times P_{Y,Z} } [ m_{j}(Y,Z,U,\theta) ] \leq 0.
		\end{align}
	 \end{enumerate}\label{decision_state_space}
\item The feasible statistical decision rules $\mathcal{D}$, with typical element $d$, given by the set of all measurable functions $d: \Psi_{n} \to \Gamma$.\label{decision_rules}\note{Note that the composition of universally measurable functions is universally measurable. Thus, it seems Borel measurability of $d$ is sufficient for the map $\psi\mapsto \inf_{s \in \mathcal{S}} I[\varphi](d(\psi),s)$ to be measurable if $\gamma \mapsto \inf_{s\in \mathcal{S}} I[\varphi](\gamma,s)$ is universally measurable.}
\item The objective function, given by a function $I[\varphi]:\Gamma \times \mathcal{S}\times \mathcal{P}_{Y,Z}\to \mathbb{R}$, called the state-dependent policy transform, which has the expression:
\begin{align}
I[\varphi](\gamma,s):= \int \varphi(v) \, d(P_{Y_{\gamma}^\star|Y,Z,U} \times P_{U|Y,Z}\times P_{Y,Z})\label{eq_welfare}
\end{align}
where $\varphi: \mathcal{V} \to \mathbb{R}$ is a measurable function (where $P_{Y,Z}$ is left implicit when writing $I[\varphi](\gamma,s)$).\todo{Varphi should depend on $\theta$, no?} \label{decision_welfare}
\end{enumerate}

\end{definition}

A few remarks on this definition of our statistical decision problem are in order. In parts \ref{decision_population} and \ref{decision_action_space}, the specification of the population and the action space are somewhat standard, and have been motivated in the previous sections. In part \ref{decision_sample_space}, the sample space is simply taken as the $n-$fold product of the observable space $(\mathcal{Y} \times \mathcal{Z})$. The measure on this space is the $n-$fold product of the true distribution $P_{Y,Z}$, from which we immediately deduce that the sample in $\psi \in \Psi_{n}$ is assumed to be i.i.d. Motivated from the framework in the previous section, part \ref{decision_state_space} indicates that the unobserved state is characterized by a distribution $P_{Y,Z}$ and the triple $(\theta,P_{U|Y,Z},P_{Y_{\gamma}^\star|Y,Z,U})$, where $\mathcal{S}$ corresponds to the set of all such triples that satisfy the model support restrictions and moment conditions introduced in the previous section. In part \ref{decision_rules}, the feasible decision rules $\mathcal{D}$ are characterized by the set of all measurable functions from the sample space $\Psi_{n}$ to the action space $\Gamma$. We will return to this point below.\footnote{Note we might instead allow for randomized decision rules by taking $\mathcal{D}$ to be the set of all measurable functions from $\Psi$ to the set of all \textit{distributions} on $\Gamma$. This is not required for what we have in mind, but is easily accommodated under slightly modified assumptions.} Furthermore, in this paper we will use the terms \textit{policy rules} and \textit{decision rules} interchangeably. Finally, part \ref{decision_welfare} of Definition \ref{definition_decision_problem} introduces the state-dependent policy transform, which is a generalization of the policy transform that allows for it's value to depend on the unknown state from part \ref{decision_state_space}. Evaluated at the true state, the state-dependent policy transform reduces to the policy transform from Definition \ref{definition_policy_transform}. 

Ex-ante (i.e. before observing the sample) each decision rule $d:\Psi_{n} \to \Gamma$ is a random variable. Under some measurability conditions, this implies the state-dependent policy transform $I[\varphi](d(\psi),s)$ is also a random variable. The remaining question is how to use the collection $\{ I[\varphi](d(\psi),s) : (s,P_{Y,Z}) \in \mathcal{S}\times \mathcal{P}_{Y,Z}\}$ to evaluate a given policy rule. It seems self-evident that a policy rule $d \in \mathcal{D}$ should be preferred to a policy rule $d' \in \mathcal{D}$ if for every $P_{Y,Z} \in \mathcal{P}_{Y,Z}$ we have $I[\varphi](d'(\psi),s) \leq I[\varphi](d(\psi),s)$ a.s. for every $s \in \mathcal{S}$; in such a case, $d$ delivers a larger value of the policy transform in every state with probability one, regardless of the distribution $P_{Y,Z}$. Any preference relation over $\mathcal{D}$ that satisfies this condition will be said to \textit{respect weak dominance}.\footnote{We refer to \cite{manski2011actualist} for a similar definition. Also note that our definition implies stochastic dominance of $I[\varphi](d(\psi),s)$ over $I[\varphi](d'(\psi),s)$ for every $(s,P_{Y,Z}) \in \mathcal{S} \times \mathcal{P}_{Y,Z}$. By Strassen's Theorem, our definition will be equivalent to stochastic dominance if we allow for alternative probability spaces for each $(s,P_{Y,Z})$ pair.} However, beyond the requirement that a preference relation respect weak dominance, it is not obvious how a policymaker should (in the prescriptive sense) choose among competing policy options given the decision problem in Definition \ref{definition_decision_problem}.\footnote{This point is raised repeatedly in the work of Charles Manski, and is summarized in \cite{manski2011actualist}.}

Although a particular preference relation is not required in order to find the results in this paper interesting, it will be useful to define our notion of optimality in the policymaker's decision problem. In particular, our results may be especially useful to policymakers that are sympathetic to the following preference relation:
\begin{definition}[PAC Maximin Preference Relation]\label{definition_pac_preference_relation} 
Fix a sample size $n$. For any $\kappa \in (0,1)$ and any $d \in \mathcal{D}$, let $c_{n}(\cdot,\kappa): \mathcal{D} \to \mathbb{R}_{++}$ be the smallest value satisfying:
\begin{align}
  \inf_{P_{Y,Z} \in \mathcal{P}_{Y,Z}} P_{Y,Z}^{\otimes n} \left(\inf_{s \in \mathcal{S}} I[\varphi](d(\psi),s) + c_{n}(d,\kappa) \geq \sup_{\gamma \in \Gamma} \inf_{s \in \mathcal{S}} I[\varphi](\gamma,s)  \right) \geq \kappa.\label{eq_method1}
\end{align}
Then decision rule $d:\Psi_{n} \to \Gamma$ is weakly preferred to (or weakly dominates) decision rule $d':\Psi_{n} \to \Gamma$ at level $\kappa$ and sample size $n$, denoted by $d' \preccurlyeq_{\kappa} d$, if and only if $c_{n}(d,\kappa) \leq c_{n}(d',\kappa)$. The decision rule $d:\Psi_{n} \to \Gamma$ is strictly preferred to (or strictly dominates) decision rule $d':\Psi_{n} \to \Gamma$, denoted by $d' \prec_{\kappa} d$, if and only if $c_{n}(d,\kappa) < c_{n}(d',\kappa)$. A decision rule $d \in \mathcal{D}$ will be called admissible with respect to $\preccurlyeq_{\kappa}$ if there is no decision rule $d' \in \mathcal{D}$ that is strictly preferred to (or strictly dominates) $d$.\note{It seems this will actually be equivalent to Wald's decision problem when $\mathcal{S}=\{s_{0}\}$.}
\end{definition}
This preference relation is named the PAC maximin preference relation given its close connection to the learning framework in the next subsection, which in turn is closely related to the PAC learning model of \cite{valiant1984theory} from computational learning theory. We refer readers to Appendix \ref{appendix_aside_PAC_learnability} where we discuss the notion of PAC learnability from computational learning theory. We will also emphasize the connection further in the next subsection.



For a fixed $\kappa \in (0,1)$, the preference relation from Definition \ref{definition_pac_preference_relation} is a total ordering, meaning any two decision rules $d$ and $d'$ can be compared according to $\preccurlyeq_{\kappa}$. In addition, it has an interpretation in terms of quantiles. In particular, suppose for simplicity that $\mathcal{P}_{Y,Z}$ contains a single distribution $\pi$ and define $Q_{\pi}(\kappa,d)$ as the $\kappa$ quantile (under distribution $\pi$) of the map:
\begin{align}
d \mapsto \sup_{\gamma \in \Gamma} \inf_{s \in \mathcal{S}} I[\varphi](\gamma,s) - \inf_{s \in \mathcal{S}} I[\varphi](d(\psi),s).\label{eq_d_map}
\end{align}
Note that the map in \eqref{eq_d_map} is always positive. Then a decision rule $d \in \mathcal{D}$ will be preferred to a decision rule $d' \in \mathcal{D}$ under $\preccurlyeq_{\kappa}$ if and only if $Q_{\pi}(\kappa,d) \leq Q_{\pi}(\kappa,d')$. Quantile utility maximization has been considered in \cite{manski1988ordinal} and \cite{manski2014quantile}, and axiomatized in \cite{rostek2010quantile}. However, our approach has major differences from these approaches, especially with regards to our treatment of the (sub-)states $s \in \mathcal{S}$.

Providing an axiomatization for the preference relation in Definition \ref{definition_pac_preference_relation} is beyond the scope of this paper. Indeed, there is no reason why a policymaker needs to have the exact preference relation from Definition \ref{definition_pac_preference_relation} in order to find the results in this paper useful or interesting. However, the following result shows that, at a minimum, $\preccurlyeq_{\kappa}$ respects weak dominance, as defined above. 

\begin{proposition}\label{proposition_pac_weak_dominance}
Suppose that Assumptions \ref{assump_preliminary}, \ref{assumption_factual_domain} and \ref{assumption_counterfactual_domain} hold, and that $\varphi: \mathcal{V} \to [\varphi_{\ell b}, \varphi_{ub}] \subseteq \mathbb{R}$ is a bounded and measurable function. Also, suppose that $\gamma \mapsto \inf_{s \in \mathcal{S}} I[\varphi](\gamma,s)$ is (universally) measurable. Let $d,d'\in \mathcal{D}$ be two decision rules, and suppose that for every $P_{Y,Z} \in \mathcal{P}_{Y,Z}$ we have $I[\varphi](d'(\psi),s) \leq I[\varphi](d(\psi),s)$ a.s. for every $s \in \mathcal{S}$. Then for any $\kappa \in (0,1)$ we have $d' \preccurlyeq_{\kappa} d$, where $\preccurlyeq_{\kappa}$ is the preference relation from Definition \ref{definition_pac_preference_relation}; that is, the preference relation $\preccurlyeq_{\kappa}$ respects weak dominance. 
\end{proposition}
\begin{proof}
See Appendix \ref{appendix_proofs}.
\end{proof}

\begin{remark}
Universal measurability is a weaker requirement then Borel measurability, and is defined in Appendix \ref{appendix_measurability}. Also, in Appendix \ref{appendix_measurability} we show that the map $\gamma \mapsto \inf_{s \in \mathcal{S}} I[\varphi](\gamma,s)$ is universally measurable, although the result and proof relies on Assumption \ref{assumption_error_bound} introduced in the next section. Since Assumption \ref{assumption_error_bound} has not yet been introduced at this point, we impose (universal) measurability of $\gamma \mapsto \inf_{s \in \mathcal{S}} I[\varphi](\gamma,s)$ as a separate assumption in this proposition. 
\end{remark}

Our main interest in the preference relation from Definition \ref{definition_pac_preference_relation}---especially versus other preference relations encountered in frequentist decision theory---is its close connection to the PAC learning framework, which allows us to use a rich set of results from statistical learning theory and empirical process theory to study its theoretical properties. Before formally introducing this connection, we will first revisit our examples to illustrate the various definitions presented in Definition \ref{definition_decision_problem}. 

\setcounter{example}{0}
\begin{example}[Simultaneous Discrete Choice (Cont'd)]
For the simultaneous discrete choice example, recall that our interest is in the properties of the counterfactual random variable $Y_{k,\gamma}^\star$, such as its mean or its conditional mean. For the sake of illustration, we will focus on the quantity:
\begin{align}
I[\varphi](\gamma) = \int_{\Omega} \mathbbm{1}\{Y_{k,\gamma}^\star(\omega) = 1 \} \,dP, \label{eq_policy_transform_sdc}
\end{align}
which is a counterfactual choice probability. Note this quantity is the policy transform of the function $\varphi(\omega,\gamma) = \mathbbm{1}\{Y_{k,\gamma}^\star(\omega) = 1 \}$. Without much additional complication, we might instead be interested in the conditional choice probability $\E[\mathbbm{1}\{Y_{k,\gamma}^\star(\omega) = 1 \}|Z]$; it is easily verified that $\tilde{I}[\varphi](\omega,\gamma)= \E[\varphi(\omega,\gamma)|Z](\omega)$, with $\varphi(\omega,\gamma) = \mathbbm{1}\{Y_{k,\gamma}^\star(\omega) = 1 \}$, is a conditional policy transform.\footnote{Indeed, by definition this quantity is measurable with respect to $\sigma(Z)$, and satisfies:
\begin{align}
I[\tilde{I}[\varphi](\cdot,\gamma)\mathbbm{1}_{A}](\gamma) = \int  \E[\varphi(\omega,\gamma)|Z](\omega) \mathbbm{1}_{A}(\omega) \, dP  = \int \mathbbm{1}\{Y_{k,\gamma}^\star(\omega) = 1 \} \mathbbm{1}_{A}(\omega)  \, dP =I[\varphi \mathbbm{1}_{A}](\gamma),\label{eq_cond_policy_sdc}
\end{align}
for every $A \in \sigma(Z)$. } Throughout we will suppose the policymaker is interested in selecting the policy $\gamma \in \Gamma$ that maximizes the quantity \eqref{eq_policy_transform_sdc}. We can now formally define the policymaker's decision problem. The population is given by the probability space $(\Omega,\mathfrak{A},P)$ and the action space is given by $(\Gamma,\mathfrak{B}(\Gamma))$, where $\Gamma$ is the set of all functions $\gamma=(\gamma_{k})_{k=1}^{K}$ with $\gamma_{k}: \mathcal{Z}\times \mathcal{Y}^{K-1} \to \mathcal{Z}\times \mathcal{Y}^{K-1}$ and $\mathfrak{B}(\Gamma)$ can be taken as the power set of $\Gamma$.\footnote{Since $\mathcal{Z}$ and $\mathcal{Y}$ are finite, both $\Gamma$ and $\mathfrak{B}(\Gamma)$ contain at most finitely many elements.} The sample space in this example is given by $\Psi_{n}$, which is all possible realizations of the $n$ vectors $\{(y_{i},z_{i})\}_{i=1}^{n}$. Each state of the world is indexed by a pair $(\theta,P_{U|Y,Z})$ satisfying the support restriction given by \eqref{eq_DC_reverse_correspondence} and the moment conditions \eqref{eq_sdc_mom1} and \eqref{eq_sdc_mom2}. The state dependent policy transform is given by:
\begin{align*}
I[\varphi](\gamma,s):= \int \mathbbm{1}\{U_{k} \leq \pi_{k}(\gamma(Z_{k},Y_{-k});\theta) \} \,dP_{U|Y,Z}dP_{Y,Z}.
\end{align*}
A feasible statistical decision rule is then any measurable function $d:\Psi_{n} \to \Gamma$ that selects a policy indexed by $\gamma$ given access to an $n-$sample from $\Psi_{n}$. 
\end{example}

\begin{example}[Program Evaluation (Cont'd)]
For the program evaluation example, recall that our interest is in the properties of the random variable $Y_{\gamma}^\star$, such as its mean or its conditional mean. For the sake of illustration, we will focus on the average outcome under some counterfactual policy $\gamma \in \Gamma$, given by $\E[Y_{\gamma}^\star]$. Note that taking $\varphi(\omega,\gamma) = Y_{\gamma}^\star(\omega) \left(:= Y^\star(\omega,\gamma)\right)$, it is then clear that $\E[Y_{\gamma}^\star] = I[\varphi](\gamma)$, so that the average effect of a counterfactual policy is the policy transform of the random variable $Y_{\gamma}^\star(\omega)$. Without much additional complication, we might instead be interested in the conditional average effect $\E[Y_{\gamma}^\star|X]$. It is easily verified that $\tilde{I}[\varphi](\omega,\gamma)= \E[\varphi(\omega,\gamma)|X](\omega)$, with $\varphi(\omega,\gamma) = Y_{\gamma}^\star(\omega)$, is a conditional policy transform.\footnote{Indeed, by definition this quantity is measurable with respect to $\sigma(X)$, and satisfies:
\begin{align}
I[\tilde{I}[\varphi](\cdot,\gamma)\mathbbm{1}_{A}](\gamma) = \int  \E[\varphi(\omega,\gamma)|X](\omega) \mathbbm{1}_{A}(\omega) \, dP  = \int Y_{\gamma}^\star(\omega) \mathbbm{1}_{A}(\omega)  \, dP =I[\varphi \mathbbm{1}_{A}](\gamma),\label{eq_cond_policy_te}
\end{align}
for every $A \in \sigma(X)$. } We will assume throughout that the policymaker is interested in maximizing the value of $\E[Y_{\gamma}^\star]$. We can now formally define the policymaker's decision problem. The population is given by the probability space $(\Omega,\mathfrak{A},P)$ and the action space is given by $(\Gamma,\mathfrak{B}(\Gamma))$, where $\Gamma$ is the set of all functions $\gamma: \mathcal{Z} \to \mathcal{Z}$ and $\mathfrak{B}(\Gamma)$ is the power set of $\Gamma$.\footnote{Since $\mathcal{Z}$ is finite, both $\Gamma$ and $\mathfrak{B}(\Gamma)$ contain at most finitely many elements.} The sample space is given by $\Psi_{n} = (\mathcal{Y} \times \{0,1\} \times \mathcal{Z})^{n}$ with a typical element $\psi= ((y_{i},d_{i},z_{i}))_{i=1}^{n}$. The state space $\mathcal{S}$ is given by $s=(\theta,P_{U_{0},U_{1},U|Y,Z}, P_{Y_{\gamma}^{\star}|U_{0},U_{1},U,Y,Z})$, where $P_{U_{0},U_{1},U|Y,Z}$ and $P_{Y_{\gamma}^{\star}|U_{0},U_{1},U,Y,Z}$ are any random variables that satisfy the support restriction \eqref{eq_POM_reverse_correspondence} and moment conditions \eqref{eq_pom_mom1} - \eqref{eq_pom_mom6}. Finally, a feasible statistical decision rule is any measurable function $d:\Psi_{n} \to \Gamma$ that selects a policy indexed by $\gamma$ given access to an $n-$sample from $\Psi_{n}$. 
\end{example}

\subsection{A Roadmap to the Theoretical Results: Ex-ante and Ex-post Analyses}

With the policymaker's decision problem defined in the previous subsection, our upcoming theoretical results can be divided according to whether they are applicable ex-ante (i.e. before observing the sample) or ex-post (i.e. after observing the sample). 

Recall the preference relation from Definition \ref{definition_pac_preference_relation}. Under this preference relation, the ``performance'' or ``quality'' of a decision rule $d$ can be measured using the value $c_{n}(d,\kappa)$. Thus, the value of $c_{n}(d,\kappa)$ will be a major focus of both the ex-ante and ex-post theoretical analyses in the remainder of the paper. Our main focus in the ex-ante theoretical results is establishing sufficient conditions for learnability of a policy space, which we will discuss further in this subsection. Our main focus for the ex-post theoretical analysis is in establishing bounds on the value of $c_{n}(d,\kappa)$ for certain decision rules, as well as bounds on the set of decision rules $d \in \mathcal{D}$ that obtain a small value of $c_{n}(d,\kappa)$.

\subsubsection{Policy Space Learnability}

To understand the ex-ante theoretical analysis, we must formally introduce the concept of policy space learnability, named because of its connection to notions of learnability from computational learning theory. Intuitively, a policy space $\Gamma$ will be learnable if, for some decision rule $d \in \mathcal{D}$, the value $c_{n}(d,\kappa)$ from Definition \ref{definition_pac_preference_relation} can be made arbitrarily small as $n$ increases. This concept will be made precise in this subsection.  

A review of concepts of learnability from computational learning theory is provided in Appendix \ref{appendix_aside_PAC_learnability}. We argue that, under the preference relation from Definition \ref{definition_pac_preference_relation}, the conceptual differences between the problem of policy choice and the problem of selecting an optimal classifier in a statistical learning setting are smaller than they may initially appear. In both settings we wish to select a decision rule based on a finite sample that will perform well, based on similar criteria, in samples yet unseen. The essential difference between the environments is that the performance of a counterfactual policy is unobservable, even for the sample in hand. Of course this is not an issue if the policymaker has an econometric model that can used to determine the counterfactual outcomes of the policy experiment. 

The general model from the previous subsections will serve exactly this purpose. Given the preference relation from Definition \ref{definition_pac_preference_relation}, the policymaker is presented with a decision problem that is remarkably similar to a learning problem, which is apparent when the following definition is compared with the definition of PAC Learnability from Appendix \ref{appendix_aside_PAC_learnability}.

\begin{definition}[PAMPAC Learnability]\label{definition_pampac_learnability}
Under Assumptions \ref{assump_preliminary}, \ref{assumption_factual_domain}, and \ref{assumption_counterfactual_domain}, a policy space $\Gamma$ is policy agnostic maximin PAC-learnable (PAMPAC) with respect to the policy transform of $\varphi:\mathcal{V} \to \mathbb{R}$ if there exists a function $\zeta_{\Gamma}: \mathbb{R}_{++}\times (0,1) \to \mathbb{N}$ such that, for any $(c,\kappa) \in  \mathbb{R}_{++}\times (0,1)$ and any distribution $P_{Y,Z}$ over $\mathcal{Y}\times \mathcal{Z}$, if $n \geq \zeta_{\Gamma}(c,\kappa)$ then there is some decision procedure $d: \Psi_{n} \to \Gamma$ satisfying:
\begin{align}
\inf_{P_{Y,Z} \in \mathcal{P}_{Y,Z}} P_{Y,Z}^{\otimes n} \left(  \inf_{s\in \mathcal{S}} I[\varphi](d(\psi),s) + c \geq \sup_{\gamma \in \Gamma}\inf_{s \in \mathcal{S}} I[\varphi](\gamma,s)  \right) \geq \kappa. \label{eq_PAC_crit}
\end{align}
\end{definition}

That is, a policy space is PAMPAC learnable if there is exists some decision rule $d:\Psi_{n} \to \mathbb{R}$ that, in the worst-case (sub-)state $s \in \mathcal{S}$, closely approximates the value:
\begin{align*}
\sup_{\gamma \in \Gamma}\inf_{s \in \mathcal{S}} I[\varphi](\gamma,s),
\end{align*}
with high probability for a sufficiently large (but finite) sample.\footnote{A nearly identical definition can be given for policy agnostic minimax PAC-learnability, with the exception that the decision procedure $d: \Psi_{n} \to \Gamma$ must satisfy:
\begin{align}
\inf_{P_{Y,Z} \in \mathcal{P}_{Y,Z}} P_{Y,Z}^{\otimes n} \left( \sup_{s\in \mathcal{S}} I[\varphi](d(\psi),s) - c \leq \inf_{\gamma \in \Gamma}\sup_{s \in \mathcal{S}} I[\varphi](\gamma,s) \right) \geq \kappa. 
\end{align}
} In terms of the preference relation from Definition \ref{definition_pac_preference_relation}, PAMPAC learnability implies that, as the sample size grows, every point in $(c,\kappa)-$space must eventually (i.e. for large enough $n$) lie above the function $c_{n}(d,\cdot):(0,1)\to \mathbb{R}_{++}$ for some decision rule $d$. This idea is illustrated in Figure \ref{fig_pampac_learnability}. Framed in this manner, we see that PAMPAC learnability is not required to determine the admissible decision rules or to make a policy choice. However, there may be substantial ex-ante limitations on the theoretical performance of any given decision rule in environments that are not PAMPAC learnable, making it an important object of theoretical analysis.
\begin{figure}[!t]
\centering
\includegraphics[scale=0.65]{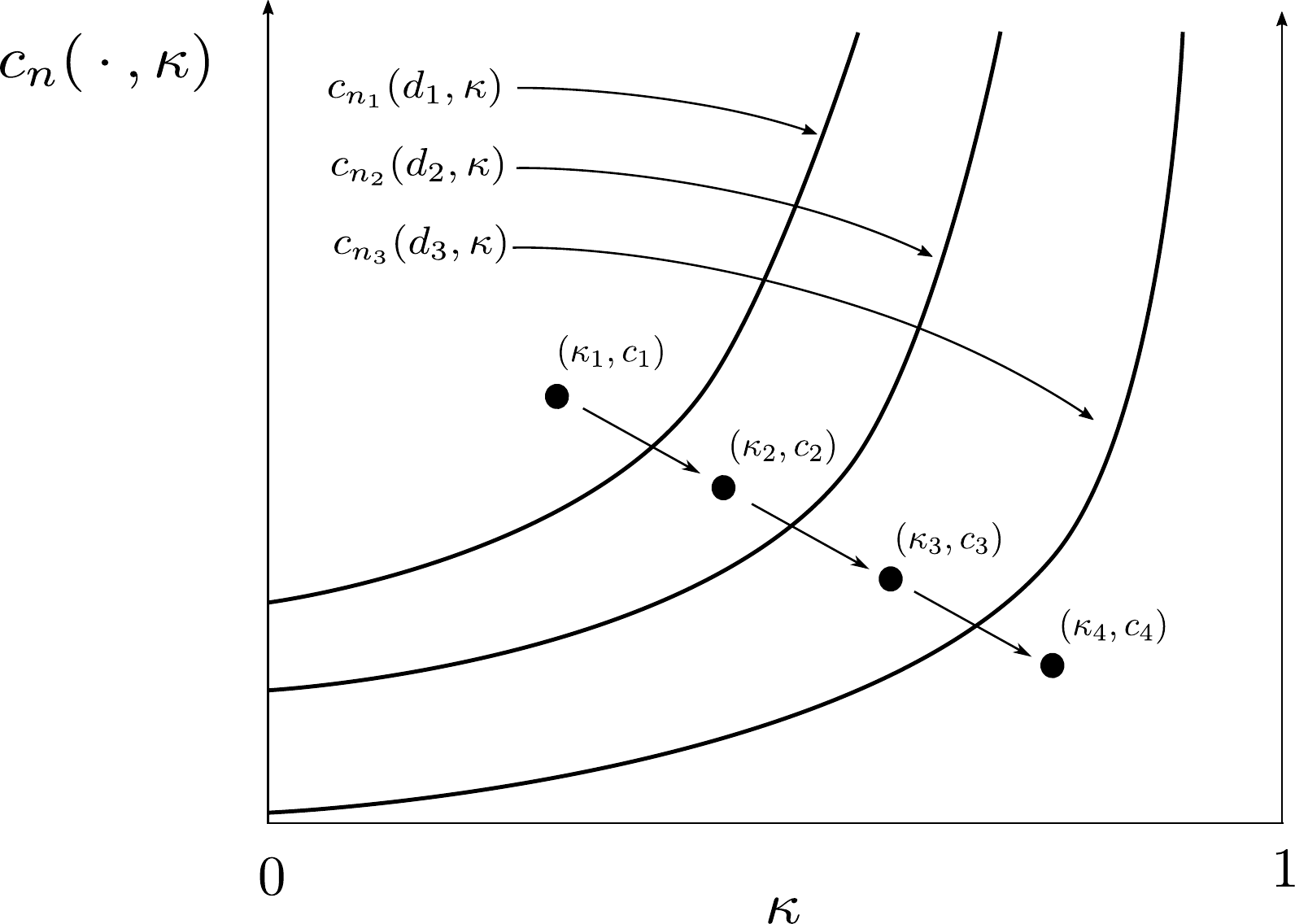}
\caption{This figure illustrates the idea of PAMPAC learnability from Definition \ref{definition_pampac_learnability}. Given a pair $(c,\kappa)$, PAMPAC learnability guarantees that there is some finite $n$ and some decision rule $d:\Psi_{n} \to \Gamma$ such that the graph of $c_{n}(d,\kappa)$ lies entirely below the point $(c,\kappa)$. For example, for $(c_{1},\kappa_{1})$ in the figure, there exists a sample size $n_{1}$ and a decision rule $d_{1}$ such that \eqref{eq_PAC_crit} is satisfied. Note that \eqref{eq_PAC_crit} is also satisfied for the points $(c_{2},\kappa_{2})$ and $(c_{3},\kappa_{3})$ at $n_{2}$ and $d_{2}$, and $n_{3}$ and $d_{3}$, respectively. To verify PAMPAC learnability, the same must hold for all points $(c,\kappa)$; in particular, in the figure we would need to find a sample size $n_{4}$ and decision rule $d_{4}$ such that the graph of $c_{n_{4}}(d_{4},\kappa)$ lies entirely below the point $(c_{4},\kappa_{4})$.}\label{fig_pampac_learnability}
\end{figure}

Despite the fact that PAMPAC learnability may appear to be a weak notion, there are trivial environments where a policy space $\Gamma$ may not be PAMPAC learnable. 

\setcounter{example}{0}
\begin{example}[Simultaneous Discrete Choice (cont'd)]
Consider the general setup of Example \ref{example_simultaneous_discrete_choice}. Suppose for simplicity that $K=1$, and consider the following modifications. Let $\mathcal{Z}=[-1,1]$ and $\Theta=[-1,1]$ and let $\pi_{k}(Z_{k},Y_{-k};\theta)=\pi_{k}(Z_{k};\theta) = \sin(Z_{k}/\theta)$. Then $Y_{k}$ is determined by the equation:
\begin{align*}
Y_{k} = \mathbbm{1}\{\sin(Z_{k}/\theta) \geq U_{k} \}.
\end{align*}
Now consider a policy space $\Gamma$ that consists of all functions $\gamma: \mathcal{Z} \to \mathcal{Z}$, and suppose we are interested in the policy transform:
\begin{align*}
I[\varphi](\gamma):= \int_{\Omega} \varphi(\omega,\gamma)\,dP= \int_{\Omega} \mathbbm{1}\{ Y_{k,\gamma}^\star(\omega)=1\} \,dP,
\end{align*}
where $\varphi(\omega,\gamma)= \mathbbm{1}\{ Y_{k,\gamma}^\star(\omega)=1\}$ and:
\begin{align*}
Y_{k,\gamma}^\star = \mathbbm{1}\{\sin(\gamma(Z_{k})/\theta) \geq U_{k} \}.
\end{align*}
In this case, we claim the policy space $\Gamma$ may not be PAMPAC learnable with respect to the policy transform of $\varphi$. \todo{Prove rigorously}
\end{example}
It is important to realize that the possible failure of PAMPAC learnability does not hinge on the choice of the sine function in this example, which is used for illustrative purposes only. Indeed, the following example shows that the idea is more general.

\setcounter{example}{1}

\begin{example}[Program Evaluation (cont'd)]
Consider the general setup of Example \ref{example_program_evaluation}, with the following modifications. Let $\mathcal{Z}=[-1,1]$ and let $\Theta$ denote the space of continuous functions with values in $[-1,1]$. Otherwise, keep all other aspects of the factual domain the same. Now consider a policy space $\Gamma$ that consists of all continuous functions $\gamma: \mathcal{Z} \to \mathcal{Z}$. Suppose we are still interested in the policy transform of $\varphi(\omega,\gamma) = Y_{\gamma}^\star(\omega)$, where:
\begin{align}
Y_{\gamma}^\star=U_{0}(1-D_{\gamma}^\star)+U_{1}D_{\gamma}^\star,
\end{align}
and where the random variable $D_{\gamma}^\star$ is given by:
\begin{align*}
D_{\gamma}^{\star} = \mathbbm{1}\{\theta_{0}(\gamma(Z)) \geq U\}.   
\end{align*}
In this case, we claim the policy space $\Gamma$ may not be PAMPAC learnable with respect to the policy transform of $\varphi$. \todo{Prove rigorously}
\end{example}
These examples illustrate that there may be limits to which policy spaces are learnable. In the first example, learnability may fail because the structural function determining the counterfactual values of $Y_{k,\gamma}^\star$ is too ``complex," and so cannot be adequately approximated (or ``learned'') with any finite amount of data. A similar explanation applies to the second example, in particular to the structural function determining the values of $D_{\gamma}^\star$. In the next sections we will explore sufficient conditions for the learnability of a policy space that are precisely related to constraints on the complexity of certain function spaces. After establishing a particular policy space is learnable, which is an ex-ante (i.e. before observing the sample) notion, we will then discuss how to evaluate particular decision rules, which is an ex-post (i.e. after observing the sample) notion. Both components will be relevant to the theoretical evaluation of the decision problem.

\subsection{The Path Forward}

As is suggested by \eqref{eq_PAC_crit} in Definition \ref{definition_pampac_learnability}, and as was discussed in the introduction, in order to determine whether a given policy space $\Gamma$ is PAMPAC learnable it is useful to first provide a characterization of the envelope functions:
\begin{align*}
I_{\ell b}[\varphi](\gamma):=\inf_{s \in \mathcal{S}} I[\varphi](\gamma,s), &&I_{u b}[\varphi](\gamma):=\sup_{s \in \mathcal{S}} I[\varphi](\gamma,s).
\end{align*}
Note that, at the true distribution $P_{Y,Z}$, the function $I_{\ell b}[\varphi](\gamma)$ serves as a lower bound on the policy transform $I[\varphi](\gamma)$. Similarly, the function $I_{u b}[\varphi](\gamma)$ serves as an upper bound. Recall that this idea was illustrated in Figure \ref{fig_policy_choice} in the introduction. 

In the case of PAMPAC learnability, if a tractable characterization of the lower envelope function $I_{\ell b}[\varphi](\gamma)$ can be provided under some conditions, then determining whether a policy space is PAMPAC learnable reduces to the problem of finding a decision rule $d: \Psi_{n} \to \Gamma$ that satisfies:
\begin{align}
\inf_{P_{Y,Z} \in \mathcal{P}_{Y,Z}} P_{Y,Z}^{\otimes n} \left( \sup_{\gamma \in \Gamma} I_{\ell b}[\varphi](\gamma)  - I_{\ell b}[\varphi](d(\psi)) \leq c \right) \geq \kappa, 
\end{align}
for large enough (but finite) $n$. Thus in the next section we focus on obtaining a tractable characterization of the envelope functions before returning to the problem of policy choice in Section \ref{sec_policy_analysis}. Once a tractable characterization of the lower (or upper) envelope function is provided, we will then present sufficient conditions for PAMPAC learnability. In addition to its importance to our ex-ante analysis, we will see that a tractable characterization of the envelope functions will also be key to our ex-post analysis of the policymaker's decision problem in Section \ref{section_ex_post_analysis}.

\section{Envelope Functions for the Policy Transform}\label{section_envelope_functions}

\subsection{Preliminaries}

In this section we derive a useful characterization of the envelope functions $I_{\ell b}[\varphi](\gamma)$ and $I_{u b}[\varphi](\gamma)$ defined in the previous section. We will show that these envelope functions can be written as the value functions of optimization problems parameterized by $\gamma \in \Gamma$. Our specific characterization will be important when deriving our learnability results, as well as for our ex-post finite-sample analysis in the next sections. However, for those interested in partial identification, the results in this section may be of substantial separate interest.

We first define the identified set for the structural parameters and policy transform before presenting our main result for this section. In general, these identified sets must be defined \textit{relative} to a distribution $P_{Y,Z}$.\footnote{See, for example, Definition 3 in \cite{chesher2017generalized} and the surrounding discussion.} For notational simplicity this is kept implicit throughout this section. 

We now begin by introducing some additional notation. For assistance with some of the notation in the next definition, the reader is referred to Appendix \ref{appendix_preliminaries}, which discusses the notion of selectionability from a random set. 
\begin{definition}[Distributions of Selections]\label{definition_collections}
The collection $\mathcal{P}_{U|Y,Z}(\theta)$ contains all regular conditional probability measures $P_{U|Y,Z}$ such that each $P_{U|Y,Z} \in \mathcal{P}_{U|Y,Z}(\theta)$ is the distribution of some selection $U \in Sel(\bm G^{-}(\cdot,\theta))$; that is:\footnote{Clearly the collection $\mathcal{P}_{U|Y,Z}(\theta)$ also depends on $P_{Y,Z}$, although we suppress this dependence for notational simplicity throughout.}
\begin{align}
\mathcal{P}_{U|Y,Z}(\theta) := \left\{P_{U|Y,Z} : U \sim P_{U|Y,Z} \text{ for some } U \in Sel(\bm G^{-}(\cdot,\theta))\right\}.\label{eq_definition_p_u_yz}
\end{align}
Furthermore, the collection $\mathcal{P}_{Y_{\gamma}^\star|Y,Z,U}(\theta,\gamma)$ contains all regular conditional probability measures $P_{Y_{\gamma}^\star|Y,Z,U}$ such that each $P_{Y_{\gamma}^\star|Y,Z,U} \in \mathcal{P}_{Y_{\gamma}^\star|Y,Z,U}(\theta,\gamma)$ is the distribution of some selection $Y_{\gamma}^\star \in Sel(\bm G^{\star}(\cdot,\theta,\gamma))$; that is:\footnote{Clearly the collection $\mathcal{P}_{Y_{\gamma}^\star|Y,Z,U}(\theta,\gamma)$ also depends on $P_{Y,Z,U}$, although we suppress this dependence for notational simplicity throughout.}
\begin{align}
\mathcal{P}_{Y_{\gamma}^\star|Y,Z,U}(\theta,\gamma) := \left\{P_{Y_{\gamma}^\star|Y,Z,U}(\theta,\gamma) : Y_{\gamma}^\star \sim P_{Y_{\gamma}^\star|Y,Z,U}(\theta,\gamma) \text{ for some } Y_{\gamma}^\star \in Sel(\bm G^{\star}(\cdot,\theta,\gamma))\right\}.\label{eq_definition_p_ystar_uyz}
\end{align}
\end{definition}

We will see shortly that compactness of $\mathcal{U}$ from Assumption \ref{assump_preliminary} is quite convenient. Indeed, note that under compactness of $\mathcal{U}$, the collection $\mathcal{P}_{U|Y,Z}(\theta)$ is uniformly tight for any $\theta$. If $\mathcal{P}_{U|Y,Z}(\theta)$ is also closed in the $\text{weak}^{*}$ topology, then the collection $\mathcal{P}_{U|Y,Z}(\theta)$ is compact in the $\text{weak}^{*}$ topology, which allows for a simplification of the statement and proofs of many of the results. However, by the fact that $\bm G^{-}$ is closed, this latter result follows directly from the fact that every selection $U \in Sel \left( \bm G^{-}(\cdot,\theta) \right)$ is supported by a compact set.\footnote{See \cite{corbae2009introduction} Theorem 9.9.2 on p. 575, as well as the surrounding discussion.} Thus, throughout our exposition we can use the fact that $\mathcal{P}_{U|Y,Z}(\theta)$ is compact in the $\text{weak}^{*}$ topology.\note{This seems to be the only place where the fact that $\bm G^{-}$ is closed is useful.} 

Beyond the simplifications that come with this result, it also solves a meaningful issue related to selections from identically distributed random sets. Indeed, two identically distributed random sets may have different sets of measurable selections, although the $\text{weak}^{*}$ closure of their measurable selections will always coincide.\footnote{See \cite{molchanov2017theory} Theorem 1.4.3 on p. 79.} The issue is thus entirely resolved by compactness of $\mathcal{U}$, which ensures the collection $\mathcal{P}_{U|Y,Z}(\theta)$ is closed in the $\text{weak}^{*}$ topology; in other words, under Assumptions \ref{assump_preliminary} and \ref{assumption_factual_domain}, this means two identically distributed random sets $\bm G^{-}(Y,Z,\theta)$ and $\bm G^{-}(Y',Z',\theta)$ (see Definition \ref{definition_identically_distributed_random_sets} in Appendix \ref{appendix_preliminaries}) will have the same set of measurable selections.

With the additional notation afforded by Definition \ref{definition_collections}, we now have the following definition of the identified set of structural parameters:
\begin{definition}[Identified Set of Structural Parameters]\label{definition_identified_set}
Under Assumptions \ref{assump_preliminary} and \ref{assumption_factual_domain}, the identified set $\Theta^{*}$ of structural parameters (with respect to the distribution $P_{Y,Z}$) is given by:
\begin{align}
\Theta^{*}:= \left\{\theta \in \Theta :\inf_{P_{U|Y,Z} \in \mathcal{P}_{U|Y,Z}(\theta)} \max_{j=1,\ldots,J} \,\,\E_{P_{U|Y,Z}\times P_{Y,Z}} \left[ m_{j}(y,z,u,\theta) \right] \leq 0 \right\}.
\end{align}
\end{definition}
Compactness of $\mathcal{P}_{U|Y,Z}(\theta)$ in the $\text{weak}^{*}$ topology, combined with boundedness of the moment conditions, ensures that the infimum in the definition of $\Theta^{*}$ is obtained.\footnote{This follows from the extreme value theorem after noting the map $P_{U|Y,Z} \mapsto \E_{P_{U|Y,Z}\times P_{Y,Z}} \left[ m_{j}(y,z,u,\theta) \right]$ is continuous when the moment function $m_{j}$ is uniformly bounded.} Although our focus in this paper is not on the identified set of structural parameters, this definition will be helpful when providing a definition of the identified set for the policy transform, as well as in the proofs. 

To state the definition of the identified set for the policy transform, it will be useful for us to first define the following function:
\begin{align}
&I^{*}[\varphi](\theta,\gamma,I, P_{Y_{\gamma}^\star|Y,Z,U}, P_{U|Y,Z})\nonumber\\
&\qquad\qquad:= \max\left\{ \,\,\left|\E_{P_{Y_{\gamma}^\star|Y,Z,U} \times P_{U|Y,Z}\times P_{Y,Z}} \left[ \varphi(V_{\gamma})-I \right] \right|,\max_{j=1,\ldots,J} \,\,\E_{P_{U|Y,Z}\times P_{Y,Z}} \left[ m_{j}(y,z,u,\theta) \right]\right\}.\label{eq_g_function}
\end{align}
Intuitively, this function is less than zero if and only if (i) all moment conditions are satisfied at the distribution $P_{Y,Z}$ and the pair $(\theta,P_{U|Y,Z})$, and (ii) if the point ``$I$'' is the resulting value of the policy transform for the inputs $(\theta,\gamma, P_{Y_{\gamma}^\star|Y,Z,U}, P_{U|Y,Z})$. As such, it represents all the conditions necessary for the point ``$I$'' to be included in the identified set for the policy transform. We now have the following definition:
\begin{definition}[Identified Set for Policy Transforms]\label{definition_identified_set_counterfactual}
Under Assumptions \ref{assump_preliminary}, \ref{assumption_factual_domain}, and \ref{assumption_counterfactual_domain}, for any $\gamma \in \Gamma$ the identified set for $I[\varphi](\gamma)$ (with respect to the distribution $P_{Y,Z}$) is given by:
\begin{align}
\mathcal{I}^{*}[\varphi](\gamma) := \bigcup_{\theta \in \Theta^{*}} \mathcal{I}[\varphi](\theta,\gamma),
\end{align}
where:
\begin{align}
\mathcal{I}[\varphi](\theta,\gamma):= \bigg\{ I \in \mathbb{R} : \exists P_{U|Y,Z} \in \mathcal{P}_{U|Y,Z}(\theta) \text{ and } &P_{Y_{\gamma}^\star|Y,Z,U} \in \mathcal{P}_{Y_{\gamma}^\star|Y,Z,U}(\theta,\gamma) \nonumber\\
&\text{ satisfying } I^{*}[\varphi]\left( \theta,\gamma,I, P_{Y_{\gamma}^\star|Y,Z,U}, P_{U|Y,Z}\right)\leq 0\bigg\}.\label{eq_identified_integral_prog_ub}
\end{align}

\end{definition}

Our main result in this section will attempt to provide a more insightful characterization of the identified set for policy transforms, which will also be vital for the problem of policy choice considered in the next section. However, before stating our main identification result, we require the following technical assumption. 

\begin{assumption}[Error Bounds]\label{assumption_error_bound}
(i) (Linear Minorant) There exists values $\delta>0$ and $C_{1}>0$ such that for every $\theta \in \Theta$:
\begin{align}
\inf_{P_{U|Y,Z} \in \mathcal{P}_{U|Y,Z}(\theta)}\max_{j=1,\ldots,J} |\E_{P_{U|Y,Z}\times P_{Y,Z}} \left[ m_{j}(y,z,u,\theta) \right]|_{+} \geq C_{1} \min\{\delta,d(\theta,\Theta^{*})\}.\label{eq_errorbound_1}
\end{align}
(ii) (Local Counterfactual Robustness) There exists a value $C_{2} \geq 0$ such that for any $\theta \in \Theta_{\delta}^{*} := \{ \theta : d(\theta, \Theta^{*}) \leq \delta\} $:
\begin{align}
&\inf_{P_{U|Y,Z} \in \mathcal{P}_{U|Y,Z}(\theta)}\inf_{P_{Y_{\gamma}^\star|Y,Z,U} \in \mathcal{P}_{Y_{\gamma}^\star|Y,Z,U}(\theta,\gamma)} \int \varphi(v)  \, dP_{V_{\gamma}}\nonumber\\
&\qquad\qquad\qquad\geq \inf_{\theta^{*} \in \Theta^{*}}\inf_{P_{U|Y,Z} \in \mathcal{P}_{U|Y,Z}(\theta^{*})}\inf_{P_{Y_{\gamma}^\star|Y,Z,U} \in \mathcal{P}_{Y_{\gamma}^\star|Y,Z,U}(\theta^{*},\gamma)} \int \varphi(v)  \, dP_{V_{\gamma}}  - C_{2} d(\theta,\Theta^{*}),\label{eq_errorbound_2}
\end{align}
and:
\begin{align}
&\sup_{P_{U|Y,Z} \in \mathcal{P}_{U|Y,Z}(\theta)}\sup_{P_{Y_{\gamma}^\star|Y,Z,U} \in \mathcal{P}_{Y_{\gamma}^\star|Y,Z,U}(\theta,\gamma)} \int \varphi(v)  \, dP_{V_{\gamma}}\nonumber\\
&\qquad\qquad\qquad\leq \sup_{\theta^{*} \in \Theta^{*}}\sup_{P_{U|Y,Z} \in \mathcal{P}_{U|Y,Z}(\theta^{*})}\sup_{P_{Y_{\gamma}^\star|Y,Z,U} \in \mathcal{P}_{Y_{\gamma}^\star|Y,Z,U}(\theta^{*},\gamma)} \int \varphi(v)  \, dP_{V_{\gamma}}  + C_{2} d(\theta,\Theta^{*}).\label{eq_errorbound_3}
\end{align}

\end{assumption}
Intuitively, Assumption \ref{assumption_error_bound} makes two statements. First, part (i) of the assumption is a global condition that requires that, whenever $\theta \in \Theta \setminus \Theta^{*}$, there is at least one moment function that can be bounded below by the function on the right side of \eqref{eq_errorbound_1}. In general this condition is very similar to previous conditions in the literature; see, for example, the  ``partial identification condition'' in \cite{chernozhukov2007} section 4.2. Also, see \cite{kaido2019constraint} for a review of similar conditions. The major difference arises from the fact that the condition must hold for all $P_{U|Y,Z} \in \mathcal{P}_{U|Y,Z}(\theta)$, owing to the fact that the moment conditions in this paper are allowed to depend on the latent variables. Verifying condition (i) can usually be done by first enumerating all scenarios which imply $\theta \notin \Theta^{*}$, and then verifying that the condition holds for each such scenario. This is exactly the strategy used when verifying the assumption in the examples. Also note that the condition is automatically satisfied if $\mathcal{P}_{U|Y,Z}(\theta)$ is empty---that is, when $\bm G^{-}(Y,Z,\theta)$ admits no measurable selections---or when none of the moment conditions depend on the structural parameters. 

Part (ii) of Assumption \ref{assumption_error_bound} appears to be entirely new. Intuitively, \eqref{eq_errorbound_2} is a local condition that requires the smallest value of the integral of $\varphi$ to not decrease too fast as we move $\theta$ slightly outside of the identified set. In the opposite direction, \eqref{eq_errorbound_3} requires that the largest value of the integral of $\varphi$ does not increase too fast as we move $\theta$ slightly outside of the the identified set. These conditions will be violated if, for example, the value of the integral can change discontinuously on the boundary of the identified set. We call the condition the \textit{local counterfactual robustness} condition because it demands that small changes in the value of the structural parameters do not generate discontinuous changes in value of the counterfactual quantity of interest. Interestingly, both of the conditions in Assumption \ref{assumption_error_bound} are related to typical assumptions made in the theory of error bounds in the optimization literature.\footnote{See \cite{pang1997error} for an introduction. } Finally, note the value of $\delta$ in parts (i) and (ii) are the same. However, this is not restrictive, since part (i) and (ii) can be established for two different values $\delta_{(i)},\delta_{(ii)}>0$, and then $\delta$ can be taken as $\delta=\min\{\delta_{(i)},\delta_{(ii)}\}$.

In practice, part (ii) of Assumption \ref{assumption_error_bound} can be challenging to verify. Because of this, we introduce the following assumption as an alternative to part (ii) of Assumption \ref{assumption_error_bound}: 
\begin{assumption}[Error Bounds (2)(ii)]\label{assumption_error_bound2}
For some $\delta > 0$, there exists values $\ell_{1},\ell_{2} \geq 0$ (possibly depending on $\delta$) such that:
\begin{align}
d(u,\bm G^{-}(y,z,\theta)) &\leq \ell_{1} \cdot d(\theta,\Theta^{-}(y,z,u)\cap \Theta_{\delta}^{*}), \,\,&&(y,z)-a.s. \,\, \text{ for all } u \in \mathcal{U} \text{ and } \theta \in \Theta_{\delta}^{*}, \label{eq_lipschitz_cond1}\\
d(y^\star,\bm G^{\star}(y,z,u,\theta,\gamma)) &\leq \ell_{2}   \cdot d(\theta,\Theta^{\star}(v,\gamma)\cap \Theta_{\delta}^{*}), \,\,&&(y,z,u)-a.s. \,\, \text{ for all } y^\star \in \mathcal{Y}^\star \text{ and } \theta \in \Theta_{\delta}^{*}.\label{eq_lipschitz_cond2}
\end{align}
where $\Theta^{-}(y,z,u)$ and $\Theta^{\star}(v,\gamma)$ are defined by:
\begin{align*}
\Theta^{-}(y,z,u)&:= \left\{ \theta : u \in \bm G^{-}(y,z,\theta)\right\}, &&\Theta^{\star}(v,\gamma):= \left\{ \theta : y^\star \in \bm G^{\star}(y,z,u,\theta,\gamma)\right\}.
\end{align*}
Furthermore, the function $\varphi: \mathcal{V} \to \mathbb{R}$ is bounded, measurable, and Lipschitz continuous in $(u,y^\star)$ with Lipschitz constant $L_{\varphi}$.
\end{assumption}

The following Lemma shows that Assumption \ref{assumption_error_bound2} is sufficient for part (ii) of Assumption \ref{assumption_error_bound}. In the process, the Lemma makes an interesting connection between Assumption \ref{assumption_error_bound} and certain Lipschitzian behaviour of the random sets $\bm G^{-}$ and $\bm G^\star$ with respect to the structural parameters $\theta \in \Theta$.  
\begin{lemma}\label{lemma_lipschitz_condition}
Suppose that Assumptions \ref{assump_preliminary}, \ref{assumption_factual_domain} and \ref{assumption_counterfactual_domain} are satisfied. Finally, suppose that $\bm G^{-}(\cdot,\theta)$ and $\bm G^{\star}(\cdot,\theta,\gamma)$ are almost-surely non-empty for each $\theta \in \Theta^{*}$. Then Assumption \ref{assumption_error_bound2} implies Assumption \ref{assumption_error_bound}(ii) with $C_{2} = L_{\varphi} \max\{\ell_{1},\ell_{2}\}$.
\end{lemma}

\begin{proof}
See Appendix \ref{appendix_proofs}.
\end{proof}

It can be shown that the conditions \eqref{eq_lipschitz_cond1} and \eqref{eq_lipschitz_cond2} are equivalent to almost-sure versions of Lipschitz continuity conditions for set-valued maps, where the distance between two sets is measured by the Pompeiu–Hausdorff distance. Localized versions of these conditions are called \textit{metric regularity} conditions, which also have a close connection to constraint qualifications from optimization theory. See \cite{dontchev2009implicit} Chapter 3.3 and \cite{ioffe2016metric} for a discussion. 

\subsection{Envelope Functions for the Policy Transform}\label{subsection_bounds_on_policy_transform}

We can finally turn to our main objective for this section, which is the problem of bounding the policy transform $I[\varphi](\gamma)$. Theoretically, bounds on $I[\varphi](\gamma)$ can be obtained by solving two (very) complicated constrained optimization problems that search over all distributions $P_{U|Y,Z}$ and $P_{Y_{\gamma}^\star|Y,Z,U}$ that satisfy our modelling assumptions for the ones that maximize and minimize the policy transform of $\varphi$. However, it is clear that such optimization problems will be infeasible in most realistic cases. The following result shows a tractable formulation of bounds on policy transforms that will be important for the next section. 
\begin{theorem}[Bounds on the Policy Transform]\label{thm_cortes}
Suppose that Assumptions \ref{assump_preliminary}, \ref{assumption_factual_domain}, \ref{assumption_counterfactual_domain} and \ref{assumption_error_bound} all hold. Also, suppose that $\varphi:\mathcal{V} \to [\varphi_{\ell b}, \varphi_{ub}]\subset \mathbb{R}$ is a bounded, measurable function, and that for each $\gamma \in \Gamma$, the random sets $\bm G^{-}(\cdot,\theta)$ and $\bm G^{\star}(\cdot,\theta,\gamma)$ are almost-surely non-empty for each $\theta \in \Theta^{*}$. Then $\overline{\text{co}}\, \mathcal{I}^{*}[\varphi](\gamma)=[I_{\ell b}[\varphi](\gamma),I_{u b}[\varphi](\gamma)]$, with: 
\begin{align}
I_{\ell b}[\varphi](\gamma) &= \inf_{\theta \in \Theta} \max_{\lambda_{j} \in \{0,1\}}\int\inf_{u \in \bm G^{-}(y,z,\theta)}  \inf_{y^{\star}  \in \bm G^{\star}(y,z,u,\theta,\gamma)} \Bigg( \varphi(v) + \mu^{*}\sum_{j=1}^{J}  \lambda_{j}  m_{j}(y,z,u,\theta)\Bigg) \, dP_{Y,Z},\label{eq_lb_varphi}\\
I_{ub}[\varphi](\gamma) &= \sup_{\theta \in \Theta} \min_{\lambda_{j} \in \{0,1\}}  \int \sup_{u \in \bm G^{-}(y,z,\theta)}  \sup_{y^{\star}  \in \bm G^{\star}(y,z,u,\theta,\gamma)} \Bigg( \varphi(v) - \mu^{*}\sum_{j=1}^{J}  \lambda_{j} m_{j}(y,z,u,\theta)\Bigg) \, dP_{Y,Z},\label{eq_ub_varphi}
\end{align}
where $\mu^{*} \in \mathbb{R}_{+}$ is any value satisfying:
\begin{align}
\mu^{*} \geq \max\left\{\frac{C_{2}}{C_{1}}, \frac{(\varphi_{ub} - \varphi_{\ell b})}{C_{1}\delta} \right\},\label{eq_mu_star_requirement}
\end{align}
and where $C_{1}$, $C_{2}$ and $\delta$ are from Assumption \ref{assumption_error_bound}.
\end{theorem}
\begin{proof}
See Appendix \ref{appendix_proofs}.
\end{proof}

Theorem \ref{thm_cortes} states that the closed, convex hull of the identified set $\mathcal{I}^{\star}[\varphi](\gamma)$ from Definition \ref{definition_identified_set_counterfactual} for the policy transform $I[\varphi](\gamma)$ can be computed as the solution to two optimization problems. Interestingly, these optimization problems are closely related to problems found in the literature on mathematical programming problems subject to equilibrium constraints (MPECs), which have previously seen applications in economics to social planning problems and Stackelberg games.\footnote{For a textbook treatment, see \cite{luo1996mathematical}.} The upper and lower envelope functions in Theorem \ref{thm_cortes} are perhaps most aptly characterized as penalized optimization problems, with $\mu^{*}$ in \eqref{eq_mu_star_requirement} serving the role of the penalty parameter. Both the statement of the result and its proof rely on the theory of exact penalty functions from the literature on error bounds in variational analysis.\footnote{See \cite{dolgopolik2016unifying} for a review.} The Theorem uses the error bounds Assumption \ref{assumption_error_bound} in order to show that the penalty $\mu^{*}$ can be taken to be finite. This is very important for the theoretical analysis of the policy decision problem to take place in the sections ahead. Furthermore, implicitly Theorem \ref{thm_cortes} shows that the values of $\lambda_{j}$ will depend only on the parameter $\theta$, a point which will used in the next sections.\todo{where exactly is this used? This point could be emphasized better later on.} 

From an identification perspective, the envelope functions will generally not give sharp bounds on the policy transform. However, under any additional conditions that ensure the identified set $\mathcal{I}^{*}(\gamma)$ is closed and convex for every $\gamma \in \Gamma$, Theorem \ref{thm_cortes} provides a (point-wise in $\gamma$) sharp characterization of the identified set for the policy transform. Finally, the result is easily modified for the case when the object of interest is a conditional policy transform.

One of the most interesting features of Theorem \ref{thm_cortes} is that, when the counterfactual object of interest is a particular form, there is no need to compute the identified set $\Theta^{*}$ of structural parameters in order to bound the counterfactual object of interest. In addition, the unobservables in the problem are profiled out, and when the identified set $\mathcal{I}^{*}(\gamma)$ is closed and convex this is without any loss of information. This point also translates into the policy decision problem studied in the next sections. The structural parameters and unobservables intuitively play the role of an intermediary connecting the factual and counterfactual domains. However, after the envelope functions from Theorem \ref{thm_cortes} are computed, they play no further role in the problem of policy choice. 

While we will not dwell on measurability issues in the main text, we note that Lemma \ref{lemma_measurable_functions} in Appendix \ref{appendix_measurability} shows that the integrands in the optimization problems are universally measurable; that is, measurable for the completion of any probability measure $P_{Y,Z}$. The proof of this result relies crucially on the fact that both $\bm G^{-}$ and $\bm G^{\star}$ are Effros-measurable. Furthermore, Proposition \ref{proposition_measurable} in Appendix \ref{appendix_measurability} shows that the maps $\gamma \mapsto I_{\ell b}[\varphi](\gamma),I_{u b}[\varphi](\gamma)$ are measurable with respect to the universal $\sigma-$algebra on $\Gamma$ (as generated by the Borel $\sigma-$algebra). These results will be important to keep in mind in the next sections on policy choice.  

We now return to the examples presented earlier to discuss our identification result. We will first verify Assumption \ref{assumption_error_bound} in our examples and will show how Lemma \ref{lemma_lipschitz_condition} can be helpful. 

\setcounter{example}{0}

\begin{example}[Simultaneous Discrete Choice (cont'd)]
Consider again Example \ref{example_simultaneous_discrete_choice} on simultaneous discrete choice, and recall that we have imposed a median zero and median independence restriction using the moment conditions in \eqref{eq_sdc_mom1} and \eqref{eq_sdc_mom2}. 

This example presents challenges for the verification of Assumption \ref{assumption_error_bound} because of the discontinuity of the function $\varphi(v)=\mathbbm{1}\{\pi_{k}(\gamma(z,y_{-k});\theta)\geq u \}$. Indeed, under our current assumptions, Assumption \ref{assumption_error_bound} is not satisfied. To appreciate the intuition, focus on Assumption \ref{assumption_error_bound}(ii). The issue for this assumption arises only when for some $k \in \{1,\ldots,K\}$ and some $z \in \mathcal{Z}$ and $y_{-k} \in \mathcal{Y}^{K-1}$ we have (i) the counterfactual cutoff value $\pi_{k}(\gamma(z,y_{-k});\theta^{*})=0$ at some $\theta^{*} \in \partial\Theta^{*}$, and if (ii) $P(Y_{k}=1|Z_{k}=z',Y_{-k}=y_{-k}')\neq 0.5$, where $(z',y_{-k}')=\gamma(z,y_{-k})$. In this knife-edge case, a very small change from $\theta^{*}\in \partial\Theta^{*}$ to some $\theta \notin \Theta^{*}$ can cause a discontinuous change in $P(Y_{\gamma,k}^\star=1)$. A full description of this failure, including illustrations of various cases, is presented in Appendix \ref{appendix_additional_details_sdc_A4_error_bound}.  

However, by slightly strengthening our moment conditions we can satisfy Assumption \ref{assumption_error_bound} in this example. The key is to introduce additional assumptions on the degree of smoothness of the distribution of $U_{k}$ around zero. In particular, we will replace the moment conditions in \eqref{eq_sdc_mom1} and \eqref{eq_sdc_mom2} with the following conditions:
\begin{align}
\E \left[\left(\mathbbm{1}\{U_{k}\leq \pi_{k}(z',y_{-k}';\theta)\} - \max\{L_{0}\pi_{k}(z',y_{-k}';\theta),0\} - 0.5\right)\mathbbm{1}\{Z_{k}=z, Y_{-k}=y_{-k}\}  \right]\leq 0,\label{eq_sdc_cons1_intext}\\
\E\left[\left(0.5 - \mathbbm{1}\{U_{k}\leq \pi_{k}(z',y_{-k}';\theta)\}-\max\{-L_{0}\pi_{k}(z',y_{-k}';\theta),0\}\right) \mathbbm{1}\{Z_{k}=z, Y_{-k}=y_{-k}\}  \right] \leq 0,\label{eq_sdc_cons2_intext}
\end{align}
for $k=1,\ldots,K$, for all $z,z' \in \mathcal{Z}$ and all $y_{-k},y_{-k}' \in \mathcal{Y}^{K-1}$. In addition to implying the median zero/median independence assumption, these new moment conditions also limit the amount of probability mass on $\mathcal{U}$ that is arbitrarily close to zero, which turns out to be key to satisfying Assumption \ref{assumption_error_bound}. Also note that, despite the fact that these moment conditions will implicitly impose constraints on the obtainable counterfactual choice probabilities, it is easily verified that they do not impose any additional constraints on the set of structural parameters $\theta \in \Theta$ that can rationalize the observed distribution (in the sense of Definition \ref{definition_identified_set}), and thus do not violate the no-backtracking principle introduced in Remark \ref{remark_no_backtracking}.   

With these new moment conditions, Assumption \ref{assumption_error_bound} can be shown to be satisfied. Recall that when first introducing Example \ref{example_simultaneous_discrete_choice} we assumed $\pi_{k}$ is a known measurable function of $(Z_{k},Y_{-k})$ that is linear in parameters $\theta$, and has a gradient (with respect to $\theta$) bounded away from zero for each $(z,y_{-k})$. We conclude that $\pi_{k}$ is Lipschitz in $\theta$, and also satisfies a ``reverse Lipschitz'' condition; that is, for each $(z,y_{-k})$ we have: 
\begin{align*}
L_{k}'||\theta-\theta^{*}||\leq |\pi_{k}(z,y_{-k};\theta)-\pi_{k}(z,y_{-k};\theta^{*})| \leq L_{k}||\theta-\theta^{*}||,
\end{align*}
for some $L_{k}',L_{k}>0$. Now define:
\begin{align}
\tau:= \min_{k}\min_{(z,y_{-k})} |0.5 - P(Y_{k}=1|Z=z,Y_{-k}=y_{-k})| && s.t. &&|0.5 - P(Y_{k}=1|Z=z,Y_{-k}=y_{-k})|>0.\label{eq_tau}
\end{align}
Then the analysis in Appendix \ref{appendix_additional_details_sdc_A4_error_bound} shows that Assumption \ref{assumption_error_bound} is verified for $C_{1} = L_{0} L'$, $C_{2}=L_{0}L$ and $\delta = \tau/(L_{0} L')$, where $L=\min_{k} L_{k}$ and $L' = \min_{k} L_{k}'$. In Theorem \ref{thm_cortes} we can thus take the penalty $\mu^{*}$ to be any value satisfying:
\begin{align*}
\mu^{*}\geq \max\left\{\frac{L}{L'}, \frac{1}{\tau} \right\}.
\end{align*}
Theorem \ref{thm_cortes} says that the lower and upper envelopes on $I[\varphi](\gamma)= P(Y_{\gamma}^\star=1)$, as a function of $\gamma$, are given by \eqref{eq_lb_varphi} and \eqref{eq_ub_varphi}, respectively.
\end{example}
\begin{remark}[Counterfactual Coherency]
Recall that Theorem \ref{thm_cortes} applies only if the random sets $\bm G^{-}(\cdot,\theta)$ and $\bm G^{\star}(\cdot,\theta,\gamma)$ are almost-surely non-empty for each $\theta \in \Theta^{*}$. In the simultaneous discrete choice example, the counterfactual map $\bm G^{\star}(\cdot,\theta,\gamma)$ can fail to be almost-surely non-empty, which is related to the well known problem of \textit{coherency} in these models. In particular, for a given instantiation of a vector of unobservables $(u_{1},\ldots,u_{K})$, there may not exist any vector of counterfactual endogenous outcome variables $(y_{1,\gamma}^\star,\ldots,y_{K,\gamma}^\star)$ that solves the system of equations represented by \eqref{eq_sdc_counterfactual_equations}. However, we note that this issue is unrelated to our particular approach, and might be resolved by (i) conditioning the analysis on the subset of $\mathcal{U}$ that ensures a solution to the system of equations in \eqref{eq_sdc_counterfactual_equations}, or (ii) imposing certain constraints on the parameter space that ensures the existence of a solution to the system of equations in \eqref{eq_sdc_counterfactual_equations}. We refer the reader to \cite{chesher2020structural} for a thorough discussion of this issue. However, whether this ``counterfactual coherency'' problem can be resolved without violating the no-backtracking principle from Remark \ref{remark_no_backtracking} appears to be an open question.\todolt{Does this violate the no back-tracking argument?} 
\end{remark}

\begin{example}[Program Evaluation (cont'd)]
Consider again Example \ref{example_program_evaluation} on program evaluation. Verification of Assumption \ref{assumption_error_bound} is presented in Appendix \ref{appendix_additional_details_te_A4_error_bound}, and uses Lemma \ref{lemma_lipschitz_condition} to verify Assumption \ref{assumption_error_bound}(ii). Remarkably, we show that Assumption \ref{assumption_error_bound} is satisfied for any value of $\delta>0$ with $C_{1}=C_{2}=1$. Thus we can take the penalty $\mu^{*}=1$. Then Theorem \ref{thm_cortes} says that the lower and upper envelopes on $I[\varphi](\gamma)= \E[Y_{\gamma}^\star]$, as a function of $\gamma$, are given by \eqref{eq_lb_varphi} and \eqref{eq_ub_varphi}, respectively.
\end{example}

\section{On the Learnability of Optimal Policies}\label{sec_policy_analysis}

In this section, we provide sufficient conditions for PAMPAC learnability. To begin, the following proposition clarifies the connection between the lower envelope function from the previous section and the notion of PAMPAC learnability. 
\begin{proposition}\label{proposition_lower_envelope_connection}
Suppose Assumptions \ref{assump_preliminary}, \ref{assumption_factual_domain}, \ref{assumption_counterfactual_domain}, and \ref{assumption_error_bound} hold. Also, suppose that $\varphi:\mathcal{V} \to [\varphi_{\ell b}, \varphi_{ub}]\subset \mathbb{R}$ is a bounded, measurable function, and that for each $\gamma \in \Gamma$, the random sets $\bm G^{-}(\cdot,\theta)$ and $\bm G^{\star}(\cdot,\theta,\gamma)$ are almost-surely non-empty for each $\theta \in \Theta^{*}$. Then a policy space $\Gamma$ is PAMPAC learnable with respect to the policy transform of $\varphi$ if and only if:
\begin{align}
\inf_{P_{Y,Z} \in \mathcal{P}_{Y,Z}} P_{Y,Z}^{\otimes n} \left( \sup_{\gamma \in \Gamma} I_{\ell b}[\varphi](\gamma)  - I_{\ell b}[\varphi](d(\psi)) \leq c \right) \geq \kappa, \label{eq_PAC_crit2}
\end{align}
where $I_{\ell b}[\varphi]: \Gamma \to \mathbb{R}$ is the lower envelope function from Theorem \ref{thm_cortes}.
\end{proposition}
\begin{remark}
By Proposition \ref{proposition_measurable} in Appendix \ref{appendix_measurability}, the map $\psi \mapsto I_{\ell b}[\varphi](d(\psi))$ is universally measurable; that is, measurable with respect to the completion of any $P_{Y,Z} \in \mathcal{P}_{Y,Z}$. Thus, the event in \eqref{eq_PAC_crit2} can always be assigned a unique probability using outer measures, if necessary. 
\end{remark}
In particular, the lower envelope function completely characterizes PAMPAC learnability of the policy space $\Gamma$ with respect to $\varphi$. Thus, it should be unsurprising that our sufficient conditions for a policy space to be PAMPAC learnable will be related to the behaviour of the lower envelope function from Theorem \ref{thm_cortes}. 

Next we introduce an entropy growth condition, which will be imposed as a constraint on the complexity allowed for both the moment functions and the function $\varphi$. To introduce the entropy growth condition, we must first define the covering number and metric entropy for a class of functions.

\begin{definition}[Covering Number, Metric Entropy]
Let $(\mathcal{T},\rho)$ be a semi-metric space. A cover of $\mathcal{T}$ is any collection of sets whose union contains $\mathcal{T}$ as a subset. For any $\varepsilon>0$, the covering number for $\mathcal{T}$, denoted by $N(\varepsilon,\mathcal{T},\rho)$, is the smallest number of $\rho-$balls needed to form a $\varepsilon-$cover. The metric entropy is the logarithm of the covering number.  
\end{definition}

\begin{definition}[Entropy Growth Condition]\label{definition_entropy_growth}
Let $\mathcal{F}$ be a measurable class of real-valued functions on a measurable space $(\mathcal{X},\mathfrak{A}_{\mathcal{X}})$ with envelope $F$. The class $\mathcal{F}$ satisfies the entropy growth condition if:
\begin{align}
\sup_{Q \in \mathcal{Q}_{n}} \log N(\varepsilon,\mathcal{F},||\cdot||_{Q,2})=o(n), \label{eq_entropy_growth}
\end{align}
for every $\varepsilon>0$, with the supremum taken over all discrete probability measures $\mathcal{Q}_{n}$ on $\mathcal{X}$ with atoms that have probabilities that are integer multiples of $1/n$. 
\end{definition}
This condition is adapted from a condition in \cite{dudley1991uniform} (Theorem 6, p. 500) that, in combination with other mild conditions, is shown to be sufficient for a class of functions to be uniform Glivenko-Cantelli.\footnote{See also \cite{van1996weak} Theorem 2.8.1 on p.167.} The entropy growth condition essentially says that, for any set $\mathcal{X}_{n}$ of $n$ points $(x_{1},\ldots,x_{n})$ in some space $\mathcal{X}$, the logarithm of the minimal number of balls of radius $\varepsilon>0$ needed to cover the set:
\begin{align*}
\mathcal{F}|_{\mathcal{X}_{n}} := \left\{ (f(x_{1}),\ldots,f(x_{n})) : f \in \mathcal{F} \right\} \subseteq \mathbb{R}^{n},
\end{align*}
is of order $o(n)$. Sufficient conditions for this to be the case can be connected to conditions previously used in the literature. For example, \eqref{eq_entropy_growth} is satisfied if the class of functions is of VC-type (c.f. \cite{chernozhukov2013gaussian}, \cite{belloni2019subvector}), if the class satisfies Pollard's manageability criterion (c.f. \cite{pollard1990empirical}, \cite{andrews2013inference}, \cite{andrews2017inference}), or if the class of functions is otherwise known to be a uniform Donsker class.

The following Theorem shows that if certain classes of functions in the policy analysis problem obey the entropy growth condition, then every policy space is PAMPAC learnable.  To state the result, we must first introduce an important class of functions. Let $\Lambda = \{0,1\}^{J}$, and for a fixed triple $(\theta,\gamma,\lambda) \in \Theta\times \Gamma\times\Lambda$, let $h_{\ell b}(\cdot,\cdot,\theta,\gamma,\lambda): \mathcal{Y} \times \mathcal{Z} \to \mathbb{R}$ be given by:
\begin{align}
 h_{\ell b}(y,z,\theta,\gamma,\lambda):=\inf_{u \in \bm G^{-}(y,z,\theta)}   \Bigg( \inf_{y^{\star}  \in \bm G^{\star}(y,z,u,\theta,\gamma)} \varphi(v) + \mu^{*}\sum_{j=1}^{J}  \lambda_{j} m_{j}(y,z,u,\theta)\Bigg).
\end{align}
Note that $h_{\ell b}(\cdot,\cdot,\theta,\gamma,\lambda)$ is exactly the integrand in the lower envelope function from Theorem \ref{thm_cortes}. Now define the class of functions:
\begin{align}
\mathcal{H}_{\ell b}:= \left\{ h_{\ell b}(\cdot,\cdot,\theta,\gamma,\lambda): \mathcal{Y} \times \mathcal{Z} \to \mathbb{R} : (\theta,\gamma,\lambda) \in \Theta \times \Gamma\times\Lambda\right\}.
\end{align}
Then we have the following result:
\begin{theorem}\label{theorem_pampac_learnability}
Suppose that Assumptions \ref{assump_preliminary}, \ref{assumption_factual_domain}, \ref{assumption_counterfactual_domain} and \ref{assumption_error_bound} hold. Also, suppose that $\varphi:\mathcal{V} \to [\varphi_{\ell b}, \varphi_{ub}]\subset \mathbb{R}$ is a bounded, measurable function, and that for each $\gamma \in \Gamma$, the random sets $\bm G^{-}(\cdot,\theta)$ and $\bm G^{\star}(\cdot,\theta,\gamma)$ are almost-surely non-empty for each $\theta \in \Theta^{*}$. Fix any $\varepsilon>0$. (i) If the class of functions $\mathcal{H}_{\ell b}$ satisfies the entropy growth condition, then every policy space is PAMPAC learnable with respect to the policy transform of $\varphi$. Furthermore, for any $c>0$ we have:
\begin{align}
\sup_{P_{Y,Z} \in \mathcal{P}_{Y,Z}} P_{Y,Z}^{\otimes n} \left( \sup_{\gamma \in \Gamma}\inf_{s \in \mathcal{S}} I[\varphi](\gamma,s)  - \inf_{s\in \mathcal{S}} I[\varphi](d(\psi),s) \geq c \right) =O(r_{1}(n)),
\end{align}
where:
\begin{align}
r_{1}(n):=\max\left\{n^{-1/2}, n^{-1/2} \sup_{Q \in \mathcal{Q}_{n}}\sqrt{  \log  N(\varepsilon,\mathcal{H}_{\ell b},||\cdot||_{Q,2})}\right\}. 
\end{align}
(ii) If the class of functions:
\begin{align}
\Phi &:= \{ \varphi(\cdot,u,y^\star):\mathcal{Y}\times\mathcal{Z} \to \mathbb{R} : (u,y^\star) \in \mathcal{U} \times \mathcal{Y}^\star \},\label{eq_thm_pampac_phi}\\
\mathcal{M}_{j}&:= \left\{ m_{j}(\cdot,u,\theta):\mathcal{Y}\times\mathcal{Z} \to \mathbb{R} : (u,\theta) \in \mathcal{U}\times\Theta \right\},\quad j=1,\ldots,J,
\end{align}
are uniformly bounded, and satisfy the entropy growth condition, then so does $\mathcal{H}_{\ell b}$. Furthermore, for any $c>0$ we have:\todolt{Eventually might be nice to have a lower bound on the rate. Also, note that often $y^\star=f(y,z,u,\theta,\gamma)$...how to verify entropy growth condition for $\varphi$ in this case? Finally, what are necessary conditions for learning?}
\begin{align}
\sup_{P_{Y,Z} \in \mathcal{P}_{Y,Z}} P_{Y,Z}^{\otimes n} \left( \sup_{\gamma \in \Gamma}\inf_{s \in \mathcal{S}} I[\varphi](\gamma,s)  - \inf_{s\in \mathcal{S}} I[\varphi](d(\psi),s) \geq c \right) =O(r_{2}(n)),
\end{align}
where:
\begin{align}
r_{2}(n):=\max\left\{n^{-1/2}, n^{-1/2} \sup_{Q \in \mathcal{Q}_{n}} \sqrt{ \log N (\varepsilon/4,\Phi,||\cdot||_{Q,2}) + \sum_{j=1}^{J} \log N (\varepsilon/2,\mathcal{M}_{j},||\cdot||_{Q,2})}\right\}. 
\end{align}
\end{theorem}
\begin{proof}
See Appendix \ref{appendix_proofs}.
\end{proof}

The proof of the part (i) proceeds by proposing a specific decision procedure, and then showing that the proposed decision procedure satisfies the requirements of PAMPAC learnability from Definition \ref{definition_pampac_learnability} when the class of functions $\mathcal{H}_{\ell b}$ satisfies the entropy growth condition. The specific decision procedure proposed in the proof is any procedure that obtains within $\varepsilon$ of the maximum of the sample analog lower envelope function for each sample $\psi \in \Psi_{n}$, for some $\varepsilon>0$. We call this rule the \textit{$\varepsilon-$maximin empirical rule}, and we will revisit it's properties in the next subsection. Here we also finally see the close connection between PAMPAC learnability and the lower envelope function from Theorem \ref{thm_cortes} in the previous section, which has been alluded to throughout the paper. The particular form of the lower envelope function from Theorem \ref{thm_cortes} makes it amenable to analysis using methods from empirical process theory, which are used in the proof of Theorem \ref{theorem_pampac_learnability}. Also note that Assumption \ref{assumption_error_bound}, which was needed to obtain a bound on the penalty $\mu^{*}$ in Theorem \ref{thm_cortes}, is also needed for this result. Without a bound on this penalty, Theorem \ref{theorem_pampac_learnability} will generally not be true.

The proof of part (ii) of Theorem \ref{theorem_pampac_learnability} shows that if each ``component'' of the lower envelope of the policy transform---namely the moment functions and the function $\varphi$---satisfy the entropy growth condition, then the metric entropy of the class $\mathcal{H}_{\ell b}$ can also be controlled. Combined with the result in Proposition \ref{proposition_lower_envelope_connection}, the proof of part (ii) of Theorem \ref{theorem_pampac_learnability} then shows that our proposed $\varepsilon-$maximin decision rule can obtain close to the maximum value (over $\gamma \in \Gamma$) of the lower envelope of the policy transform with high probability.

It may seem surprising that our learnability result holds for any policy space. However, this is a result of the fact that the complexity of the policy space is tempered by the class of functions $\Phi$ from \eqref{eq_thm_pampac_phi}, since it is only through functions in this class that the policy can affect the policy transform. By imposing that the class $\Phi$ satisfy the entropy growth condition, we are implicitly imposing constraints on the complexity of the policy space. Note that the Theorem provides only sufficient conditions for PAMPAC learnability, and alternative results that impose complexity constraints on the policy space $\Gamma$ directly, rather than on $\Phi$, may be possible.   

We will now turn to our motivating examples to verify learnability of the involved policy spaces.

\setcounter{example}{0}

\begin{example}[Simultaneous Discrete Choice (cont'd)]
Consider again Example \ref{example_simultaneous_discrete_choice} on simultaneous discrete choice. In this case we have:
\begin{align}
\Phi &:= \{ \mathbbm{1}\{\pi_{k}(\gamma(\,\cdot\,);\theta) \geq u \} : (u,\theta) \in \mathcal{U}\times\Theta \},
\end{align}
with the moment conditions:
\begin{align}
\E \left[\left(\mathbbm{1}\{U_{k}\leq \pi_{k}(z',y_{-k}';\theta)\} - \max\{L_{0}\pi_{k}(z',y_{-k}';\theta),0\} - 0.5\right)\mathbbm{1}\{Z_{k}=z, Y_{-k}=y_{-k}\}  \right]\leq 0,\\
\E\left[\left(0.5 - \mathbbm{1}\{U_{k}\leq \pi_{k}(z',y_{-k}';\theta)\}-\max\{-L_{0}\pi_{k}(z',y_{-k}';\theta),0\}\right) \mathbbm{1}\{Z_{k}=z, Y_{-k}=y_{-k}\}  \right] \leq 0,
\end{align}
for $k=1,\ldots,K$, for all $z,z' \in \mathcal{Z}$ and all $y_{-k},y_{-k}' \in \mathcal{Y}^{K-1}$. Details on the verification of the entropy growth condition for both $\Phi$ and the class of moment functions associated with the moment conditions above are presented in Appendix \ref{appendix_additional_details_sdc_verify_learnable}. Furthermore, under our assumptions for this example, the rate of convergence derived from Theorem \ref{theorem_pampac_learnability} is found to be $O(n^{-1/2})$. 
\end{example}

\begin{example}[Program Evaluation (cont'd)]
Consider again Example \ref{example_program_evaluation} on program evaluation. In this case we have:
\begin{align}
\Phi &:= \{ \mathbbm{1}\{g(\gamma(z)) \geq u \}(u_{1}-u_{0}) + u_{0} : (u_{0},u_{1},u,g) \in \mathcal{U}\times\mathcal{G} \},
\end{align}
with the moment conditions:
\begin{align}
\E[\left(D - g(Z_{0},X)\right)\mathbbm{1}\{Z_{0}=z_{0},X=x\}] &\leq 0,\qquad \forall z_{0} \in \mathcal{Z}_{0}, \, x \in \mathcal{X},\\
\E[\left(g(Z_{0},X) - D\right)\mathbbm{1}\{Z_{0}=z_{0},X=x\}] &\leq 0,\qquad \forall z_{0} \in \mathcal{Z}_{0}, \, x \in \mathcal{X},\\
\E[\left(\mathbbm{1}\{U \leq g(z_{0},x)\} - g(z_{0},x)\right)\mathbbm{1}\{X=x\}] &\leq 0,\qquad \forall z_{0} \in \mathcal{Z}_{0}, \, x \in \mathcal{X}, \\
\E[\left(g(z_{0},x) - \mathbbm{1}\{U \leq g(z_{0},x)\}\right)\mathbbm{1}\{X=x\}] &\leq 0, \qquad\forall z_{0} \in \mathcal{Z}_{0},\, x \in \mathcal{X},\\
\E\left[t(z_{0},x)-\mathbbm{1}\{Z=z_{0},X=x\}\right] &\leq 0,\qquad\forall z_{0} \in \mathcal{Z}_{0},\, \forall x \in \mathcal{X}, \\
\E\left[\mathbbm{1}\{Z=z_{0},X=x\}-t(z_{0},x)\right] &\leq 0,\qquad\forall z_{0} \in \mathcal{Z}_{0},\, \forall x \in \mathcal{X},
\end{align}
and:
\begin{align}
\E\left[U_{d}\left(\mathbbm{1}\{Z=z_{0},X=x\}\sum_{z_{0} \in \mathcal{Z}_{0}} t(z_{0},x)- \mathbbm{1}\{X=x\}t(z_{0},x) \right)\right] &\leq 0, \,\,\forall z_{0} \in \mathcal{Z}_{0},\, x \in \mathcal{X},\, d\in \{0,1\},\\
\E\left[U_{d}\left(\mathbbm{1}\{X=x\}t(z_{0},x)-\mathbbm{1}\{Z=z_{0},X=x\}\sum_{z_{0} \in \mathcal{Z}_{0}} t(z_{0},x) \right) \right] &\leq 0, \,\,\forall z_{0} \in \mathcal{Z}_{0},\, x \in \mathcal{X},\, d\in \{0,1\}.
\end{align}
Details on the verification of the entropy growth condition for both $\Phi$ and the class of functions associated with the moment functions above are presented in Appendix \ref{appendix_additional_details_te_verify_learnable}. Furthermore, under our assumptions for this example, the rate of convergence derived from Theorem \ref{theorem_pampac_learnability} is found to be $O(n^{-1/2})$. 
\end{example}

\section{Ex-Post Theoretical Results}\label{section_ex_post_analysis}

Theorem \ref{theorem_pampac_learnability} shows sufficient conditions for PAMPAC learnability in a given environment. However, while the result shows that it may be possible \textit{ex-ante} (i.e. before observing a particular sample) to learn a given policy space, it does not provide us any useful \textit{ex-post} (i.e. after observing the sample) information on the performance of our decision rule. This reflects a well-known complaint of PAC learnability, and has given rise to the literature on data-dependent excess risk bounds in statistical learning literature; see \cite{bartlett2002model}, \cite{koltchinskii2001rademacher}, and \cite{koltchinskii2006local} for examples, and \cite{boucheron2005theory} or \cite{koltchinskii2011oracle} for a review. Thus after establishing learnability of a particular class of policies, it may be of separate interest to evaluate the finite sample performance of a given decision rule for a given sample.  

This is accomplished in the next subsections. We will focus our attention on the particular decision rule used in the proof of Theorem \ref{theorem_pampac_learnability} which was shown to satisfy the requirements of PAMPAC learnability under the assumptions of the theorem. The decision rule used was allowed to be any $\varepsilon-$maximizer of the empirical version of the lower envelope function $I_{\ell b}[\varphi](\gamma)$, which is why we will call it the \textit{$\varepsilon-$maximin empirical rule.}

\begin{definition}[$\varepsilon-$maximin empirical welfare]
Fix any $\varepsilon\geq 0$ and let $\widehat{I}_{\ell b}[\varphi](\gamma)$ denote the lower envelope from Theorem \ref{thm_cortes} evaluated at the empirical measure for $(Y,Z)$. Then $d: \Psi_{n} \to \Gamma$ is a $\varepsilon-$maximin empirical (eME) rule if:
\begin{align}
\widehat{I}_{\ell b}[\varphi](d(\psi)) + \varepsilon \geq \sup_{\gamma \in \Gamma} \widehat{I}_{\ell b}[\varphi](\gamma). \label{maximin_swf_empirical}
\end{align}
\end{definition} 
\begin{remark}
Note that in general the ``$\varepsilon$'' is necessary (although it can be made arbitrarily small), owing to the fact that the supremum of $\widehat{I}_{\ell b}[\varphi](\cdot)$ may not be obtained. 
\end{remark}

Furthermore, unlike our result on PAMPAC learnability, all of the results in the next subsections are data-dependent, and do not depend on any particular properties (beyond measurability) of any function classes involved in the policy decision problem. Thus, there is no need to verify the entropy growth condition, or any other condition sufficient for learnability to use the results ahead. In practice, we still recommend that the sufficient conditions for learnability of a policy space be verified prior to using the results. 

\subsection{Theoretical Results for the Maximin Empirical Rule}\label{sec_welfare_guarantees_MEW}

In this section we obtain a bound on the value of $c_{n}(d,\kappa)$ for any fixed $\kappa$ taking $d$ to be the eME rule. To describe our procedure, we will first introduce a data-dependent complexity measure for the class $\mathcal{H}_{\ell b}$. The complexity measure we use is based on the empirical Rademacher complexity, advocated by \cite{bartlett2002model}, \cite{koltchinskii2001rademacher}, and \cite{koltchinskii2006local} (among others) in the context of empirical risk minimization.

\begin{definition}[Empirical Rademacher Complexity]
Let $\mathcal{F}$ be a class of measurable functions $f:\mathcal{Y}\times\mathcal{Z} \to \mathbb{R}$. The \textit{empirical Rademacher complexity} of $\mathcal{F}$ is given as:
\begin{align}
    ||\mathfrak{R}_{n}||(\mathcal{F}) := \sup_{f \in \mathcal{F}} \left|\frac{1}{n} \sum_{i=1}^{n} \xi_{i} \cdot f(y_{i},z_{i})\right|,\label{eq_rademacher}
\end{align}
where $\xi_{i}$ are realizations of Rademacher random variables; that is, $\xi \in \{-1,1\}$ and $P(\xi_{i}=-1)=P(\xi=1)=1/2$.
\end{definition}
\begin{remark}
A technical point worth emphasizing is that, when seen as a function of the underlying product probability space, the empirical Rademacher complexity may not be a measurable function. We suppress these difficulties in the statement of our results, although we show in Appendix \ref{appendix_measurability} that the Rademacher complexity $||\mathfrak{R}_{n}||(\mathcal{H}_{\ell b})$ is universally measurable (with respect to the product Borel $\sigma-$algebra on $(\mathcal{Y}\times\mathcal{Z})^{n}$), which is sufficient for the purposes in this paper.
\end{remark}

In our context, the empirical Rademacher complexity of the class $\mathcal{H}_{\ell b}$ depends only on the observed empirical distribution and on $n$ draws of a Rademacher random variable; it can therefore be computed after simulating from the Rademacher distribution. With this new definition in hand, we have the following result:  

\begin{theorem}\label{thm_finite_sample1}
Suppose that Assumptions \ref{assump_preliminary}, \ref{assumption_factual_domain}, \ref{assumption_counterfactual_domain}, and \ref{assumption_error_bound} hold. Let $\varphi:\mathcal{V} \to [\varphi_{\ell b}, \varphi_{ub}]\subset \mathbb{R}$ be a bounded, measurable function, and suppose that for each $\gamma \in \Gamma$, the random sets $\bm G^{-}(\cdot,\theta)$ and $\bm G^{\star}(\cdot,\theta,\gamma)$ are almost-surely non-empty for each $\theta \in \Theta^{*}$. Let $\{(y_{i},z_{i})\}_{i=1}^{n}$ be i.i.d. from some distribution $P_{Y,Z}$ satisfying our assumptions and let $d:\Psi_{n} \to \Gamma$ be an eME decision rule for some $\varepsilon>0$. Furthermore, let $\overline{H}<\infty$ satisfy $|h| \leq \overline{H}$ for every $h \in \mathcal{H}_{\ell b}$, and let:
\begin{align}
c_{n}(\kappa)&=  4 ||\mathfrak{R}_{n}||(\mathcal{H}_{\ell b}) + \sqrt{\frac{72 \ln(2/(2-\kappa)) \overline{H}^{2}}{n}} + 5\varepsilon.
\end{align}
Then for any sample size $n$, and any $\kappa \in (0,1)$ we have: 
\begin{align}
\inf_{P_{Y,Z} \in \mathcal{P}_{Y,Z}} P_{Y,Z}^{\otimes n} \left( \sup_{\gamma \in \Gamma}\inf_{s \in \mathcal{S}} I[\varphi](\gamma,s)  - \inf_{s\in \mathcal{S}} I[\varphi](d(\psi),s) \leq c_{n}(\kappa) \right) \geq \kappa.\label{eq_concen_1}
\end{align}
\end{theorem}

\begin{proof}
See Appendix \ref{appendix_proofs}.
\end{proof}

Theorem \ref{thm_finite_sample1} shows two closely related results. First, for any fixed value of $\kappa \in (0,1)$ the Theorem shows that, when in the worst-case state, the eME rule obtains within $c_{n}(\kappa)$ of the maximin value of the state-dependent policy transform with probability at least $\kappa$.  Simple comparative statics show that the value of $c_{n}(\kappa)$ is smaller when $n$ is larger and/or $||\mathfrak{R}_{n}||(\mathcal{H}_{\ell b})$ and $\overline{H}$ are smaller. The only difficult part of computing $c_{n}(\kappa)$ is computing the Rademacher complexity, which is approximately as difficult computationally as computing the empirical version of the lower bound in Theorem \ref{thm_cortes}.  

We again see a close connection between PAMPAC learnability and the lower envelope function from Theorem \ref{thm_cortes}. The particular form of the lower envelope function from Theorem \ref{thm_cortes} makes it especially amenable to analysis using concentration concentration inequalities, which are used in the proof of Theorem \ref{thm_finite_sample1}. Again Assumption \ref{assumption_error_bound} is required for this result: without a finite (and known) value for the penalty $\mu^{*}$, derivation of the finite sample results in Theorem \ref{thm_finite_sample1} would not be possible.  

Finally we mention again that, unlike Theorem \ref{theorem_pampac_learnability} on PAMPAC learnability, Theorem \ref{thm_finite_sample1} does not impose any restrictions on the underlying class of functions $\mathcal{H}_{\ell b}$. In particular, this class need not satisfy the entropy growth condition from Definition \ref{definition_entropy_growth}, nor any other sufficient conditions for learnability, meaning Theorem \ref{thm_finite_sample1} is applicable even when $\Gamma$ is not PAMPAC learnable. As a result, Theorem \ref{thm_finite_sample1} is able to provide finite sample guarantees for the eME rule, but necessarily remains silent about rates of convergence.

\subsection{Bounds on the Set of Optimal Policies}\label{sec_set_of_optimal_policies}

The previous subsection uses a specific rule, the eME rule, and derives finite sample theoretical guarantees on the performance of this rule. However, the eME rule is only one particular rule, and for a variety of reasons it may not be the rule selected by the policymaker.  

In order to complement the results of the previous subsection, in this subsection we will provide some theoretical results on alternative policy rules. To understand the approach, let us define the function:
\begin{align}
\mathscr{E}^{*}(\gamma)&:= \sup_{\gamma \in \Gamma}\inf_{s \in \mathcal{S}} I[\varphi](\gamma,s) - \inf_{s \in \mathcal{S}} I[\varphi](\gamma,s) =\sup_{\gamma \in \Gamma} I_{\ell b}[\varphi](\gamma) - I_{\ell b}[\varphi](\gamma),\label{eq_gamma_error}
\end{align}
and the set:
\begin{align}
\mathscr{G}^{*}(\delta)&:= \{ \gamma \in \Gamma : \mathscr{E}^{*}(\gamma) \leq \delta\}.
\end{align}
We call the set $\mathscr{G}^{*}(\delta)$ the $\delta-$level set. Our objective in this subsection will be to provide an approximation of the $\delta-$level set that holds with probability at least $\kappa$. If we can do so, then by constructionn any decision rule $d:\Psi_{n}\to \Gamma$ that maps within our approximation of the $\delta-$level set will have $c_{n}(d,\kappa)\leq \delta$. There may be many decision rules that map within our approximation to the $\delta-$level set, so our theoretical results will be applicable to a large number of decision rules. As a by product of our analysis, we will also show that for certain values of $\delta$ the eME rule will be contained in the $\delta-$level set with probability at least $\kappa$. Again, the results of this section do not impose any restrictions on the underlying class of functions $\mathcal{H}_{\ell b}$, and are applicable even when $\Gamma$ is not PAMPAC learnable.

To introduce our results for the $\delta-$level set, we must first introduce some additional notation. In particular, define:
\begin{align}
\mathscr{E}_{n}(\gamma)&:=\sup_{\gamma \in \Gamma}\inf_{s \in \mathcal{S}} \widehat{I}[\varphi](\gamma,s) - \inf_{s \in \mathcal{S}} \widehat{I}[\varphi](\gamma,s) =\sup_{\gamma \in \Gamma} \widehat{I}_{\ell b}[\varphi](\gamma) - \widehat{I}_{\ell b}[\varphi](\gamma) ,
\end{align}
and for $\delta>0$ define the set:
\begin{align}
\mathscr{G}_{n}(\delta)&:= \{ \gamma \in \Gamma : \mathscr{E}_{n}(\gamma) \leq \delta\}.
\end{align}
The set $\mathscr{G}_{n}(\delta)$ represents the empirical version of the $\delta-$level set. 

The following theorem shows that, for sufficiently large $\delta$, the $\delta-$level set is contained within an enlargement of, and contains a contraction of, the empirical $\delta-$level set with high probability. 

\begin{theorem}\label{theorem_delta_minimal}
Suppose that Assumptions \ref{assump_preliminary}, \ref{assumption_factual_domain}, \ref{assumption_counterfactual_domain}, and \ref{assumption_error_bound} hold. Also suppose that $\varphi:\mathcal{V} \to [\varphi_{\ell b}, \varphi_{ub}]\subset \mathbb{R}$ is a bounded, measurable function, and that for each $\gamma \in \Gamma$, the random sets $\bm G^{-}(\cdot,\theta)$ and $\bm G^{\star}(\cdot,\theta,\gamma)$ are almost-surely non-empty for each $\theta \in \Theta^{*}$. Let $\overline{H}<\infty$ satisfy $|h| \leq \overline{H}$ for every $h \in \mathcal{H}_{\ell b}$, and suppose that $\{(y_{i},z_{i})\}_{i=1}^{n}$ is i.i.d. from some distribution $P_{Y,Z}$ satisfying our assumptions. Define:
\begin{align*}
\mathcal{H}_{n,\ell b}'(\delta)&:=\{ h_{\ell b}(\cdot,\cdot,\theta,\gamma,\lambda) -  h_{\ell b}(\cdot,\cdot,\theta',\gamma',\lambda') : \theta,\theta' \in \Theta,\,\, \gamma,\gamma' \in \mathscr{G}_{n}(\delta),\lambda,\lambda' \in \{0,1\}^{J} \},
\end{align*}
where $\mathcal{H}_{n,\ell b}'(\delta)$ has a uniform bound $\overline{H}_{n}'(\delta)\leq 2\overline{H} <\infty$. Furthermore, let $t_{j}:=\sqrt{c_{1} \log(c_{2} j )}$ with $c_{1}=5$ and $c_{2}= (3/(2(1-\kappa)))^{2/5}$, and let $\{\delta_{j}\}_{j=0}^{\infty}$ be a sequence decreasing to zero with $\delta_{0} > 2\overline{H}$. Choose some $\mathfrak{a} \in (1,\infty)$, let $\mathfrak{b}=2-1/\mathfrak{a}$, and let:
\begin{align}
T_{n}(\delta)&: = 
\begin{cases}
	2||\mathfrak{R}_{n}||(\mathcal{H}_{n,\ell b}'(\mathfrak{b}\delta_{j})) + \frac{3t_{j}\overline{H}_{n}'(\mathfrak{b}\delta_{j})}{\sqrt{n}}, &\text{ if } \delta \in (\delta_{j+1},\delta_{j}] \text{ for some $j \geq 0$}\\
	0, &\text{ otherwise, } 
\end{cases}
\end{align}
and:
\begin{align}
T_{n}^\flat(\sigma)&:=\sup_{\delta \geq \sigma} \frac{T_{n}(\delta)}{\delta},\\
T_{n}^\sharp(\eta)&:=\inf\left\{ \sigma>0 : T_{n}^\flat(\sigma)\leq \eta\right\}.
\end{align}
Finally, set $\delta^{*}> T_{n}^\sharp(1-1/\mathfrak{a})$. Then for any $\delta \geq \mathfrak{a} \delta^{*}$ we have:
\begin{align*}
\inf_{P_{Y,Z} \in \mathcal{P}_{Y,Z}} P_{Y,Z}^{\otimes n} \left(\mathscr{G}_{n}(\delta/\mathfrak{a}) \subseteq \mathscr{G}^{*}(\delta) \subseteq \mathscr{G}_{n}(\mathfrak{b}\delta)  \right)  \geq \kappa.
\end{align*}
\end{theorem}

\begin{proof}
See Appendix \ref{appendix_proofs}.
\end{proof}

Theorem \ref{theorem_delta_minimal} closely mimics results in the statistical learning literature, namely in the problem of bounding excess risk in empirical risk minimization problems. In particular, the proof of the result uses techniques developed by \cite{koltchinskii2006local} and \cite{koltchinskii2011oracle}, where the latter gives a textbook treatment.\footnote{The $\flat-$ and $\sharp-$transforms are taken from \cite{koltchinskii2006local}, and the properties of these transforms can be found in Appendix A.3. of \cite{koltchinskii2011oracle}.} Theorem \ref{theorem_delta_minimal} gives a novel application of these techniques to the problem of policy choice in the presence of partial identification. Similar to the other results in this paper, Theorem \ref{theorem_delta_minimal} relies crucially on the form of the lower envelope function from Theorem \ref{thm_cortes}. Again Assumption \ref{assumption_error_bound} is required, since Theorem \ref{theorem_delta_minimal} requires a finite (and known) value for the penalty parameter $\mu^{*}$.

Intuitively, Theorem \ref{theorem_delta_minimal} says that for a suitably large value of $\delta$ the $\delta-$level sets $\mathscr{G}_{n}(\delta)$ of the function $\mathscr{E}_{n}(\cdot)$ can be used to approximate the $\delta-$level sets $\mathscr{G}^{*}(\cdot)$ of the function $\mathscr{E}^{*}(\cdot)$. The substantial component of the results is the selection of such a ``suitably large value of $\delta$.'' In particular, the value of $\delta$ needed for our approximation to work must be larger than the value of $\delta^{*}$ from the Theorem, where $\delta^{*}$ is related to the solution of a fixed point equation. The connection of the functions $T_{n}(\,\cdot\,)$, $T_{n}^{\flat}(\,\cdot\,)$ and $T_{n}^{\sharp}(\,\cdot\,)$ to fixed point equations is illustrated in Figure \ref{fig_tnm} and is described in its associated caption. As illustrated in the Figure, the function $T_{n}(\delta)$ is a left-continuous step function that is greater than or equal to zero on the interval $[0,\delta_{0}]$, and zero otherwise. 

\begin{figure}[!ht]
\centering
\includegraphics[scale=0.6]{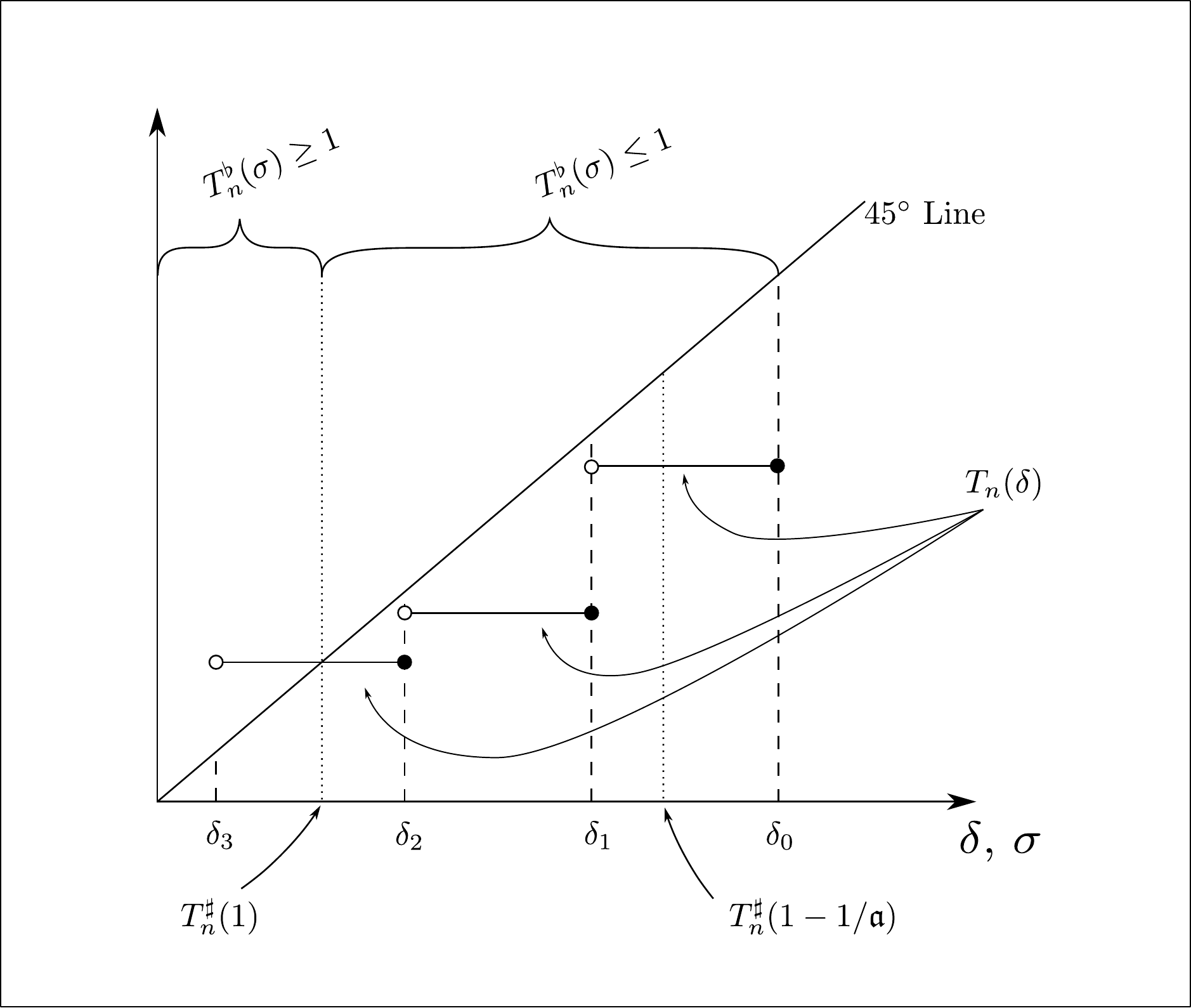}
\caption{This figure illustrates step (iv) in the procedure to determine the $\delta-$level set. After choosing a decreasing sequence $\{\delta_{j}\}_{j=0}^{\infty}$, the policymaker finds the value $\delta^{*}$ such that $\delta^{*}>T_{n}^{\sharp}(1-1/\mathfrak{a})$. In the figure, this occurs in the interval $(\delta_{1},\delta_{0}]$ (although, of course, this need not be the case). The figure also illustrates the fact that $T_{n}(\delta)$ is a step function. Finally, the figure illustrates how the $\flat-$ and $\sharp-$transforms of $T_{n}(\delta)$ are related to fixed point-equations. In particular, the figure illustrates the fixed point of $T_{n}(\delta) = \delta$, which is given exactly by $T^{\sharp}_{n}(1)$. In addition, the fixed point of $T_{n}(\delta) = \delta(1-1/\mathfrak{a})$ is given by $T^{\sharp}_{n}(1-1/\mathfrak{a})$.}\label{fig_tnm}
\end{figure}

The proof of Theorem \ref{theorem_delta_minimal} relies on Lemma \ref{lemma_subsets}, and the best way to understand Theorem \ref{theorem_delta_minimal} is to first understand Lemma \ref{lemma_subsets}.

\begin{lemma}\label{lemma_subsets}
Suppose that the assumptions of Theorem \ref{theorem_delta_minimal} all hold. Define:
\begin{align*}
\mathcal{H}_{\ell b}'(\delta)&:=\{ h_{\ell b}(\cdot,\cdot,\theta,\gamma,\lambda) -  h_{\ell b}(\cdot,\cdot,\theta',\gamma',\lambda') : \theta,\theta' \in \Theta,\,\,  \gamma,\gamma' \in \mathscr{G}^{*}(\delta),\lambda,\lambda' \in \Lambda \},
\end{align*}
where $\mathcal{H}_{\ell b}'(\delta)$ has a uniform bound $\overline{H}'(\delta)\leq 2\overline{H} < \infty$. Furthermore, let $t_{j}:=\sqrt{c_{1} \log(c_{2} j )}$ with $c_{1}=5$ and $c_{2}= (3/(2(1-\kappa)))^{2/5}$, and let $\{\delta_{j}\}_{j=0}^{\infty}$ be a sequence decreasing to zero with $\delta_{0} > 2\overline{H}$. Also, let:
\begin{align}
T(\delta)&: = 
\begin{cases}
	2||\mathfrak{R}_{n}||(\mathcal{H}_{\ell b}'(\delta_{j})) + \frac{3t_{j}\overline{H}'(\delta_{j})}{\sqrt{n}}, &\text{ if } \delta \in (\delta_{j+1},\delta_{j}],\\
	0, &\text{ otherwise, } 
\end{cases}\label{eq_T}
\end{align}
and:
\begin{align}
T^\flat(\sigma)&:=\sup_{\delta \geq \sigma} \frac{T(\delta)}{\delta},\\
T^\sharp(\eta)&:=\inf\left\{ \sigma>0 : T^\flat(\sigma)\leq \eta\right\}.
\end{align}
Finally, suppose $\delta^{**}> T^\sharp(1-1/\mathfrak{a})$ for some $\mathfrak{a}\in (1,\infty)$. Then for any $\delta \geq \mathfrak{a} \delta^{**}$ we have:
\begin{align*}
\inf_{P_{Y,Z} \in \mathcal{P}_{Y,Z}} P_{Y,Z}^{\otimes n} \left(\mathscr{G}_{n}(\delta/\mathfrak{a}) \subseteq \mathscr{G}^{*}(\delta) \subseteq \mathscr{G}_{n}((2-1/\mathfrak{a})\delta) \right) \geq \kappa.
\end{align*}
\end{lemma}
\begin{proof}
See Appendix \ref{appendix_proofs}.
\end{proof}
Note that Lemma \ref{lemma_subsets} is very similar to Theorem \ref{theorem_delta_minimal}, with a major exception being that the class of functions $\mathcal{H}_{\ell b}'(\delta)$ in Lemma \ref{lemma_subsets} differs from the class of functions $\mathcal{H}_{n,\ell b}'(\delta)$ in Theorem \ref{theorem_delta_minimal}. Note that $\mathcal{H}_{n,\ell b}'(\delta)$ represents a ``feasible version'' of $\mathcal{H}_{\ell b}'(\delta)$, since $\mathcal{H}_{\ell b}'(\delta)$ depends on the unknown $\delta-$level set $\mathscr{G}^{*}(\delta)$, where $\mathcal{H}_{n,\ell b}'(\delta)$ depends on the empirical $\delta-$level set $\mathscr{G}_{n}(\delta)$.

A heuristic proof may help provide some sense of how these results work. A necessary step in proving either Theorem \ref{theorem_delta_minimal} or Lemma \ref{lemma_subsets} is to relate the quantities $\mathscr{E}_{n}(\gamma)$ and $\mathscr{E}^{*}(\gamma)$, which is exactly what is done in the proof of Lemma \ref{lemma_subsets}. Among other things, the proof of Lemma \ref{lemma_subsets} demonstrates that an important object connecting the quantities $\mathscr{E}_{n}(\gamma)$ and $\mathscr{E}^{*}(\gamma)$ is given by:
\begin{align}
\delta \mapsto \sup_{\theta, \theta' \in \Theta} \sup_{\gamma, \gamma' \in \mathscr{G}^{*}(\delta)} \sup_{\lambda,\lambda' \in \Lambda}\left|\left( \P_{n} h_{\ell b}(\cdot,\theta,\gamma,\lambda) -  \P_{n} h_{\ell b}(\cdot,\theta',\gamma',\lambda')\right) -  \left( P h_{\ell b}(\cdot,\theta,\gamma,\lambda) -  P h_{\ell b}(\cdot,\theta',\gamma',\lambda') \right)\right|,\label{eq_supnorm_empirical}
\end{align}
where:
\begin{align*}
\P_{n} h_{\ell b}(\cdot,\theta,\gamma,\lambda)&:= \frac{1}{n}\sum_{i=1}^{n}h_{\ell b}(y_{i},z_{i},\theta,\gamma,\lambda), && P h_{\ell b}(\cdot,\theta,\gamma,\lambda):= \int h_{\ell b}(y,z,\theta,\gamma,\lambda)\,dP_{Y,Z}.
\end{align*}
The quantity \eqref{eq_supnorm_empirical} is easily seen to be the sup-norm of a particular empirical process. Note that this empirical process depends on unknown population quantities through both $\mathscr{G}^{*}(\delta)$ and through the functions $ P h_{\ell b}(\cdot,\theta,\gamma,\lambda)$ and $ P h_{\ell b}(\cdot,\theta',\gamma',\lambda')$, which depend on the unknown true probability measure. While the dependence on $\mathscr{G}^{*}(\delta)$ is unavoidable for now, the dependence on $ P h_{\ell b}(\cdot,\theta,\gamma,\lambda)$ and $ P h_{\ell b}(\cdot,\theta',\gamma',\lambda')$ can be removed by working with the function $T(\delta)$ from \eqref{eq_T}.\footnote{Note that, technically speaking, the dependence of \eqref{eq_supnorm_empirical} on $ P h_{\ell b}(\cdot,\theta,\gamma,\lambda)$ and $ P h_{\ell b}(\cdot,\theta',\gamma',\lambda')$ is removed using a symmetrization inequality (c.f. \cite{van1996weak} Lemma 2.3.1) and a Hoeffding-type concentration inequality, which leads exactly to the upper-bound $T(\delta)$, which holds with high probability.} Thus, the function $T(\delta)$ in Lemma \ref{lemma_subsets}---which is slightly different from $T_{n}(\delta)$ in Theorem \ref{theorem_delta_minimal}---is constructed to serve as an upper envelope of the quantity in \eqref{eq_supnorm_empirical}, for every $\delta \in [0,\delta_{0}]$, on some event $E_{n}$ with probability at least $\kappa$. 

With \eqref{eq_supnorm_empirical} replaced by its upper bound $T(\delta)$, the proof of Lemma \ref{lemma_subsets} then shows that, if $\sigma:=\mathscr{E}^{*}(\gamma)$, the following inequalities hold on the event $E_{n}$:
\begin{align}
\mathscr{E}^{*}(\gamma) &\leq \mathscr{E}_{n}(\gamma) + T(\sigma),\label{eq_e_ineq1}\\
\mathscr{E}_{n}(\gamma) &\leq \mathscr{E}^{*}(\gamma) + T(\sigma).\label{eq_e_ineq2}
\end{align}
Now note that if $\delta^{**} = T^{\sharp}(1-1/\mathfrak{a})+\varepsilon$ for any $\varepsilon>0$, then $T(\delta)\leq (1-1/\mathfrak{a}) \cdot \delta$ for every $\delta \geq \delta^{**}$. Furthermore, by construction the value of $\delta^{**}$ will be close to the smallest possible value for which this is true. Now fix any $\gamma$ with $\sigma = \mathscr{E}^{*}(\gamma) \geq \delta^{**}$. Then clearly:\note{Here (in this paragraph) the fixed point value of $\delta^{**}$ seems useful.} 
\begin{align}
T(\sigma)  \leq \left(1-\frac{1}{\mathfrak{a}}\right) \mathscr{E}^{*}(\gamma),
\end{align}
Combining this result with \eqref{eq_e_ineq1} and \eqref{eq_e_ineq2} we obtain that for any $\gamma$ satisfying $\mathscr{E}^{*}(\gamma) \geq \delta^{**}$ we have:
\begin{align}
\mathscr{E}^{*}(\gamma) &\leq \mathfrak{a} \mathscr{E}_{n}(\gamma),\label{eq_e_ineq3}\\
\mathscr{E}_{n}(\gamma) &\leq \mathfrak{b} \mathscr{E}^{*}(\gamma).
\end{align}
The remainder of the proof of Lemma \ref{lemma_subsets} is dedicated to showing that the following inequalities hold for any $\gamma \in \Gamma$ on the event $E_{n}$:\footnote{Note that the first inequality is trivial, since \eqref{eq_e_ineq3} is satisfied when $\mathscr{E}^{*}(\gamma) \geq \delta^{**}$, and if $\mathscr{E}^{*}(\gamma) \leq \delta^{**}$, then $\mathscr{E}^{*}(\gamma) \leq \mathfrak{a}\delta^{**}$, since $\mathfrak{a}>1$. The second of these inequalities is non-trivial, and relies on an auxiliary result given by Lemma \ref{lemma_gamma_hat_high_prob} in the Appendix.}
\begin{align}
\mathscr{E}^{*}(\gamma) &\leq \mathfrak{a} \left(\mathscr{E}_{n}(\gamma)\vee \delta^{**}\right),\\
\mathscr{E}_{n}(\gamma) &\leq  \mathfrak{b}\left( \mathscr{E}^{*}(\gamma)\vee \delta^{**}\right), 
\end{align}
After these inequalities are established, it is straightforward to argue that $\mathscr{G}_{n}(\delta/\mathfrak{a}) \subseteq \mathscr{G}^{*}(\delta) \subseteq \mathscr{G}_{n}(\mathfrak{b}\delta)$ on the event $E_{n}$ when $\delta \geq \mathfrak{a} \delta^{**}$. Intuitively, the proof of Theorem \ref{theorem_delta_minimal} then shows that $\mathcal{H}_{\ell b}'(\delta)$, $T(\cdot)$ (and its $\flat-$ and $\sharp-$transform) and $\delta^{**}$ defined in Lemma \ref{lemma_subsets} can be replaced with their feasible versions $\mathcal{H}_{n,\ell b}'(\delta)$, $T_{n}(\cdot)$ (and its $\flat-$ and $\sharp-$transform) and $\delta^{*}$ defined in Theorem \ref{theorem_delta_minimal}.

Theorem \ref{theorem_delta_minimal} suggests the following procedure to approximate the $\delta-$level set. The policymaker begins by computing $\mathscr{E}_{n}(\gamma)$ as a function of $\gamma$ (for example, by establishing a grid over $\Gamma$). The policymaker fixes some value $\mathfrak{a} \in (1,\infty)$ and constructs a sequence $\{\delta_{j}\}_{j=0}^{\infty}$ decreasing to zero with $(1-1/\mathfrak{a})\delta_{0} > 2\overline{H}$. In general the procedure will give a tighter bound if the sequence $\{\delta_{j}\}_{j=0}^{\infty}$ has small initial increments. The policymaker then computes $\delta^{*}>T_{n}^\sharp(1-1/\mathfrak{a})$. This is done by the following procedure:
\begin{enumerate}[label=(\roman*)]
	\item The policymaker takes $n$ i.i.d. draws of a Rademacher random variable $\xi$.
 	\item At the $j^{th}$ step (beginning at step $0$) the policymaker uses $\mathscr{E}_{n}(\gamma)$ to compute the Rademacher complexity $||\mathfrak{R}_{n}||(\mathcal{H}_{n,\ell b}'(\mathfrak{b}\delta_{j}))$ with the formula \eqref{eq_rademacher}.
 	\item The policymaker uses $\mathscr{E}_{n}(\gamma)$ to compute a uniform upper bound $\overline{H}_{n}(\delta_{j})$ for $\mathcal{H}_{n,\ell b}'(\delta_{j})$ (or she can simply use $2\overline{H}$).
 	\item The policymaker determines if there is any value $\delta \in (\delta_{j+1},\delta_{j}]$ such that $T_{n}(\delta_{j})/\delta \geq 1-1/\mathfrak{a}$. 
 	\begin{itemize}
 		\item If so, the policymaker stops and sets $\delta^{*}= \delta+\eta$, where $\eta>0$ and $\delta\in (\delta_{j+1},\delta_{j}]$ is equal to any value satisfying $T_{n}(\delta_{j})/\delta \leq 1-1/\mathfrak{a}$.
 		\item If not, the policymaker repeats steps (i) and (ii) for iteration $j+1$.
 	\end{itemize}
 \end{enumerate}
An illustration of this step is provided in Figure \ref{fig_tnm}. By Theorem \ref{theorem_delta_minimal}, the policymaker then knows that for every $\delta \geq \delta^{*}$, the $\delta-$minimal set $\mathscr{G}(\delta)$ will be contained within the sample analogue $\delta-$minimal set $\mathscr{G}_{n}(\mathfrak{b}\delta)$, and will contain the sample analogue $\delta-$minimal set $\mathscr{G}_{n}(\delta/\mathfrak{a})$ with probability at least $\kappa$. Note that the computational bottleneck in this procedure arises from repeatedly computing the Rademacher complexity.

In addition to being interesting in its own right, Theorem \ref{theorem_delta_minimal} also sheds light on the results from the previous subsection. In particular, the proof of Theorem \ref{theorem_delta_minimal} and Lemma \ref{lemma_gamma_hat_high_prob} lead to the following result, which is stated as a corollary of Theorem \ref{theorem_delta_minimal}. 

\begin{corollary}\label{corollary_gamma_hat_high_prob2}
Suppose the assumptions of Theorem \ref{theorem_delta_minimal} hold, and let $\delta^{*}$ be as in Theorem \ref{theorem_delta_minimal}. For any $\varepsilon>0$ let $\hat{\gamma}\in \Gamma$ be the policy selected by the eME decision rule. If $\delta \geq \delta^{*}\geq \varepsilon>0$, then:
\begin{align*}
\inf_{P_{Y,Z} \in \mathcal{P}_{Y,Z}} P_{Y,Z}^{\otimes n} \left( \mathscr{E}^{*}(\hat{\gamma})\leq \delta \right) \geq \kappa. 
\end{align*}
That is $\hat{\gamma} \in \mathscr{G}^{*}(\delta)$ with high probability when $\delta \geq \delta^{*}\geq \varepsilon>0$.\note{It feels like there is some way to turn this into necessary conditions for learnability. In particular, it seems that the eME rule is always eventually in the $\delta-$minimal set. If $\delta^{*}$ tends to zero, does it suffice to focus on the eME rule?}
\end{corollary}
This result shows that, if $\varepsilon\leq \delta^{*}$ then our eME rule from the previous subsection will be contained in the $\delta-$level set $\mathscr{G}^{*}(\delta)$ when $\delta \geq \delta^{*}$ with high probability. This should serve as some additional justification for using the eME rule, since it shows that, when both $\delta^{*}$ and $\varepsilon$ are small, the procedure suggested by Theorem \ref{theorem_delta_minimal} will not lead to decision rules that vastly outperform the eME rule.



\section{Conclusion}\label{section_conclusion}

The purpose of the paper is to develop a general and novel framework for bounding counterfactual quantities and for making policy decisions. Our framework is applicable in models that partially identified and/or incomplete. Furthermore, we do not require parametric distributional assumptions for the latent variables, and we allow for moment conditions that depend on latent variables. We introduce the \textit{policy transform}, and argue that many counterfactual quantities can be written as the policy transform of some function. We then introduce a preference relation that respects weak dominance, and discuss the problem of policy choice using a framework similar to the PAC model of learnability from computational learning theory. Our theoretical results are divided into those that are applicable ex-ante (i.e. before observing the sample) and ex-post (i.e. after observing the sample). For our ex-ante results, we introduce the notion of ``learning'' a policy space, and provide sufficient conditions for a policy space to be learnable. For our ex-post results, we provide theoretical guarantees on the performance of particular policy rules. Throughout the paper we also demonstrate how to apply the results to a simultaneous discrete choice example and a program evaluation example.   

There are many obvious extensions of this work that might be interesting. This paper has been particularly focused on theoretical developments, with examples serving mainly a pedagogical purpose. Further development of the examples and empirical applications are needed to clearly illustrate and fully investigate the strengths and weaknesses of the method in practice. In addition, the paper has been largely silent on implementation, which may be computationally complex in certain environments. Further development of efficient algorithms to implement the procedures is clearly needed. Finally, the relation between PAC learnability and the literature on frequentist decision theory requires further investigation and clarification. We believe all of these extensions to be fruitful avenues of future research.


\pagebreak

\bibliographystyle{apa}
\bibliography{bibfile}
\pagebreak
\begin{appendix}

\section{Preliminaries}\label{appendix_preliminaries}
\subsection{Preliminaries on Random Set Theory}

This Appendix introduces some key elements of random set theory. Since measurability issues play a significant role in random set theory, we begin by providing the definition of an Effros-measurable multifunction, and show its connection with the definition of a random set. 

\begin{definition}[Effros-Measurability, Random Set]
Let $(\Omega,\mathfrak{A},P)$ be a probability space, let $\mathcal{V}$ be a Polish space, and let $\mathcal{O}_{\mathcal{V}}$ denote the collection of all open sets on $\mathcal{V}$. A multifunction $\bm V: \Omega \to \mathfrak{F}_{\mathcal{V}}$ is called Effros-measurable if for every $A \in \mathcal{O}_{\mathcal{V}}$ we have $\bm V^{-}(A):=\{ \omega \in \Omega : \bm V(\omega) \cap A \neq \emptyset \} \in \mathfrak{A}$. An Effros-measurable closed-valued multifunction on a probability space $(\Omega,\mathfrak{A},P)$ is called a random closed set. 
\end{definition}

From this definition, we see that a random closed set is an Effros-measurable closed multifunction which takes elements from the underlying probability space to the collection of closed sets on some Polish space $\mathcal{V}$. An Effros-measurable closed multifunction is also sometimes called \textit{weakly measurable}.\footnote{See \cite{aliprantis2006infinite} Ch. 18} When the underlying probability space $(\Omega,\mathfrak{A},P)$ is complete Effros-measurability is equivalent to both (i) $\bm V^{-}(B) \in \mathfrak{A}$ for all $B \in \mathfrak{B}(\mathcal{V})$ (Borel measurability) and (ii) $\bm V^{-}(F) \in \mathfrak{A}$ for all $F \in \mathfrak{F}_{\mathcal{V}}$ (strong measurability).\footnote{See \cite{molchanov2017theory} Theorem 1.3.3, p.59. } Our main interest in the paper is in the case when $\mathcal{V}$ is a subset of finite-dimensional euclidean space, although the framework is more general.

While Effros-measurability is the proper notion of measurability for many of the results, it can be difficult to verify. There are other conditions that are sufficient for Effros measurability, but we find one condition to be particularly helpful in the examples. Let $d$ denote the metric on a Polish space $\mathcal{V}$, and let $\bm V: \Omega \to \mathfrak{F}_{\mathcal{V}}$ be a multifunction. The distance to the set $\bm V(\omega)$ on $\mathcal{V}$ is given by:
\begin{align*}
d(v, \bm V(\omega)) := \inf\{  d(v,v') : v' \in \bm V(\omega) \}.
\end{align*}
By a result of \cite{himmelberg1975measurable}, Effros measurability of the multifunction $\bm V$ is equivalent to measurability of $d(v, \bm V(\omega))$ (as a random variable from $\Omega$ to $[0,\infty]$) for each $v \in \mathcal{V}$. 

Throughout the paper it is also important to understand what it means for two random sets to be identically distributed, which is provided in the next definition.

\begin{definition}[Identically Distributed Random Sets]\label{definition_identically_distributed_random_sets}
Let $(\Omega,\mathfrak{A},P)$ be a probability space, let $\mathcal{V}$ be a Polish space. We say that two random sets $\bm V$ and $\bm V^{*}$ are identically distributed, denoted by $\bm V \sim \bm V^{*}$, if for every $A \in \mathcal{O}_{\mathcal{V}}$ we have $P(\omega : \bm V(\omega) \cap A \neq \emptyset) = P(\omega : \bm V^{*}(\omega) \cap A \neq \emptyset)$. 
\end{definition}

Finally, an important concept in random set theory is that of a selection from a random set. Intuitively, a random set $\bm V$ can be understood as a collection of random variables $V$ satisfying $V(\omega) \in \bm V(\omega)$ $P-$a.s. Such random variables are called selections from the random set $\bm V$, which is made precise in the following definition.   

\begin{definition}[Selections, Conditional Selections]\label{definition_selections}
A random element $V: \Omega \to \mathcal{V}$ is called a (measurable) selection of $\bm V$ if $V(\omega) \in \bm V(\omega)$ for $P-$almost all $\omega \in \Omega$. The family of all measurable selections of a random set $\bm V$ will be denoted by $Sel(\bm V)$.
\end{definition}

Although it is suppressed in the notation, the family of selections $Sel(\cdot)$ depends both on the distribution of the random set $\bm V$, and on the underlying probability space. Indeed, two identically distributed random sets on the same probability space may have different families of selections.\footnote{See Example 1.4.2 in \cite{molchanov2017theory}, p. 79.} However, the weak closed convex hulls of the family of selections from two random closed sets on the same probability space coincide. In addition, when the underlying probability space is non-atomic, it is not necessary to take convex hulls. See the discussion following Definition \ref{definition_collections} in the main text.

\subsection{PAC Learnability}\label{appendix_aside_PAC_learnability}

As described in the introduction, our definition of learnability is related to the definition of learnability prescribed in \cite{valiant1984theory}. It will thus be useful to understand the concept of learnability from computational learning theory. We will omit technical details in the pursuit of clarity.

In a supervised learning problem, the researcher is presumed to have an i.i.d. sample $\psi= ((y_{i},z_{i}))_{i=1}^{n}$ from the true measure $P_{Y,Z}$. The researcher is also assumed to have a class of functions $\mathcal{F}$ in mind, called the hypothesis space. The researcher's objective is to select a function $f: \mathcal{Z} \to \mathcal{Y}$, called a hypothesis (or a classifier or a predictor), from the hypothesis space $\mathcal{F}$ that can accurately predict values in $\mathcal{Y}$ given values in $\mathcal{Z}$. The performance of a given function $f \in \mathcal{F}$ is measured according to a loss function. That is, it is assumed the researcher has some function $L: \mathcal{Y} \times \mathcal{Y} \to \mathbb{R}$ such that $L(y,f(z))$ measures the loss incurred when a prediction $f(z)$ is made and the true value of the outcome is $y$. The problem of selecting a good hypothesis $f$ is then translated into the problem of choosing $f \in \mathcal{F}$ to minimize expected loss, or risk. A decision rule in this context is a measurable map $d: \Psi_{n} \to \mathcal{F}$ that selects a hypothesis from the hypothesis space; in learning theory, this decision rule is called an algorithm.

So far the reader should note a resemblance to decision problems seen in statistics and econometrics. However, important differences between the fields arise when evaluating a given statistical decision rule. In particular, computer scientists are interested in rules that achieve close to the minimum possible risk with high probability in finite samples. To define this rigorously, let $\hat{f} \in \mathcal{F}$ be the hypothesis selected by some decision rule (or algorithm) $d: \Psi_{n} \to \mathcal{F}$. Since $\hat{f} \in \mathcal{F}$ depends on the observed sample, ex-ante it will be a random variable. Now fix any values $(c,\kappa) \in \mathbb{R}_{++}\times (0,1)$. Then $\hat{f}$ closely approximates the performance of the optimal decision rule in finite samples if:
\begin{align}
\inf_{P_{Y,Z} \in \mathcal{P}_{Y,Z}} P_{Y,Z}^{\otimes n} \left( \left| \inf_{f \in \mathcal{F}} \E [L(y,f(z))] - \E [L(y,\hat{f}(z))] \right| \leq c \right) \geq \kappa,\label{eq_pac_learnable}
\end{align}
for a small value of $c \in \mathbb{R}_{+}$ and a large value of $\kappa \in (0,1)$ at sample size $n$. Here $\mathcal{P}_{Y,Z}$ is the collection of all Borel probability measures on $\mathcal{Y}\times\mathcal{Z}$, and thus the performance of a decision rule is uniform over all possible distributions $P_{Y,Z} \in \mathcal{P}_{Y,Z}$.\footnote{Note that taking the outer probability is necessary because the sampling uncertainty from the choice of $\hat{f}$ is not resolved by the inner expectation.} We can now introduce the notion of (agnostic) PAC learnability initially proposed by \cite{haussler1992decision}.

\begin{definition}[Agnostic PAC Learnability]\label{definition_pac_learnability}
A hypothesis class $\mathcal{F}$ is (agnostic) probably approximately correct (PAC) learnable with respect to the loss function $L$ if there exists a function $\zeta_{\mathcal{F}}: \mathbb{R}_{+}\times (0,1) \to \mathbb{N}$ such that, for any $(c,\kappa) \in \mathbb{R}_{++} \times (0,1) \to \mathbb{N}$, if $n \geq \zeta_{\mathcal{F}}(c,\kappa)$ then there is some decision procedure $d: \Psi_{n} \to \mathcal{F}$ such that $\hat{f}:=d(\psi)$ satisfies \eqref{eq_pac_learnable}.
\end{definition} 
\begin{remark}
This definition omits an important component of the original definition of PAC learnability found in the paper of \cite{valiant1984theory}, which also requires that the algorithm (decision rule) can be processed in polynomial time (relative to the length of its input). For some this may be a serious omission, as the requirement that an algorithm can be efficiently processed is seen as a core component of learnability in computational learning theory.\footnote{This perspective is apparent in \cite{valiant2013probably}.}    
\end{remark}

In other words, a hypothesis space is (agnostic) PAC learnable if we can guarantee that \eqref{eq_pac_learnable} holds for any choice of the pair $(c,\kappa) \in \mathbb{R}_{++} \times (0,1)$ for large enough $n$. Here $c$ is called the error tolerance parameter, and $\kappa$ is called the confidence parameter. The ``agnostic'' component of the definition refers to the fact that the hypothesis class $\mathcal{F}$ may or may not include the true labelling function $f^{*}:\mathcal{Z} \to \mathcal{Y}$; indeed, such a ``true'' labelling function may not even exist.

One major advantage of the PAC framework---relative to other frequentist methods of evaluating decision rules---is its analytical tractability and amenability to analysis via concentration inequalities, and techniques from empirical process theory. Indeed, in the case when the decision rule $d:\Psi_{n}\to \mathcal{F}$ corresponds to the empirical risk minimization rule, it is well known that PAC learnability is implied by uniform convergence (over both $\mathcal{P}_{Y,Z}$ and $\mathcal{F}$) of the empirical risk to the population risk.\footnote{See, for example, \cite{shalev2014understanding} Lemma 4.2.} In specific learning problems this uniform convergence is equivalent to learnability (see the discussion in \cite{alon1997scale} and \cite{shalev2010learnability}). This means well-developed tools in empirical process theory can be used to establish the learnability of a particular class of functions. Intuitively, whether or not a particular class of functions $\mathcal{F}$ is learnable depends on the ``complexity'' of the function class. There are various ways to measure the complexity of $\mathcal{F}$, some of which are encountered in the current paper. In general, classes that exhibit less complexity are easier to learn than classes that exhibit more complexity, and if a class of functions is too complex, it may not be learnable.\todolt{Finish the section: ``On the Relation Between PAC Learnability and Decision Theory'' Discuss stochastic dominance.}

\section{Proofs}\label{appendix_proofs}
\begin{remark}[Common Notation]\label{remark_common_notation}
To avoid repetition we introduce some common notation for use in the proofs of Theorem \ref{theorem_pampac_learnability}, Theorem \ref{thm_finite_sample1}, Lemma \ref{lemma_subsets} and Lemma \ref{lemma_gamma_hat_high_prob}. In particular, for any $\theta \in \Theta$ and $\gamma \in \Gamma$ let $\lambda^{*}(\theta,\gamma)$, and $\hat{\lambda}(\theta,\gamma)$ satisfy:
\begin{align}
Ph_{\ell b}(\cdot, \theta,\gamma,\lambda^{*}(\theta,\gamma)) = \max_{\lambda \in \Lambda} Ph_{\ell b}(\cdot,\theta,\gamma,\lambda),\label{eq_lambda_star2}\\
\P_{n}h_{\ell b}(\cdot, \theta,\gamma,\hat{\lambda}(\theta,\gamma)) = \max_{\lambda \in \Lambda}  \P_{n}h_{\ell b}(\cdot,\theta,\gamma,\lambda).\label{eq_lambda_hat2}
\end{align}
Now for any $\gamma \in \Gamma$, let $\theta^{*}$ and $\hat{\theta}$ satisfy:
\begin{align}
Ph_{\ell b}(\cdot, \theta^{*}(\gamma),\gamma,\lambda^{*}(\theta^{*}(\gamma),\gamma)) \leq  \inf_{\theta \in \Theta} Ph_{\ell b}(\cdot, \theta,\gamma,\lambda^{*}(\theta,\gamma))+\varepsilon,\label{eq_theta_star2}\\
\P_{n}h_{\ell b}(\cdot, \hat{\theta}(\gamma),\gamma,\hat{\lambda}(\hat{\theta}(\gamma),\gamma)) \leq \inf_{\theta \in \Theta} \P_{n}h_{\ell b}(\cdot,\theta,\gamma,\hat{\lambda}(\theta,\gamma)) +\varepsilon,\label{eq_theta_hat2}
\end{align}
Finally, let $\gamma^{*}$ and $\hat{\gamma}$ satisfy:
\begin{align}
Ph_{\ell b}(\cdot, \theta^{*}(\gamma^{*}),\gamma^{*},\lambda^{*}(\theta^{*}(\gamma^{*}),\gamma^{*})) \geq \sup_{\gamma \in \Gamma} Ph_{\ell b}(\cdot, \theta^{*}(\gamma),\gamma,\lambda^{*}(\theta^{*}(\gamma),\gamma)) -\varepsilon,\label{eq_gamma_star2}\\
 \P_{n}h_{\ell b}(\cdot, \hat{\theta}(\hat{\gamma}),\hat{\gamma},\hat{\lambda}(\hat{\theta}(\hat{\gamma}),\hat{\gamma})) \geq \sup_{\gamma \in \Gamma}  \P_{n}h_{\ell b}(\cdot,\hat{\theta}(\gamma),\gamma,\hat{\lambda}(\hat{\theta}(\gamma),\gamma)) -\varepsilon.\label{eq_gamma_hat2}
\end{align}
\todoltinline{Whether there exists a Borel measurable $d:\Psi_{n} \mapsto \Gamma$ such that $d(\psi) = \hat{\gamma}$ appears to be unresolved. The issue is that $\hat{\theta}:\Gamma\to \Theta$ may only be analytically measurable (e.g. see the stinchcombe and white paper). Even if $h_{\ell b}$ is a Borel measurable function of $\theta$ (which it is generally not) the composition of $h_{\ell b}$ and $\hat{\theta}(\cdot)$ may not be analytically measurable (since it is the composition of a Borel function and analytic function). This might be resolved by assuming (i) the supremum over $\Theta$ is obtained, (ii) $\Gamma$ has at most countably many points, and (iii) allowing $d$ to be analytically measurable: then the function $\P_{n}h_{\ell b}(\cdot, \hat{\theta}(\gamma),\gamma,\hat{\lambda}(\hat{\theta}(\gamma),\gamma))$ will be lower semi-analytic, and the countable maximum of lower semi-analytic functions is lower semi-analytic.}
With these definitions, it is straightforward to show:
\begin{align}
\sup_{\gamma \in \Gamma}\inf_{\theta \in \Theta} \max_{\lambda \in \Lambda} P h_{\ell b}(\cdot,\theta,\gamma,\lambda) &\leq \inf_{\theta \in \Theta} \max_{\lambda \in \Lambda} P h_{\ell b}(\cdot,\theta,\gamma^{*},\lambda) +3\varepsilon,\\
\sup_{\gamma \in \Gamma}\inf_{\theta \in \Theta} \max_{\lambda \in \Lambda} \P_{n} h_{\ell b}(\cdot,\theta,\gamma,\lambda) &\leq \inf_{\theta \in \Theta} \max_{\lambda \in \Lambda} \P_{n} h_{\ell b}(\cdot,\theta,\hat{\gamma},\lambda) +3\varepsilon.
\end{align}
Furthermore, we can always choose $\gamma^{*}$ and $\hat{\gamma}$ to satisfy:\todo{Verify}
\begin{align}
\inf_{\theta \in \Theta} \max_{\lambda \in \Lambda} P h_{\ell b}(\cdot,\theta,\hat{\gamma},\lambda) \leq \inf_{\theta \in \Theta} \max_{\lambda \in \Lambda} P h_{\ell b}(\cdot,\theta,\gamma^{*},\lambda),\\
\inf_{\theta \in \Theta} \max_{\lambda \in \Lambda} \P_{n} h_{\ell b}(\cdot,\theta,\gamma^{*},\lambda) \leq \inf_{\theta \in \Theta} \max_{\lambda \in \Lambda} \P_{n} h_{\ell b}(\cdot,\theta,\hat{\gamma},\lambda).
\end{align}

\end{remark}
\begin{remark}[Measurability]
We will not comment on measurability issues in every proof, and instead we refer readers to the discussion Appendix \ref{appendix_measurability} (namely, Proposition \ref{proposition_measurable} and Corollary \ref{corollary_universal_rademacher}). There it is shown that certain quantities in this paper that are not typically (Borel) measurable are still universally measurable. This allows us to use outer measures to resolve measurability issues, although this is left implicit in many of the proofs. However, we also note that all measurability issues can also be resolved by restricting $\Theta$ and $\Gamma$ to have at most countably many points. 
\end{remark}

\subsection{Proofs of the Main Results}

\begin{proof}[Proof of Proposition \ref{proposition_pac_weak_dominance}]
Recall by assumption we have $\gamma\mapsto \inf_{s \in \mathcal{S}} I[\varphi](\gamma,s)$ is universally measurable. By (Borel) measurability of each decision rule $d:\Psi_{n} \to \Gamma$ (and thus universal measurability), and the fact that universally measurable functions are closed under composition, this implies that the map $\psi\mapsto \inf_{s \in \mathcal{S}} I[\varphi](d(\psi),s)$ is universally measurable. The result then follows from Lemma \ref{lemma_on_weak_dominance} after noting that $\sup_{\gamma \in \Gamma}\inf_{s \in \mathcal{S}} I[\varphi](\gamma,s)$ is a constant for each $P_{Y,Z} \in \mathcal{P}_{Y,Z}$ (and thus plays the role of ``$c(P)$'' from Lemma \ref{lemma_on_weak_dominance}).
\end{proof}

\begin{proof}[Proof of Lemma \ref{lemma_lipschitz_condition}]
Fix a value of $\delta>0$ satisfying Assumption \ref{assumption_error_bound2}. We will focus on proving \eqref{eq_errorbound_2} holds, as the proof of \eqref{eq_errorbound_3} is similar. By iterated application of Lemma \ref{lemma_joint_to_marginal_selection}, \eqref{eq_errorbound_2} can be rewritten as:
\begin{align*}
\inf_{\theta^{*} \in \Theta^{*}} \int \inf_{u \in \bm G^{-}(y,z,\theta^{*})} \inf_{y^\star \in \bm G^{\star}(y,z,u,\theta,\gamma)} \varphi(v)  \, dP_{Y,Z} - \int \inf_{u \in \bm G^{-}(y,z,\theta)} \inf_{y^\star \in \bm G^{\star}(y,z,u,\theta,\gamma)} \varphi(v)  \, dP_{Y,Z} \leq C_{2} d(\theta,\Theta^{*}).
\end{align*}
Note that this inequality is trivially satisfied for any $C_{2} \geq 0$ when $\theta \in \Theta^{*}$. Thus, it suffices to focus on the case when $\theta \in \Theta_{\delta}^{*}\setminus \Theta^{*}$. Furthermore, for this latter case it suffices to find a value of $C_{2}\geq 0$ satisfying:
\begin{align*}
\int \left( \inf_{u \in \bm G^{-}(y,z,\theta_{1})} \inf_{y^\star \in \bm G^{\star}(y,z,u,\theta_{1},\gamma)} \varphi(v) - \inf_{u \in \bm G^{-}(y,z,\theta_{2})} \inf_{y^\star \in \bm G^{\star}(y,z,u,\theta_{2},\gamma)} \varphi(v) \right) \, dP_{Y,Z} \leq C_{2} d(\theta_{1},\theta_{2}),
\end{align*}
for any $\theta_{1},\theta_{2} \in \Theta_{\delta}^{*}$. However, to find $C_{2}$ in the previous display, it suffices to find $C_{2}$ such that:
\begin{align}
\inf_{u \in \bm G^{-}(y,z,\theta_{1})} \inf_{y^\star \in \bm G^{\star}(y,z,u,\theta_{1},\gamma)} \varphi(v) - \inf_{u \in \bm G^{-}(y,z,\theta_{2})} \inf_{y^\star \in \bm G^{\star}(y,z,u,\theta_{2},\gamma)} \varphi(v) \leq C_{2} d(\theta_{1},\theta_{2}),\label{eq_C2_pick}
\end{align}
$(y,z)-a.s.$ Fix any $\varepsilon>0$ and let $(y,z) \in \mathcal{Y}\times\mathcal{Z}$ be any pair (outside the null sets in \eqref{eq_lipschitz_cond1} and \eqref{eq_lipschitz_cond2}). For any $\theta_{1},\theta_{2} \in \Theta_{\delta}^{*}$ let $u_{1}^{*}$, $u_{2}^{*}$, $y_{1}^{*}$ and $y_{2}^{*}$ satisfy:
\begin{align*}
u_{1}^{*}\in \bm G^{-}(y,z,\theta_{1}),&&y_{1}^{*}\in \bm G^{\star}(y,z,u_{1}^{*},\theta_{1},\gamma),\\
u_{2}^{*}\in \bm G^{-}(y,z,\theta_{2}),&&y_{2}^{*}\in \bm G^{\star}(y,z,u_{2}^{*},\theta_{2},\gamma),
\end{align*}
and:
\begin{align*}
\varphi(y,z,u_{1}^{*},y_{1}^{*}) \leq \inf_{u \in \bm G^{-}(y,z,\theta_{1})} \inf_{y^\star \in \bm G^{\star}(y,z,u,\theta_{1},\gamma)} \varphi(v)+\varepsilon,\\
\varphi(y,z,u_{2}^{*},y_{2}^{*}) \leq \inf_{u \in \bm G^{-}(y,z,\theta_{2})} \inf_{y^\star \in \bm G^{\star}(y,z,u,\theta_{2},\gamma)} \varphi(v)+\varepsilon.
\end{align*}
For simplicity we will denote $v_{1}^{*}:=(y,z,u_{1}^{*},y_{1}^{*})$ and $v_{2}^{*}:=(y,z,u_{2}^{*},y_{2}^{*})$. Now, by Proposition 3C.1 in \cite{dontchev2009implicit}, condition \eqref{eq_lipschitz_cond2} implies:
\begin{align*}
d_{H}(\bm G^{\star}(y,z,u,\theta_{1},\gamma),\bm G^{\star}(y,z,u,\theta_{2},\gamma)) \leq \ell_{2} d(\theta_{1},\theta_{2}), &&\forall \theta_{1},\theta_{2} \in \Theta_{\delta}^{*}
\end{align*}
$(y,z,u)-$a.s.\footnote{Recall the Hausdorff distance between two non-empty subsets $A$ and $B$ of a metric space $(\mathcal{X},d)$ is given by:
\begin{align*}
d_{H}(A,B):= \max\left\{\sup_{a\in A}\inf_{b \in B} d(a,b), \sup_{b\in B}\inf_{a \in A} d(a,b)\right\}.
\end{align*}
} Thus, since $y_{2}^{*}\in \bm G^{\star}(y,z,u,\theta_{2},\gamma)$ by assumption, there exists $y_{1} \in \bm G^{\star}(y,z,u,\theta_{1},\gamma)$ such that $d(y_{1},y_{2}^{*}) \leq \ell_{2} d(\theta_{1},\theta_{2})$. Furthermore, by Proposition 3C.1 in \cite{dontchev2009implicit}, condition \eqref{eq_lipschitz_cond1} implies:
\begin{align*}
d_{H}(\bm G^{-}(y,z,\theta_{1}),\bm G^{-}(y,z,\theta_{2})) \leq \ell_{1} d(\theta_{1},\theta_{2}), &&\forall \theta_{1},\theta_{2} \in \Theta_{\delta}^*.
\end{align*}
Thus, since $u_{2}^{*}\in \bm G^{-}(y,z,\theta_{2})$ by assumption, there exists $u_{1} \in \bm G^{-}(y,z,\theta_{1})$ such that $d(u_{1},u_{2}^{*}) \leq \ell_{1} d(\theta_{1},\theta_{2})$. Now let us define $v_{1}:=(y,z,u_{1},y_{1})$. Then we have:
\begin{align*}
&\inf_{u \in \bm G^{-}(y,z,\theta_{1})} \inf_{y^\star \in \bm G^{\star}(y,z,u,\theta_{1},\gamma)} \varphi(v) - \inf_{u \in \bm G^{-}(y,z,\theta_{2})} \inf_{y^\star \in \bm G^{\star}(y,z,u,\theta_{2},\gamma)}\varphi(v)\\
&\qquad\qquad\leq \varphi(v_{1}^{*}) - \varphi(v_{2}^{*}) +\varepsilon \\
&\qquad\qquad\leq \varphi(v_{1}) - \varphi(v_{2}^{*}) +2\varepsilon\\
&\qquad\qquad\leq L_{\varphi} d((y_{1},u_{1}),(u_{2}^{*},y_{2}^{*})) +2\varepsilon\\
&\qquad\qquad\leq L_{\varphi} \max\{d(y_{1},y_{2}^{*}),d(u_{1},u_{2}^{*})\} +2\varepsilon\\
&\qquad\qquad\leq L_{\varphi} \max\{\ell_{1},\ell_{2}\}d(\theta_{1},\theta_{2}) +2\varepsilon,
\end{align*}
which holds for all $\theta_{1}, \theta_{2} \in \Theta_{\delta}^{*}$.\footnote{Here we take the product metric as the sup metric; that is, if $(\mathcal{X},d)$ and $(\mathcal{X}',d')$ are two metric spaces, then the product metric $d_{\infty}$ on $\mathcal{X} \times \mathcal{X}'$ is defined as $d_{\infty}((x_{1},x_{1}'),(x_{2},x_{2}')) = \max\left\{d(x_{1},x_{2}),d'(x_{1}',x_{2}')  \right\}$.} Since $\varepsilon>0$ is arbitrary, we conclude that $C_{2}$ in \eqref{eq_C2_pick} can be taken equal to $L_{\varphi} \max\{\ell_{1},\ell_{2}\}$. This completes the proof.
\end{proof}

\begin{proof}[Proof of Theorem \ref{thm_cortes}]\label{proof_thm_cortes}
We will show the lower bound, as the proof for the upper bound is symmetric. We will prove the following sequence of equalities and inequalities:
\begin{align}
I[\varphi](\gamma)&:=\int \varphi(v) \, dP_{V_{\gamma}} \nonumber\\
&\geq \inf_{\theta \in \Theta^{*}}\inf_{P_{U|Y,Z} \in \mathcal{P}_{U|Y,Z}(\theta)}\inf_{P_{Y_{\gamma}^\star|Y,Z,U} \in \mathcal{P}_{Y_{\gamma}^\star|Y,Z,U}(\theta,\gamma)} \int \varphi(v)  \, dP_{V_{\gamma}}\label{eq1}\\
&=\inf_{\theta \in \Theta}  \inf_{P_{U|Y,Z} \in \mathcal{P}_{U|Y,Z}(\theta)} \Bigg(\inf_{P_{Y_{\gamma}^\star|Y,Z,U} \in \mathcal{P}_{Y_{\gamma}^\star|Y,Z,U}(\theta,\gamma)} \int \varphi(v)  \, dP_{V_{\gamma}}\nonumber\\
&\qquad\qquad\qquad\qquad\qquad\qquad\qquad\qquad\qquad\qquad+ \mu^{*}\sum_{j=1}^{J}  \lambda_{j}^{\ell b}(\theta,P_{Y,Z}) \E_{P_{Y,Z,U}} [ m_{j}(y,z,u,\theta)]\Bigg)\label{eq2}\\
&=\inf_{\theta \in \Theta}  \inf_{P_{U|Y,Z} \in \mathcal{P}_{U|Y,Z}(\theta)} \Bigg(\int  \inf_{y^{\star}  \in \bm G^{\star}(y,z,u,\theta,\gamma)}\varphi(v)  \, dP_{Y,Z,U}\nonumber\\
&\qquad\qquad\qquad\qquad\qquad\qquad\qquad\qquad\qquad\qquad+ \mu^{*}\sum_{j=1}^{J} \lambda_{j}^{\ell b}(\theta,P_{Y,Z})  \E_{P_{Y,Z,U}} [ m_{j}(y,z,u,\theta)]\Bigg)\label{eq3}\\
&=\inf_{\theta \in \Theta}  \int \Bigg(\inf_{u \in \bm G^{-}(y,z,\theta)}  \inf_{y^{\star}  \in \bm G^{\star}(y,z,u,\theta,\gamma)}\varphi(v) + \mu^{*}\sum_{j=1}^{J} \lambda_{j}^{\ell b}(\theta,P_{Y,Z}) m_{j}(y,z,u,\theta)\Bigg) \, dP_{Y,Z}\label{eq4}\\
&=\inf_{\theta \in \Theta} \max_{\lambda_{j} \in \{0,1\}} \int \Bigg(\inf_{u \in \bm G^{-}(y,z,\theta)}  \inf_{y^{\star}  \in \bm G^{\star}(y,z,u,\theta,\gamma)}\varphi(v) + \mu^{*}\sum_{j=1}^{J} \lambda_{j} m_{j}(y,z,u,\theta)\Bigg) \, dP_{Y,Z}\label{eq5}.
\end{align}
Inequality \eqref{eq1} is obvious. Equality \eqref{eq2} follows from Lemma \ref{lemma_error_bound}. Equalities \eqref{eq3} and \eqref{eq4} follow from Lemma \ref{lemma_joint_to_marginal_selection}. Finally, \eqref{eq5} follows from Lemma \ref{lemma_multiplier_switch}.
\end{proof}

\begin{proof}[Proof of Theorem \ref{theorem_pampac_learnability}]
Let $\mathcal{F}$ be a class of real-valued functions, and let $\psi = ((y_{i},z_{i}))_{i=1}^{n}$ denote a particular sample vector taking values in the sample space $\Psi_{n}$. For any $f ,f' \in \mathcal{F}$ define the norm:
\begin{align*}
||f-f'||_{\psi,2} := \left( \sum_{i=1}^{n} \left( f(y_{i},z_{i}) - f'(y_{i},z_{i}) \right)^{2}  \right)^{1/2} .
\end{align*}
Recall that:
\begin{align*}
h_{\ell b}(y,z,\theta,\gamma,\lambda):= \inf_{u \in \bm G^{-}(y,z,\theta)}  \inf_{y^{\star}  \in \bm G^{\star}(y,z,u,\theta,\gamma)} \Bigg( \varphi(v) + \mu^{*}\sum_{j=1}^{J}  \lambda_{j} m_{j}(y,z,u,\theta)\Bigg).
\end{align*}
For notational simplicity we will define:
\begin{align*}
\P_{n}h_{\ell b}(\cdot, \theta,\gamma,\lambda) &:= \frac{1}{n}\sum_{i=1}^{n} \inf_{u_{i} \in \bm G^{-}(y_{i},z_{i},\theta)}  \inf_{y_{i}^{\star}  \in \bm G^{\star}(y_{i},z_{i},u_{i},\theta,\gamma)} \Bigg( \varphi(v_{i}) + \mu^{*}\sum_{j=1}^{J}  \lambda_{j}  m_{j}(y_{i},z_{i},u_{i},\theta)\Bigg),\\
P h_{\ell b}(\cdot, \theta,\gamma,\lambda) &:= \int \inf_{u \in \bm G^{-}(y,z,\theta)}  \inf_{y^{\star}  \in \bm G^{\star}(y,z,u,\theta,\gamma)} \Bigg( \varphi(v) + \mu^{*}\sum_{j=1}^{J}  \lambda_{j} m_{j}(y,z,u,\theta)\Bigg)\,dP_{Y,Z}.
\end{align*}
For any decision rule $d:\Psi_{n} \to \Gamma$ and any $P_{Y,Z} \in \mathcal{P}_{Y,Z}$, we have by Markov's inequality and Theorem \ref{thm_cortes}:\footnote{To be mindful of measurability issues, we can use the outer-measures version of Markov's inequality given in Lemma 6.10 in \cite{kosorok2008introduction}.}
\begin{align}
P_{Y,Z}^{\otimes n} \left( \sup_{\gamma \in \Gamma}\inf_{s \in \mathcal{S}} I[\varphi](\gamma,s)  - \inf_{s\in \mathcal{S}} I[\varphi](d(\psi),s) \geq c \right) &\leq \frac{1}{c} \E \left(\sup_{\gamma \in \Gamma}\inf_{s \in \mathcal{S}} I[\varphi](\gamma,s)  - \inf_{s\in \mathcal{S}} I[\varphi](d(\psi),s)\right)\nonumber\\ 
&=\frac{1}{c} \E \left(\sup_{\gamma \in \Gamma} I_{\ell b}[\varphi](\gamma)  - I_{\ell b}[\varphi](d(\psi))\right).\label{eq_markov}
\end{align}
Now note by symmetrization (e.g. \cite{van1996weak} Lemma 2.3.1) we have:
\begin{align}
&\sup_{\gamma \in \Gamma}\sup_{\theta \in \Theta} \max_{\lambda \in \Lambda}\bigg| \E \left(\P_{n}h_{\ell b}(\cdot, \theta,\gamma,\lambda) -  Ph_{\ell b}(\cdot, \theta,\gamma,\lambda) \right)\bigg|\nonumber\\
&\qquad\qquad\qquad\qquad\leq  \E  \sup_{\gamma \in \Gamma}\sup_{\theta \in \Theta} \max_{\lambda \in \Lambda}\bigg| \P_{n}h_{\ell b}(\cdot, \theta,\gamma,\lambda) -  Ph_{\ell b}(\cdot, \theta,\gamma,\lambda) \bigg| \leq 2\E  ||\mathfrak{R}_{n}||(\mathcal{H}_{\ell b}),\label{eq_ml12}
\end{align}
where the final outer expectation is a joint expectation that is also taken over the Rademacher random variables. Now let $\lambda^{*}(\theta,\gamma)$, $\hat{\lambda}(\theta,\gamma)$,  $\theta^{*}(\gamma)$, $\hat{\theta}(\gamma)$, $\gamma^{*}$ and $\hat{\gamma}$ be as in Remark \ref{remark_common_notation}, and set $d(\psi)=\hat{\gamma}$.\todolt{It seems this may not always be possible, since there may not exist a Borel measurable $d$ satisfying this requirement.} Then we have:
\begin{align*}
&\E I_{\ell b}[\varphi](d(\psi))\\
&=\E \inf_{\theta \in \Theta} \max_{\lambda \in \Lambda} Ph_{\ell b}(\cdot, \theta,d(\psi),\lambda), &&\text{(by Theorem \ref{thm_cortes}),}\\
&=\E \inf_{\theta \in \Theta} Ph_{\ell b}(\cdot, \theta,d(\psi),\lambda^{*}(\theta,d(\psi))), &&\text{(since $\lambda^{*}$ is optimal at $P$ for any $(\theta,\gamma)$),}\\
&=\E Ph_{\ell b}(\cdot, \theta^{*}(d(\psi)),d(\psi),\lambda^{*}(\theta^{*},d(\psi)))-\varepsilon, &&\text{(since $\theta^{*}$ is $\varepsilon-$optimal at $(P,\lambda^{*})$ for any $\gamma$),}\\
&\geq \E  Ph_{\ell b}(\cdot, \theta^{*}(d(\psi)),d(\psi),\hat{\lambda}(\theta^{*}(d(\psi)),d(\psi)))-\varepsilon, &&\text{(since $\lambda^{*}$ was optimal at $P$ for any $(\theta,\gamma)$),}\\
&\geq \E  \P_{n}h_{\ell b}(\cdot, \theta^{*}(d(\psi)),d(\psi),\hat{\lambda}(\theta^{*}(d(\psi)),d(\psi))) -2\E  ||\mathfrak{R}_{n}||(\mathcal{H}_{\ell b})-\varepsilon, &&\text{(by \eqref{eq_ml12}),}\\
&\geq \E  \P_{n}h_{\ell b}(\cdot, \hat{\theta}(d(\psi)),d(\psi),\hat{\lambda}(\hat{\theta}(d(\psi)),d(\psi))) -2\E  ||\mathfrak{R}_{n}||(\mathcal{H}_{\ell b})-2\varepsilon, &&\text{(since $\hat{\theta}$ is $\varepsilon$-optimal at $(\P_{n},\hat{\lambda})$ for any $\gamma$),}\\
&\geq \E  \P_{n}h_{\ell b}(\cdot, \hat{\theta}(\gamma^{*}),\gamma^{*},\hat{\lambda}(\hat{\theta}(\gamma^{*}),\gamma^{*})) -2\E  ||\mathfrak{R}_{n}||(\mathcal{H}_{\ell b})-3\varepsilon, &&\text{(since $d(\psi)$ was $\varepsilon$-optimal at $(\P_{n},\hat{\lambda},\hat{\theta})$),}\\
&\geq \E  \P_{n}h_{\ell b}(\cdot, \hat{\theta}(\gamma^{*}),\gamma^{*},\lambda^{*}(\hat{\theta}(\gamma^{*}),\gamma^{*})) -2\E  ||\mathfrak{R}_{n}||(\mathcal{H}_{\ell b})-3\varepsilon, &&\text{(since $\hat{\lambda}$ was optimal at $\P_{n}$ for any $(\theta,\gamma)$),}\\
&\geq \E  Ph_{\ell b}(\cdot, \hat{\theta}(\gamma^{*}),\gamma^{*},\lambda^{*}(\hat{\theta}(\gamma^{*}),\gamma^{*})) -4\E  ||\mathfrak{R}_{n}||(\mathcal{H}_{\ell b}) - 3\varepsilon, &&\text{(by \eqref{eq_ml12}),}\\
&\geq \E  Ph_{\ell b}(\cdot, \theta^{*}(\gamma^{*}),\gamma^{*},\lambda^{*}(\theta^{*}(\gamma^{*}),\gamma^{*})) -4\E  ||\mathfrak{R}_{n}||(\mathcal{H}_{\ell b}) - 4\varepsilon, &&\text{(since $\theta^{*}$ is $\varepsilon$-optimal at $(P,\lambda^{*})$ for any $\gamma$),}\\
&\geq \E  \sup_{\gamma \in \Gamma} Ph_{\ell b}(\cdot, \theta^{*}(\gamma),\gamma,\lambda^{*}(\theta^{*}(\gamma),\gamma)) -4\E  ||\mathfrak{R}_{n}||(\mathcal{H}_{\ell b}) - 5\varepsilon, &&\text{(since $\gamma^{*}$ was $\varepsilon$-optimal at $(P,\lambda^{*},\theta^{*})$),}\\
&\geq \E  \sup_{\gamma \in \Gamma} \inf_{\theta \in \Theta} Ph_{\ell b}(\cdot, \theta,\gamma,\lambda^{*}(\theta,\gamma)) -4\E  ||\mathfrak{R}_{n}||(\mathcal{H}_{\ell b}) - 5\varepsilon, &&\text{(since $\theta^{*}$ was $\varepsilon$-optimal at $(P,\lambda^{*})$ for any $\gamma$),}\\
&\geq  \E \sup_{\gamma \in \Gamma} \inf_{\theta \in \Theta}\max_{\lambda \in \Lambda} Ph_{\ell b}(\cdot, \theta,\gamma,\lambda) -4\E  ||\mathfrak{R}_{n}||(\mathcal{H}_{\ell b}) - 5\varepsilon, &&\text{(since $\lambda^{*}$ was optimal at $P$ for any $(\theta,\gamma)$),}\\
&\geq \E  \sup_{\gamma \in \Gamma} I_{\ell b}[\varphi](\gamma) -4\E  ||\mathfrak{R}_{n}||(\mathcal{H}_{\ell b}) - 5\varepsilon, &&\text{(by Theorem \ref{thm_cortes}).}
\end{align*}
Since $\varepsilon>0$ can be taken arbitrarily small, we conclude that:
\begin{align}
\E  \left(\sup_{\gamma \in \Gamma}I[\varphi](\gamma)  -I[\varphi](d(\psi))\right)  \leq 4\E  ||\mathfrak{R}_{n}||(\mathcal{H}_{\ell b}).\label{eq_symmetrize}
\end{align}
It thus suffices to bound the Rademacher complexity, given by:
\begin{align*}
&\E  ||\mathfrak{R}_{n}||(\mathcal{H}_{\ell b}) \\
&= \E  \left( \sup_{\gamma \in \Gamma} \sup_{\theta \in \Theta}\max_{\lambda \in \Lambda} \left| \frac{1}{n}\sum_{i=1}^{n} \xi_{i} \left(\inf_{u_{i} \in \bm G^{-}(y_{i},z_{i},\theta)}   \Bigg( \inf_{y_{i}^{\star}  \in \bm G^{\star}(y_{i},z_{i},u_{i},\theta,\gamma)} \varphi(v_{i}) + \mu^{*}\sum_{j=1}^{J}  \lambda_{j}  m_{j}(y_{i},z_{i},u_{i},\theta)\Bigg)\right)\right|\right). 
\end{align*}
If $\mathcal{H}_{\ell b}$ is not closed under symmetry, then redefine it as $\mathcal{H}_{\ell b}\cup (- \mathcal{H}_{\ell b})$; for our purposes this is without loss of generality, since this operation can only increase the value of $\E  ||\mathfrak{R}_{n}||(\mathcal{H}_{\ell b})$. We then have from Lemma \ref{lemma_bartlett} that for any $\varepsilon>0$:
\begin{align}
\E  ||\mathfrak{R}_{n}||(\mathcal{H}_{\ell b}) \leq \frac{2\varepsilon}{\sqrt{n}} + 2\text{Diam}_{\psi,2}(\mathcal{H}_{\ell b})  \sqrt{ \frac{\log  N(\varepsilon,\mathcal{H}_{\ell b},||\cdot||_{\psi,2})}{n}}.\label{eq_rademacher_bound}
\end{align}
Since the class of functions $\mathcal{H}_{\ell b}$ is uniformly bounded, we have $\text{Diam}_{\psi,2}(\mathcal{H}_{\ell b})<\infty$. It remains to bound the metric entropy. To do so, we will define:
\begin{align}
&\mathcal{H}_{I}:=\Bigg\{ h(\cdot,u,\theta,\gamma,\lambda): \mathcal{Y}\times \mathcal{Z} \to \mathbb{R} : \nonumber\\
&\qquad\qquad\qquad\qquad h(y,z,u,\theta,\gamma)= \inf_{y^{\star}  \in \bm G^{\star}(y,z,u,\theta,\gamma)} \varphi(v) + \mu^{*}\sum_{j=1}^{J}  \lambda_{j}  m_{j}(y,z,u,\theta),\nonumber\\
&\qquad\qquad\qquad\qquad\qquad\qquad\qquad\qquad\qquad\qquad\qquad\qquad\qquad\qquad(u,\theta,\gamma,\lambda) \in \mathcal{U} \times \Theta \times \Gamma \times \Lambda\Bigg\},\\
\nonumber\\
&\mathcal{H}_{II}:=\left\{ h(\cdot,u,\theta,\gamma): \mathcal{Y}\times \mathcal{Z} \to \mathbb{R} : h(y,z,u,\theta,\gamma)=\inf_{y^{\star}  \in \bm G^{\star}(y,z,u,\theta,\gamma)} \varphi(v), \quad (u,\theta,\gamma) \in \mathcal{U} \times \Theta \times \Gamma  \right\},\\
\nonumber\\
&\mathcal{H}_{III}:=\left\{ h(\cdot,u,y^{\star}): \mathcal{Y}\times \mathcal{Z} \to \mathbb{R} : h(y,z,u,y^\star)=\varphi(y,z,u,y^\star), \quad (u,y^\star) \in \mathcal{U} \times \mathcal{Y}^{\star} \right\},\\
\nonumber\\
&\mathcal{H}_{IV}:=\left\{ h(\cdot,u,\theta,\lambda): \mathcal{Y}\times \mathcal{Z} \to \mathbb{R} : h(y,z,u,\theta)=\sum_{j=1}^{J}  \lambda_{j} m_{j}(y,z,u,\theta), \,\, (u,\theta,\lambda) \in \mathcal{U} \times \Theta \times \Lambda \right\}.
\end{align}
By Lemma \ref{lemma_inf_covering}, we have:
\begin{align*}
N (\varepsilon,\mathcal{H}_{\ell b},||\cdot||_{\psi,2})\leq N (\varepsilon/2,\mathcal{H}_{I},||\cdot||_{\psi,2}).
\end{align*}
By Lemma \ref{lemma_useful_for_donsker} we also have:
\begin{align*}
N (\varepsilon/2,\mathcal{H}_{I},||\cdot||_{\psi,2}) \leq N (\varepsilon/2,\mathcal{H}_{II},||\cdot||_{\psi,2})  N (\varepsilon/2,\mathcal{H}_{IV},||\cdot||_{\psi,2}). 
\end{align*}
Applying Lemma \ref{lemma_inf_covering} again we have:
\begin{align*}
N (\varepsilon/2,\mathcal{H}_{II},||\cdot||_{\psi,2}) \leq N (\varepsilon/4,\mathcal{H}_{III},||\cdot||_{\psi,2}).
\end{align*}
Finally, from iterated application of Lemma \ref{lemma_useful_for_donsker}:
\begin{align*}
N (\varepsilon/2,\mathcal{H}_{IV},||\cdot||_{\psi,2}) \leq \prod_{j=1}^{J} N (\varepsilon/(2J),\mathcal{M}_{j},||\cdot||_{\psi,2}),
\end{align*}
We conclude that:
\begin{align*}
\log N (\varepsilon,\mathcal{H}_{\ell b},||\cdot||_{\psi,2})&\leq \log N (\varepsilon/4,\mathcal{H}_{III},||\cdot||_{\psi,2}) + \sum_{j=1}^{J} \log N (\varepsilon/(2J),\mathcal{M}_{j},||\cdot||_{\psi,2})\\
&\leq \sup_{Q \in \mathcal{Q}_{n}} \log N (\varepsilon/4,\mathcal{H}_{III},||\cdot||_{Q,2}) + \sum_{j=1}^{J} \sup_{Q \in \mathcal{Q}_{n}} \log N (\varepsilon/(2J),\mathcal{M}_{j},||\cdot||_{Q,2}),
\end{align*}
with the supremum taken over all discrete probability measures $\mathcal{Q}_{n}$ on $\mathcal{X}$ with atoms that have probabilities that are integer multiples of $1/n$. Since by assumption $\mathcal{H}_{III}$ and $\mathcal{M}_{j}$ satisfy the entropy growth condition, the right side of the previous display is of order $o(n)$. Combining this with \eqref{eq_rademacher_bound}, we see that for any $(c,\kappa)$ pair, there exists some $n$ such that $4 \E ||\mathfrak{R}_{n}||(\mathcal{H}_{\ell b}) \leq c(1-\kappa)$. Combining this with \eqref{eq_symmetrize} and \eqref{eq_markov}, the proof is complete.

\end{proof}

\begin{proof}[Proof of Theorem \ref{thm_finite_sample1}]\label{proof_thm_finite_sample1}
Recall that:
\begin{align*}
h_{\ell b}(y,z,\theta,\gamma,\lambda):= \inf_{u \in \bm G^{-}(y,z,\theta)}  \inf_{y^{\star}  \in \bm G^{\star}(y,z,u,\theta,\gamma)} \Bigg( \varphi(v) + \mu^{*}\sum_{j=1}^{J}  \lambda_{j} m_{j}(y,z,u,\theta)\Bigg).
\end{align*}
For notational simplicity we will define:
\begin{align*}
\P_{n}h_{\ell b}(\cdot, \theta,\gamma,\lambda) &:= \frac{1}{n}\sum_{i=1}^{n} \inf_{u_{i} \in \bm G^{-}(y_{i},z_{i},\theta)}  \inf_{y_{i}^{\star}  \in \bm G^{\star}(y_{i},z_{i},u_{i},\theta,\gamma)} \Bigg( \varphi(v_{i}) + \mu^{*}\sum_{j=1}^{J}  \lambda_{j}  m_{j}(y_{i},z_{i},u_{i},\theta)\Bigg),\\
P h_{\ell b}(\cdot, \theta,\gamma,\lambda) &:= \int \inf_{u \in \bm G^{-}(y,z,\theta)}  \inf_{y^{\star}  \in \bm G^{\star}(y,z,u,\theta,\gamma)} \Bigg( \varphi(v) + \mu^{*}\sum_{j=1}^{J}  \lambda_{j}  m_{j}(y,z,u,\theta)\Bigg)\,dP_{Y,Z}.
\end{align*}
We claim that it suffices to set $c_{n}(\kappa)=2\tilde{c}_{n}(\psi,\kappa) + 5\varepsilon$, where $\tilde{c}_{n}(\psi,\kappa)$ satisfies:
\begin{align}
\sup_{\gamma \in \Gamma}\sup_{\theta \in \Theta} \max_{\lambda \in \Lambda} \bigg| \P_{n}h_{\ell b}(\cdot, \theta,\gamma,\lambda) -  Ph_{\ell b}(\cdot, \theta,\gamma,\lambda) \bigg| \leq \tilde{c}_{n}(\psi,\kappa),\label{eq_ml1}
\end{align}
with probability at least $\kappa/2$. Let $\lambda^{*}(\theta,\gamma)$, $\hat{\lambda}(\theta,\gamma)$,  $\theta^{*}(\gamma)$, $\hat{\theta}(\gamma)$, $\gamma^{*}$ and $\hat{\gamma}$ be as in Remark \ref{remark_common_notation} and set $d(\psi)=\hat{\gamma}$.\todolt{It seems this may not always be possible, since there may not exist a Borel measurable $d$ satisfying this requirement.} Then we have:
\begin{align*}
&I_{\ell b}[\varphi](d(\psi))\\
&=\inf_{\theta \in \Theta} \max_{\lambda \in \Lambda} Ph_{\ell b}(\cdot, \theta,d(\psi),\lambda), &&\text{(by Theorem \ref{thm_cortes}),}\\
&=\inf_{\theta \in \Theta} Ph_{\ell b}(\cdot, \theta,d(\psi),\lambda^{*}(\theta,d(\psi))), &&\text{(since $\lambda^{*}$ is optimal at $P$ for any $(\theta,\gamma)$),}\\
&=Ph_{\ell b}(\cdot, \theta^{*}(d(\psi)),d(\psi),\lambda^{*}(\theta^{*},d(\psi)))-\varepsilon, &&\text{(since $\theta^{*}$ is $\varepsilon-$optimal at $(P,\lambda^{*})$ for any $\gamma$),}\\
&\geq  Ph_{\ell b}(\cdot, \theta^{*}(d(\psi)),d(\psi),\hat{\lambda}(\theta^{*}(d(\psi)),d(\psi)))-\varepsilon, &&\text{(since $\lambda^{*}$ was optimal at $P$ for any $(\theta,\gamma)$),}\\
&\geq_{(\kappa/2)} \P_{n}h_{\ell b}(\cdot, \theta^{*}(d(\psi)),d(\psi),\hat{\lambda}(\theta^{*}(d(\psi)),d(\psi))) -\tilde{c}_{n}(\psi,\kappa)-\varepsilon, &&\text{(by \eqref{eq_ml1}),}\\
&\geq \P_{n}h_{\ell b}(\cdot, \hat{\theta}(d(\psi)),d(\psi),\hat{\lambda}(\hat{\theta}(d(\psi)),d(\psi))) -\tilde{c}_{n}(\psi,\kappa)-2\varepsilon, &&\text{(since $\hat{\theta}$ is $\varepsilon$-optimal at $(\P_{n},\hat{\lambda})$ for any $\gamma$),}\\
&\geq \P_{n}h_{\ell b}(\cdot, \hat{\theta}(\gamma^{*}),\gamma^{*},\hat{\lambda}(\hat{\theta}(\gamma^{*}),\gamma^{*})) -\tilde{c}_{n}(\psi,\kappa)-3\varepsilon, &&\text{(since $d(\psi)$ was $\varepsilon$-optimal at $(\P_{n},\hat{\lambda},\hat{\theta})$),}\\
&\geq \P_{n}h_{\ell b}(\cdot, \hat{\theta}(\gamma^{*}),\gamma^{*},\lambda^{*}(\hat{\theta}(\gamma^{*}),\gamma^{*})) -\tilde{c}_{n}(\psi,\kappa)-3\varepsilon, &&\text{(since $\hat{\lambda}$ was optimal at $\P_{n}$ for any $(\theta,\gamma)$),}\\
&\geq_{(\kappa/2)}  Ph_{\ell b}(\cdot, \hat{\theta}(\gamma^{*}),\gamma^{*},\lambda^{*}(\hat{\theta}(\gamma^{*}),\gamma^{*})) -2\tilde{c}_{n}(\psi,\kappa) - 3\varepsilon, &&\text{(by \eqref{eq_ml1}),}\\
&\geq Ph_{\ell b}(\cdot, \theta^{*}(\gamma^{*}),\gamma^{*},\lambda^{*}(\theta^{*}(\gamma^{*}),\gamma^{*})) -2\tilde{c}_{n}(\psi,\kappa) - 4\varepsilon, &&\text{(since $\theta^{*}$ is $\varepsilon$-optimal at $(P,\lambda^{*})$ for any $\gamma$),}\\
&\geq\sup_{\gamma \in \Gamma} Ph_{\ell b}(\cdot, \theta^{*}(\gamma),\gamma,\lambda^{*}(\theta^{*}(\gamma),\gamma)) -2\tilde{c}_{n}(\psi,\kappa) - 5\varepsilon, &&\text{(since $\gamma^{*}$ was $\varepsilon$-optimal at $(P,\lambda^{*},\theta^{*})$),}\\
&\geq  \sup_{\gamma \in \Gamma} \inf_{\theta \in \Theta} Ph_{\ell b}(\cdot, \theta,\gamma,\lambda^{*}(\theta,\gamma)) -2\tilde{c}_{n}(\psi,\kappa) - 5\varepsilon, &&\text{(since $\theta^{*}$ was $\varepsilon$-optimal at $(P,\lambda^{*})$ for any $\gamma$),}\\
&\geq \sup_{\gamma \in \Gamma} \inf_{\theta \in \Theta}\max_{\lambda \in \Lambda} Ph_{\ell b}(\cdot, \theta,\gamma,\lambda) -2\tilde{c}_{n}(\psi,\kappa) - 5\varepsilon, &&\text{(since $\lambda^{*}$ was optimal at $P$ for any $(\theta,\gamma)$),}\\
&\geq \sup_{\gamma \in \Gamma} I_{\ell b}[\varphi](\gamma) -2\tilde{c}_{n}(\psi,\kappa) - 5\varepsilon, &&\text{(by Theorem \ref{thm_cortes}).}
\end{align*}
where each inequality ``$\geq_{(\kappa/2)}$'' holds with probability at least $\kappa/2$. Note that this shows:
\begin{align*}
\sup_{\gamma \in \Gamma}  I_{\ell b}[\varphi](\gamma)- I_{\ell b}[\varphi](\hat{\gamma}) \leq 2\tilde{c}_{n}(\psi,\kappa) +5\varepsilon,
\end{align*}
with probability at least $\kappa$. To satisfy \eqref{eq_ml1} it clearly suffices to choose $\tilde{c}_{n}(\psi,\kappa)$ to satisfy:
\begin{align}
\sup_{P_{Y,Z} \in \mathcal{P}_{Y,Z}} P_{Y,Z}^{\otimes n}\bigg(\sup_{\gamma \in \Gamma}\sup_{\theta \in \Theta}\max_{\lambda \in \Lambda} \bigg| \P_{n}h_{\ell b}(\cdot, \theta,\gamma,\lambda) -  Ph_{\ell b}(\cdot, \theta,\gamma,\lambda) \bigg|\geq \tilde{c}_{n}(\psi,\kappa)\bigg) \leq 1-\kappa/2.
\end{align}
From \cite{koltchinskii2011oracle} Theorem 4.6 we have for any $t>0$:
\begin{align*}
\sup_{P_{Y,Z} \in \mathcal{P}_{Y,Z}} P_{Y,Z}^{\otimes n}\left( \sup_{\gamma \in \Gamma}\sup_{\theta \in \Theta}\max_{\lambda \in \Lambda} \bigg| \P_{n}h_{\ell b}(\cdot, \theta,\gamma,\lambda) -  Ph_{\ell b}(\cdot, \theta,\gamma,\lambda) \bigg| \geq 2 ||\mathfrak{R}_{n}||(\mathcal{H}_{\ell b}) + \frac{3t \overline{H} }{\sqrt{n}} \right) &\leq \exp \left( -\frac{t^{2}}{2} \right).
\end{align*}
Now set:
\begin{align*}
\tilde{c}_{n}(\psi,\kappa)&=  2||\mathfrak{R}_{n}||(\mathcal{H}_{\ell b}) + \sqrt{\frac{18 \ln(2/(2-\kappa)) \overline{H}^{2}}{n}}.
\end{align*}
Then we have:
\begin{align*}
c_{n}(\kappa)&=  4 ||\mathfrak{R}_{n}||(\mathcal{H}_{\ell b}) + \sqrt{\frac{72 \ln(2/(2-\kappa)) \overline{H}^{2}}{n}} + 5\varepsilon.
\end{align*}
Then we conclude \eqref{eq_concen_1}. 

\end{proof}

\begin{proof}[Proof of Theorem \ref{theorem_delta_minimal}]
Let $T, T^\flat$, and $T^\sharp$ be as defined in Lemma \ref{lemma_subsets}.\note{Note that the objective of the proof is to show $\delta^{**} \leq \delta^{*}$ on $E_{n}$, the same event for which Lemma \ref{lemma_subsets} applies. Knowing the feasible-to-compute $\delta^{*}$ then allows us to apply Lemma \ref{lemma_subsets}.} In this proof, it is useful to note the following facts:
\begin{enumerate}[label=(\roman*)]
	\item The functions $\delta \mapsto T_{n}(\delta), T(\delta)$ are non-decreasing left-continuous step functions that are greater than or equal to zero on the interval $[0,\delta_{0}]$, and zero otherwise.
	\item The functions $\sigma \mapsto T_{n}^\flat(\sigma), T^\flat(\sigma)$, are non-increasing and left-continuous with their only possible points of discontinuity at the points $\{\delta_{j}\}_{j=0}^{\infty}$.\noteinline{It seems they are continuous on any interval where $T_{n}(\delta)/\delta$ and $T(\delta)/\delta$ are increasing, and left continuous otherwise. The left-continuity seems to come from the fact that, if $T(\delta)/\delta$ is decreasing in $\delta$, then it is a decreasing and left-continuous function, so that also $\sigma\mapsto\sup_{\delta \geq \sigma} T(\delta)/\delta$ is also left continuous.} 
	\item The functions $\eta \mapsto T_{n}^\sharp(\eta), T^\sharp(\eta)$ are non-increasing and continuous.\noteinline{In fact $T^\sharp$ will be constant at the points of discontinuity of $T^\flat$. The continuity comes from the fact that $T^\flat$ is a non-increasing left-continuous function. At any point $\sigma^{*}$ of discontinuity of $T^\flat$ we will have $T^\flat(\sigma^{*}) > \lim_{\sigma \downarrow \sigma^{*}} T^\flat(\sigma)$, and can have a situation where $T^\flat(\sigma^{*}) > \eta \geq T^\flat(\sigma)$ for $\sigma>\sigma^{*}$. However, in such a case we will have $T^\sharp(\eta) = \sigma^{*}$, since the infimum in the definition of $T^\sharp$ behaves in a continuous way. At such points of discontinuity of $T^\flat$, it seems $T^\sharp$ will be constant.}
\end{enumerate}
Now for any $\eta>0$, let:
\begin{align*}
\delta^{*}&= T_{n}^{\sharp}(1-1/\mathfrak{a})+\eta',\\
\delta^{**}&= T^{\sharp}(1-1/\mathfrak{a})+\eta,
\end{align*}
where $\eta'=\eta+\varepsilon$ for some $\varepsilon>0$. Note that choosing $\delta^{**}$ slightly larger than $T^\sharp(1-1/\mathfrak{a})$ ensures that $T^\flat(\delta^{**})\leq 1-1/\mathfrak{a}$. A similar note applies to $\delta^{*}$ and $T_{n}^{\sharp}(1-1/\mathfrak{a})$.

From the proof of Lemma \ref{lemma_subsets} we know there exists an event $E_{n}$ with $P_{Y,Z}^{\otimes n} \left(E_{n} \right) \geq \kappa$ such that on $E_{n}$ we have $\mathscr{G}^{*}(\delta) \subseteq \mathscr{G}_{n}(\mathfrak{b}\delta)$ for every $\delta \geq \delta^{**}$. Thus, for every $\delta \geq \delta^{**}$ we have on $E_{n}$ that $T(\delta)\leq T_{n}(\delta)$, which implies:
\begin{align*}
\frac{T(\delta)}{\delta} \leq \frac{T_{n}(\delta)}{\delta}, 
\end{align*}
for all $\delta \geq \delta^{**}$. Thus, on $E_{n}$ we have $T^\flat(\sigma) \leq T_{n}^\flat(\sigma)$ for any $\sigma \geq \delta^{**}$, and in particular we have:
\begin{align}
T^\flat(\delta^{**}):= \sup_{\delta \geq \delta^{**}}\frac{T(\delta)}{\delta} \leq \sup_{\delta \geq \delta^{**}}\frac{T_{n}(\delta)}{\delta}=:T_{n}^\flat(\delta^{**}), \label{eq_Tflat_order}
\end{align}
Recall our choice of $\delta^{**}$ ensures that $T^\flat(\delta^{**})\leq 1-1/\mathfrak{a}$. We can now distinguish two cases on the event $E_{n}$:
\begin{enumerate}
	\item We have:
	\begin{align*}
	\sup_{\delta \geq \delta^{**}}\frac{T(\delta)}{\delta} \leq 1-\frac{1}{\mathfrak{a}} \leq \sup_{\delta \geq \delta^{**}}\frac{T_{n}(\delta)}{\delta}. 
	\end{align*}
	In this case, we have $T_{n}^\flat(\delta^{**}) \geq 1-1/\mathfrak{a}$, and thus $T_{n}^\sharp(1-1/\mathfrak{a}) \geq \delta^{**}$, so that $\delta^{*}>\delta^{**}$ (see the definitions of $\delta^{*}$ and $\delta^{**}$ above).
	\item We have:
	\begin{align*}
	\sup_{\delta \geq \delta^{**}}\frac{T(\delta)}{\delta} \leq \sup_{\delta \geq \delta^{**}}\frac{T_{n}(\delta)}{\delta} < 1-\frac{1}{\mathfrak{a}}.
	\end{align*}
	This implies either (i) $T^\sharp(1-1/\mathfrak{a}) \leq T_{n}^\sharp(1-1/\mathfrak{a}) < \delta^{**}$, or (ii) $T_{n}^\sharp(1-1/\mathfrak{a}) < T^\sharp(1-1/\mathfrak{a}) < \delta^{**}$. In case (i) we clearly have $\delta^{*} \geq \delta^{**}$. In case (ii), let:
	\begin{align*}
	c := T^\sharp(1-1/\mathfrak{a}) - T_{n}^\sharp(1-1/\mathfrak{a})>0.
	\end{align*}
	Then:
	\begin{align*}
	\delta^{**} - \delta^{*} &= T^\sharp(1-1/\mathfrak{a}) +\eta - T_{n}^\sharp(1-1/\mathfrak{a}) - \eta'\\
	&=c -\varepsilon,
	\end{align*}
	where the last line follows from the definition of $\eta'$.	Now suppose that $c>\varepsilon$ for our $\varepsilon>0$ chosen at the beginning of the proof. We will show that this produces a contradiction. To understand the approach, note the value of $c$ does not depend on the value of $\eta>0$, so the assumption that $c>\varepsilon$ must trivially hold for every $\eta>0$. If we can show that $c<\varepsilon$ for some $\eta>0$, we will have arrived at our desired contradiction.  

	Recall that on $E_{n}$ we have $T^\flat(\sigma) \leq T_{n}^\flat(\sigma)$ for any $\sigma \geq \delta^{**}$. This implies that, for any $r>0$, if $T^\sharp(r) \geq \delta^{**}$ then $T_{n}^\sharp(r) \geq T^\sharp(r)$. Now choose a value $r_{\eta}\in \mathbb{R}$ closest to $1-1/\mathfrak{a}$ such that $r_{\eta} \leq 1-1/\mathfrak{a}$ and:
	\begin{align*}
	T^\sharp(r_{\eta}) =T^{\sharp}(1-1/\mathfrak{a})+\eta =\delta^{**}.
	\end{align*}
	Such a choice is always possible by continuity of $T^{\sharp}$, and by the fact that $T^\sharp$ is non-increasing. By taking $\eta$ (and thus also $\delta^{**}$) small enough we conclude by continuity of $T^\sharp$ that the point $r_{\eta}$ can also always be chosen arbitrarily close to $1-1/\mathfrak{a}$.\todo{You better check this, I am not certain it is correct.} Recall by continuity of $T_{n}^\sharp$ that there exists $\varepsilon'>0$ such that $T_{n}^\sharp(x) - T_{n}^\sharp(1-1/\mathfrak{a}) < \varepsilon$ whenever $(1-1/\mathfrak{a}) < x + \varepsilon'$. Now by choosing $r_{\eta} \leq 1-1/\mathfrak{a}$ such that $1-1/\mathfrak{a}<r_{\eta}+\varepsilon'$, we have:  
	\begin{align*}
	c &= T^\sharp(1-1/\mathfrak{a}) - T_{n}^\sharp(1-1/\mathfrak{a})\\
	&< T^\sharp(r_{\eta}) - T_{n}^\sharp(1-1/\mathfrak{a})\\
	&\leq T_{n}^\sharp(r_{\eta}) - T_{n}^\sharp(1-1/\mathfrak{a})\\
	&<\varepsilon.
	\end{align*}
	This of course contradicts the fact that $c>\varepsilon$ for every choice of $\eta>0$. We conclude that $c\leq\varepsilon$, and since $\delta^{**}-\delta^{*}=c-\varepsilon$, we have $\delta^{**}\leq \delta^{*}$.\note{This proof is extremely tricky, but essentially seems related to the fact that $\delta^{**}$ is a function of $\eta$, and the fact that if $T^\sharp(r) \geq \delta^{**}(\eta)$ then $T_{n}^\sharp(r) \geq T^\sharp(r)$. } 
\end{enumerate}
We conclude in all cases that $\delta^{**}\leq \delta^{*}$ on $E_{n}$. The result then follows directly from Lemma \ref{lemma_subsets}.

\end{proof}

\begin{proof}[Proof of Lemma \ref{lemma_subsets}]
Recall that:
\begin{align*}
h_{\ell b}(y,z,\theta,\gamma,\lambda):= \inf_{u \in \bm G^{-}(y,z,\theta)}  \inf_{y^{\star}  \in \bm G^{\star}(y,z,u,\theta,\gamma)} \Bigg( \varphi(v) + \mu^{*}\sum_{j=1}^{J}  \lambda_{j} m_{j}(y,z,u,\theta)\Bigg).
\end{align*}
For notational simplicity we will define:
\begin{align*}
\P_{n}h_{\ell b}(\cdot, \theta,\gamma,\lambda) &:= \frac{1}{n}\sum_{i=1}^{n} \inf_{u_{i} \in \bm G^{-}(y_{i},z_{i},\theta)}  \inf_{y_{i}^{\star}  \in \bm G^{\star}(y_{i},z_{i},u_{i},\theta,\gamma)} \Bigg( \varphi(v_{i}) + \mu^{*}\sum_{j=1}^{J}  \lambda_{j} m_{j}(y_{i},z_{i},u_{i},\theta)\Bigg),\\
P h_{\ell b}(\cdot, \theta,\gamma,\lambda) &:= \int \inf_{u \in \bm G^{-}(y,z,\theta)}  \inf_{y^{\star}  \in \bm G^{\star}(y,z,u,\theta,\gamma)} \Bigg( \varphi(v) + \mu^{*}\sum_{j=1}^{J}  \lambda_{j}  m_{j}(y,z,u,\theta)\Bigg)\,dP_{Y,Z}.
\end{align*}
Define the events:
\begin{align*}
E_{n,j} := \left\{\sup_{\theta, \theta' \in \Theta} \sup_{\gamma, \gamma' \in \mathscr{G}^{*}(\delta_{j})} \sup_{\lambda,\lambda' \in \Lambda}\left|\left( \P_{n} h_{\ell b}(\cdot,\theta,\gamma,\lambda) -  \P_{n} h_{\ell b}(\cdot,\theta',\gamma',\lambda')\right) -  \left( P h_{\ell b}(\cdot,\theta,\gamma,\lambda) -  P h_{\ell b}(\cdot,\theta',\gamma',\lambda') \right)\right| \leq T(\delta_{j})\right\},
\end{align*}
and:
\begin{align}
E_{n}:= \bigcap_{\{ j : \delta_{j}\geq \delta^{**}\}} E_{n,j}.\label{eq_event_Enm}
\end{align}
Note the value $2 \overline{H}$ is an upper bound for any function in $\mathcal{H}_{\ell b}'(\delta)$ for any $\delta> 0$. By our choice of $\delta_{0}> 2\overline{H}$ we have:
\begin{align*}
\sup_{P_{Y,Z} \in \mathcal{P}_{Y,Z}} P_{Y,Z}^{\otimes n} \left( E_{n,0}^{c}\right)=0.
\end{align*}
Furthermore, from the uniform version of Hoeffding's inequality (e.g. \cite{koltchinskii2011oracle} Theorem 4.6, p.71) we have:
\begin{align*}
\sup_{P_{Y,Z} \in \mathcal{P}_{Y,Z}} P_{Y,Z}^{\otimes n}\left( E_{n,j}^{c}\right) \leq  \exp \left( -\frac{t_{j}^{2}}{2} \right), 
\end{align*}
for each $j\in \mathbb{N}$. We conclude by the union bound that:
\begin{align*}
\inf_{P_{Y,Z} \in \mathcal{P}_{Y,Z}} P_{Y,Z}^{\otimes n}\left( E_{n}\right) \geq 1- \sum_{\{j : \delta_{j} \geq \delta^{**}\}} \exp \left( -\frac{t_{j}^{2}}{2} \right). 
\end{align*}
Now note that with $c_{1}=5$, $c_{2}= (3/(2\kappa))^{2/5}$ and $t_{j} = \sqrt{c_{1} \log(c_{2} \cdot j)}$, we have:
\begin{align*}
\sum_{\{j : \delta_{j} \geq \delta^{**}\}} \exp \left( -\frac{t_{j}^{2}}{2} \right) &\leq  \sum_{j=1}^{\infty} \exp \left( -\frac{t_{j}^{2}}{2} \right)\\
&= \sum_{j=1}^{\infty} \exp \left( -\frac{c_{1} \log(c_{2} \cdot j)}{2} \right) \\
&= \sum_{j=1}^{\infty} (c_{2}\cdot j)^{-\frac{c_{1}}{2}} \\
&=\frac{2(1-\kappa)}{3}  \sum_{j=1}^{\infty} \left( \frac{1}{j} \right)^{5/2}\\
&\leq \frac{2(1-\kappa) }{3} \left( \frac{3}{2}\right)\\
&=1-\kappa.
\end{align*}
Thus we conclude:
\begin{align}
\inf_{P_{Y,Z} \in \mathcal{P}_{Y,Z}} P_{Y,Z}^{\otimes n} \left( E_{n} \right) \geq \kappa. 
\end{align}
The remainder of the proof proceeds in two parts:
\begin{enumerate}
	\item We will show that on the event $E_{n}$ we have for any $\gamma \in \Gamma$, $\mathscr{E}_{n}(\gamma) \leq (2-1/\mathfrak{a})\left(\mathscr{E}^{*}(\gamma)\vee \delta^{**}\right)$. We will then use this fact to argue that, on $E_{n}$, for any $\delta \geq \delta^{**}$ we have $\mathscr{G}^{*}(\delta) \subseteq \mathscr{G}_{n}((2-1/\mathfrak{a})\delta)$.
	\item We will show that on the event $E_{n}$ we have for any $\gamma \in \Gamma$, $\mathscr{E}^{*}(\gamma) \leq \mathfrak{a}\left(\mathscr{E}_{n}(\gamma)\vee \delta^{**}\right)$. We will then use this fact to argue that, on $E_{n}$, for any $\delta \geq \mathfrak{a}\delta^{**}$ we have $\mathscr{G}_{n}(\delta/\mathfrak{a}) \subseteq \mathscr{G}^{*}(\delta)$.
\end{enumerate}
Throughout this proof, let $\lambda^{*}(\theta,\gamma)$, $\hat{\lambda}(\theta,\gamma)$,  $\theta^{*}(\gamma)$, $\hat{\theta}(\gamma)$, $\gamma^{*}$ and $\hat{\gamma}$ be as in Remark \ref{remark_common_notation}.\\~\\ 
\noindent \textbf{Part 1:} We will prove that on the event $E_{n}$ we have $\mathscr{E}_{n}(\gamma) \leq (2-1/\mathfrak{a})\left(\mathscr{E}^{*}(\gamma)\vee \delta^{**}\right)$ for any $\gamma \in \Gamma$. First, consider any $\gamma$ with $\sigma:= \mathscr{E}^{*}(\gamma) \geq \delta^{**}$. Pick any $\varepsilon>0$ such that $\delta^{**}\geq\varepsilon$, which is possible since $\delta^{**}>T^\sharp(1-1/\mathfrak{a}) \geq 0$. Then on the event $E_{n}$ we have:
\begin{align*}
\mathscr{E}_{n}(\gamma)&:=\sup_{\gamma \in \Gamma}\inf_{\theta \in \Theta} \max_{\lambda \in \Lambda} \P_{n} h_{\ell b}(\cdot,\theta,\gamma,\lambda) - \inf_{\theta \in \Theta} \max_{\lambda \in \Lambda} \P_{n} h_{\ell b}(\cdot,\theta,\gamma,\lambda)\\
&\leq \inf_{\theta \in \Theta} \max_{\lambda \in \Lambda} \P_{n} h_{\ell b}(\cdot,\theta,\hat{\gamma},\lambda) - \inf_{\theta \in \Theta} \max_{\lambda \in \Lambda} \P_{n} h_{\ell b}(\cdot,\theta,\gamma,\lambda)+3\varepsilon\\
&= \inf_{\theta \in \Theta} \max_{\lambda \in \Lambda} P h_{\ell b}(\cdot,\theta,\hat{\gamma},\lambda) - \inf_{\theta \in \Theta} \max_{\lambda \in \Lambda} P h_{\ell b}(\cdot,\theta,\gamma,\lambda)\\
&\qquad\qquad+ \left( \inf_{\theta \in \Theta} \max_{\lambda \in \Lambda} P h_{\ell b}(\cdot,\theta,\gamma,\lambda)- \inf_{\theta \in \Theta} \max_{\lambda \in \Lambda} P h_{\ell b}(\cdot,\theta,\hat{\gamma},\lambda) \right)\\
&\qquad\qquad\qquad\qquad-\left(\inf_{\theta \in \Theta}  \max_{\lambda \in \Lambda}\P_{n} h_{\ell b}(\cdot,\theta,\gamma,\lambda)-\inf_{\theta \in \Theta} \max_{\lambda \in \Lambda} \P_{n} h_{\ell b}(\cdot,\theta,\hat{\gamma},\lambda)\right)+3\varepsilon\\
&\leq\sup_{\gamma \in \Gamma} \inf_{\theta \in \Theta}  \max_{\lambda \in \Lambda} P h_{\ell b}(\cdot,\theta,\gamma,\lambda) - \inf_{\theta \in \Theta} \max_{\lambda \in \Lambda}  P h_{\ell b}(\cdot,\theta,\gamma,\lambda)\\
&\qquad\qquad+ \left( \inf_{\theta \in \Theta} \max_{\lambda \in \Lambda} P h_{\ell b}(\cdot,\theta,\gamma,\lambda)- \inf_{\theta \in \Theta} \max_{\lambda \in \Lambda} P h_{\ell b}(\cdot,\theta,\hat{\gamma},\lambda) \right)\\
&\qquad\qquad\qquad\qquad-\left(\inf_{\theta \in \Theta} \max_{\lambda \in \Lambda} \P_{n} h_{\ell b}(\cdot,\theta,\gamma,\lambda)-\inf_{\theta \in \Theta}\max_{\lambda \in \Lambda}  \P_{n} h_{\ell b}(\cdot,\theta,\hat{\gamma},\lambda)\right)+3\varepsilon\\
&=\mathscr{E}^{*}(\gamma)+ \left( \inf_{\theta \in \Theta} \max_{\lambda \in \Lambda} P h_{\ell b}(\cdot,\theta,\gamma,\lambda)- \inf_{\theta \in \Theta} \max_{\lambda \in \Lambda} P h_{\ell b}(\cdot,\theta,\hat{\gamma},\lambda) \right)\\
&\qquad\qquad\qquad\qquad-\left(\inf_{\theta \in \Theta} \max_{\lambda \in \Lambda} \P_{n} h_{\ell b}(\cdot,\theta,\gamma,\lambda)-\inf_{\theta \in \Theta}\max_{\lambda \in \Lambda}  \P_{n} h_{\ell b}(\cdot,\theta,\hat{\gamma},\lambda)\right)+3\varepsilon.
\end{align*}
Now note:
\begin{align*}
&\inf_{\theta \in \Theta} \max_{\lambda \in \Lambda} P h_{\ell b}(\cdot,\theta,\gamma,\lambda)- \inf_{\theta \in \Theta} \max_{\lambda \in \Lambda} P h_{\ell b}(\cdot,\theta,\hat{\gamma},\lambda)\\
&\leq \inf_{\theta \in \Theta} \max_{\lambda \in \Lambda} P h_{\ell b}(\cdot,\theta,\gamma,\lambda)- \max_{\lambda \in \Lambda} P h_{\ell b}(\cdot,\theta^{*}(\hat{\gamma}),\hat{\gamma},\lambda)+\varepsilon\\
&\leq  \max_{\lambda \in \Lambda} P h_{\ell b}(\cdot,\hat{\theta}(\gamma),\gamma,\lambda)- \max_{\lambda \in \Lambda} P h_{\ell b}(\cdot,\theta^{*}(\hat{\gamma}),\hat{\gamma},\lambda)+2\varepsilon\\
&\leq  \max_{\lambda \in \Lambda} P h_{\ell b}(\cdot,\hat{\theta}(\gamma),\gamma,\lambda)-  P h_{\ell b}(\cdot,\theta^{*}(\hat{\gamma}),\hat{\gamma},\hat{\lambda}(\theta^{*}(\hat{\gamma}),\hat{\gamma}))+2\varepsilon\\
&= P h_{\ell b}(\cdot,\hat{\theta}(\gamma),\gamma,\lambda^{*}(\hat{\theta}(\gamma),\gamma))-  P h_{\ell b}(\cdot,\theta^{*}(\hat{\gamma}),\hat{\gamma},\hat{\lambda}(\theta^{*}(\hat{\gamma}),\hat{\gamma}))+2\varepsilon.
\end{align*}
Similarly:
\begin{align*}
&\inf_{\theta \in \Theta}\max_{\lambda \in \Lambda}  \P_{n} h_{\ell b}(\cdot,\theta,\hat{\gamma},\lambda) - \inf_{\theta \in \Theta} \max_{\lambda \in \Lambda} \P_{n} h_{\ell b}(\cdot,\theta,\gamma,\lambda)\\
&\leq \inf_{\theta \in \Theta}\max_{\lambda \in \Lambda}  \P_{n} h_{\ell b}(\cdot,\theta,\hat{\gamma},\lambda) - \max_{\lambda \in \Lambda} \P_{n} h_{\ell b}(\cdot,\hat{\theta}(\gamma),\gamma,\lambda) + \varepsilon\\
&\leq \max_{\lambda \in \Lambda}  \P_{n} h_{\ell b}(\cdot,\theta^{*}(\hat{\gamma}),\hat{\gamma},\lambda) - \max_{\lambda \in \Lambda} \P_{n} h_{\ell b}(\cdot,\hat{\theta}(\gamma),\gamma,\lambda) + 2\varepsilon\\
&\leq \max_{\lambda \in \Lambda}  \P_{n} h_{\ell b}(\cdot,\theta^{*}(\hat{\gamma}),\hat{\gamma},\lambda) - \P_{n} h_{\ell b}(\cdot,\hat{\theta}(\gamma),\gamma,\lambda^{*}(\hat{\theta}(\gamma),\gamma)) + 2\varepsilon\\
&= \P_{n} h_{\ell b}(\cdot,\theta^{*}(\hat{\gamma}),\hat{\gamma},\hat{\lambda}(\theta^{*}(\hat{\gamma}),\hat{\gamma})) - \P_{n} h_{\ell b}(\cdot,\hat{\theta}(\gamma),\gamma,\lambda^{*}(\hat{\theta}(\gamma),\gamma)) + 2\varepsilon.
\end{align*}
Thus we conclude:
\begin{align*}
\mathscr{E}_{n}(\gamma)&\leq \mathscr{E}^{*}(\gamma)+ 7\varepsilon + P h_{\ell b}(\cdot,\hat{\theta}(\gamma),\gamma,\lambda^{*}(\hat{\theta}(\gamma),\gamma))-  P h_{\ell b}(\cdot,\theta^{*}(\hat{\gamma}),\hat{\gamma},\hat{\lambda}(\theta^{*}(\hat{\gamma}),\hat{\gamma}))\\
&\qquad\qquad\qquad\qquad\qquad\qquad- \left(\P_{n} h_{\ell b}(\cdot,\hat{\theta}(\gamma),\gamma,\lambda^{*}(\hat{\theta}(\gamma),\gamma))  - \P_{n} h_{\ell b}(\cdot,\theta^{*}(\hat{\gamma}),\hat{\gamma},\hat{\lambda}(\theta^{*}(\hat{\gamma}),\hat{\gamma}))\right).
\end{align*}
However, $\gamma \in \mathscr{G}^{*}(\sigma)$ by assumption, and by Lemma \ref{lemma_gamma_hat_high_prob} we have $\hat{\gamma} \in \mathscr{G}^{*}(\sigma)$ on the event $E_{n}$. Thus, the right side of the previous display can be bounded above:
\begin{align*}
&P h_{\ell b}(\cdot,\hat{\theta}(\gamma),\gamma,\lambda^{*}(\hat{\theta}(\gamma),\gamma))-  P h_{\ell b}(\cdot,\theta^{*}(\hat{\gamma}),\hat{\gamma},\hat{\lambda}(\theta^{*}(\hat{\gamma}),\hat{\gamma}))- \left(\P_{n} h_{\ell b}(\cdot,\hat{\theta}(\gamma),\gamma,\lambda^{*}(\hat{\theta}(\gamma),\gamma))  - \P_{n} h_{\ell b}(\cdot,\theta^{*}(\hat{\gamma})\right)\\
&\leq \sup_{\theta, \theta' \in \Theta} \sup_{\gamma, \gamma' \in \mathscr{G}^{*}(\sigma)} \sup_{\lambda,\lambda' \in \Lambda} \left| \left(\P_{n} h_{\ell b}(\cdot,\theta,\gamma,\lambda) - \P_{n} h_{\ell b}(\cdot,\theta',\gamma',\lambda')\right) - \left( \P_{n} h_{\ell b}(\cdot,\theta,\gamma,\lambda) - \P_{n} h_{\ell b}(\cdot,\theta',\gamma',\lambda') \right) \right|.
\end{align*}
Furthermore, for any $\sigma\geq \delta^{**}$, on the event $E_{n}$ this final quantity is bounded above by $T(\sigma)$; this follows from the definition of $T(\sigma)$ and the monotonicity of the map:
\begin{align*}
x \mapsto \sup_{\theta, \theta' \in \Theta} \sup_{\gamma, \gamma' \in \mathscr{G}^{*}(x)} \max_{\lambda,\lambda' \in \Lambda} \left|\left( \P_{n} h_{\ell b}(\cdot,\theta,\gamma,\lambda) -  \P_{n} h_{\ell b}(\cdot,\theta',\gamma',\lambda')\right) -  \left( P h_{\ell b}(\cdot,\theta,\gamma,\lambda) -  P h_{\ell b}(\cdot,\theta',\gamma',\lambda') \right)\right|.
\end{align*}
Thus on $E_{n}$:
\begin{align*}
\mathscr{E}_{n}(\gamma)&\leq \mathscr{E}^{*}(\gamma)+ T(\sigma)+7\varepsilon\\
&= \mathscr{E}^{*}(\gamma)+ \frac{T(\sigma)}{\sigma} \sigma+7\varepsilon\\
&\leq \mathscr{E}^{*}(\gamma)+ \sup_{\delta \geq \sigma}\left( \frac{T(\delta)}{\delta}\right)\sigma+7\varepsilon\\
&= \mathscr{E}^{*}(\gamma) + T^\flat(\sigma)\sigma+7\varepsilon\\
&= \mathscr{E}^{*}(\gamma) + T^\flat(\sigma)\mathscr{E}^{*}(\gamma)+7\varepsilon.
\end{align*}
Now, since $\sigma \geq \delta^{**} > T^\sharp(1-1/\mathfrak{a})$ we have $T^\flat(\sigma)\leq T^\flat(\delta^{**}) \leq 1- 1/\mathfrak{a}$. Thus, on the event $E_{n}$, if $\gamma$ is such that $\mathscr{E}^{*}(\gamma) \geq \delta^{**}$, we have:
\begin{align*}
\mathscr{E}_{n}(\gamma) \leq  \left(2-\frac{1}{\mathfrak{a}}\right)\mathscr{E}^{*}(\gamma)+7\varepsilon.
\end{align*}
Since $\varepsilon>0$ is any value such that $\delta^{**}\geq \varepsilon$, and thus can be made arbitrarily small, we conclude that on the event $E_{n}$ we have for any $\gamma$ with $\mathscr{E}^{*}(\gamma) \geq \delta^{**}$:
\begin{align*}
\mathscr{E}_{n}(\gamma) \leq  \left(2-\frac{1}{\mathfrak{a}}\right)\mathscr{E}^{*}(\gamma). 
\end{align*}
Now consider the case when $\sigma:=\mathscr{E}^{*}(\gamma) \leq \delta^{**}$. By the same derivation as above we obtain:
\begin{align*}
&\mathscr{E}_{n}(\gamma)\\
&\leq \mathscr{E}^{*}(\gamma)  + \sup_{\theta, \theta' \in \Theta} \sup_{\gamma, \gamma' \in \mathscr{G}^{*}(\sigma)}\max_{\lambda,\lambda' \in \Lambda} \left|\left( \P_{n} h_{\ell b}(\cdot,\theta,\gamma,\lambda) -  \P_{n} h_{\ell b}(\cdot,\theta',\gamma',\lambda')\right) -  \left( P h_{\ell b}(\cdot,\theta,\gamma,\lambda) -  P h_{\ell b}(\cdot,\theta',\gamma',\lambda') \right)\right| +7\varepsilon.
\end{align*}
By monotonicity, we have:
\begin{align*}
&\sup_{\theta, \theta' \in \Theta} \sup_{\gamma, \gamma' \in \mathscr{G}^{*}(\sigma)}\max_{\lambda,\lambda' \in \Lambda} \left|\left( \P_{n} h_{\ell b}(\cdot,\theta,\gamma,\lambda) -  \P_{n} h_{\ell b}(\cdot,\theta',\gamma',\lambda')\right) -  \left( P h_{\ell b}(\cdot,\theta,\gamma,\lambda) -  P h_{\ell b}(\cdot,\theta',\gamma',\lambda') \right)\right|\\
 &\leq \sup_{\theta, \theta' \in \Theta} \sup_{\gamma, \gamma' \in \mathscr{G}^{*}(\delta^{**})}\max_{\lambda,\lambda' \in \Lambda} \left|\left( \P_{n} h_{\ell b}(\cdot,\theta,\gamma,\lambda) -  \P_{n} h_{\ell b}(\cdot,\theta',\gamma',\lambda')\right) -  \left( P h_{\ell b}(\cdot,\theta,\gamma,\lambda) -  P h_{\ell b}(\cdot,\theta',\gamma',\lambda') \right)\right|. 
\end{align*}
Furthermore, on the event $E_{n}$ we have:
\begin{align*}
\sup_{\theta, \theta' \in \Theta} \sup_{\gamma, \gamma' \in \mathscr{G}^{*}(\delta^{**})}\max_{\lambda,\lambda' \in \Lambda} \left|\left( \P_{n} h_{\ell b}(\cdot,\theta,\gamma,\lambda) -  \P_{n} h_{\ell b}(\cdot,\theta',\gamma',\lambda')\right) -  \left( P h_{\ell b}(\cdot,\theta,\gamma,\lambda) -  P h_{\ell b}(\cdot,\theta',\gamma',\lambda') \right)\right| \leq T(\delta^{**}).
\end{align*}
Thus, on the event $E_{n}$:
\begin{align*}
\mathscr{E}_{n}(\gamma)&\leq \mathscr{E}^{*}(\gamma) + T(\delta^{**})+7\varepsilon\\
&\leq \mathscr{E}^{*}(\gamma) + \sup_{\delta \geq \delta^{**}}\left( \frac{T(\delta)}{\delta} \right) \delta^{**} +7\varepsilon\\
&= \mathscr{E}^{*}(\gamma) + T^\flat(\delta^{**})\delta^{**} +7\varepsilon\\
&\leq \mathscr{E}^{*}(\gamma) + \left( 1-\frac{1}{\mathfrak{a}} \right) \delta^{**} +7\varepsilon\\
&\leq \delta^{**} + \left( 1-\frac{1}{\mathfrak{a}} \right) \delta^{**} +7\varepsilon\\
&=\left( 2-\frac{1}{\mathfrak{a}} \right) \delta^{**} +7\varepsilon.
\end{align*}
Since $\varepsilon>0$ is any value such that $\delta^{**}\geq \varepsilon$, and thus can be made arbitrarily small, we conclude that on the event $E_{n}$ we have for any $\gamma$:
\begin{align*}
\mathscr{E}_{n}(\gamma) \leq  \left(2-\frac{1}{\mathfrak{a}}\right)\left( \mathscr{E}^{*}(\gamma)\vee \delta^{**}\right). 
\end{align*}
We will use this result to argue that, on the event $E_{n}$, if $\delta\geq \delta^{**}$ then $\mathscr{E}^{*}(\gamma) \leq \delta \implies \mathscr{E}_{n}(\gamma) \leq (2-1/\mathfrak{a})\delta$. There are two cases:
\begin{enumerate}[label=(\roman*)]
	\item $\mathscr{E}^{*}(\gamma) \leq \delta^{**} \leq \delta$, which implies on the event $E_{n}$:
	\begin{align*}
\mathscr{E}_{n}(\gamma) \leq  \left(2-\frac{1}{\mathfrak{a}}\right)\left( \mathscr{E}^{*}(\gamma)\vee \delta^{**}\right) =\left(2-\frac{1}{\mathfrak{a}}\right)\delta^{**} \leq \left(2-\frac{1}{\mathfrak{a}}\right)\delta.  
\end{align*}
	\item $\delta^{**} \leq \mathscr{E}^{*}(\gamma) \leq \delta$, which implies on the event $E_{n}$:
	\begin{align*}
\mathscr{E}_{n}(\gamma) \leq  \left(2-\frac{1}{\mathfrak{a}}\right)\left( \mathscr{E}^{*}(\gamma)\vee \delta^{**}\right) = \left(2-\frac{1}{\mathfrak{a}}\right) \mathscr{E}^{*}(\gamma)\leq \left(2-\frac{1}{\mathfrak{a}}\right) \delta.  
\end{align*}
\end{enumerate}
Thus we conclude that for any $\delta \geq \delta^{**}$, on $E_{n}$ we have that $\mathscr{E}^{*}(\gamma) \leq \delta \implies \mathscr{E}_{n}(\gamma) \leq (2-1/\mathfrak{a})\delta$. Now recall that we have $\mathscr{E}^{*}(\gamma) \leq \delta \iff \gamma \in \mathscr{G}^{*}(\delta)$ and $\mathscr{E}_{n}(\gamma) \leq (2-1/\mathfrak{a})\delta \iff \gamma \in \mathscr{G}_{n}((2-1/\mathfrak{a})\delta)$. Thus, we conclude that for any $\delta \geq \delta^{**}$, on the event $E_{n}$:
\begin{align*}
\mathscr{G}^{*}(\delta) \subseteq \mathscr{G}_{n}((2-1/\mathfrak{a})\delta),
\end{align*}
as desired.\\

\noindent \textbf{Part 2:} We will prove that on the event $E_{n}$ we have $\mathscr{E}^{*}(\gamma) \leq \mathfrak{a}\left(\mathscr{E}_{n}(\gamma)\vee \delta^{**}\right)$ for any $\gamma \in \Gamma$. If $\gamma$ is such that $\mathscr{E}^{*}(\gamma) \leq \delta^{**}$ then this is trivially true (since $\mathfrak{a}>1$). Now consider any $\gamma$ with $\sigma:= \mathscr{E}^{*}(\gamma) \geq \delta^{**}$. Pick any $\varepsilon>0$ such that $\delta^{**}\geq\varepsilon$, which is possible since $\delta^{**}>T^\sharp(1-1/\mathfrak{a}) \geq 0$. Then on the event $E_{n}$ we have:
\begin{align*}
\mathscr{E}^{*}(\gamma)&:=\sup_{\gamma \in \Gamma}\inf_{\theta \in \Theta} \max_{\lambda \in \Lambda} P h_{\ell b}(\cdot,\theta,\gamma,\lambda) - \inf_{\theta \in \Theta} \max_{\lambda \in \Lambda} P h_{\ell b}(\cdot,\theta,\gamma,\lambda)\\
&\leq \inf_{\theta \in \Theta} \max_{\lambda \in \Lambda} P h_{\ell b}(\cdot,\theta,\gamma^{*},\lambda) - \inf_{\theta \in \Theta} \max_{\lambda \in \Lambda} P h_{\ell b}(\cdot,\theta,\gamma,\lambda)+3\varepsilon\\
&= \inf_{\theta \in \Theta} \max_{\lambda \in \Lambda} P h_{\ell b}(\cdot,\theta,\gamma^{*},\lambda) - \inf_{\theta \in \Theta} \max_{\lambda \in \Lambda} P h_{\ell b}(\cdot,\theta,\gamma,\lambda)\\
&\qquad\qquad+ \left( \sup_{\gamma \in \Gamma} \inf_{\theta \in \Theta} \max_{\lambda \in \Lambda} \P_{n} h_{\ell b}(\cdot,\theta,\gamma,\lambda)- \inf_{\theta \in \Theta} \max_{\lambda \in \Lambda} \P_{n} h_{\ell b}(\cdot,\theta,\gamma,\lambda) \right)\\
&\qquad\qquad\qquad\qquad-\left(\sup_{\gamma \in \Gamma}\inf_{\theta \in \Theta}  \max_{\lambda \in \Lambda}\P_{n} h_{\ell b}(\cdot,\theta,\gamma,\lambda)-\inf_{\theta \in \Theta} \max_{\lambda \in \Lambda} \P_{n} h_{\ell b}(\cdot,\theta,\gamma,\lambda)\right)+3\varepsilon\\
&=\mathscr{E}^{*}(\gamma)+ \left( \inf_{\theta \in \Theta} \max_{\lambda \in \Lambda} P h_{\ell b}(\cdot,\theta,\gamma^{*},\lambda) - \inf_{\theta \in \Theta} \max_{\lambda \in \Lambda} P h_{\ell b}(\cdot,\theta,\gamma,\lambda)\right)\\
&\qquad\qquad\qquad\qquad-\left(\sup_{\gamma \in \Gamma}\inf_{\theta \in \Theta}  \max_{\lambda \in \Lambda}\P_{n} h_{\ell b}(\cdot,\theta,\gamma,\lambda)-\inf_{\theta \in \Theta} \max_{\lambda \in \Lambda} \P_{n} h_{\ell b}(\cdot,\theta,\gamma,\lambda)\right)+3\varepsilon.
\end{align*}
Now note:
\begin{align*}
&\inf_{\theta \in \Theta} \max_{\lambda \in \Lambda} P h_{\ell b}(\cdot,\theta,\gamma^{*},\lambda) - \inf_{\theta \in \Theta} \max_{\lambda \in \Lambda} P h_{\ell b}(\cdot,\theta,\gamma,\lambda) \\
&\leq\inf_{\theta \in \Theta} \max_{\lambda \in \Lambda} P h_{\ell b}(\cdot,\theta,\gamma^{*},\lambda) -  \max_{\lambda \in \Lambda} P h_{\ell b}(\cdot,\theta^{*}(\gamma),\gamma,\lambda)+\varepsilon \\
&\leq\max_{\lambda \in \Lambda} P h_{\ell b}(\cdot,\hat{\theta}(\gamma^{*}),\gamma^{*},\lambda) -  \max_{\lambda \in \Lambda} P h_{\ell b}(\cdot,\theta^{*}(\gamma),\gamma,\lambda)+2\varepsilon \\
&\leq\max_{\lambda \in \Lambda} P h_{\ell b}(\cdot,\hat{\theta}(\gamma^{*}),\gamma^{*},\lambda) -  P h_{\ell b}(\cdot,\theta^{*}(\gamma),\gamma,\hat{\lambda}(\theta^{*}(\gamma),\gamma))+2\varepsilon \\
&\leq P h_{\ell b}(\cdot,\hat{\theta}(\gamma^{*}),\gamma^{*},\lambda^{*}(\hat{\theta}(\gamma^{*}),\gamma^{*})) -  P h_{\ell b}(\cdot,\theta^{*}(\gamma),\gamma,\hat{\lambda}(\theta^{*}(\gamma),\gamma))+2\varepsilon.
\end{align*}
Similarly:
\begin{align*}
&\inf_{\theta \in \Theta} \max_{\lambda \in \Lambda} \P_{n} h_{\ell b}(\cdot,\theta,\gamma,\lambda) - \sup_{\gamma \in \Gamma}\inf_{\theta \in \Theta}  \max_{\lambda \in \Lambda}\P_{n} h_{\ell b}(\cdot,\theta,\gamma,\lambda)\\
&\leq \inf_{\theta \in \Theta} \max_{\lambda \in \Lambda} \P_{n} h_{\ell b}(\cdot,\theta,\gamma,\lambda) - \inf_{\theta \in \Theta}  \max_{\lambda \in \Lambda}\P_{n} h_{\ell b}(\cdot,\theta,\gamma^{*},\lambda) + 3\varepsilon\\
&\leq \inf_{\theta \in \Theta} \max_{\lambda \in \Lambda} \P_{n} h_{\ell b}(\cdot,\theta,\gamma,\lambda) -  \max_{\lambda \in \Lambda}\P_{n} h_{\ell b}(\cdot,\hat{\theta}(\gamma^{*}),\gamma^{*},\lambda) + 4\varepsilon\\
&\leq  \max_{\lambda \in \Lambda} \P_{n} h_{\ell b}(\cdot,\theta^{*}(\gamma),\gamma,\lambda) -  \max_{\lambda \in \Lambda}\P_{n} h_{\ell b}(\cdot,\hat{\theta}(\gamma^{*}),\gamma^{*},\lambda) + 5\varepsilon\\
&\leq \P_{n} h_{\ell b}(\cdot,\theta^{*}(\gamma),\gamma,\hat{\lambda}(\theta^{*}(\gamma),\gamma)) -  \P_{n} h_{\ell b}(\cdot,\hat{\theta}(\gamma^{*}),\gamma^{*},\lambda^{*}(\hat{\theta}(\gamma^{*}),\gamma^{*})) + 5\varepsilon.
\end{align*}
Thus we conclude:
\begin{align*}
\mathscr{E}^{*}(\gamma)&\leq \mathscr{E}_{n}(\gamma)+ 10\varepsilon + P h_{\ell b}(\cdot,\hat{\theta}(\gamma^{*}),\gamma^{*},\lambda^{*}(\hat{\theta}(\gamma^{*}),\gamma^{*})) -  P h_{\ell b}(\cdot,\theta^{*}(\gamma),\gamma,\hat{\lambda}(\theta^{*}(\gamma),\gamma))\\
&\qquad\qquad\qquad\qquad\qquad\qquad- \left(\P_{n} h_{\ell b}(\cdot,\hat{\theta}(\gamma^{*}),\gamma^{*},\lambda^{*}(\hat{\theta}(\gamma^{*}),\gamma^{*}))  - \P_{n} h_{\ell b}(\cdot,\theta^{*}(\gamma),\gamma,\hat{\lambda}(\theta^{*}(\gamma),\gamma))\right).
\end{align*}
However, $\gamma \in \mathscr{G}^{*}(\sigma)$ by assumption, and $\mathscr{E}^{*}(\gamma^{*}) \leq \varepsilon \leq \mathscr{E}^{*}(\gamma) = \sigma$ implies $\gamma^{*} \in \mathscr{G}^{*}(\sigma)$. Thus, the right side of the previous display can be bounded above:
\begin{align*}
&P h_{\ell b}(\cdot,\hat{\theta}(\gamma^{*}),\gamma^{*},\lambda^{*}(\hat{\theta}(\gamma^{*}),\gamma^{*})) -  P h_{\ell b}(\cdot,\theta^{*}(\gamma),\gamma,\hat{\lambda}(\theta^{*}(\gamma),\gamma))\\
&\qquad\qquad\qquad\qquad\qquad\qquad- \left(\P_{n} h_{\ell b}(\cdot,\hat{\theta}(\gamma^{*}),\gamma^{*},\lambda^{*}(\hat{\theta}(\gamma^{*}),\gamma^{*}))  - \P_{n} h_{\ell b}(\cdot,\theta^{*}(\gamma),\gamma,\hat{\lambda}(\theta^{*}(\gamma),\gamma))\right)\\
&\leq \sup_{\theta, \theta' \in \Theta} \sup_{\gamma, \gamma' \in \mathscr{G}^{*}(\sigma)}\max_{\lambda,\lambda' \in \Lambda} \left|\left( \P_{n} h_{\ell b}(\cdot,\theta,\gamma,\lambda) -  \P_{n} h_{\ell b}(\cdot,\theta',\gamma',\lambda')\right) -  \left( P h_{\ell b}(\cdot,\theta,\gamma,\lambda) -  P h_{\ell b}(\cdot,\theta',\gamma',\lambda') \right)\right|.
\end{align*}
Furthermore, for any $\sigma\geq \delta^{**}$, on the event $E_{n}$ this final quantity is bounded above by $T(\sigma)$; this follows from the definition of $T(\sigma)$ and the monotonicity of the map:
\begin{align*}
x \mapsto \sup_{\theta, \theta' \in \Theta} \sup_{\gamma, \gamma' \in \mathscr{G}^{*}(x)} \max_{\lambda,\lambda' \in \Lambda} \left|\left( \P_{n} h_{\ell b}(\cdot,\theta,\gamma,\lambda) -  \P_{n} h_{\ell b}(\cdot,\theta',\gamma',\lambda')\right) -  \left( P h_{\ell b}(\cdot,\theta,\gamma,\lambda) -  P h_{\ell b}(\cdot,\theta',\gamma',\lambda') \right)\right|.
\end{align*}
Thus on $E_{n}$:
\begin{align*}
\mathscr{E}^{*}(\gamma)&\leq \mathscr{E}_{n}(\gamma)+ T(\sigma)+10\varepsilon\\
&= \mathscr{E}_{n}(\gamma)+ \frac{T(\sigma)}{\sigma} \sigma+10\varepsilon\\
&\leq \mathscr{E}_{n}(\gamma)+ \sup_{\delta \geq \sigma}\left( \frac{T(\delta)}{\delta}\right) \sigma+10\varepsilon\\
&= \mathscr{E}_{n}(\gamma) + T^\flat(\sigma)\sigma+10\varepsilon\\
&= \mathscr{E}_{n}(\gamma) + T^\flat(\sigma)\mathscr{E}^{*}(\gamma)+10\varepsilon.
\end{align*}
Now, since $\sigma \geq \delta^{**} > T^\sharp(1-1/\mathfrak{a})$ we have $T^\flat(\sigma)\leq T^\flat(\delta^{**}) \leq 1- 1/\mathfrak{a}$.\note{This seems to be the main reason why you pick $\delta^{**} > T^\sharp(1-1/\mathfrak{a}$.} Thus, on the event $E_{n}$, if $\gamma$ is such that $\sigma=\mathscr{E}^{*}(\gamma) \geq \delta^{**}$, we have:
\begin{align*}
&\mathscr{E}^{*}(\gamma) \leq \mathscr{E}_{n}(\gamma) + \left(1-\frac{1}{\mathfrak{a}}\right)\mathscr{E}^{*}(\gamma)+10\varepsilon \\
&\qquad\qquad\qquad\qquad\implies \mathscr{E}^{*}(\gamma) \leq \mathfrak{a}\mathscr{E}_{n}(\gamma) +10\mathfrak{a}\varepsilon.
\end{align*}
Since $\varepsilon>0$ is any value such that $\delta^{**}\geq \varepsilon$, and thus can be made arbitrarily small, we conclude that on the event $E_{n}$ we have for any $\gamma$:
\begin{align*}
\mathscr{E}^{*}(\gamma) \leq \mathfrak{a}\left(\mathscr{E}_{n}(\gamma)\vee \delta^{**}\right). 
\end{align*}
We will use this result to argue that, on the event $E_{n}$, if $\delta/\mathfrak{a}\geq \delta^{**}$ then $\mathscr{E}^{*}(\gamma) \leq \delta$. There are two cases:
\begin{enumerate}[label=(\roman*)]
	\item $\mathscr{E}_{n}(\gamma) \leq \delta^{**} \leq \delta/\mathfrak{a}$, which implies on the event $E_{n}$:
	\begin{align*}
\mathscr{E}^{*}(\gamma) \leq  \mathfrak{a}\left( \mathscr{E}_{n}(\gamma)\vee \delta^{**}\right) =\mathfrak{a}\delta^{**} \leq \delta.  
\end{align*}
	\item $\delta^{**} \leq \mathscr{E}_{n}(\gamma) \leq\delta/\mathfrak{a}$, which implies on the event $E_{n}$:
	\begin{align*}
\mathscr{E}^{*}(\gamma) \leq \mathfrak{a}\left( \mathscr{E}_{n}(\gamma)\vee \delta^{**}\right) = \mathfrak{a}\mathscr{E}^{*}(\gamma)\leq  \delta.  
\end{align*}
\end{enumerate}
Thus we conclude that for any $\delta/\mathfrak{a} \geq \delta^{**}$, on $E_{n}$ we have that $\mathscr{E}_{n}(\gamma) \leq \delta/\mathfrak{a} \implies \mathscr{E}^{*}(\gamma) \leq \delta$. Now recall that we have $\mathscr{E}_{n}(\gamma) \leq \delta/\mathfrak{a} \iff \gamma \in \mathscr{G}_{n}(\delta/\mathfrak{a})$ and $\mathscr{E}^{*}(\gamma) \leq \delta \iff \gamma \in \mathscr{G}^{*}(\delta)$. Thus, we conclude that for any $\delta \geq \mathfrak{a}\delta^{**}$, on the event $E_{n}$:
\begin{align*}
\mathscr{G}_{n}(\delta/\mathfrak{a}) \subseteq \mathscr{G}^{*}(\delta) ,
\end{align*}
as desired. This completes the proof.

\end{proof}

\subsection{Auxiliary Results and Proofs}

\subsubsection{On Issues of Measurability}\label{appendix_measurability}\label{appendix_on_issues_of_measurability}

The following discussion mirrors the discussion in \cite{dudley2010real} Section 3.3 and \cite{dudley2014uniform} Section 5.3. Let $\mathcal{X}$ be a Polish space, and let $\mathfrak{B}(\mathcal{X})$ be the Borel $\sigma-$algebra on $\mathcal{X}$. Then $(\mathcal{X},\mathfrak{B}(\mathcal{X}))$ is a measurable space. If $P$ is a probability law on $\mathfrak{B}(\mathcal{X})$, then $(\mathcal{X},\mathfrak{B}(\mathcal{X}),P)$ is a probability space. Now for any $B \subset \mathcal{X}$, we can define the outer measure $P^{*}$ on $B$ as:
\begin{align*}
P^{*}(B) := \inf\{P(C) : B \subset C \text{ and } C \in \mathfrak{B}(\mathcal{X})\}.
\end{align*}
By Theorem 3.3.1 in \cite{dudley2010real}, there always exists $C \in \mathfrak{B}(\mathcal{X})$ such that $P^{*}(B) = P(C)$, and such a set $C$ is called a measurable cover of $B$. Now define the collection of null sets for $P$ as:
\begin{align*}
Null(P):= \{A \subset \mathcal{X} : P^{*}(A) =0 \}.
\end{align*}
Furthermore, let $\mathfrak{B}_{P}^{*}(\mathcal{X})$ denote the smallest $\sigma-$algebra generated by $\mathfrak{B}(\mathcal{X})\cup Null(P)$. By Proposition 3.3.2 in \cite{dudley2010real}, we have:
\begin{align*}
\mathfrak{B}_{P}^{*}(\mathcal{X}):= \{B \subset \mathcal{X} : B \Delta C \in Null(P) \text{ for some } C \in \mathfrak{B}(\mathcal{X}) \},
\end{align*}
where $B \Delta C = (B\setminus C)\cup(C\setminus B)$. We can now extend the measure $P$ from $\mathfrak{B}(\mathcal{X})$ to a measure $\overline{P}$ on $\mathfrak{B}_{P}^{*}(\mathcal{X})$ as follows: if $B\Delta C \in Null(P)$ and $C \in \mathfrak{B}(\mathcal{X})$, then set $\overline{P}(B) = P(C)$. Proposition 3.3.3 in \cite{dudley2010real} verifies this is a valid extension; that is, $\overline{P}$ is a measure on $\mathfrak{B}_{P}^{*}(\mathcal{X})$ and $\overline{P}$ agrees with $P$ for all sets in $\mathfrak{B}(\mathcal{X})$.

However, note that the collection $\mathfrak{B}_{P}^{*}(\mathcal{X})$ clearly depends on the probability measure $P$. Indeed, if $Q$ is another measure on $\mathfrak{B}(\mathcal{X})$, and $\mathfrak{B}_{Q}^{*}(\mathcal{X})$ is defined in an analogous manner to $\mathfrak{B}_{P}^{*}(\mathcal{X})$, then it is possible for the two collections $\mathfrak{B}_{P}^{*}(\mathcal{X})$ and $\mathfrak{B}_{Q}^{*}(\mathcal{X})$ to differ because the null sets of $P$ and $Q$ differ. On the other hand, clearly both $\mathfrak{B}_{P}^{*}(\mathcal{X})$ and $\mathfrak{B}_{Q}^{*}(\mathcal{X})$ must have many elements in common; for example, both collections must contain the Borel sets $\mathfrak{B}(\mathcal{X})$.

A set $B \in \mathfrak{B}_{P}^{*}(\mathcal{X})$ is called \textit{measurable for the completion of $P$}. If for every probability measure $P$ the set $B$ is measurable for the completion of $P$, then we call $B$ \textit{universally measurable}. We will denote the universally measurable sets as $\mathfrak{B}^{*}(\mathcal{X})$; it is easily verified that $\mathfrak{B}^{*}(\mathcal{X})$ is also a $\sigma-$algebra.\footnote{This follows from the fact that an arbitrary intersection of $\sigma-$algebras is a $\sigma-$algebra.} By definition, for any two probability measures $P$ and $Q$, both $\mathfrak{B}_{P}^{*}(\mathcal{X})$ and $\mathfrak{B}_{Q}^{*}(\mathcal{X})$ contain the universally measurable sets. Also note that, in our example, clearly the Borel sets $\mathfrak{B}(\mathcal{X})$ are universally measurable.

A subset $A\subset \mathcal{X}$ of a Polish space $\mathcal{X}$ (with the Borel $\sigma-$field) is called $\mathfrak{B}(\mathcal{X})$-\textit{analytic} if there exists a compact metric space $\mathcal{Y}$ such that $A$ is the projection onto $\mathcal{X}$ of some $B \in \mathfrak{B}(\mathcal{X})\otimes \mathfrak{B}(\mathcal{Y})$.\footnote{We note that this is one of many equivalent definitions of an analytic set. See Chapter 8 of \cite{cohn2013measure}. Our definition is from \cite{stinchcombe1992some}.} A function $f:A \to [-\infty,\infty]$ is called lower (or upper) semi-analytic if $A$ is an analytic set and $\{x \in A : f(x) < c\}$ (or $\{x \in A : f(x) \geq c\}$) is an analytic set for every $c \in \mathbb{R}$; that is, if the epigraph (or hypograph) of $f$ is an analytic set. In a Polish space, every analytic set is universally measurable. We refer the reader to Chapter 8 of \cite{cohn2013measure} for additional details.

\begin{lemma}[Infimum over Random Sets is Lower Semi-Analytic]\label{lemma_measurable_functions}
Suppose that Assumptions \ref{assump_preliminary}, \ref{assumption_factual_domain} and \ref{assumption_counterfactual_domain} hold. Then for any lower semi-analytic function $f:\mathcal{V}\times\Gamma\times\Theta\times\{0,1\}^{J} \to \mathbb{R}$, the function $f_{lb,1}(y,z,u,\theta,\gamma,\lambda)$ given by:
\begin{align}
f_{lb,1}(y,z,u,\theta,\gamma,\lambda)&:=\inf_{y^{\star} \in \bm G^\star(y,z,u,\theta,\gamma)} f(v,\theta,\gamma,\lambda),\label{eq_measurable_2}
\end{align}
is lower semi-analytic; that is, $\{(y,z,u,\gamma,\theta,\lambda): f_{lb,1}(y,z,u,\theta,\gamma,\lambda) <r\}$ is an analytic set for every $r \in \mathbb{R}$, and thus is universally measurable. In addition, the function $f_{lb,2}(y,z,\theta,\gamma,\lambda)$ given by:
\begin{align}
f_{lb,2}(y,z,\theta,\gamma,\lambda)&:=\inf_{u \in \bm G^{-}(y,z,\theta)}f_{lb,1}(y,z,u,\theta,\gamma,\lambda),\label{eq_measurable_4}
\end{align}
is also lower semi-analytic; that is, $\{(y,z,\theta,\gamma,\lambda): f_{lb,2}(y,z,\theta,\gamma,\lambda)<r\}$ is an analytic set for every $r \in \mathbb{R}$, and thus is universally measurable.
\end{lemma}
\begin{remark}
Defining $f_{ub,1}(y,z,u,\theta,\gamma,\lambda)$ and $f_{ub,2}(y,z,u,\theta,\gamma,\lambda)$ as the analogous functions with the infimum replaced with the supremum, it is possible to show that $f_{ub,1}(y,z,u,\theta,\gamma,\lambda)$ and $f_{ub,2}(y,z,u,\theta,\gamma,\lambda)$ are upper semi-analytic. 
\end{remark}
\begin{proof}[Proof of Lemma \ref{lemma_measurable_functions}]
Recall that under Assumption \ref{assumption_counterfactual_domain}, the multifunction $\bm G^{\star}(y,z,u,\theta,\gamma)$ is Effros measurable with respect to the product Borel $\sigma-$algebra $\mathfrak{B}(\mathcal{Y})\otimes \mathfrak{B}(\mathcal{Z})\otimes \mathfrak{B}(\mathcal{U})\otimes\mathfrak{B}(\Gamma)$. By \cite{molchanov2017theory} Theorem 1.3.3 this implies that:
\begin{align*}
\text{Graph}(\bm G^{\star}) \in \mathfrak{B}(\mathcal{Y})\otimes \mathfrak{B}(\mathcal{Z})\otimes \mathfrak{B}(\mathcal{U})\otimes\mathfrak{B}(\Theta)\otimes\mathfrak{B}(\Gamma).
\end{align*}
Thus $\text{Graph}(\bm G^{\star})$ is a Borel (and thus also an analytic) set. Now note that $\bm G^\star(y,z,u,\theta,\gamma)$ can be rewritten as:
\begin{align*}
\bm G^\star(y,z,u,\theta,\gamma):= \left\{y^\star \in \mathcal{Y}^\star : (y,z,u,y^\star,\theta,\gamma) \in \text{Graph}(\bm G^{\star}) \right\}.
\end{align*}
The fact that $f_{lb,1}: \mathcal{V}\times\Gamma\times\Theta\times\{0,1\}^{J} \to \mathbb{R}$ is lower semi-analytic then follows directly from the selection Theorem of \cite{shreve1978alternative}, p. 968.\footnote{See also \cite{bertsekas1978stochastic} Proposition 7.47, p. 179.} Taking $f_{lb,1}(y,z,u,\theta,\gamma,\lambda)$ as lower semi-analytic, a nearly identical proof shows that $f_{lb,2}(y,z,\theta,\gamma,\lambda)$ is also lower semi-analytic. 
\end{proof}

\begin{proposition}\label{proposition_measurable}
Suppose the assumptions of Theorem \ref{thm_cortes} hold. Then the maps $\gamma \mapsto I_{\ell b}[\varphi](\gamma), I_{u b}[\varphi](\gamma)$ are universally measurable.
\end{proposition}
\begin{proof}
We will focus on the map $\gamma \mapsto I_{\ell b}[\varphi](\gamma)$, as the proof for the upper envelope function is symmetric. By Theorem \ref{thm_cortes} we have:
\begin{align*}
I_{\ell b}[\varphi](\gamma) = \inf_{\theta \in \Theta} \max_{\lambda \in \{0,1\}^{J}} \int h_{\ell b}(y,z,\theta,\gamma,\lambda) \, dP_{Y,Z},
\end{align*}
where:
\begin{align*}
h_{\ell b}(y,z,\theta,\gamma,\lambda):= \inf_{u \in \bm G^{-}(y,z,\theta)}  \inf_{y^{\star}  \in \bm G^{\star}(y,z,u,\theta,\gamma)} \Bigg( \varphi(v) + \mu^{*}\sum_{j=1}^{J}  \lambda_{j} m_{j}(y,z,u,\theta)\Bigg).
\end{align*}
Suppose that $h_{\ell b}(y,z,\theta,\gamma,\lambda)$ is lower semi-analytic (we will return to this in a moment). Then by proposition 7.46 in \cite{bertsekas1978stochastic}, the map:
\begin{align}
(\theta,\gamma,\lambda) \mapsto \int h_{\ell b}(y,z,\theta,\gamma,\lambda) \, dP_{Y,Z},\label{eq_semi_analytic1}
\end{align}
is lower semi-analytic. Furthermore, suppose that $g_{1}: \mathbb{R} \to \mathbb{R}$ and $g_{2}:\mathbb{R} \to \mathbb{R}$ are lower semi-analytic. The function $g(x) = g_{1}(x)\vee g_{2}(x)$ satisfies:
\begin{align*}
g^{-1}((-\infty,r)) = g_{1}^{-1}((-\infty,r))\cup g_{2}^{-1}((-\infty,r)).
\end{align*}
Since analytic sets are closed under (countable) unions and intersections (\cite{parthasarathy2005probability} Theorem 3.1), we have that $g$ is lower semi-analytic whenever $g_{1}$ and $g_{2}$ are lower semi-analytic. From this we conclude that the function:
\begin{align*}
(\theta,\gamma) \mapsto \max_{\lambda \in \{0,1\}^{J}} \int h_{\ell b}(y,z,\theta,\gamma,\lambda) \, dP_{Y,Z},
\end{align*}
is lower semi-analytic, being the pointwise maximum of at most $2^{J}$ lower semi-analytic functions of the form \eqref{eq_semi_analytic1}. Finally, by the selection Theorem of \cite{shreve1978alternative}, p. 968 (see also \cite{bertsekas1978stochastic} Proposition 7.47) we have the map:
\begin{align*}
\gamma \mapsto \sup_{\theta \in \Theta} \max_{\lambda \in \{0,1\}^{J}} \int h_{\ell b}(y,z,\theta,\gamma) \, dP_{Y,Z},
\end{align*}
is lower semi-analytic, and thus universally measurable. It thus remains only to show that $h_{\ell b}(y,z,\theta,\gamma,\lambda)$ is lower semi-analytic. By Lemma \ref{lemma_measurable_functions}, $h_{\ell b}(y,z,\theta,\gamma)$ will be lower semi-analytic if we can show the function:
\begin{align}
(v,\theta,\gamma,\lambda)\mapsto \varphi(v) + \mu^{*}\sum_{j=1}^{J}  \lambda_{j}  m_{j}(y,z,u,\theta),\label{eq_semi_analytic2}
\end{align}
is lower semi-analytic. Since both $\varphi(v)$ and $\{m_{j}(y,z,u,\theta)\}_{j=1}^{J}$ are Borel measurable by assumption, since the composition of Borel measurable functions are Borel measurable we conclude that \eqref{eq_semi_analytic2} is Borel measurable. The conclusion then follows from the fact that every Borel measurable function is lower semi-analytic.
\end{proof}

A nearly identical argument shows that, for every fixed sequence $(\xi_{1},\ldots,\xi_{n})\in \{-1,1\}^{n}$, the Rademacher complexity:
\begin{align*}
((y_{1},z_{1}),\ldots,(y_{n},z_{n})) \mapsto ||\mathfrak{R}||(\mathcal{H}_{\ell b}), 
\end{align*}
is universally measurable. This is stated as a corollary of the previous result for easy reference. 
\begin{corollary}\label{corollary_universal_rademacher}
Suppose the assumptions of Theorem \ref{thm_cortes} hold, and suppose that the sequence $(Y_{1},Z_{1})$, $\ldots$, $(Y_{n},Z_{n})$ are the coordinate projections of the product probability space $((\mathcal{Y}\times\mathcal{Z})^{n},(\mathfrak{B}(\mathcal{Y})\otimes\mathfrak{B}(\mathcal{Z}))^{\otimes n},P_{Y,Z}^{\otimes n}))$. Then the map:
\begin{align*}
((Y_{1},Z_{1}),\ldots,(Y_{n},Z_{n})) \mapsto ||\mathfrak{R}||(\mathcal{H}_{\ell b}), 
\end{align*}
is universally measurable; that is, it is measurable for the completion of $P_{Y,Z}^{\otimes n}$ for any $P_{Y,Z} \in \mathcal{P}_{Y,Z}$. 
\end{corollary}

\subsubsection{Respect for Weak Dominance of the Preference Relation in Definition \ref{definition_pac_preference_relation}}

\begin{lemma}\label{lemma_on_weak_dominance}
Let $(\Omega,\mathfrak{A})$ be a measurable space, and let $X_{1},X_{2}: \Omega\times \mathcal{T} \to \mathbb{R}$ be two stochastic processes such that $X_{1}(\cdot,t)$ and $X_{2}(\cdot,t)$ are measurable for each $t$, and $\omega\mapsto \inf_{t \in \mathcal{T}} X_{1}(\omega,t),\inf_{t \in \mathcal{T}} X_{2}(\omega,t)$ are universally measurable; that is, measurable with respect to the completion of any probability measure on $(\Omega,\mathfrak{A})$. Furthermore, suppose that for any probability measure on $(\Omega,\mathfrak{A})$ we have $X_{1}(\omega,t)\leq X_{2}(\omega,t)$ a.s. for every $t \in \mathcal{T}$, and let $c:\mathcal{P} \to \mathbb{R}_{++}$ be any value depending only on $P$, where $\mathcal{P}$ is the set of all probability measure on $(\Omega,\mathfrak{A})$. Finally, let $c_{1},c_{2}: (0,1)\times \mathcal{P}\to \mathbb{R}_{++}$ be the smallest values satisfying:
\begin{align*}
P \left( \inf_{t \in \mathcal{T}} X_{1}(\omega,t) + c_{1}(\kappa,P) \geq c(P) \right) \geq \kappa,\\
P \left( \inf_{t \in \mathcal{T}} X_{2}(\omega,t) + c_{2}(\kappa,P) \geq c(P) \right) \geq \kappa,
\end{align*}
for each $\kappa \in (0,1)$. Then for every $P\in \mathcal{P}$ we have $c_{2}(\kappa,P)\leq c_{1}(\kappa,P)$ for every $\kappa \in (0,1)$.
\end{lemma}
\begin{proof}
Fix any probability measure $P\in \mathcal{P}$. Then by assumption:
\begin{align*}
X_{1}(\omega,t) \leq X_{2}(\omega,t)\,\, a.s.\qquad \forall t\in \mathcal{T}. 
\end{align*}
This implies:
\begin{align*}
\inf_{t \in \mathcal{T}} X_{1}(\omega,t) \leq X_{2}(\omega,t)\,\, a.s.\qquad \forall t\in \mathcal{T},
\end{align*}
which in turn implies:
\begin{align*}
\inf_{t \in \mathcal{T}} X_{1}(\omega,t) \leq \inf_{t \in \mathcal{T}} X_{2}(\omega,t)\,\, a.s.,
\end{align*}
and thus:
\begin{align*}
\inf_{t \in \mathcal{T}} X_{1}(\omega,t) -c(P) \leq \inf_{t \in \mathcal{T}} X_{2}(\omega,t)-c(P)\,\, a.s.
\end{align*}
Let $N$ denote the null set for which this relation is not true (this set may depend on $P \in \mathcal{P}$). Then we have for every $x \in \mathbb{R}$:
\begin{align*}
\left\{\omega : \inf_{t \in \mathcal{T}} X_{1}(\omega,t)-c(P) > x \right\}\cap N^{c} \subseteq \left\{\omega : \inf_{t \in \mathcal{T}} X_{2}(\omega,t)-c(P) > x \right\}\cap N^{c}, 
\end{align*}
By assumption, these events belong to the universal $\sigma-$algebra generated by $\mathfrak{A}$, and so are measurable with respect to the completion of any $P \in \mathcal{P}$. This implies that for every $x \in \mathbb{R}$:
\begin{align*}
P\left(\omega : \inf_{t \in \mathcal{T}} X_{1}(\omega,t)-c(P) > x \right) \leq P\left(\omega : \inf_{t \in \mathcal{T}} X_{2}(\omega,t)-c(P) > x \right). 
\end{align*}
Now taking any $\kappa \in (0,1)$ and setting $x=-c_{1}(\kappa,P)$ we have:
\begin{align*}
\kappa \leq P\left(\omega : \inf_{t \in \mathcal{T}} X_{1}(\omega,t) + c_{1}(\kappa,P) > c(P) \right) \leq P\left(\omega : \inf_{t \in \mathcal{T}} X_{2}(\omega,t) +c_{1}(\kappa,P) > c(P) \right). 
\end{align*}
By definition, this implies $c_{2}(\kappa,P)$ can be no larger than $c_{1}(\kappa,P)$; that is, $c_{2}(\kappa,P)\leq c_{1}(\kappa,P)$. Since $\kappa \in (0,1)$ was arbitrary, we conclude that $c_{2}(\kappa,P)\leq c_{1}(\kappa,P)$ for every $\kappa \in (0,1)$. Since $P \in \mathcal{P}$ was also arbitrary we conclude that for every $P \in \mathcal{P}$ we have $c_{2}(\kappa,P)\leq c_{1}(\kappa,P)$ for every $\kappa \in (0,1)$. This completes the proof. 
\end{proof}

\subsubsection{Results on Interchanging Integrals and Supremum/Infimum}

\begin{lemma}[Interchange of Integral and Supremum/Infimum]\label{lemma_joint_to_marginal_selection}
Let $(\mathcal{V},\mathfrak{B}(\mathcal{V}))$ and $(\mathcal{V}',\mathfrak{B}(\mathcal{V}'))$ be measurable spaces with $\mathcal{V}$ and $\mathcal{V}'$ as Polish spaces. Let $V \in \mathcal{V}$ be any random variable defined on the probability space $(\Omega,\mathfrak{A},P)$ with (marginal) distribution $P_{V} = P \circ V^{-1}$. Furthermore, let $\bm G: \mathcal{V} \to \mathcal{V}'$ be an almost surely non-empty Effros-measurable multifunction. Then for any bounded and measurable function $\varphi: \mathcal{V} \times \mathcal{V}'\to \mathbb{R}$, we have:
\begin{align}
\int \sup_{v' \in \bm G(v)} \varphi(v,v') \,dP_{V} = \sup_{V' \in Sel(\bm G)} \int \varphi(v,V'(v)) \,dP_{V}, \label{eq_graf_lemma_sup_cond}\\
\int \inf_{v' \in \bm G(v)} \varphi(v,v') \,dP_{V} = \inf_{V' \in Sel(\bm G)} \int \varphi(v,V'(v)) \,dP_{V},\label{eq_graf_lemma_inf_cond}
\end{align} 
In particular, if:
\begin{align*}
\mathcal{P}_{V'|V}:= \{P_{V'|V} : V' \sim P_{V'|V}, \, \,V': \mathcal{V} \to \mathcal{V}' \text{ is measurable and } P_{V'|V}(V' \in \bm G(V)|V=v)=1\,\,a.s. \},
\end{align*}
then:
\begin{align}
\int \sup_{v' \in \bm G(v)} \varphi(v,v') \,dP_{V} = \sup_{P_{V'|V} \in \mathcal{P}_{V'|V}} \int \varphi(v,v') \,d(P_{V'|V} \times P_{V}), \\
\int \sup_{v' \in \bm G(v)} \varphi(v,v') \,dP_{V} = \inf_{P_{V'|V} \in \mathcal{P}_{V'|V}} \int \varphi(v,v') \,d(P_{V'|V} \times P_{V}).
\end{align} 
\end{lemma}
\begin{proof}[Proof of Lemma \ref{lemma_joint_to_marginal_selection}]
Since $\bm G$ is Effros measurable, by Theorem 1.3.3 in \cite{molchanov2017theory} we have that $\text{gr}(\bm G) \in \mathfrak{B}(\mathcal{V})\otimes \mathfrak{B}(\mathcal{V}')$, and thus $\text{gr}(\bm G)$ is trivially an analytic set. Now define:
\begin{align*}
\text{gr}_{v}(\bm G):= \{ v' \in \mathcal{V}' : (v,v') \in \text{gr}(\bm G) \}.
\end{align*}
Now let:
\begin{align*}
\varphi^{*}(v):=\sup_{v' \in \bm G(v)} \varphi(v,v') = \sup_{v' \in \text{gr}_{v}(\bm G)} \varphi(v,v').
\end{align*}
Furthermore, define the set:
\begin{align*}
M:= \left\{v \in \Pi_{\mathcal{V}}(\text{gr}(\bm G)) : \exists v' \in \text{gr}_{v}(\bm G) \text{ s.t. } \varphi(v,v') = \varphi^{*}(v) \right\}.
\end{align*}
where $\Pi_{\mathcal{V}}: \mathcal{V} \times \mathcal{V}' \to \mathcal{V}$ is the projection operator. Fix any $\varepsilon>0$. By the Exact Selection Theorem (\cite{shreve1979universally}, p.16) there exists a universally measurable function $\tilde{v}':\mathcal{V} \to \mathcal{V}'$ such that $(v,\tilde{v}'(v)) \in \text{gr}(\bm G)$ for every $v \in \Pi_{\mathcal{V}}(\text{gr}(\bm G))$ and:\note{Here boundedness is essential to rule out $+\infty$}
\begin{align*}
\varphi(v,\tilde{v}'(v)) \begin{cases}
	= \varphi^{*}(v), &\text{ if } v \in M,\\
	\geq \varphi^{*}(v) - \varepsilon, &\text{ if } v \notin M.
\end{cases}
\end{align*}
This allows us to write:
\begin{align*}
\int \sup_{v' \in \bm G(v)} \varphi(v,v') \, dP_{V} \leq \int  \varphi(v,\widetilde{v}'(v)) \, dP_{V} + \varepsilon. 
\end{align*}
Since $\tilde{v}'$ is a (universally) measurable selection, clearly we have:
\begin{align*}
\int  \varphi(v,\widetilde{v}'(v)) \, dP_{V} \leq \sup_{V' \in Sel(\bm G)} \int \varphi(v,V(v)) \, dP_{V}
\end{align*}
It suffices to show:
\begin{align*}
\sup_{V' \in Sel(\bm G)} \int \varphi(v,V'(v)) \, dP_{V} \leq \int \sup_{v' \in \bm G(v)} \varphi(v,v') \, dP_{V} .
\end{align*}
For any $\varepsilon>0$, let $V_{\varepsilon}' \in Sel(\bm G)$ satisfy:
\begin{align*}
\sup_{V' \in Sel(\bm G)} \int \varphi(v,V'(v)) \, dP_{V}  \leq \int \varphi(v,V_{\varepsilon}'(v)) \, dP_{V} + \varepsilon.
\end{align*}
Furthermore, let $N:= \{v : V_{\varepsilon}(v) \notin \bm G(v)\}$. Then by definition of $Sel(\bm G)$ we have $P(N)=0$. Thus:
\begin{align*}
\int \varphi(v,V_{\varepsilon}(v)) \, dP_{V} = \int_{N^c} \varphi(v,V_{\varepsilon}'(v)) \, dP_{V} \leq \int_{N^{c}}  \sup_{v' \in \bm G(v)} \varphi(v,v')\,dP_{V} \leq \int \sup_{v' \in \bm G(v)} \varphi(v,v') \,dP_{V}. 
\end{align*}
Combining everything we have:
\begin{align*}
\int \sup_{v' \in \bm G(v)} \varphi(v,v') \, dP_{V} \leq \sup_{V' \in Sel(\bm G)} \int \varphi(v,V(v)) \, dP_{V} + \varepsilon \leq \int \sup_{v' \in \bm G(v)} \varphi(v,v') \,dP_{V} + 2 \varepsilon
\end{align*}
Since $\varepsilon>0$ was arbitrary, we conclude:
\begin{align*}
\int \sup_{v' \in \bm G(v)} \varphi(v,v') \,dP_{V} = \sup_{V' \in Sel_{u.m.}(\bm G)} \int \varphi(v,V'(v)) \,dP_{V}.
\end{align*}
Since each $V' \in Sel_{u.m.}(\bm G)$ is universally measurable, each $V'$ can be associated with a $\mathfrak{B}(\mathcal{V})-$measurable random variable $\tilde{V}'$ such that $V' = \tilde{V}'$ a.s. Thus we can conclude:
\begin{align*}
\int \sup_{v' \in \bm G(v)} \varphi(v,v') \,dP_{V} = \sup_{V' \in Sel(\bm G)} \int \varphi(v,V'(v)) \,dP_{V}.
\end{align*}
To show the final claim, note that for any $V':\mathcal{V} \to \mathcal{V}'$ we have:
\begin{align*}
P_{V'|V}(V'=v'|V=v) = \mathbbm{1}\{V'(v)=v'\}.
\end{align*}
i.e. the conditional distribution of $V'$ is degenerate. Thus for any $V' \in Sel(\bm G)$:
\begin{align*}
\int \int \varphi(v,v') \,d(P_{V'|V}\times P_{V}) &= \int \varphi(v,v') \mathbbm{1}\{V'(v)=v'\} \,dP_{V}\\
&= \int \varphi(v,V'(v)) \,dP_{V}.
\end{align*}
By definition of $\mathcal{P}_{V'|V}$, we conclude that:\todo{Check with Joon, as this still feels a bit shaky.}
\begin{align*}
\sup_{V' \in Sel(\bm G)} \int \varphi(v,V'(v)) \,dP_{V} = \sup_{P_{V'|V} \in \mathcal{P}_{V'|V}} \int \varphi(v,v') \,d(P_{V'|V}\times P_{V}).
\end{align*}
\noteinline{Recall your previous approach. You tried to show that $\int \sup_{v' \in \bm G(v)} \varphi(v,v') $ could be written as an integral with respect to a capacity functional. You could then use the interchange results of Graf or Miranda to show the integral with respect to the capacity functional could be written as $\sup_{V \in Sel(\bm G)} \int \varphi(v,V(v)) \P_{V}$. However, this argument only works well for problems like $\int \sup_{v' \in \bm G(v)} \varphi(v')$, in which case it seems all the measurability problems with projections do not arise; for example, with $\varphi\geq 0$ this can be written as $\int_{0}^{\infty} P( \bm G(v) \cap \{ v' : \varphi(v') \geq t\}\neq \emptyset )$, whereas in your problem you get $\int_{0}^{\infty} P( \text{gr}(\bm G(v)) \cap \{(v,v') : \varphi(v,v') \geq t\}\neq \emptyset )$, which is not necessarily the capacity functional of a random closed set. }

\end{proof}

\subsubsection{Results on Error Bounds}

In the next Lemma we focus on the lower envelope function, although clearly an analogous result is true for the upper envelope function. For notational simplicity, denote:
\begin{align}
\varphi^{*}:= \inf_{\theta \in \Theta^{*}}\inf_{P_{U|Y,Z} \in \mathcal{P}_{U|Y,Z}(\theta)}\inf_{P_{Y_{\gamma}^\star|Y,Z,U} \in \mathcal{P}_{Y_{\gamma}^\star|Y,Z,U}(\theta,\gamma)} \int \varphi(v)  \, dP_{V_{\gamma}}.
\end{align}
We now have the following result:
\begin{lemma}[Equality Between Primal and Penalized Problems]\label{lemma_error_bound}
Suppose the Assumptions of Theorem \ref{thm_cortes} hold. Then there exists functions $\lambda_{j}^{\ell b}: \Theta \times \mathcal{P}_{Y,Z} \to \{0,1\}$, $j=1,\ldots,J$, depending only on $\theta$ and the distribution $P_{Y,Z}$, such that:
\begin{align*}
\varphi^{*} = \inf_{\theta \in \Theta}  \inf_{P_{U|Y,Z} \in \mathcal{P}_{U|Y,Z}(\theta)}\Bigg(\inf_{P_{Y_{\gamma}^\star|Y,Z,U} \in \mathcal{P}_{Y_{\gamma}^\star|Y,Z,U}(\theta,\gamma)} \int \varphi(v)  \, dP_{V_{\gamma}} + \mu^{*}\sum_{j=1}^{J} \lambda_{j}^{\ell b}(\theta,P_{Y,Z}) \E_{P_{Y,Z,U}} [ m_{j}(y,z,u,\theta)]\Bigg).
\end{align*}

\end{lemma}
\begin{remark}
Recall that $\mathcal{P}_{Y,Z}$ is the set of all Borel probability measures on $\mathcal{Y}\times \mathcal{Z}$.
\end{remark}

\begin{proof}[Proof of Lemma \ref{lemma_error_bound}]
First, note that for any functions $\lambda_{j}^{\ell b}: \Theta \times \mathcal{P}_{Y,Z} \to \{0,1\}$, $j=1,\ldots,J$, we have:
\begin{align*}
\varphi^{*}&:= \inf_{\theta \in \Theta^{*}}\inf_{P_{U|Y,Z} \in \mathcal{P}_{U|Y,Z}(\theta)}\inf_{P_{Y_{\gamma}^\star|Y,Z,U} \in \mathcal{P}_{Y_{\gamma}^\star|Y,Z,U}(\theta,\gamma)} \int \varphi(v)  \, dP_{V_{\gamma}}\\
&= \inf_{\theta\in \Theta^{*}}\inf_{P_{U|Y,Z} \in \mathcal{P}_{U|Y,Z}(\theta)}\inf_{P_{Y_{\gamma}^\star|Y,Z,U} \in \mathcal{P}_{Y_{\gamma}^\star|Y,Z,U}(\theta,\gamma)} \sup_{\lambda \in \mathbb{R}_{+}^{J}} \Bigg(\int \varphi(v)  \, dP_{V_{\gamma}} + \mu^{*} \sum_{j=1}^{J} \lambda_{j} \E_{P_{Y,Z,U}} [ m_{j}(y,z,u,\theta)]\Bigg)\\
&\geq \inf_{\theta\in \Theta^{*}}\inf_{P_{U|Y,Z} \in \mathcal{P}_{U|Y,Z}(\theta)}\inf_{P_{Y_{\gamma}^\star|Y,Z,U} \in \mathcal{P}_{Y_{\gamma}^\star|Y,Z,U}(\theta,\gamma)} \Bigg(\int \varphi(v)  \, dP_{V_{\gamma}} + \mu^{*} \sum_{j=1}^{J} \lambda_{j}^{\ell b}(\theta,P_{Y,Z}) \E_{P_{Y,Z,U}} [ m_{j}(y,z,u,\theta)]\Bigg)\\
&\geq \inf_{\theta\in \Theta}\inf_{P_{U|Y,Z} \in \mathcal{P}_{U|Y,Z}(\theta)}\Bigg(\inf_{P_{Y_{\gamma}^\star|Y,Z,U} \in \mathcal{P}_{Y_{\gamma}^\star|Y,Z,U}(\theta,\gamma)}  \int \varphi(v)  \, dP_{V_{\gamma}} + \mu^{*} \sum_{j=1}^{J} \lambda_{j}^{\ell b}(\theta,P_{Y,Z}) \E_{P_{Y,Z,U}} [ m_{j}(y,z,u,\theta)]\Bigg).
\end{align*}
It thus suffices to show that there exists functions $\lambda_{j}^{\ell b}: \Theta \times \mathcal{P}_{Y,Z} \to \{0,1\}$ for $j=1,\ldots,J$ satisfying the reverse inequality. This is done constructively. In particular, define:
\begin{align*}
\mathcal{J}^{*}(\theta,P_{Y,Z})&:= \bigg\{ j \in \{1,\ldots,J\} :  \inf_{P_{U|Y,Z} \in \mathcal{P}_{U|Y,Z}(\theta)}\E_{P_{U|Y,Z}\times P_{Y,Z}} \left[ m_{j}(y,z,u,\theta) \right]\\
&\qquad\qquad\qquad\qquad\qquad\qquad= \inf_{P_{U|Y,Z} \in \mathcal{P}_{U|Y,Z}(\theta)}\max_{j=1,\ldots,J} |\E_{P_{U|Y,Z}\times P_{Y,Z}} \left[ m_{j}(y,z,u,\theta) \right]|_{+} \bigg\}.
\end{align*}
That is, the set $\mathcal{J}^{*}(\theta,P_{Y,Z})$ returns the indices of the weakly positive (i.e. weakly violated) moment functions that obtain the inner maximum in the problem:\note{Intuitively, $\lambda_{j}$ will select the most problematic moment (farthest from being satisfied for even the best case possible $P_{U|Y,Z}$)}
\begin{align*}
\inf_{P_{U|Y,Z} \in \mathcal{P}_{U|Y,Z}(\theta)}\max_{j=1,\ldots,J} |\E_{P_{U|Y,Z}\times P_{Y,Z}} \left[ m_{j}(y,z,u,\theta) \right]|_{+}.
\end{align*}
Now set:
\begin{align}
\lambda_{j}^{\ell b}(\theta,P_{Y,Z}):= \mathbbm{1}\left\{j \in \mathcal{J}^{*}(\theta,P_{Y,Z}) \right\}.\label{eq_lambda_j_definition}
\end{align}
To show why this choice works, begin by fixing any $\theta \in \Theta_{\delta}^{*}$. By Assumption \ref{assumption_error_bound}(ii) we have:
\begin{align}
C_{2} d(\theta,\Theta^{*}) &\geq \varphi^{*}-\inf_{P_{U|Y,Z} \in \mathcal{P}_{U|Y,Z}(\theta)}\inf_{P_{Y_{\gamma}^\star|Y,Z,U} \in \mathcal{P}_{Y_{\gamma}^\star|Y,Z,U}(\theta,\gamma)} \int \varphi(v)  \, dP_{V_{\gamma}}.\label{eq_error_bound1}
\end{align}
Furthermore, from Assumption \ref{assumption_error_bound}(i), since $\theta \in \Theta_{\delta}^{*}$ by assumption, we have:
\begin{align}
C_{1} d(\theta,\Theta^{*})&=C_{1} \min\{\delta, d(\theta,\Theta^{*})\}\\
 &\leq \inf_{P_{U|Y,Z} \in \mathcal{P}_{U|Y,Z}(\theta)}\max_{j=1,\ldots,J} |\E_{P_{U|Y,Z}\times P_{Y,Z}} \left[ m_{j}(y,z,u,\theta) \right]|_{+} \nonumber\\
&=\inf_{P_{U|Y,Z} \in \mathcal{P}_{U|Y,Z}(\theta)} \lambda_{j}^{\ell b}(\theta,P_{Y,Z}) |\E_{P_{U|Y,Z}\times P_{Y,Z}} \left[ m_{j}(y,z,u,\theta) \right]|_{+}  \nonumber\\
&=\inf_{P_{U|Y,Z} \in \mathcal{P}_{U|Y,Z}(\theta)} \lambda_{j}^{\ell b}(\theta,P_{Y,Z}) \E_{P_{U|Y,Z}\times P_{Y,Z}} \left[ m_{j}(y,z,u,\theta) \right]  \nonumber\\
&\leq \sum_{j=1}^{J} \inf_{P_{U|Y,Z} \in \mathcal{P}_{U|Y,Z}(\theta)} \lambda_{j}^{\ell b}(\theta,P_{Y,Z}) \E_{P_{U|Y,Z}\times P_{Y,Z}} \left[ m_{j}(y,z,u,\theta) \right]  \nonumber\\
&\leq \inf_{P_{U|Y,Z} \in \mathcal{P}_{U|Y,Z}(\theta)}\sum_{j=1}^{J}\lambda_{j}^{\ell b}(\theta,P_{Y,Z}) \E_{P_{Y,Z,U}} [ m_{j}(y,z,u,\theta)].\label{eq_error_bound2}
\end{align}
Now by construction we have $\mu^{*}\geq C_{2}/C_{1}$. Thus:
\begin{align}
C_{2}d(\theta,\Theta^{*}) \leq \mu^{*}C_{1}d(\theta,\Theta^{*}).\label{eq_error_bound3}
\end{align}
Now use \eqref{eq_error_bound3} to combine \eqref{eq_error_bound1} and \eqref{eq_error_bound2} and rearrange to obtain:
\begin{align*}
\varphi^{*}
&\leq  \inf_{P_{U|Y,Z} \in \mathcal{P}_{U|Y,Z}(\theta)}\inf_{P_{Y_{\gamma}^\star|Y,Z,U} \in \mathcal{P}_{Y_{\gamma}^\star|Y,Z,U}(\theta,\gamma)} \int \varphi(v)  \, dP_{V_{\gamma}} \\
&\qquad\qquad\qquad\qquad\qquad\qquad\qquad+ \mu^{*}\inf_{P_{U|Y,Z} \in \mathcal{P}_{U|Y,Z}(\theta)}\sum_{j=1}^{J}\lambda_{j}^{\ell b}(\theta,P_{Y,Z}) \E_{P_{Y,Z,U}} [ m_{j}(y,z,u,\theta)] \\
&\leq  \inf_{P_{U|Y,Z} \in \mathcal{P}_{U|Y,Z}(\theta)}\Bigg(\inf_{P_{Y_{\gamma}^\star|Y,Z,U} \in \mathcal{P}_{Y_{\gamma}^\star|Y,Z,U}(\theta,\gamma)} \int \varphi(v)  \, dP_{V_{\gamma}}+ \mu^{*}\sum_{j=1}^{J} \lambda_{j}^{\ell b}(\theta,P_{Y,Z}) \E_{P_{Y,Z,U}} [ m_{j}(y,z,u,\theta)] \Bigg),
\end{align*}
which holds for all $\theta \in \Theta_{\delta}^{*}$. To complete the proof, consider any $\theta \in \Theta \setminus \Theta_{\delta}^{*}$. Recall from the assumptions of Theorem \ref{thm_cortes} that $\varphi:\mathcal{V} \to [\varphi_{\ell b}, \varphi_{ub}] \subset \mathbb{R}$. Then using Assumption \ref{assumption_error_bound} we have: 
\begin{align*}
& \inf_{P_{U|Y,Z} \in \mathcal{P}_{U|Y,Z}(\theta)}\left(\inf_{P_{Y_{\gamma}^\star|Y,Z,U} \in \mathcal{P}_{Y_{\gamma}^\star|Y,Z,U}(\theta,\gamma)} \int \varphi(v)  \, dP_{V_{\gamma}} + \mu^{*}\sum_{j=1}^{J}\lambda_{j}^{\ell b}(\theta,P_{Y,Z}) \E_{P_{U|Y,Z}\times P_{Y,Z}} [ m_{j}(y,z,u,\theta)]\right)\\
&\qquad\geq \varphi_{\ell b} + \mu^{*} \inf_{P_{U|Y,Z} \in \mathcal{P}_{U|Y,Z}(\theta)}\int \left(  \sum_{j=1}^{J}\lambda_{j}^{\ell b}(\theta,P_{Y,Z})  m_{j}(y,z,u,\theta) \right)\, d(P_{U|Y,Z}\times P_{Y,Z})\\
&\qquad\geq  \varphi_{\ell b} + \mu^{*} \inf_{P_{U|Y,Z} \in \mathcal{P}_{U|Y,Z}(\theta)}\max_{j=1,\ldots,J} |\E_{P_{U|Y,Z}\times P_{Y,Z}} \left[ m_{j}(y,z,u,\theta) \right]|_{+}\\
&\qquad\geq  \varphi_{\ell b} + \mu^{*} C_{1}\min\{\delta, d(\theta,\Theta^{*})\}\\
&\qquad= \varphi_{\ell b} + \mu^{*} C_{1} \delta\\
&\qquad\geq \varphi^{*},
\end{align*}
where the last line follows since $\mu^{*} \geq (\varphi_{ub} - \varphi_{\ell b})/(C_{1}\delta) \geq (\varphi^{*} - \varphi_{\ell b})/(C_{1}\delta)$. Since we have shown the inequality holds for every $\theta \in \Theta$, we have:
\begin{align*}
&\inf_{\theta \in \Theta}  \inf_{P_{U|Y,Z} \in \mathcal{P}_{U|Y,Z}(\theta)}\Bigg(\inf_{P_{Y_{\gamma}^\star|Y,Z,U} \in \mathcal{P}_{Y_{\gamma}^\star|Y,Z,U}(\theta,\gamma)} \int \varphi(v)  \, dP_{V_{\gamma}}+ \mu^{*}\sum_{j=1}^{J} \lambda_{j}^{\ell b}(\theta,P_{Y,Z}) \E_{P_{Y,Z,U}} [ m_{j}(y,z,u,\theta)]\Bigg) \geq \varphi^{*}.
\end{align*}
This completes the proof.

\end{proof}

\begin{lemma}\label{lemma_multiplier_switch}
Suppose the assumptions of Theorem \ref{thm_cortes} hold, and define:
\begin{align*}
h_{\ell b}(y,z,\theta,\gamma)&:= \inf_{u \in \bm G^{-}(y,z,\theta)}  \inf_{y^{\star}  \in \bm G^{\star}(y,z,u,\theta,\gamma)} \Bigg( \varphi(v) + \mu^{*}\sum_{j=1}^{J}  \lambda_{j}^{\ell b}(\theta,P_{Y,Z})  m_{j}(y,z,u,\theta)\Bigg).
\end{align*} 
where $\lambda_{j}^{\ell b}: \Theta \times \mathcal{P}_{Y,Z} \to \{0,1\}$, $j=1,\ldots,J$, are as from Lemma \ref{lemma_error_bound}. Then we have:
\begin{align}
&\int h_{\ell b}(y,z,\theta,\gamma)\, dP_{Y,Z}\nonumber\\
&\qquad\qquad\leq \max_{\lambda_{j} \in \{0,1\}} \int \inf_{u \in \bm G^{-}(y,z,\theta)}  \inf_{y^{\star}  \in \bm G^{\star}(y,z,u,\theta,\gamma)} \Bigg( \varphi(v) + \mu^{*}\sum_{j=1}^{J}\lambda_{j}m_{j}(y,z,u,\theta)\Bigg)\, dP_{Y,Z},\label{eq_lemma_final1}
\end{align}
with equality holding in \eqref{eq_lemma_final1} if $\theta \in \Theta^{*}$. 
\end{lemma}
\begin{proof}[Proof of Lemma \ref{lemma_multiplier_switch}]
We have:
\begin{align*}
&\int h_{\ell b}(y,z,\theta,\gamma)\, dP_{Y,Z}\\
&:=  \int \inf_{u \in \bm G^{-}(y,z,\theta)}  \inf_{y^{\star}  \in \bm G^{\star}(y,z,u,\theta,\gamma)} \Bigg( \varphi(v) + \mu^{*}\sum_{j=1}^{J}  \lambda_{j}^{\ell b}(\theta,P_{Y,Z})  m_{j}(y,z,u,\theta)\Bigg)\, dP_{Y,Z}\\
&=\max_{\lambda_{j} \in \{0,1\} \text{ s.t } \lambda_{j} =\lambda_{j}^{\ell b}(\theta,P_{Y,Z})}\int \inf_{u \in \bm G^{-}(y,z,\theta)}  \inf_{y^{\star}  \in \bm G^{\star}(y,z,u,\theta,\gamma)} \Bigg( \varphi(v) + \mu^{*}\sum_{j=1}^{J}  \lambda_{j}  m_{j}(y,z,u,\theta)\Bigg)\, dP_{Y,Z}\\
&\leq\max_{\lambda_{j} \in \{0,1\} } \int \inf_{u \in \bm G^{-}(y,z,\theta)}  \inf_{y^{\star}  \in \bm G^{\star}(y,z,u,\theta,\gamma)} \Bigg( \varphi(v) + \mu^{*}\sum_{j=1}^{J}  \lambda_{j}  m_{j}(y,z,u,\theta)\Bigg)\, dP_{Y,Z}.
\end{align*}
The first line holds by definition, the second line holds since the $\lambda_{j}(\theta,P_{Y,Z})$ depends only on $\theta$, and third line holds because the unconstrained maximum is always weakly larger than the constrained maximum. 

It remains only to show that the last inequality holds with equality whenever $\theta \in \Theta^{*}$. To do so it suffices to show that for any $\theta \in \Theta^{*}$:
\begin{align}
\int h_{\ell b}(y,z,\theta,\gamma)\, dP_{Y,Z} \geq \max_{\lambda_{j} \in \{0,1\} } \int \inf_{u \in \bm G^{-}(y,z,\theta)}  \inf_{y^{\star}  \in \bm G^{\star}(y,z,u,\theta,\gamma)} \Bigg( \varphi(v) + \mu^{*}\sum_{j=1}^{J}  \lambda_{j}  m_{j}(y,z,u,\theta)\Bigg)\, dP_{Y,Z}.\label{eq_minimax_goal}
\end{align}
To this end, note that by Lemma \ref{lemma_joint_to_marginal_selection} we have:
\begin{align}
&\int h_{\ell b}(y,z,\theta,\gamma)\, dP_{Y,Z}\nonumber\\
&=\int \inf_{u \in \bm G^{-}(y,z,\theta)}  \inf_{y^{\star}  \in \bm G^{\star}(y,z,u,\theta,\gamma)} \Bigg( \varphi(v) + \mu^{*}\sum_{j=1}^{J}  \lambda_{j}^{\ell b}(\theta,P_{Y,Z})  m_{j}(y,z,u,\theta)\Bigg)\, dP_{Y,Z}\nonumber\\
&=\inf_{P_{U|Y,Z} \in \mathcal{P}_{U|Y,Z}(\theta)} \int  \inf_{y^{\star}  \in \bm G^{\star}(y,z,u,\theta,\gamma)} \Bigg( \varphi(v) + \mu^{*}\sum_{j=1}^{J}  \lambda_{j}^{\ell b}(\theta,P_{Y,Z})  m_{j}(y,z,u,\theta)\Bigg)\, d(P_{U|Y,Z}\times P_{Y,Z}).\label{eq_final_countdown1}
\end{align}
Now since the infimum of the sum is larger than the sum of the infimums, we have:
\begin{align}
&\inf_{P_{U|Y,Z} \in \mathcal{P}_{U|Y,Z}(\theta)} \int  \inf_{y^{\star}  \in \bm G^{\star}(y,z,u,\theta,\gamma)} \Bigg( \varphi(v) + \mu^{*}\sum_{j=1}^{J}  \lambda_{j}^{\ell b}(\theta,P_{Y,Z})  m_{j}(y,z,u,\theta)\Bigg)\, d(P_{U|Y,Z}\times P_{Y,Z})\nonumber\\
&\qquad\qquad\qquad\geq \inf_{P_{U|Y,Z} \in \mathcal{P}_{U|Y,Z}(\theta)} \int  \inf_{y^{\star}  \in \bm G^{\star}(y,z,u,\theta,\gamma)}  \varphi(v) \,d(P_{U|Y,Z}\times P_{Y,Z})\nonumber \\
&\qquad\qquad\qquad\qquad\qquad\qquad+ \inf_{P_{U|Y,Z} \in \mathcal{P}_{U|Y,Z}(\theta)} \mu^{*} \sum_{j=1}^{J} \lambda_{j}^{\ell b}(\theta,P_{Y,Z})  \E_{P_{U|Y,Z}\times P_{Y,Z}} [m_{j}(y,z,u,\theta)].\label{eq_final_countdown2}
\end{align}
Now recall that $\lambda_{j}^{\ell b}(\theta,P_{Y,Z})=1$ if and only if: 
\begin{align*}
\inf_{P_{U|Y,Z} \in \mathcal{P}_{U|Y,Z}(\theta)} \E_{P_{U|Y,Z}\times P_{Y,Z}} [m_{j}(y,z,u,\theta)]= \inf_{P_{U|Y,Z} \in \mathcal{P}_{U|Y,Z}(\theta)} \max_{j=1,\ldots,J}  |\E_{P_{U|Y,Z}\times P_{Y,Z}} [m_{j}(y,z,u,\theta)]|_{+}.
\end{align*}
From here we conclude:
\begin{align*}
&\inf_{P_{U|Y,Z} \in \mathcal{P}_{U|Y,Z}(\theta)} \max_{j=1,\ldots,J} |\E_{P_{U|Y,Z}\times P_{Y,Z}} [m_{j}(y,z,u,\theta)]|_{+}\\
&=\inf_{P_{U|Y,Z} \in \mathcal{P}_{U|Y,Z}(\theta)} \max_{j=1,\ldots,J} \lambda_{j}^{\ell b}(\theta,P_{Y,Z}) \E_{P_{U|Y,Z}\times P_{Y,Z}} [m_{j}(y,z,u,\theta)]\\
&\leq \inf_{P_{U|Y,Z} \in \mathcal{P}_{U|Y,Z}(\theta)} \sum_{j=1}^{J} \lambda_{j}^{\ell b}(\theta,P_{Y,Z}) \E_{P_{U|Y,Z}\times P_{Y,Z}} [m_{j}(y,z,u,\theta)]
\end{align*}
Thus:
\begin{align}
&\inf_{P_{U|Y,Z} \in \mathcal{P}_{U|Y,Z}(\theta)} \int  \inf_{y^{\star}  \in \bm G^{\star}(y,z,u,\theta,\gamma)}  \varphi(v) \,dP_{Y,Z,U} \nonumber\\
&\qquad\qquad\qquad\qquad+ \inf_{P_{U|Y,Z} \in \mathcal{P}_{U|Y,Z}(\theta)} \mu^{*} \sum_{j=1}^{J}  \lambda_{j}^{\ell b}(\theta,P_{Y,Z})  \E_{P_{U|Y,Z}\times P_{Y,Z}} [m_{j}(y,z,u,\theta)]\nonumber\\
&\geq \inf_{P_{U|Y,Z} \in \mathcal{P}_{U|Y,Z}(\theta)} \int  \inf_{y^{\star}  \in \bm G^{\star}(y,z,u,\theta,\gamma)}  \varphi(v) \,dP_{Y,Z,U} \nonumber\\
&\qquad\qquad\qquad\qquad+ \mu^{*} \inf_{P_{U|Y,Z} \in \mathcal{P}_{U|Y,Z}(\theta)} \max_{j=1,\ldots,J} |\E_{P_{U|Y,Z}\times P_{Y,Z}} [m_{j}(y,z,u,\theta)]|_{+}.\label{eq_final_countdown3}
\end{align}
However, since $\theta \in \Theta^{*}$ by assumption, we have:
\begin{align}
\inf_{P_{U|Y,Z} \in \mathcal{P}_{U|Y,Z}(\theta)} \max_{j=1,\ldots,J} |\E_{P_{U|Y,Z}\times P_{Y,Z}} [m_{j}(y,z,u,\theta)]|_{+}= 0.\label{eq_final_countdown4}
\end{align}
Thus, combining \eqref{eq_final_countdown1}, \eqref{eq_final_countdown2}, \eqref{eq_final_countdown3} and \eqref{eq_final_countdown4} we can conclude:
\begin{align}
&\int h_{\ell b}(y,z,\theta,\gamma)\, dP_{Y,Z}\geq  \inf_{P_{U|Y,Z} \in \mathcal{P}_{U|Y,Z}(\theta)} \int  \inf_{y^{\star}  \in \bm G^{\star}(y,z,u,\theta,\gamma)}  \varphi(v) \,d(P_{U|Y,Z}\times P_{Y,Z}).\label{eq_final_countdown5}
\end{align}
Now, applying Lemma \ref{lemma_joint_to_marginal_selection} again we have:
\begin{align}
&\inf_{P_{U|Y,Z} \in \mathcal{P}_{U|Y,Z}(\theta)} \int  \inf_{y^{\star}  \in \bm G^{\star}(y,z,u,\theta,\gamma)}  \varphi(v) \,d(P_{U|Y,Z}\times P_{Y,Z})\nonumber\\
&\qquad\qquad= \inf_{P_{U|Y,Z} \in \mathcal{P}_{U|Y,Z}(\theta)} \inf_{P_{Y_{\gamma}^{\star}|Y,Z,U} \in \mathcal{P}_{Y_{\gamma}^{\star}|Y,Z,U}(\theta,\gamma)} \int  \varphi(v)\, P_{V_{\gamma}}.\label{eq_final_countdown6}
\end{align}
Now note for $\theta \in \Theta^{*}$:
\begin{align}
&\inf_{P_{U|Y,Z} \in \mathcal{P}_{U|Y,Z}(\theta)} \inf_{P_{Y_{\gamma}^{\star}|Y,Z,U} \in \mathcal{P}_{Y_{\gamma}^{\star}|Y,Z,U}(\theta,\gamma)} \int  \varphi(v)\, P_{V_{\gamma}}\nonumber\\
&=  \inf_{P_{U|Y,Z} \in \mathcal{P}_{U|Y,Z}(\theta)} \inf_{P_{Y_{\gamma}^{\star}|Y,Z,U} \in \mathcal{P}_{Y_{\gamma}^{\star}|Y,Z,U}(\theta,\gamma)} \sup_{\lambda_{j} \in \mathbb{R}_{+}} \Bigg(\int  \varphi(v)\, P_{V_{\gamma}} + \mu^{*}\sum_{j=1}^{J}  \lambda_{j}  \E_{P_{U|Y,Z}\times P_{Y,Z}} [m_{j}(y,z,u,\theta)]\Bigg)\nonumber\\
&\geq  \inf_{P_{U|Y,Z} \in \mathcal{P}_{U|Y,Z}(\theta)} \inf_{P_{Y_{\gamma}^{\star}|Y,Z,U} \in \mathcal{P}_{Y_{\gamma}^{\star}|Y,Z,U}(\theta,\gamma)} \max_{\lambda_{j} \in \{0,1\} } \Bigg(\int  \varphi(v)\, P_{V_{\gamma}} + \mu^{*}\sum_{j=1}^{J}  \lambda_{j}  \E_{P_{U|Y,Z}\times P_{Y,Z}} [m_{j}(y,z,u,\theta)]\Bigg).\label{eq_final_countdown7}
\end{align}
Now by the minimax inequality:
\begin{align}
&\inf_{P_{U|Y,Z} \in \mathcal{P}_{U|Y,Z}(\theta)} \inf_{P_{Y_{\gamma}^{\star}|Y,Z,U} \in \mathcal{P}_{Y_{\gamma}^{\star}|Y,Z,U}(\theta,\gamma)} \max_{\lambda_{j} \in \{0,1\} } \Bigg(\int  \varphi(v)\, P_{V_{\gamma}} + \mu^{*}\sum_{j=1}^{J}  \lambda_{j}  \E_{P_{U|Y,Z}\times P_{Y,Z}} [m_{j}(y,z,u,\theta)]\Bigg)\nonumber\\
&\geq \max_{\lambda_{j} \in \{0,1\} } \inf_{P_{U|Y,Z} \in \mathcal{P}_{U|Y,Z}(\theta)} \inf_{P_{Y_{\gamma}^{\star}|Y,Z,U} \in \mathcal{P}_{Y_{\gamma}^{\star}|Y,Z,U}(\theta,\gamma)} \int \Bigg( \varphi(v) + \mu^{*}\sum_{j=1}^{J}  \lambda_{j}  m_{j}(y,z,u,\theta)\Bigg)\, dP_{V_{\gamma}}.\label{eq_final_countdown8}
\end{align}
Finally, by iterated application of Lemma \ref{lemma_joint_to_marginal_selection} we have:
\begin{align}
&\max_{\lambda_{j} \in \{0,1\} } \inf_{P_{U|Y,Z} \in \mathcal{P}_{U|Y,Z}(\theta)} \inf_{P_{Y_{\gamma}^{\star}|Y,Z,U} \in \mathcal{P}_{Y_{\gamma}^{\star}|Y,Z,U}(\theta,\gamma)} \int \Bigg( \varphi(v) + \mu^{*}\sum_{j=1}^{J}  \lambda_{j}  m_{j}(y,z,u,\theta)\Bigg)\, dP_{V_{\gamma}}\nonumber\\
&\geq \max_{\lambda_{j} \in \{0,1\} } \int \inf_{u \in \bm G^{-}(y,z,\theta)}  \inf_{y^{\star}  \in \bm G^{\star}(y,z,u,\theta,\gamma)} \Bigg( \varphi(v) + \mu^{*}\sum_{j=1}^{J}  \lambda_{j}m_{j}(y,z,u,\theta)\Bigg)\, dP_{Y,Z}.\label{eq_final_countdown9}
\end{align}
Combining \eqref{eq_final_countdown5}, \eqref{eq_final_countdown6}, \eqref{eq_final_countdown7}, \eqref{eq_final_countdown8}, and \eqref{eq_final_countdown9} we have:
\begin{align*}
\int h_{\ell b}(y,z,\theta,\gamma)\, dP_{Y,Z}\geq\max_{\lambda_{j} \in \{0,1\} } \int \inf_{u \in \bm G^{-}(y,z,\theta)}  \inf_{y^{\star}  \in \bm G^{\star}(y,z,u,\theta,\gamma)} \Bigg( \varphi(v) + \mu^{*}\sum_{j=1}^{J}  \lambda_{j}m_{j}(y,z,u,\theta)\Bigg)\, dP_{Y,Z},
\end{align*}
whenever $\theta \in \Theta^{*}$. This concludes the proof.
\end{proof}

\subsubsection{Lemmas Supporting Theorem \ref{theorem_pampac_learnability} on Learnability}

\begin{lemma}\label{lemma_inf_covering}
Suppose that $\mathcal{F}:=\{ f_{\alpha}(\cdot,\theta): \mathcal{X} \to \mathbb{R} : \theta \in \Theta, \alpha \in \mathcal{A}\}$ is a totally bounded parametric class of measurable real-valued functions on the metric space $(\mathcal{X},d)$, where $(\mathcal{A},d_{a})$ and $(\Theta,d_{\theta})$ are also metric spaces. Furthermore let $\mathcal{G}$ be a class of real-valued functions with each element $g(\cdot,\theta): \mathcal{X} \to \mathbb{R}$ defined by:
\begin{align*}
g(x,\theta):= \inf_{\alpha \in C(x,\theta)} f_{\alpha}(x,\theta),
\end{align*}
for some $f \in \mathcal{F}$, where $C(x,\theta)$ is a non-empty multifunction for each $(x,\theta)$ pair. Then for any probability measure $Q$ we have:
\begin{align*}
N(\varepsilon,\mathcal{G},||\cdot||_{Q,2}) \leq  N(\varepsilon/2,\mathcal{F},||\cdot||_{Q,2}). 
\end{align*}
\end{lemma}
\begin{proof}[Proof of Lemma \ref{lemma_inf_covering}]
As a parametric class of functions (parameterized by $(\alpha,\theta)$), the $\varepsilon/2-$cover of $\mathcal{F}$ can be characterized by a collection of points $\{(\alpha_{i},\theta_{i})\}_{i=1}^{n}$, where $n = N(\varepsilon/2,\mathcal{F},||\cdot||_{Q,2})$. Denote such a collection by $\mathcal{N}(\mathcal{F})$. We will show that for any $g \in \mathcal{G}$ there exists a pair $(\alpha',\theta') \in \mathcal{N}(\mathcal{F})$ such that:
\begin{align*}
|g(x,\theta) - f_{\alpha'}(x,\theta')| \leq \varepsilon.
\end{align*}
Since every $g \in \mathcal{G}$ can be expressed as:
\begin{align*}
g(x,\theta)=\inf_{\alpha \in C(x,\theta)} f_{\alpha}(x,\theta),
\end{align*}
it suffices to show there exists a pair $(\alpha',\theta') \in \mathcal{N}(\mathcal{F})$ such that:
\begin{align*}
\left|\inf_{\alpha \in C(x,\theta)} f_{\alpha}(x,\theta) - f_{\alpha'}(x,\theta')\right| \leq \varepsilon.
\end{align*}
Now let $\alpha^{*}$ be any value satisfying:
\begin{align*}
\left|\inf_{\alpha \in C(x,\theta)} f_{\alpha}(x,\theta) - f_{\alpha^{*}}(x,\theta)\right|\leq \varepsilon/2.
\end{align*}
That is, $\alpha^{*}$ is a $\varepsilon/2$ solution to the minimization problem. Now choose the pair $(\alpha',\theta') \in \mathcal{N}(\mathcal{F})$ such that $|f_{\alpha^{*}}(x,\theta) - f_{\alpha'}(x,\theta')|\leq \varepsilon/2$ (such a choice is always possible since $\mathcal{N}(\mathcal{F})$ is a $\varepsilon/2-$cover of $\mathcal{F}$). Then we have:
\begin{align*}
|g(x,\theta) - f_{\alpha'}(x,\theta')| &= \left|\inf_{\alpha \in C(x,\theta)} f_{\alpha}(x,\theta) - f_{\alpha'}(x,\theta')\right|\\
&\leq \left|\inf_{\alpha \in C(x,\theta)} f_{\alpha}(x,\theta) - f_{\alpha^{*}}(x,\theta)\right| + \left|f_{\alpha^{*}}(x,\theta) - f_{\alpha'}(x,\theta')\right|\\
&\leq\varepsilon/2 + \varepsilon/2\\
&=\varepsilon.
\end{align*}
This completes the proof.

\end{proof}

\begin{lemma}\label{lemma_bartlett}
Let $\mathcal{F}$ be a symmetric class of measurable real-valued functions on a Polish space $\mathcal{X}$, and let $\psi=(x_{i})_{i=1}^{n}$ denote an arbitrary vector of $n$ points from $\mathcal{X}$. Then for any $\varepsilon>0$:
\begin{align*}
\E ||\mathfrak{R}_{n}||(\mathcal{F}) \leq \frac{2\varepsilon}{\sqrt{n}} + 2\text{Diam}_{\psi,2}(\mathcal{F})  \sqrt{ \frac{ \log  N(\varepsilon,\mathcal{F},||\cdot||_{\psi,2})}{n}}.
\end{align*}

\end{lemma}
\begin{proof}[Proof of Lemma \ref{lemma_bartlett}]
Note that:
\begin{align*}
n\E ||\mathfrak{R}_{n}||(\mathcal{F})  &= n \E\sup_{f \in \mathcal{F} }\left|\frac{1}{n}\sum_{i=1}^{n} \xi_{i} f(x_{i}) \right|= \E\sup_{f \in \mathcal{F} }\left|\sum_{i=1}^{n} \xi_{i} f(x_{i}) \right|.
\end{align*}
Now recall that the Rademacher process $\sum_{i=1}^{n} \xi_{i} f(x_{i})$ is sub-Gaussian with respect to the euclidean distance between the vectors $\left( f(x_{1}),\ldots, f(x_{n}) \right)$ and $\left( f'(x_{1}),\ldots, f'(x_{n}) \right)$ for $f,f'\in \mathcal{F}$. We denote this euclidean distance by $||f-f'||_{\psi,2}$ to emphasize that the vector $\psi = (x_{i})_{i=1}^{n}$ is fixed. Fix an minimal $\varepsilon-$net $\mathcal{F}^{*}\subset \mathcal{F}$ in the $||\cdot||_{\psi,2}$ norm. There exists at least one function $f' \in \mathcal{F}^{*}$ such that:
\begin{align*}
\E\sup_{f \in \mathcal{F}}\left|\sum_{i=1}^{n} \xi_{i} f(x_{i}) \right| \leq \E\sup_{f \in \mathcal{F} }\left|\sum_{i=1}^{n} \xi_{i} f(x_{i}) - \sum_{i=1}^{n} \xi_{i} f'(x_{i}) \right|+\varepsilon\sqrt{n}.
\end{align*}
(For example, we can always take $f'$ to be the element in $\mathcal{F}^{*}$ to be closest to $-f$ in the $||\cdot||_{\psi,2}$ norm, which is an element of $\mathcal{F}$ by symmetry.) Now for any $f \in \mathcal{F} $, let $f^{*}(f) \in \mathcal{F}^{*}$ be a function with $||f- f^{*}(f)||_{\psi,2} \leq \varepsilon$. Then:  
\begin{align*}
&\left|\sum_{i=1}^{n} \xi_{i} f(x_{i}) - \sum_{i=1}^{n} \xi_{i} f'(x_{i})\right|\\
&= \left|\sum_{i=1}^{n} \xi_{i} f(x_{i}) - \sum_{i=1}^{n} \xi_{i} f^{*}(f)(x_{i}) + \sum_{i=1}^{n} \xi_{i} f^{*}(f)(x_{i}) - \sum_{i=1}^{n} \xi_{i} f'(x_{i})\right|\\
&\leq \left| \sum_{i=1}^{n} \xi_{i} f(x_{i}) - \sum_{i=1}^{n} \xi_{i} f^{*}(f)(x_{i})\right| + \left|\sum_{i=1}^{n} \xi_{i} f^{*}(f)(x_{i}) - \sum_{i=1}^{n} \xi_{i} f'(x_{i})\right|\\
&\leq \sup_{||f -f^{*}||_{\psi,2} \leq \varepsilon} \left| \sum_{i=1}^{n} \xi_{i} f(x_{i}) - \sum_{i=1}^{n} \xi_{i} f^{*}(x_{i})\right| +  \sup_{f^{*},f' \in \mathcal{F}^{*}}\left|\sum_{i=1}^{n} \xi_{i} f^{*}(x_{i}) - \sum_{i=1}^{n} \xi_{i} f'(x_{i})\right|\\
&\leq \sup_{||f -f^{*}||_{\psi,2} \leq \varepsilon} \sum_{i=1}^{n} \left|  f(x_{i}) - f^{*}(x_{i})\right| +  \sup_{f^{*},f' \in \mathcal{F}^{*}}\left|\sum_{i=1}^{n} \xi_{i} f^{*}(x_{i}) - \sum_{i=1}^{n} \xi_{i} f'(x_{i})\right|\\
&\leq \sup_{||f -f^{*}||_{\psi,2} \leq \varepsilon} \sqrt{n} ||f-f^{*}||_{\psi,2} +  \sup_{f^{*},f' \in \mathcal{F}^{*}}\left|\sum_{i=1}^{n} \xi_{i} f^{*}(x_{i}) - \sum_{i=1}^{n} \xi_{i} f'(x_{i})\right|\\
&\leq \sqrt{n} \varepsilon +  \sup_{f^{*},f' \in \mathcal{F}^{*}}\left|\sum_{i=1}^{n} \xi_{i} f^{*}(x_{i}) - \sum_{i=1}^{n} \xi_{i} f'(x_{i})\right|.
\end{align*}
Note we have used the inequality $||f-f'||_{\psi,1}\leq \sqrt{n} ||f-f'||_{\psi,2}$, where $||f-f'||_{\psi,1}$ denotes the $L^{1}$ distance between $f$ and $f'$ at the points $\psi=(x_{i})_{i=1}^{n}$. Now for any value $a>0$ we have:
\begin{align*}
&\exp \left( a \E \max_{f^{*},f' \in \mathcal{F}^{*}}\left|\sum_{i=1}^{n} \xi_{i} f^{*}(x_{i}) - \sum_{i=1}^{n} \xi_{i} f'(x_{i})\right| \right)\\
&\leq \E  \exp \left( a \max_{f^{*},f' \in \mathcal{F}^{*}}\left|\sum_{i=1}^{n} \xi_{i} f^{*}(x_{i}) - \sum_{i=1}^{n} \xi_{i} f'(x_{i})\right| \right)\\
&= \E \max_{f^{*},f' \in \mathcal{F}^{*}} \exp \left(a \left|\sum_{i=1}^{n} \xi_{i} f^{*}(x_{i}) - \sum_{i=1}^{n} \xi_{i} f'(x_{i})\right|\right)\\
&\leq \sum_{f,f^{*}\in \mathcal{F}^{*}} \E \exp \left(a \left|\sum_{i=1}^{n} \xi_{i} f^{*}(x_{i}) - \sum_{i=1}^{n} \xi_{i} f'(x_{i})\right|\right)\\
&\leq \sum_{f,f^{*}\in \mathcal{F}^{*}} \exp \left(a^{2} \text{Diam}_{\psi,2}^{2}(\mathcal{F})/2\right)\\
&\leq N^{2}(\varepsilon,\mathcal{F},||\cdot||_{\psi,2}) \exp \left(a^{2} \text{Diam}_{\psi,2}^{2}(\mathcal{F})/2\right),
\end{align*}
where the second-last inequality follows from the fact that the Rademacher process is sub-Gaussian with parameter $\text{Diam}_{\psi,2}^{2}(\mathcal{F})$.\footnote{Recall a stochastic process $(\omega,t)\mapsto X(\omega,t)$ on a metric space $(T,d)$ is sub-Gaussian with respect to the metric $d$ if $\E \exp \left( \lambda \left( X_{t} - X_{s} \right)  \right)\leq \exp(\lambda^{2}d(t,s)^{2}/2)$. For example, the Rademacher process is sub-Gaussian with respect to the euclidean metric.}
Taking logs and dividing both sides by $a>0$, we have:
\begin{align*}
\E \max_{f^{*},f' \in \mathcal{F}^{*}}\left|\sum_{i=1}^{n} \xi_{i} f^{*}(x_{i}) - \sum_{i=1}^{n} \xi_{i} f'(x_{i})\right| \leq \frac{2\log  N(\varepsilon,\mathcal{F},||\cdot||_{\psi,2})}{a} + \frac{a\text{Diam}_{\psi,2}^{2}(\mathcal{F})}{2}.
\end{align*}
Minimizing the upper bound with respect to ``$a$'' yields:\footnote{The minimizing value is $a =2 \left( \log N(\varepsilon,\mathcal{F},||\cdot||_{\psi,2})/ \text{Diam}_{\psi,2}^{2}(\mathcal{F}) \right)^{1/2}$.}
\begin{align*}
\E \max_{f^{*},f' \in \mathcal{F}^{*}}\left|\sum_{i=1}^{n} \xi_{i} f^{*}(x_{i}) - \sum_{i=1}^{n} \xi_{i} f'(x_{i})\right| \leq 2\text{Diam}_{\psi,2}(\mathcal{F})  \sqrt{ \log  N(\varepsilon,\mathcal{F},||\cdot||_{\psi,2})} 
\end{align*}
We conclude that:
\begin{align*}
n\E ||\mathfrak{R}_{n}||(\mathcal{F}) \leq 2\sqrt{n}\varepsilon +2\text{Diam}_{\psi,2}(\mathcal{F})  \sqrt{ \log  N(\varepsilon,\mathcal{F},||\cdot||_{\psi,2})}.
\end{align*}

\end{proof}

\begin{lemma}\label{lemma_useful_for_donsker}
Let $\mathcal{G}$ and $\mathcal{H}$ be two classes of functions and let $\mathcal{F} := \{ g + h : g \in \mathcal{G}, h \in \mathcal{H} \}$. Then: 
\begin{align*}
N \left( \varepsilon, \mathcal{F}, ||\cdot|| \right) \leq  N(\varepsilon/2,\mathcal{G},||\cdot||) N(\varepsilon/2,\mathcal{H},||\cdot||),
\end{align*}
where $||\cdot||$ is any norm.
\end{lemma}
\begin{remark}
Note that a nearly identical proof of this result can be used to show that:
$$N \left( \varepsilon, \mathcal{F}, ||\cdot|| \right) \leq  N(\varepsilon\cdot a,\mathcal{G},||\cdot||) N(\varepsilon \cdot b,\mathcal{H},||\cdot||),$$ where $a, b> 0$ are any values satisfying $a+b=1$.
\end{remark}
\begin{proof}[Proof of Lemma \ref{lemma_useful_for_donsker}]
Suppose that $N \left( \varepsilon/2, \mathcal{G},||\cdot||  \right)=n$ and $N(\varepsilon/2,\mathcal{H},||\cdot|| )=m$. It suffices to show $N \left( \varepsilon, \mathcal{F}, ||\cdot|| \right)\leq nm$. Let $\mathcal{N}(\mathcal{G})$ denote the centres of the balls that obtain the $n-$cover of $\mathcal{G}$ and let $\mathcal{N}(\mathcal{H})$ denote the centres of the balls that obtain the $m-$cover of $\mathcal{H}$. Enumerate the elements of $\mathcal{N}(\mathcal{G})$ as $g_{1},\ldots, g_{n}$ and enumerate the elements of $\mathcal{N}(\mathcal{H})$ as $h_{1},\ldots,h_{m}$. Now define the following collections:
\begin{align*}
G_{j} := \{ g \in \mathcal{G} : ||g-g_{j}||  \leq \varepsilon/2\}, &&H_{k} := \{ h \in \mathcal{H} : ||h-h_{k}||  \leq \varepsilon/2\}, 
\end{align*}
for $j=1,\ldots,n$ and $k=1,\ldots,m$. Then $\{G_{j}\}$ forms a $\varepsilon/2-$cover of $\mathcal{G}$ and $\{H_{k}\}$ forms a $\varepsilon/2-$cover of $\mathcal{H}$. Now for any $g_{j} \in \mathcal{N}(\mathcal{G})$ and $h_{k} \in \mathcal{N}(\mathcal{H})$ let $f_{jk}=g_{j}+h_{k}$, and define:
\begin{align*}
F_{jk}:= \{f : ||f-f_{jk}|| \leq \varepsilon \}.
\end{align*}
We will now argue that $\{F_{jk}\}$ is a $\varepsilon-$cover of $\mathcal{F}$. Note that if we can establish this, the proof will be complete, since there are only $nm$ sets $F_{jk}$. By construction each $F_{jk}$ is a $||\cdot||-$ball of radius $\varepsilon$, so it only remains to check that $\{F_{jk}\}$ covers $\mathcal{F}$. To do so, fix any $f \in \mathcal{F}$. Then by definition $f=g+h$ for some $g \in \mathcal{G}$ and $h \in \mathcal{H}$. Since $\{G_{j}\}$ forms a $\varepsilon/2-$cover of $\mathcal{G}$ and $\{H_{k}\}$ forms a $\varepsilon/2-$cover of $\mathcal{H}$, we know there is some $g_{j} \in \mathcal{N}(\mathcal{G})$ and some $h_{k} \in \mathcal{N}(\mathcal{H})$ such that $||g-g_{j}|| \leq \varepsilon/2$ and $||h-h_{k}||\leq \varepsilon/2$. But since $f_{jk}=g_{j}+h_{k}$ we have that:
\begin{align*}
||f - f_{jk}|| = ||(g+h) - (g_{j}+h_{k})|| \leq ||g - g_{j}|| + ||h - h_{k}|| \leq \varepsilon/2 + \varepsilon/2 = \varepsilon,
\end{align*}
so that $f \in F_{jk}$, and so is an element of the cover $\{F_{jk}\}$. Since $f \in \mathcal{F}$ was arbitrary, we conclude that $\{F_{jk}\}$ covers $\mathcal{F}$. This completes the proof. 

\end{proof}

\subsubsection{A Lemma Supporting Theorem \ref{theorem_delta_minimal} and Lemma \ref{lemma_subsets}}
\begin{lemma}\label{lemma_gamma_hat_high_prob}
Let $\delta^{**}$ be as in Lemma \ref{lemma_subsets}. If $\delta \geq \delta^{**}\geq \varepsilon>0$, then:
\begin{align*}
\sup_{P_{Y,Z} \in \mathcal{P}_{Y,Z}} P_{Y,Z}^{\otimes n} \left( \mathscr{E}^{*}(\hat{\gamma})\geq \delta \right) \leq 1-\kappa. 
\end{align*}
That is $\hat{\gamma} \in \mathscr{G}^{*}(\delta)$ with high probability when $\delta \geq \delta^{**}\geq \varepsilon>0$.\todo{Add reference}
\end{lemma}
\begin{proof}
Throughout this proof, let $\lambda^{*}(\theta,\gamma)$, $\hat{\lambda}(\theta,\gamma)$,  $\theta^{*}(\gamma)$, $\hat{\theta}(\gamma)$, $\gamma^{*}$ and $\hat{\gamma}$ be as in Remark \ref{remark_common_notation}. Fix any $\delta>\delta^{**}$ (the case when $\delta=\delta^{**}$ follows from continuity). If $\sigma:=\mathscr{E}^{*}(\hat{\gamma}) \geq \delta \geq \varepsilon >0$, then:
\begin{align*}
\mathscr{E}^{*}(\hat{\gamma})&:= \sup_{\gamma \in \Gamma} \inf_{\theta \in \Theta} \max_{\lambda\in \Lambda} P h_{\ell b}(\cdot,\theta,\gamma,\lambda)-\inf_{\theta \in \Theta} \max_{\lambda\in \Lambda} P h_{\ell b}(\cdot,\theta,\hat{\gamma},\lambda) \\
&\leq  \inf_{\theta \in \Theta} \max_{\lambda\in \Lambda} P h_{\ell b}(\cdot,\theta,\gamma^{*},\lambda)-\inf_{\theta \in \Theta} \max_{\lambda\in \Lambda} P h_{\ell b}(\cdot,\theta,\hat{\gamma},\lambda)+3\varepsilon\\ 
&=  \inf_{\theta \in \Theta} \max_{\lambda\in \Lambda} P h_{\ell b}(\cdot,\theta,\gamma^{*},\lambda)-\inf_{\theta \in \Theta} \max_{\lambda\in \Lambda} P h_{\ell b}(\cdot,\theta,\hat{\gamma},\lambda)\\
&\qquad\qquad+ \left(\inf_{\theta \in \Theta} \max_{\lambda\in \Lambda}  \P_{n} h_{\ell b}(\cdot,\theta,\gamma^{*},\lambda)-\inf_{\theta \in \Theta} \max_{\lambda\in \Lambda}  \P_{n} h_{\ell b}(\cdot,\theta,\hat{\gamma},\lambda)\right)\\
&\qquad\qquad\qquad\qquad- \left(\inf_{\theta \in \Theta} \max_{\lambda\in \Lambda}  \P_{n} h_{\ell b}(\cdot,\theta,\gamma^{*},\lambda)-\inf_{\theta \in \Theta}\max_{\lambda\in \Lambda}   \P_{n} h_{\ell b}(\cdot,\theta,\hat{\gamma},\lambda)\right) +3\varepsilon\\
&\leq  \inf_{\theta \in \Theta} \max_{\lambda\in \Lambda} P h_{\ell b}(\cdot,\theta,\gamma^{*},\lambda)-\inf_{\theta \in \Theta} \max_{\lambda\in \Lambda} P h_{\ell b}(\cdot,\theta,\hat{\gamma},\lambda)\\
&\qquad\qquad- \left(\inf_{\theta \in \Theta} \max_{\lambda\in \Lambda}  \P_{n} h_{\ell b}(\cdot,\theta,\gamma^{*},\lambda)-\inf_{\theta \in \Theta}\max_{\lambda\in \Lambda}   \P_{n} h_{\ell b}(\cdot,\theta,\hat{\gamma},\lambda)\right) +4\varepsilon
\end{align*}
Now note:
\begin{align*}
&\inf_{\theta \in \Theta} \max_{\lambda\in \Lambda} P h_{\ell b}(\cdot,\theta,\gamma^{*},\lambda)-\inf_{\theta \in \Theta} \max_{\lambda\in \Lambda} P h_{\ell b}(\cdot,\theta,\hat{\gamma},\lambda)\\
&\leq \inf_{\theta \in \Theta} \max_{\lambda\in \Lambda} P h_{\ell b}(\cdot,\theta,\gamma^{*},\lambda)- \max_{\lambda\in \Lambda} P h_{\ell b}(\cdot,\theta^{*}(\hat{\gamma}),\hat{\gamma},\lambda)+\varepsilon\\
&\leq \max_{\lambda\in \Lambda} P h_{\ell b}(\cdot,\hat{\theta}(\gamma^{*}),\gamma^{*},\lambda)- \max_{\lambda\in \Lambda} P h_{\ell b}(\cdot,\theta^{*}(\hat{\gamma}),\hat{\gamma},\lambda)+2\varepsilon\\
&\leq \max_{\lambda\in \Lambda} P h_{\ell b}(\cdot,\hat{\theta}(\gamma^{*}),\gamma^{*},\lambda)- P h_{\ell b}(\cdot,\theta^{*}(\hat{\gamma}),\hat{\gamma},\hat{\lambda}(\theta^{*}(\hat{\gamma}),\hat{\gamma}))+2\varepsilon\\
&\leq P h_{\ell b}(\cdot,\hat{\theta}(\gamma^{*}),\gamma^{*},\lambda^{*}(\hat{\theta}(\gamma^{*}),\gamma^{*}))- P h_{\ell b}(\cdot,\theta^{*}(\hat{\gamma}),\hat{\gamma},\hat{\lambda}(\theta^{*}(\hat{\gamma}),\hat{\gamma}))+2\varepsilon.
\end{align*}
Similarly:
\begin{align*}
&\inf_{\theta \in \Theta}\max_{\lambda\in \Lambda}   \P_{n} h_{\ell b}(\cdot,\theta,\hat{\gamma},\lambda) - \inf_{\theta \in \Theta} \max_{\lambda\in \Lambda}  \P_{n} h_{\ell b}(\cdot,\theta,\gamma^{*},\lambda)\\
&\leq \inf_{\theta \in \Theta}\max_{\lambda\in \Lambda}   \P_{n} h_{\ell b}(\cdot,\theta,\hat{\gamma},\lambda) -  \max_{\lambda\in \Lambda}  \P_{n} h_{\ell b}(\cdot,\hat{\theta}(\gamma^{*}),\gamma^{*},\lambda)+\varepsilon\\
&\leq \max_{\lambda\in \Lambda}\P_{n} h_{\ell b}(\cdot,\theta^{*}(\hat{\gamma}),\hat{\gamma},\lambda) -  \max_{\lambda\in \Lambda}  \P_{n} h_{\ell b}(\cdot,\hat{\theta}(\gamma^{*}),\gamma^{*},\lambda)+2\varepsilon\\
&\leq \max_{\lambda\in \Lambda}\P_{n} h_{\ell b}(\cdot,\theta^{*}(\hat{\gamma}),\hat{\gamma},\lambda) -  \max_{\lambda\in \Lambda}  \P_{n} h_{\ell b}(\cdot,\hat{\theta}(\gamma^{*}),\gamma^{*},\lambda^{*}(\hat{\theta}(\gamma^{*}),\gamma^{*}))+2\varepsilon\\
&\leq \P_{n} h_{\ell b}(\cdot,\theta^{*}(\hat{\gamma}),\hat{\gamma},\hat{\lambda}(\theta^{*}(\hat{\gamma}),\hat{\gamma})) -  \max_{\lambda\in \Lambda}  \P_{n} h_{\ell b}(\cdot,\hat{\theta}(\gamma^{*}),\gamma^{*},\lambda^{*}(\hat{\theta}(\gamma^{*}),\gamma^{*}))+2\varepsilon.
\end{align*}
Thus we have:
\begin{align*}
\mathscr{E}^{*}(\hat{\gamma})&\leq P h_{\ell b}(\cdot,\hat{\theta}(\gamma^{*}),\gamma^{*},\lambda^{*}(\hat{\theta}(\gamma^{*}),\gamma^{*}))- P h_{\ell b}(\cdot,\theta^{*}(\hat{\gamma}),\hat{\gamma},\hat{\lambda}(\theta^{*}(\hat{\gamma}),\hat{\gamma}))\\
&\qquad\qquad- \left(\max_{\lambda\in \Lambda}  \P_{n} h_{\ell b}(\cdot,\hat{\theta}(\gamma^{*}),\gamma^{*},\lambda^{*}(\hat{\theta}(\gamma^{*}),\gamma^{*})) - \P_{n} h_{\ell b}(\cdot,\theta^{*}(\hat{\gamma}),\hat{\gamma},\hat{\lambda}(\theta^{*}(\hat{\gamma}),\hat{\gamma}))  \right) +  8\varepsilon.   
\end{align*}
Furthermore, $\sigma = \mathscr{E}^{*}(\hat{\gamma}) \geq  \mathscr{E}^{*}(\gamma^{*})$ implies that $\hat{\gamma},\gamma^{*} \in \mathscr{G}(\sigma)$. Thus:
\begin{align*}
\mathscr{E}^{*}(\hat{\gamma}) \leq \sup_{\theta, \theta' \in \Theta} \sup_{\gamma, \gamma' \in \mathscr{G}^{*}(\sigma)} \max_{\lambda,\lambda' \in \Lambda} \left|\left( \P_{n} h_{\ell b}(\cdot,\theta,\gamma,\lambda) -  \P_{n} h_{\ell b}(\cdot,\theta',\gamma',\lambda')\right) -  \left( P h_{\ell b}(\cdot,\theta,\gamma,\lambda) -  P h_{\ell b}(\cdot,\theta',\gamma',\lambda') \right)\right|, 
\end{align*}
which follows since $\varepsilon>0$ can be made arbitrarily small. Now define:
\begin{align}
E_{n,j}:= \left\{ \sup_{\theta, \theta' \in \Theta} \sup_{\gamma, \gamma' \in \mathscr{G}^{*}(\sigma)} \max_{\lambda,\lambda' \in \Lambda} \left|\left( \P_{n} h_{\ell b}(\cdot,\theta,\gamma,\lambda) -  \P_{n} h_{\ell b}(\cdot,\theta',\gamma',\lambda')\right) -  \left( P h_{\ell b}(\cdot,\theta,\gamma,\lambda) -  P h_{\ell b}(\cdot,\theta',\gamma',\lambda') \right)\right| \leq T(\delta_{j}) \right\},
\end{align}
and:
\begin{align*}
E_{n}:= \bigcap_{\{ j : \delta_{j}\geq \delta^{**}\}} E_{n,j}.
\end{align*}
Note by our choice of $\delta_{0}> 2\overline{H}$ we have:
\begin{align*}
\sup_{P_{Y,Z} \in \mathcal{P}_{Y,Z}} P_{Y,Z}^{\otimes n} \left( E_{n,0}^{c}\right)=0.
\end{align*}
Furthermore, from the uniform version of Hoeffding's inequality (e.g. \cite{koltchinskii2011oracle} Theorem 4.6, p.71) we have:
\begin{align*}
\sup_{P_{Y,Z} \in \mathcal{P}_{Y,Z}} P_{Y,Z}^{\otimes n} \left( E_{n,j}^{c}\right) \leq  \exp \left( -\frac{t_{j}^{2}}{2} \right), 
\end{align*}
for each $j\in \mathbb{N}$. We conclude by the union bound that:
\begin{align*}
\sup_{P_{Y,Z} \in \mathcal{P}_{Y,Z}} P_{Y,Z}^{\otimes n} \left( E_{n}^{c}\right) \leq \sum_{\{j : \delta_{j} \geq \delta^{**}\}} \exp \left( -\frac{t_{j}^{2}}{2} \right) \leq \sum_{j=0}^{\infty} \exp \left( -\frac{t_{j}^{2}}{2} \right) \leq 1-\kappa. 
\end{align*}
Now on the event $E_{n}$, for every $\delta \geq \delta^{**}$ we have:
\begin{align*}
\sup_{\theta, \theta' \in \Theta} \sup_{\gamma, \gamma' \in \mathscr{G}^{*}(\delta)} \max_{\lambda,\lambda' \in \Lambda} \left|(\P_{n}-P)\left( h_{\ell b}(\cdot,\theta,\gamma,\lambda) -  h_{\ell b}(\cdot,\theta',\gamma',\lambda')\right) \right| \leq T(\delta). 
\end{align*}
Now suppose by way of contradiction that $\{\mathscr{E}^{*}(\hat{\gamma}) \geq \delta\} \cap E_{n}\neq\emptyset$. Then on this event we have:
\begin{align*}
\sigma&:= \mathscr{E}^{*}(\hat{\gamma})\\
&\leq \sup_{\theta, \theta' \in \Theta} \sup_{\gamma, \gamma' \in \mathscr{G}^{*}(\sigma)} \max_{\lambda,\lambda' \in \Lambda} \left|\left( \P_{n} h_{\ell b}(\cdot,\theta,\gamma,\lambda) -  \P_{n} h_{\ell b}(\cdot,\theta',\gamma',\lambda')\right) -  \left( P h_{\ell b}(\cdot,\theta,\gamma,\lambda) -  P h_{\ell b}(\cdot,\theta',\gamma',\lambda') \right)\right|\\
&\leq T(\sigma).
\end{align*}
However, note that this implies that $\sigma \leq \delta^{**}$ on the event $E_{n}$. But since $\sigma \geq \delta > \delta^{**}$ by assumption, we have a contradiction. We conclude that $\{\mathscr{E}^{*}(\hat{\gamma})\} \geq \delta\} \cap E_{n}=\emptyset$, or equivalently that $\{\mathscr{E}^{*}(\hat{\gamma}) \geq \delta\} \subseteq  E_{n}^{c}$, where the event $E_{n}^{c}$ has probability at most $1-\kappa$. 

\end{proof}

\section{Additional Details for the Examples}\label{appendix_additional_details_for_examples}

\subsection{Example \ref{example_simultaneous_discrete_choice}: Simultaneous Discrete Choice}\label{appendix_additional_details_sdc}

\subsubsection{Verification of Assumptions \ref{assump_preliminary}, \ref{assumption_factual_domain} and \ref{assumption_counterfactual_domain}}\label{appendix_additional_details_sdc_A1_to_A3}

We will now proceed to verify Assumption \ref{assump_preliminary}, \ref{assumption_factual_domain} and \ref{assumption_counterfactual_domain}. First note that Assumption \ref{assump_preliminary} is trivially satisfied, since the probability space $(\Omega,\mathfrak{A},P)$ is complete, and both $\mathcal{U}$ and $\Theta$ are compact metric spaces with the euclidean norm.

To verify Assumption \ref{assumption_factual_domain}, note that the multifunction for the factual domain can be rewritten as:\todolt{\cite{aliprantis2006infinite} have some other ways to check sufficient conditions}
\begin{align}
\bm G^{-} \left( Y,Z,\theta \right)= \left\{ u \in \mathcal{U} : \begin{array}{l}
	u_{k} \in [\pi_{k} ( Z_{k},Y_{-k};\theta ),1], \text{ if } Y_{k}=0,\\
	u_{k} \in [-1,\pi_{k} (Z_{k},Y_{-k};\theta )], \text{ if } Y_{k}=1.
\end{array}\right\}.\label{eq_DC_reverse_correspondence2}
\end{align}
From here we conclude that, for any $u \in \mathcal{U}$:
\begin{align} 
&d(u,\bm G^{-} \left( Y,Z,\theta \right))\nonumber\\
&= \max_{k} \bigg( \mathbbm{1}\{Y_{k}=0\} |\pi_{k} ( Z_{k},Y_{-k};\theta ) - u_{k}|_{+} + \mathbbm{1}\{Y_{k}=1\} |u_{k} - \pi_{k} ( Z_{k},Y_{-k};\theta ) |_{+}\bigg).
\end{align}
Under our assumptions, this distance is the maximum of $K$ measurable functions, and so is itself measurable. Since $u \in \mathcal{U}$ was arbitrary, by the result of \cite{himmelberg1975measurable} (see also Theorem 1.3.3 in \cite{molchanov2017theory}) this implies that $\bm G^{-}$ is an Effros-measurable multifunction (w.r.t. $\mathfrak{B}(\mathcal{Y})\otimes \mathfrak{B}(\mathcal{Z})\otimes \mathfrak{B}(\Theta)$), as desired. It is then easily seen that the conditional distribution of the vector $U$ given $(Y,Z)$ satisfies \eqref{eq_puyz} in Assumption \ref{assumption_factual_domain} using the multifunction in \eqref{eq_DC_reverse_correspondence2} with $\theta=\theta_{0}$. To complete the verification of Assumption \ref{assumption_factual_domain}, note that all the moment functions from the moment conditions in \eqref{eq_sdc_mom1} and \eqref{eq_sdc_mom2} are bounded in absolute value and Borel measurable (w.r.t. $\mathfrak{B}(\mathcal{Y})\otimes \mathfrak{B}(\mathcal{Z})\otimes \mathfrak{B}(\Theta)$). 

We now turn to the verification of Assumption \ref{assumption_counterfactual_domain}. Recall the counterfactual multifunction:
\begin{align}
\bm G^\star(Z,U,\theta,\gamma):=\left\{ y^\star \in \mathcal{Y}  : y_{k}^\star=\mathbbm{1}\{ \pi_{k} \left(\gamma(Z_{k},y_{-k}^\star);\theta \right)  \geq U_{k}\},\,\, k=1,\ldots,K. \right\}.
\end{align}
Close inspection reveals that:
\begin{align} 
&d(y^\star,\bm G^{\star} \left(Z,U,\theta,\gamma\right))= \max_{k} \left|y_{k}^\star - \mathbbm{1}\{ \pi_{k} \left(\gamma(Z_{k},y_{-k}^\star);\theta \right)  \geq U_{k}\}\right|.
\end{align}
Under our assumptions, this distance is also the maximum of $K$ measurable functions, and so is itself measurable. Since $y^\star \in \mathcal{Y}^\star$ was arbitrary, by the result of \cite{himmelberg1975measurable} (see also Theorem 1.3.3 in \cite{molchanov2017theory}) this implies that $\bm G^{\star}$ is an Effros-measurable multifunction  (w.r.t. $\mathfrak{B}(\mathcal{Y})\otimes \mathfrak{B}(\mathcal{Z})\otimes \mathfrak{B}(\mathcal{Z})\otimes \mathfrak{B}(\Theta)\otimes \mathfrak{B}(\Gamma)$), as desired. Finally, it is easily seen that the conditional distribution of the vector $Y_{\gamma}^\star$ given $(Y,Z,U)$ satisfies \eqref{eq_pygamma} in Assumption \ref{assumption_counterfactual_domain} using the multifunction in \eqref{eq_DC_counterfactual_correspondence} with $\theta=\theta_{0}$.

\subsubsection{Verification of Assumption \ref{assumption_error_bound}}\label{appendix_additional_details_sdc_A4_error_bound}

We will first verify Assumption \ref{assumption_error_bound}(ii) for some $C_{2}\geq 0$ and $\delta>0$, and then will show that Assumption \ref{assumption_error_bound}(i) is also satisfied for our choice of $\delta>0$. 

As was mentioned in the main text, under our current assumptions for this example Assumption \ref{assumption_error_bound}(ii) is not satisfied. The issue is illustrated in Figures \ref{fig_concentrated_mass2} and \ref{fig_concentrated_mass3}, and a case where Assumption \ref{assumption_error_bound}(ii) is satisfied is illustrated in Figure \ref{fig_concentrated_mass4}. The issue arises only when for some $k \in \{1,\ldots,K\}$ and some $z \in \mathcal{Z}$ and $y_{-k} \in \mathcal{Y}_{-k}$ we have: (i) the object of interest is $P(Y_{\gamma,k}^\star=1|Z_{k}=z',Y_{-k}=y_{-k}')$ or $P(Y_{\gamma,k}^\star=1)$, (ii) the counterfactual cutoff value $\pi_{k}(\gamma(z,y_{-k});\theta^{*})=0$ at some $\theta^{*} \in \partial \Theta^{*}$, and (iii) if $P(Y_{k}=1|Z_{k}=z',Y_{-k}=y_{-k}')\neq 0.5$, where $(z',y_{-k}')=\gamma(z,y_{-k})$. In this knife-edge case, a very small change in $\theta^{*}$ to some $\theta \notin \Theta^{*}$ can cause a discontinuous change in $P(Y_{\gamma,k}^\star=1|Z_{k}=z',Y_{-k}=y_{-k}')$ or $P(Y_{\gamma,k}^\star=1)$. 
\begin{figure}[!t]
\centering
\includegraphics[scale=1]{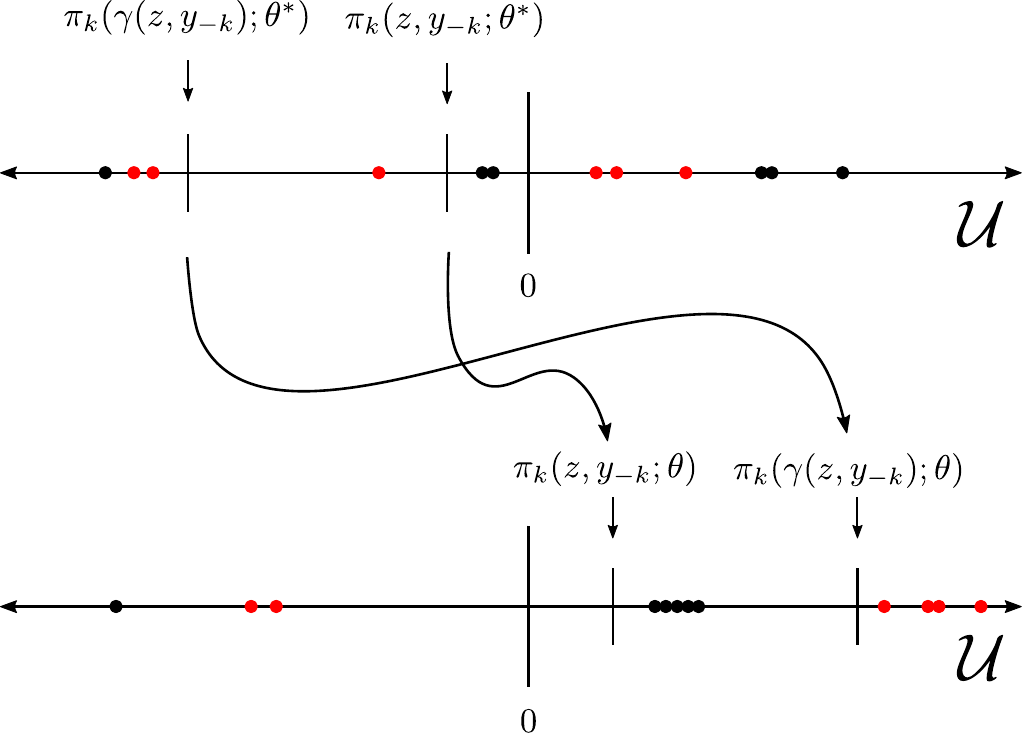}
\caption{This figure illustrates a case violating Assumption \ref{assumption_error_bound}(ii). The black dots $\bullet$ represent equal probability masses ($1/6$) assigned by the conditional distribution of $U_{k}$ given $(z,y_{-k})$. The red dots \color{red} $\bullet$ \color{black} represent equal probability masses ($1/6$) assigned by the conditional distribution of $U_{k}$ given $(z',y_{-k}')=\gamma(z,y_{-k})$. In the upper portion of the figure we have $\theta^{*} \in \Theta^{*}$, the median zero assumption is satisfied (three black dots $\bullet$ and three red dots \color{red} $\bullet$ \color{black} on either side of zero) and the maximum value of $P(Y_{\gamma}^\star=1|Z_{k}=z,Y_{-k}=y_{-k})$ at $\theta^{*}$ is obtained at $1/6$. However, in the bottom portion of the figure a small change in the value of $\theta^{*} \in \Theta^{*}$ to $\theta \notin \Theta^{*}$ causes a violation of the median zero assumption for the points $(z,y_{-k})$ and $(z',y_{-k}')$. At the new value $\theta \notin \Theta^{*}$ we have the maximum value of $P(Y_{\gamma}^\star=1|Z_{k}=z,Y_{-k}=y_{-k})$ is $1$. Note that the scale of the figure can be made arbitrarily small. }\label{fig_concentrated_mass2}
\end{figure}
\begin{figure}[!t]
\centering
\includegraphics[scale=1]{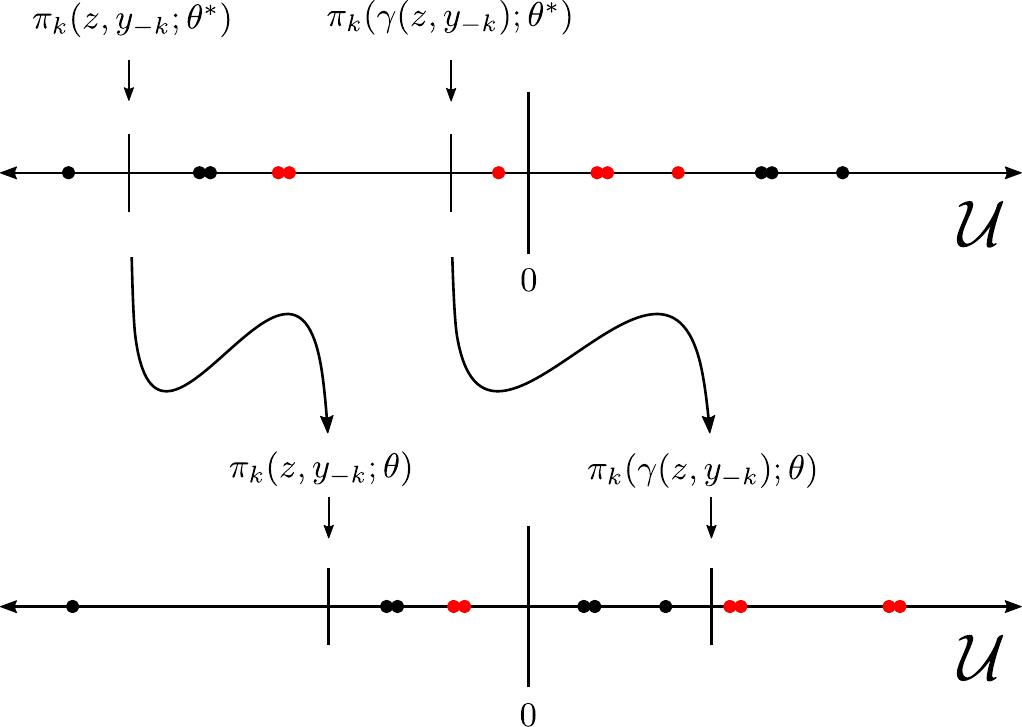}
\caption{This figure illustrates a case violating Assumption \ref{assumption_error_bound}(ii). The black dots $\bullet$ represent equal probability masses ($1/6$) assigned by the conditional distribution of $U_{k}$ given $(z,y_{-k})$. The red dots \color{red} $\bullet$ \color{black} represent equal probability masses ($1/6$) assigned by the conditional distribution of $U_{k}$ given $(z',y_{-k}')=\gamma(z,y_{-k})$. In the upper portion of the figure we have $\theta^{*} \in \Theta^{*}$, the median zero assumption is satisfied (three black dots $\bullet$ and three red dots \color{red} $\bullet$ \color{black} on either side of zero) and the maximum value of $P(Y_{\gamma}^\star=1|Z_{k}=z,Y_{-k}=y_{-k})$ at $\theta^{*}$ is obtained at $1/2$. However, in the bottom portion of the figure a small change in the value of $\theta^{*} \in \Theta^{*}$ to $\theta \notin \Theta^{*}$ causes a violation of the median zero assumption for the point $(z',y_{-k}')$. At the new value $\theta \notin \Theta^{*}$ we have the maximum value of $P(Y_{\gamma}^\star=1|Z_{k}=z,Y_{-k}=y_{-k})$ is $1$. Note that the scale of the figure can be made arbitrarily small.}\label{fig_concentrated_mass3}
\end{figure}
\begin{figure}[!t]
\centering
\includegraphics[scale=1]{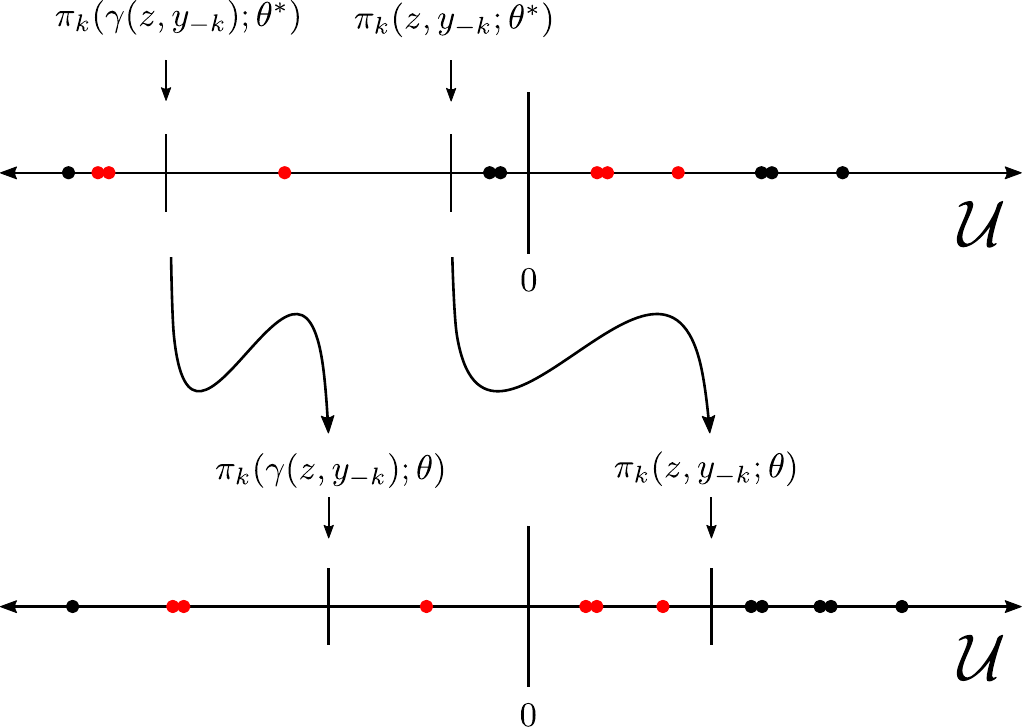}
\caption{This figure illustrates a case that does not violate Assumption \ref{assumption_error_bound}(ii). The black dots $\bullet$ represent equal probability masses ($1/6$) assigned by the conditional distribution of $U_{k}$ given $(z,y_{-k})$. The red dots \color{red} $\bullet$ \color{black} represent equal probability masses ($1/6$) assigned by the conditional distribution of $U_{k}$ given $(z',y_{-k}')=\gamma(z,y_{-k})$. In the upper portion of the figure we have $\theta^{*} \in \Theta^{*}$, the median zero assumption is satisfied (three black dots $\bullet$ and three red dots \color{red} $\bullet$ \color{black} on either side of zero) and $P(Y_{\gamma}^\star=1|Z_{k}=z,Y_{-k}=y_{-k})=1/6$. In the bottom portion of the figure a small change in the value of $\theta^{*} \in \Theta^{*}$ to $\theta \notin \Theta^{*}$ causes a violation of the median zero assumption for the point $(z,y_{-k})$. However, at the new value $\theta \notin \Theta^{*}$ we still have that the maximum obtainable value of $P(Y_{\gamma}^\star=1|Z_{k}=z,Y_{-k}=y_{-k})$ is $1/6$. }\label{fig_concentrated_mass4}
\end{figure}

To prevent such discontinuities in the value of the policy transform, we can introduce additional assumptions on the degree of smoothness of the distribution of $U_{k}$ around zero. In particular, instead of the moment conditions in \eqref{eq_sdc_mom1} and \eqref{eq_sdc_mom2} we propose imposing the constraints:
\begin{align}
P \left(U_{k}\leq \pi_{k}(z',y_{-k}';\theta) | Z_{k}=z, Y_{-k}=y_{-k}  \right) -0.5 \leq \max\{L_{0}\pi_{k}(z',y_{-k}';\theta),0\},\label{eq_sdc_cons1}\\
0.5 - P \left(U_{k}\leq\pi_{k}(z',y_{-k}';\theta) | Z_{k}=z, Y_{-k}=y_{-k}  \right) \leq \max\{-L_{0}\pi_{k}(z',y_{-k}';\theta),0\},\label{eq_sdc_cons2}
\end{align}
for some $L_{0}> 0$, for $k=1,\ldots,K$, and for all $z,z' \in \mathcal{Z}$ and $y_{-k}, y_{-k}' \in \mathcal{Y}^{K-1}$. These constraints impose a local Lipschitzian constraint on the distribution of $U_{k}$ around zero. Note that by taking $L_{0}$ sufficiently large, these constraints will only be active when $\pi_{k}(z',y_{-k}';\theta)$ is close to zero. It is also easily verified that the new moment conditions implied by \eqref{eq_sdc_cons1} and \eqref{eq_sdc_cons2} also satisfy Assumption \ref{assumption_factual_domain}.

We claim that the constraints \eqref{eq_sdc_cons1} and \eqref{eq_sdc_cons2} imply that $U_{k}$ is median zero and median independent of $(Z,Y_{-k})$. To see this, note that $U_{k}$ has a median of zero given $(z,y_{-k})$ if and only if: 
\begin{enumerate}[label=(\Roman*)]
	\item $\pi_{k}(z_{k},y_{-k};\theta) \leq 0$ and $P(U_{k}\leq \pi_{k}(z_{k},y_{-k};\theta)  |Z=z_{k},Y_{-k}=y_{-k}) \leq 0.5$; or
	\item $\pi_{k}(z_{k},y_{-k};\theta) > 0$ and $P(U_{k}> \pi_{k}(z_{k},y_{-k};\theta) |Z=z_{k},Y_{-k}=y_{-k})\leq 0.5$. 
\end{enumerate}
The idea behind these conditions is illustrated in Figure \ref{fig_median_independence}. Conversely, $U_{k}$ does not have a median of zero conditional on $(z,y_{-k})$ if and only if:
\begin{enumerate}[label=(\roman*)]
	\item $\pi_{k}(z_{k},y_{-k};\theta)  > 0$ and $P(U_{k}\leq \pi_{k}(z_{k},y_{-k};\theta)  |Z=z_{k},Y_{-k}=y_{-k}) < 0.5$; or
	\item $\pi_{k}(z_{k},y_{-k};\theta) \leq 0$ and $P(U_{k}> \pi_{k}(z_{k},y_{-k};\theta) |Z=z_{k},Y_{-k}=y_{-k})< 0.5$. 
\end{enumerate}
Note that if (i) holds then \eqref{eq_sdc_cons2} fails, and if (ii) holds then \eqref{eq_sdc_cons1} fails. This implies that if both \eqref{eq_sdc_cons1} and \eqref{eq_sdc_cons2} hold, then (i) and (ii) do not hold, and thus $U_{k}$ is median zero and median independent of $(Z,Y_{-k})$. However, note that it is possible that either (I) or (II) is satisfied but one of \eqref{eq_sdc_cons1} or \eqref{eq_sdc_cons2} fails, owing to the fact that together \eqref{eq_sdc_cons1} and \eqref{eq_sdc_cons2} are stronger than the median zero and median independence restrictions initially imposed in \eqref{eq_sdc_mom1} and \eqref{eq_sdc_mom2}.

\begin{figure}[!t]
\centering
\includegraphics[scale=0.9]{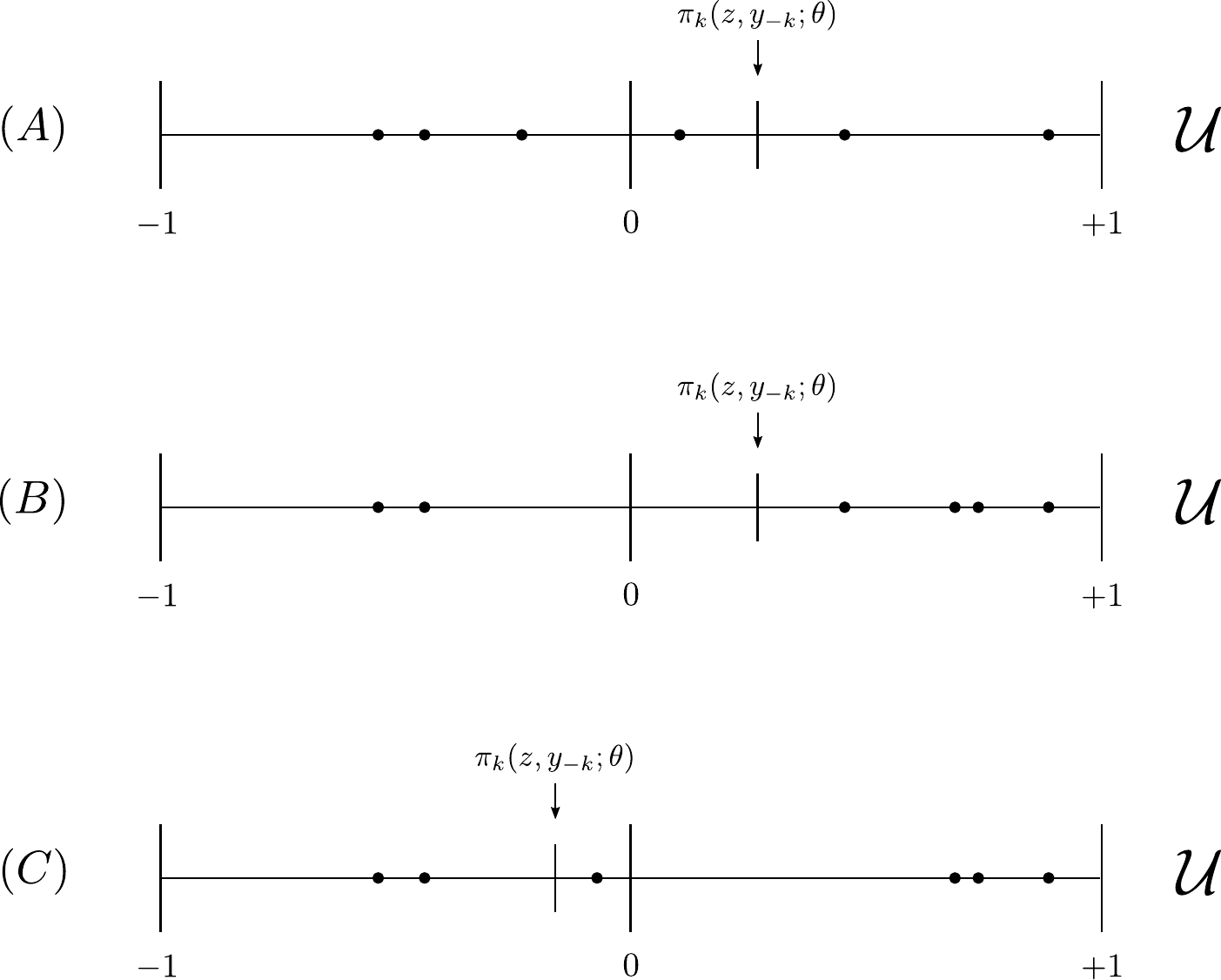}
\caption{This figure illustrates three scenarios, each scenario involving a different allocation of probability mass for $U_{k}$, represented by the $6$ dots $\bullet$ representing equal probability mass, and a different value of the cutoff $\pi_{k}(z,y_{-k};\theta)$. In scenario (A), $\pi_{k}(z_{k},y_{-k};\theta) > 0$ and $P(Y_{k}=0|Z=z_{k},Y_{-k}=y_{-k})\leq 0.5$. In this case the median zero condition can be satisfied, for example, by the allocation of probability mass displayed in the figure. In scenario (B), $\pi_{k}(z_{k},y_{-k};\theta) > 0$ and $P(Y_{k}=0|Z=z_{k},Y_{-k}=y_{-k})> 0.5$. Here there is no way of satisfying the median zero assumption, since too much mass will always be assigned above zero. In scenario (C), $\pi_{k}(z_{k},y_{-k};\theta) < 0$ and $P(Y_{k}=0|Z=z_{k},Y_{-k}=y_{-k})> 0.5$. In this case the median zero condition can again be satisfied, for example, by the allocation of probability mass displayed in the figure.   }\label{fig_median_independence}
\end{figure}

We will now proceed to verify Assumption \ref{assumption_error_bound}. First recall from the discussion in the text that $\pi_{k}$ is a known measurable function of $(Z_{k},Y_{-k},\theta)$ that is linear in parameters $\theta$ and has a gradient (with respect to $\theta$) bounded away from zero for each $(z,y_{-k})$. Thus, $\pi_{k}$ is Lipschitz in $\theta$, and also satisfies a ``reverse Lipschitz'' condition; that is, for each $(z,y_{-k})$ we have: 
\begin{align*}
L_{k}'||\theta-\theta^{*}||\leq |\pi_{k}(z,y_{-k};\theta)-\pi_{k}(z,y_{-k};\theta^{*})| \leq L_{k}||\theta-\theta^{*}||,
\end{align*}
for some $L_{k}',L_{k}>0$. Now, if one of the constraints \eqref{eq_sdc_cons1} or \eqref{eq_sdc_cons2} is violated, we have one of the following inequalities:
\begin{align}
P (\widetilde{U}_{k}\leq \pi_{k}(z',y_{-k}';\theta) | Z_{k}=z, Y_{-k}=y_{-k}  ) -0.5 > \max\{L_{0}\pi_{k}(z',y_{-k}';\theta),0\},\label{eq_sdc_cons3}\\
0.5 -P (\widetilde{U}_{k}\leq \pi_{k}(z',y_{-k}';\theta) | Z_{k}=z, Y_{-k}=y_{-k}  ) > \max\{-L_{0}\pi_{k}(z',y_{-k}';\theta),0\},\label{eq_sdc_cons4}
\end{align}
Subtracting \eqref{eq_sdc_cons3} from \eqref{eq_sdc_cons1} and taking $(z',y_{-k}')=\gamma(z,y_{-k})$, we have:
\begin{align}
&P \left(U_{k}\leq \pi_{k}(\gamma(z,y_{-k});\theta^{*}) | Z_{k}=z, Y_{-k}=y_{-k}  \right) -P (\widetilde{U}_{k}\leq \pi_{k}(\gamma(z,y_{-k});\theta) | Z_{k}=z, Y_{-k}=y_{-k} )\nonumber\\
&\qquad=P \left(U_{k}\leq \pi_{k}(z',y_{-k}';\theta^{*}) | Z_{k}=z, Y_{-k}=y_{-k}  \right) -P (\widetilde{U}_{k}\leq \pi_{k}(z',y_{-k}';\theta) | Z_{k}=z, Y_{-k}=y_{-k} )\nonumber\\
&\qquad< \max\{L_{0}\pi_{k}(z',y_{-k}';\theta^{*}),0\} -\max\{L_{0}\pi_{k}(z',y_{-k}';\theta),0\}\nonumber\\
&\qquad\leq \max\{L_{0}\pi_{k}(z',y_{-k}';\theta^{*})-L_{0}\pi_{k}(z',y_{-k}';\theta),0\}\nonumber\\
&\qquad\leq L_{0} |\pi_{k}(z',y_{-k}';\theta^{*})-\pi_{k}(z',y_{-k}';\theta)|\nonumber\\
&\qquad\leq L_{0} L_{k} ||\theta - \theta^{*}||.
\end{align}
Furthermore, subtracting \eqref{eq_sdc_cons4} from \eqref{eq_sdc_cons2} and again taking $(z',y_{-k}')=\gamma(z,y_{-k})$, we have:
\begin{align}
&P (\widetilde{U}_{k}\leq \pi_{k}(\gamma(z,y_{-k});\theta) | Z_{k}=z, Y_{-k}=y_{-k} ) -P \left(U_{k}\leq \pi_{k}(\gamma(z,y_{-k});\theta^{*}) | Z_{k}=z, Y_{-k}=y_{-k}  \right)\nonumber\\
&\qquad=P (\widetilde{U}_{k}\leq \pi_{k}(z',y_{-k}';\theta) | Z_{k}=z, Y_{-k}=y_{-k} ) -P \left(U_{k}\leq \pi_{k}(z',y_{-k}';\theta^{*}) | Z_{k}=z, Y_{-k}=y_{-k}  \right)\nonumber\\
&\qquad<  \max\{-L_{0}\pi_{k}(z',y_{-k}';\theta^{*}),0\} -\max\{-L_{0}\pi_{k}(z',y_{-k}';\theta),0\}\nonumber\\
&\qquad\leq  \max\{L_{0}\pi_{k}(z',y_{-k}';\theta)-L_{0}\pi_{k}(z',y_{-k}';\theta^{*}),0\} \nonumber\\
&\qquad\leq L_{0} |\pi_{k}(z',y_{-k}';\theta^{*})-\pi_{k}(z',y_{-k}';\theta)|\nonumber\\
&\qquad\leq L_{0} L_{k} ||\theta - \theta^{*}||.
\end{align}
From here we can deduce that Assumption \ref{assumption_error_bound}(ii) is satisfied for any $\delta>0$ with $C_{2}=L_{0}L$ where $L= \min_{k} L_{k}$.

To verify Assumption \ref{assumption_error_bound}(i), we will first introduce the following Lemma and provide a sketch of its proof:

\begin{lemma}\label{lemma_sufficient_condition_sdc}
Consider the simultaneous discrete choice environment of Example \ref{example_simultaneous_discrete_choice}, but with the new moment conditions \eqref{eq_sdc_cons1} and \eqref{eq_sdc_cons2} in place of \eqref{eq_sdc_mom1} and \eqref{eq_sdc_mom2}. Now fix some value $\theta \in \Theta$. If there exists a random variable $U$ with distribution $P_{U|Y,Z} \in \mathcal{P}_{U|Y,Z}(\theta)$ satisfying:
\begin{align}
P \left(U_{k}\leq \pi_{k}(z,y_{-k};\theta) | Z_{k}=z, Y_{-k}=y_{-k}  \right) -0.5 \leq \max\{L_{0}\pi_{k}(z,y_{-k};\theta),0\}, \label{eq_sdc_cons5}\\
0.5 - P \left(U_{k}\leq\pi_{k}(z,y_{-k};\theta) | Z_{k}=z, Y_{-k}=y_{-k}  \right) \leq \max\{-L_{0}\pi_{k}(z,y_{-k};\theta),0\},\label{eq_sdc_cons6}
\end{align}
for $k=1,\ldots,K$ and for every $(z,y_{-k}) \in \mathcal{Z} \times \mathcal{Y}^{K-1}$, then $\theta \in \Theta^{*}$. 
\end{lemma}
\begin{remark}
Note that, precisely because of the result in this Lemma, the new moment conditions implied by \eqref{eq_sdc_cons1} and \eqref{eq_sdc_cons2} satisfy the no-backtracking principle from Remark \ref{remark_no_backtracking}. Indeed, this Lemma shows that \eqref{eq_sdc_cons5} and \eqref{eq_sdc_cons6} are sufficient to characterize the identified set. Since these moment conditions do not depend on the counterfactual $\gamma$ of interest, the no-backtracking principle is satisfied. 
\end{remark}
\begin{proof}[Proof]
Note by assumption there exists a random variable $U$ with distribution $P_{U|Y,Z} \in \mathcal{P}_{U|Y,Z}(\theta)$ satisfying \eqref{eq_sdc_cons5} and \eqref{eq_sdc_cons6} for $k=1,\ldots,K$ and for every $(z,y_{-k}) \in \mathcal{Z} \times \mathcal{Y}^{K-1}$. 
Take $\widetilde{U}$ to be a random vector satisfying:
\begin{align*}
P (\widetilde{U}_{k}\leq \pi_{k}(z,y_{-k};\theta) | Z_{k}=z, Y_{-k}=y_{-k}  ) = P \left(U_{k}\leq \pi_{k}(z,y_{-k};\theta) | Z_{k}=z, Y_{-k}=y_{-k}  \right),
\end{align*}
for $k=1,\ldots,K$ and for every $(z,y_{-k}) \in \mathcal{Z} \times \mathcal{Y}^{K-1}$, so that $\widetilde{U}$ satisfies \eqref{eq_sdc_cons5} and \eqref{eq_sdc_cons6}. We must show that we can fix probabilities of the form $P(\widetilde{U}_{k} \leq \pi_{k}(z',y_{-k}';\theta)|Z_{k}=z,Y_{-k}=y_{-k})$ for $(z',y_{-k}')\neq (z,y_{-k})$ in a way that satisfies the remaining constraints from \eqref{eq_sdc_cons1} and \eqref{eq_sdc_cons2}, as well as the constraints:
\begin{align*}
&P(U_{k} \leq \pi_{k}(z',y_{-k}';\theta)|Z_{k}=z,Y_{-k}=y_{-k}) \leq P(U_{k} \leq \pi_{k}(z,y_{-k};\theta)|Z_{k}=z,Y_{-k}=y_{-k}),
\end{align*}
if $\pi_{k}(z',y_{-k}';\theta) \leq \pi_{k}(z,y_{-k};\theta)$, and:
\begin{align*}
P(U_{k} \leq \pi_{k}(z,y_{-k};\theta)|Z_{k}=z,Y_{-k}=y_{-k}) \leq P(U_{k} \leq \pi_{k}(z',y_{-k}';\theta)|Z_{k}=z,Y_{-k}=y_{-k}),
\end{align*}
if $\pi_{k}(z,y_{-k};\theta) \leq \pi_{k}(z',y_{-k}';\theta)$. However, such an allocation of probability is clearly always possible.   
\end{proof}

The contrapositive of this result says that if $\theta \notin \Theta^{*}$, then there is no random variable $U$ with distribution $P_{U|Y,Z} \in \mathcal{P}_{U|Y,Z}(\theta)$ satisfying \eqref{eq_sdc_cons5} and \eqref{eq_sdc_cons6}; in other words, if $\theta \notin \Theta^{*}$, then every distribution $P_{U|Y,Z} \in \mathcal{P}_{U|Y,Z}(\theta)$ violates either \eqref{eq_sdc_cons5} or \eqref{eq_sdc_cons6}. Thus, the Lemma suggests that when analysing violations of the moment conditions in order to verify Assumption \ref{assumption_error_bound}(i), it suffices to focus on the moment conditions \eqref{eq_sdc_cons5} and \eqref{eq_sdc_cons6}.

Finally, there is an important property that will be utilized repeatedly when verifying Assumption \ref{assumption_error_bound}(i): for any $P_{U|Y,Z}\in \mathcal{P}_{U|Y,Z}(\theta)$ and any $P_{U'|Y,Z}\in \mathcal{P}_{U|Y,Z}(\theta')$, we must have:
\begin{align}
P \left(U_{k}\leq \pi_{k}(z,y_{-k};\theta) | Z_{k}=z, Y_{-k}=y_{-k}  \right) = P \left(U_{k}'\leq \pi_{k}(z,y_{-k};\theta') | Z_{k}=z, Y_{-k}=y_{-k}  \right)\label{eq_sdc_property}
\end{align}
for $k=1,\ldots,K$ and for every $(z,y_{-k}) \in \mathcal{Z} \times \mathcal{Y}^{K-1}$. Indeed, this property follows from the fact that both $P_{U|Y,Z}$ and $P_{U'|Y,Z}$ satisfy the support restrictions for the simultaneous discrete choice model at $\theta$ and $\theta'$, respectively, and thus they must both rationalize the same observed conditional choice probabilities.

Now we are prepared to verify Assumption \ref{assumption_error_bound}(i). First fix some value of $\theta \notin \Theta$. If $\mathcal{P}_{U|Y,Z}(\theta)$ is empty, then Assumption \ref{assumption_error_bound}(i) is satisfied for any $C_{1},\delta>0$. Thus, we will focus attention on the non-trivial case where $\mathcal{P}_{U|Y,Z}(\theta)$ is non-empty. Note that if $P(Y_{k}=1|Z=z,Y_{-k}=y_{-k})=0.5$ for $k=1,\ldots,K$ and for every $(z,y_{-k}) \in \mathcal{Z} \times \mathcal{Y}^{K-1}$, then \eqref{eq_sdc_cons5} and \eqref{eq_sdc_cons6} will be satisfied for any $P_{U|Y,Z} \in \mathcal{P}_{U|Y,Z}(\theta)$. By Lemma \ref{lemma_sufficient_condition_sdc} this implies $\theta \in \Theta^{*}$, contradicting the fact that $\theta \notin \Theta$. We conclude that if $P(Y_{k}=1|Z=z,Y_{-k}=y_{-k})=0.5$ for $k=1,\ldots,K$ and for every $(z,y_{-k}) \in \mathcal{Z} \times \mathcal{Y}^{K-1}$ then $\theta \notin \Theta^{*}$ implies $\mathcal{P}_{U|Y,Z}(\theta)$ is empty, a case we have ruled out. Thus, we will take as a starting point that there exists at least one $k$ and one pair $(z,y_{-k}) \in \mathcal{Z} \times \mathcal{Y}^{K-1}$ such that $P(Y_{k}=1|Z=z,Y_{-k}=y_{-k})\neq 0.5$. Now define:
\begin{align}
\tau:= \min_{k}\min_{(z,y_{-k})} |0.5 - P(Y_{k}=1|Z=z,Y_{-k}=y_{-k})| && s.t. &&|0.5 - P(Y_{k}=1|Z=z,Y_{-k}=y_{-k})|>0.
\end{align}
By assumption and by construction we have $\tau>0$. We now consider violations of the moment conditions \eqref{eq_sdc_cons5} and \eqref{eq_sdc_cons6} in turn. First, consider a violation of \eqref{eq_sdc_cons5}. In particular, for our fixed value of $\theta \notin \Theta$ suppose:
\begin{align}
P (\widetilde{U}_{k}\leq \pi_{k}(z,y_{-k};\theta) | Z_{k}=z, Y_{-k}=y_{-k} ) -0.5 > \max\{L_{0}\pi_{k}(z,y_{-k};\theta),0\},\label{eq_v1}
\end{align}
for some $k$ and $(z,y_{-k})$ pair, where $\widetilde{U}_{k}$ is a subvector of $\widetilde{U}$ whose distribution is a member of $\mathcal{P}_{U|Y,Z}(\theta)$. Furthermore, let $\theta^{*} \in \Theta^{*}$ be the element of $\Theta^{*}$ closest to $\theta$ (such an element exists since $\Theta^{*}$ will be closed, which follows from continuity of the payoff functions).  There are four cases to consider:
\begin{enumerate}
	\item $\pi_{k}(z,y_{-k};\theta^{*}) \leq 0$ and $\pi_{k}(z,y_{-k};\theta)\leq 0$. Then we have:
	\begin{align}
	\max\{L_{0}\pi_{k}(z,y_{-k};\theta),0\} = 0. \label{eq_c1c1}
	\end{align}
	However, since $\pi_{k}(z,y_{-k};\theta^{*}) \leq 0$ it must be that:
	\begin{align*}
	0.5 \geq P \left(U_{k}\leq \pi_{k}(z,y_{-k};\theta^{*}) | Z_{k}=z, Y_{-k}=y_{-k}  \right) = P (\widetilde{U}_{k}\leq \pi_{k}(z,y_{-k};\theta) | Z_{k}=z, Y_{-k}=y_{-k} ),
	\end{align*}
	where we have used property \eqref{eq_sdc_property} and the fact that $\theta^{*}$ satisfies both \eqref{eq_sdc_cons1} and \eqref{eq_sdc_cons2}. But then this implies:
	\begin{align}
	P(\widetilde{U}_{k}\leq \pi_{k}(z,y_{-k};\theta) | Z_{k}=z, Y_{-k}=y_{-k} ) - 0.5 \leq 0. \label{eq_c1c2}
	\end{align}
	Combining \eqref{eq_c1c1} and \eqref{eq_c1c2} contradicts the assumption of \eqref{eq_v1}. Thus, this case is not possible under the assumption of \eqref{eq_v1}.

	\item $\pi_{k}(z,y_{-k};\theta^{*}) \leq 0$ and $\pi_{k}(z,y_{-k};\theta)> 0$. Then we have:
	\begin{align}
	\max\{L_{0}\pi_{k}(z,y_{-k};\theta),0\} = L_{0}\pi_{k}(z,y_{-k};\theta). \label{eq_c2c1}
	\end{align}
	However, since $\pi_{k}(z,y_{-k};\theta^{*}) \leq 0$ then it must be that:
	\begin{align*}
	0.5 \geq P \left(U_{k}\leq \pi_{k}(z,y_{-k};\theta^{*}) | Z_{k}=z, Y_{-k}=y_{-k}  \right) = P(\widetilde{U}_{k}\leq \pi_{k}(z,y_{-k};\theta) | Z_{k}=z, Y_{-k}=y_{-k} ),
	\end{align*} 
	where we have used property \eqref{eq_sdc_property} and the fact that $\theta^{*}$ satisfies both \eqref{eq_sdc_cons1} and \eqref{eq_sdc_cons2}. But then this implies:
	\begin{align}
	 P (\widetilde{U}_{k}\leq \pi_{k}(z,y_{-k};\theta) | Z_{k}=z, Y_{-k}=y_{-k}) - 0.5 \leq 0. \label{eq_c2c2}
	\end{align}
	Combining \eqref{eq_c2c1} and \eqref{eq_c2c2} contradicts the assumption of \eqref{eq_v1}. Thus, this case is not possible under the assumption of \eqref{eq_v1}.

	\item $\pi_{k}(z,y_{-k};\theta^{*}) > 0$ and $\pi_{k}(z,y_{-k};\theta)\leq 0$. Then we have:
	\begin{align*}
	\max\{L_{0}\pi_{k}(z,y_{-k};\theta),0\}=0.
	\end{align*}
	Then:
	\begin{align*}
	&P (\widetilde{U}_{k}\leq \pi_{k}(z,y_{-k};\theta) | Z_{k}=z, Y_{-k}=y_{-k} ) -0.5 -\max\{L_{0}\pi_{k}(z,y_{-k};\theta),0\}\\
	&=P (\widetilde{U}_{k}\leq \pi_{k}(z,y_{-k};\theta) | Z_{k}=z, Y_{-k}=y_{-k} ) -0.5\\
	&\geq \tau,
	\end{align*}
	where the last line follows from the fact that $P (\widetilde{U}_{k}\leq \pi_{k}(z,y_{-k};\theta) | Z_{k}=z, Y_{-k}=y_{-k} )-0.5>0$ by assumption of \eqref{eq_v1} and the fact $\pi_{k}(z,y_{-k};\theta)\leq 0$, and by the definition of $\tau$ from \eqref{eq_tau}.

	\item $\pi_{k}(z,y_{-k};\theta^{*}) > 0$ and $\pi_{k}(z,y_{-k};\theta)> 0$. First note that by assumption we have:
	\begin{align*}
	&P (\widetilde{U}_{k}\leq \pi_{k}(z,y_{-k};\theta) | Z_{k}=z, Y_{-k}=y_{-k} )-0.5  - \max\{L_{0}\pi_{k}(z,y_{-k};\theta),0\}\\
	&>0\\
	&\geq P (U_{k}\leq \pi_{k}(z,y_{-k};\theta^{*}) | Z_{k}=z, Y_{-k}=y_{-k} )-0.5  - \max\{L_{0}\pi_{k}(z,y_{-k};\theta^{*}),0\}.
	\end{align*}
	Using \eqref{eq_sdc_property} and the fact that $\pi_{k}(z,y_{-k};\theta^{*}) >0$ and $\pi_{k}(z,y_{-k};\theta)> 0$, this implies $\pi_{k}(z,y_{-k};\theta^{*}) > \pi_{k}(z,y_{-k};\theta)$.	Now let $\theta'$ be a convex combination of $\theta^{*}$ and $\theta$ satisfying:
	\begin{align*}
	P(U_{k}'\leq \pi_{k}(z,y_{-k};\theta') | Z_{k}=z, Y_{-k}=y_{-k} ) - 0.5  - L_{0}\pi_{k}(z,y_{-k};\theta')=0,
	\end{align*}
	for some selection $U_{k}'$. Such an element always exists by linearity of $\pi_{k}$. Then:
	\begin{align*}
	&P(\widetilde{U}_{k}\leq \pi_{k}(z,y_{-k};\theta) | Z_{k}=z, Y_{-k}=y_{-k} ) - 0.5 - \max\{L_{0}\pi_{k}(z,y_{-k};\theta),0\}\\
	&=P(\widetilde{U}_{k}\leq \pi_{k}(z,y_{-k};\theta) | Z_{k}=z, Y_{-k}=y_{-k} ) - 0.5 - \max\{L_{0}\pi_{k}(z,y_{-k};\theta),0\}\\
	&\qquad+ L_{0}\pi_{k}(z,y_{-k};\theta') - L_{0}\pi_{k}(z,y_{-k};\theta')\\
	&= L_{0}\pi_{k}(z,y_{-k};\theta') - L_{0}\pi_{k}(z,y_{-k};\theta)\\
	&= L_{0}|\pi_{k}(z,y_{-k};\theta') - \pi_{k}(z,y_{-k};\theta)|\\
	&\geq L_{0}L_{k}'||\theta'-\theta||\\
	&\geq L_{0}L_{k}'||\theta^{*}-\theta||.
	\end{align*}
	In the third last line we used the fact that $\pi_{k}(z,y_{-k};\theta^{*}) > \pi_{k}(z,y_{-k};\theta)$. In the second last line we have used the reverse Lipschitz condition, and in the final line we have used the fact that $\theta'$ lies between $\theta$ and $\theta^{*}$, by virtue of being a convex combination of these elements. 
\end{enumerate}

Next, consider a violation of \eqref{eq_sdc_cons6}. In particular, for our fixed $\theta \notin \Theta$ suppose:
\begin{align}
0.5 - P (\widetilde{U}_{k}\leq \pi_{k}(z,y_{-k};\theta) | Z_{k}=z, Y_{-k}=y_{-k} )  > \max\{-L_{0}\pi_{k}(z,y_{-k};\theta),0\},\label{eq_v2}
\end{align}
for some $k$ and $(z,y_{-k})$ pair, where $\widetilde{U}_{k}$ is a subvector of $\widetilde{U}$ whose distribution is is a member of $\mathcal{P}_{U|Y,Z}(\theta)$. Again, let $\theta^{*} \in \Theta^{*}$ be the element of $\Theta^{*}$ closest to $\theta$. There are again four cases to consider:
\begin{enumerate}
	\item $\pi_{k}(z,y_{-k};\theta^{*}) \leq 0$ and $\pi_{k}(z,y_{-k};\theta)\leq 0$. First note that by assumption we have:
	\begin{align*}
	&0.5 - P (\widetilde{U}_{k}\leq \pi_{k}(z,y_{-k};\theta) | Z_{k}=z, Y_{-k}=y_{-k} )  - \max\{-L_{0}\pi_{k}(z,y_{-k};\theta),0\}\\
	&>0\\
	&\geq 0.5 - P (U_{k}\leq \pi_{k}(z,y_{-k};\theta^{*}) | Z_{k}=z, Y_{-k}=y_{-k} )  - \max\{-L_{0}\pi_{k}(z,y_{-k};\theta^{*}),0\}.
	\end{align*}
	Using \eqref{eq_sdc_property} and the fact that $\pi_{k}(z,y_{-k};\theta^{*}) \leq 0$ and $\pi_{k}(z,y_{-k};\theta)\leq 0$, this implies $\pi_{k}(z,y_{-k};\theta^{*}) < \pi_{k}(z,y_{-k};\theta)$. Now let $\theta'$ be a convex combination of $\theta^{*}$ and $\theta$ satisfying:
	\begin{align*}
	0.5- P(U_{k}'\leq \pi_{k}(z,y_{-k};\theta') | Z_{k}=z, Y_{-k}=y_{-k} )  + L_{0}\pi_{k}(z,y_{-k};\theta')=0,
	\end{align*}
	for some selection $U_{k}'$. Such an element always exists by linearity of $\pi_{k}$.  Then:
	\begin{align*}
	&0.5- P(U_{k}\leq \pi_{k}(z,y_{-k};\theta) | Z_{k}=z, Y_{-k}=y_{-k} ) - \max\{-L_{0}\pi_{k}(z,y_{-k};\theta),0\}\\
	&=0.5- P(U_{k}\leq \pi_{k}(z,y_{-k};\theta) | Z_{k}=z, Y_{-k}=y_{-k} ) - \max\{-L_{0}\pi_{k}(z,y_{-k};\theta),0\}\\
	&\qquad+ L_{0}\pi_{k}(z,y_{-k};\theta') - L_{0}\pi_{k}(z,y_{-k};\theta')\\
	&= L_{0}\pi_{k}(z,y_{-k};\theta) - L_{0}\pi_{k}(z,y_{-k};\theta')\\
	&= L_{0}|\pi_{k}(z,y_{-k};\theta) - \pi_{k}(z,y_{-k};\theta')|\\
	&\geq L_{0}L_{k}'||\theta-\theta'||\\
	&\geq L_{0}L_{k}'||\theta-\theta^{*}||.
	\end{align*} 
	In the third last line we used the fact that $\pi_{k}(z,y_{-k};\theta^{*}) < \pi_{k}(z,y_{-k};\theta)$. In the second last line we have used the reverse Lipschitz condition, and in the final line we have used the fact that $\theta'$ lies between $\theta$ and $\theta^{*}$, by virtue of being a convex combination of these elements. 
	\item $\pi_{k}(z,y_{-k};\theta^{*}) \leq 0$ and $\pi_{k}(z,y_{-k};\theta)> 0$. Then we have:
	\begin{align*}
	\max\{-L_{0}\pi_{k}(z,y_{-k};\theta),0\}=0.
	\end{align*}
	Then:
	\begin{align*}
	& 0.5 - P (\widetilde{U}_{k}\leq \pi_{k}(z,y_{-k};\theta) | Z_{k}=z, Y_{-k}=y_{-k} ) -\max\{-L_{0}\pi_{k}(z,y_{-k};\theta),0\}\\
	&=0.5 - P (\widetilde{U}_{k}\leq \pi_{k}(z,y_{-k};\theta) | Z_{k}=z, Y_{-k}=y_{-k} )\\
	&\geq \tau,
	\end{align*} 
	where the last line follows from the fact that $0.5 - P (\widetilde{U}_{k}\leq \pi_{k}(z,y_{-k};\theta) | Z_{k}=z, Y_{-k}=y_{-k} )>0$ by assumption of \eqref{eq_v2} and the fact $\pi_{k}(z,y_{-k};\theta)> 0$, and by the definition of $\tau$ from \eqref{eq_tau}. 

	\item $\pi_{k}(z,y_{-k};\theta^{*}) > 0$ and $\pi_{k}(z,y_{-k};\theta)\leq 0$. Then we have:
	\begin{align}
	\max\{-L_{0}\pi_{k}(z,y_{-k};\theta),0\} = -L_{0}\pi_{k}(z,y_{-k};\theta). \label{eq_c2c1v2}
	\end{align}
	However, since $\pi_{k}(z,y_{-k};\theta^{*}) > 0$ then it must be that:
	\begin{align*}
	0.5 \leq P \left(U_{k}\leq \pi_{k}(z,y_{-k};\theta^{*}) | Z_{k}=z, Y_{-k}=y_{-k}  \right) = P(\widetilde{U}_{k}\leq \pi_{k}(z,y_{-k};\theta) | Z_{k}=z, Y_{-k}=y_{-k} ),
	\end{align*} 
	where we have used property \eqref{eq_sdc_property} and the fact that $\theta^{*}$ satisfies both \eqref{eq_sdc_cons1} and \eqref{eq_sdc_cons2}. But then this implies:
	\begin{align}
	0.5 - P(\widetilde{U}_{k}\leq \pi_{k}(z,y_{-k};\theta) | Z_{k}=z, Y_{-k}=y_{-k} )  \leq 0. \label{eq_c2c2v2}
	\end{align}
	Combining \eqref{eq_c2c1v2} and \eqref{eq_c2c2v2} contradicts the assumption of \eqref{eq_v2}. Thus, this case is not possible under the assumption of \eqref{eq_v2}.

	\item $\pi_{k}(z,y_{-k};\theta^{*}) > 0$ and $\pi_{k}(z,y_{-k};\theta)> 0$. Then we have:
	\begin{align}
	\max\{-L_{0}\pi_{k}(z,y_{-k};\theta),0\} = 0. \label{eq_c1c1v2}
	\end{align}
	However, since $\pi_{k}(z,y_{-k};\theta^{*}) > 0$ then it must be that:
	\begin{align*}
	0.5 \leq P \left(U_{k}\leq \pi_{k}(z,y_{-k};\theta^{*}) | Z_{k}=z, Y_{-k}=y_{-k}  \right) = P (\widetilde{U}_{k}\leq \pi_{k}(z,y_{-k};\theta) | Z_{k}=z, Y_{-k}=y_{-k} ),
	\end{align*}
	where we have used property \eqref{eq_sdc_property} and the fact that $\theta^{*}$ satisfies both \eqref{eq_sdc_cons1} and \eqref{eq_sdc_cons2}. But then this implies:
	\begin{align}
	0.5 - P(\widetilde{U}_{k}\leq \pi_{k}(z,y_{-k};\theta) | Z_{k}=z, Y_{-k}=y_{-k} )  \leq 0. \label{eq_c1c2v2}
	\end{align}
	Combining \eqref{eq_c1c1v2} and \eqref{eq_c1c2v2} contradicts the assumption of \eqref{eq_v2}. Thus, this case is not possible under the assumption of \eqref{eq_v2}.

\end{enumerate}

Combining everything, we conclude that Assumption \ref{assumption_error_bound} holds with $C_{1} = L_{0} L'$ and $\delta = \tau/(L_{0} L')$, where $L' = \min_{k} L_{k}'$. 

\subsubsection{Verification of Learnability}\label{appendix_additional_details_sdc_verify_learnable}

By the assumed linearity of $\pi_{k}$ with respect to $\theta$, and since $\pi_{k}$ depends only on the subvector $\theta_{k}$ of $\theta$, the function $(u,\theta) \mapsto \pi_{k}(\gamma(z,y_{-k});\theta) - u$ is a hyperplane in $\mathbb{R}^{d_{k}}$ for each $(z,y_{-k})$, where $d_{k}$ is the dimension of $\theta_{k}$. By Lemma 2.6.15 in \cite{van1996weak}, for example, $\Phi$ is a Vapnik-Chervonenkis (VC) class with VC dimension at most $d_{k}+2$. Furthermore, recalle that $\Phi$ can be taken to be uniformly bounded in absolute value by $1$. Using, for example, Theorem 2.6.7 in \cite{van1996weak}, we can deduce:
\begin{align*}
\sup_{Q \in \mathcal{Q}_{n}} \log N \left( \varepsilon, \Phi, ||\cdot||_{Q,2} \right) = O(1),
\end{align*}
so that $\Phi$ easily satisfies the entropy growth condition. Now let $j$ index a generic moment function:
\begin{align*}
m_{j}(Y_{-k},Z,U,\theta) = \left(\mathbbm{1}\{U_{k}\leq \pi_{k}(z',y_{-k}';\theta)\} - \max\{L_{0}\pi_{k}(z',y_{-k}';\theta),0\} - 0.5\right)\mathbbm{1}\{Z_{k}=z, Y_{-k}=y_{-k}\},
\end{align*}
and let $\mathcal{M}_{j}$ be the associated class of functions: 
\begin{align*}
\mathcal{M}_{j}= \left\{ m_{j}(\cdot,u,\theta):\mathcal{Y}\times\mathcal{Z} \to \mathbb{R} : (u,\theta) \in \mathcal{U}\times\Theta \right\}.
\end{align*}
Note that the values $(z',y_{-k}')$ are not arguments of the function, but instead are associated with the index $j$. Since $\pi_{k}$ takes values in the interval $[-1,1]$, the class $\mathcal{M}_{j}$ is uniformly bounded. We claim that there exists no set of size $2$ shattered by $\mathcal{M}_{j}$, implying $\mathcal{M}_{j}$ is a VC-subgraph class. We will prove this by way of contradiction. In particular, suppose that there exists two points $(y_{1},z_{1})$ and $(y_{2},z_{2})$, and values $t_{1},t_{2} \in \mathbb{R}$ such that:
\begin{align}
\left| \left\{\begin{bmatrix}\mathbbm{1}\{m_{j}(y_{1},z_{1},u,\theta) \geq t_{1}\}\\ \mathbbm{1}\{m_{j}(y_{2},z_{2},u,\theta) \geq t_{2}\}\end{bmatrix}  : (u,\theta) \in \mathcal{U} \times \Theta \right\}\right| = 4.\label{eq_sdc_shattering}
\end{align}
In other words, we suppose the set $\{(y_{1},z_{1}),(y_{2},z_{2})\}$ is shattered by $\mathcal{M}_{j}$, and that $t_{1},t_{2} \in \mathbb{R}$ witness the shattering. We have:
\begin{align*}
m_{j}(y_{1},z_{1},u,\theta) = \left(\mathbbm{1}\{u_{k}\leq \pi_{k}(z',y_{-k}';\theta)\} - \max\{L_{0}\pi_{k}(z',y_{-k}';\theta),0\} - 0.5\right)\mathbbm{1}\{z_{1,k}=z, y_{1,-k}=y_{-k}\},\\
m_{j}(y_{2},z_{2},u,\theta) = \left(\mathbbm{1}\{u_{k}\leq \pi_{k}(z',y_{-k}';\theta)\} - \max\{L_{0}\pi_{k}(z',y_{-k}';\theta),0\} - 0.5\right)\mathbbm{1}\{z_{2,k}=z, y_{2,-k}=y_{-k}\}.
\end{align*}
Now consider two cases:
\begin{enumerate}
	\item $\mathbbm{1}\{z_{1,k}=z, y_{1,-k}=y_{-k}\}=\mathbbm{1}\{z_{1,k}=z, y_{1,-k}=y_{-k}\}$: In this case the two functions $m_{j}(y_{1},z_{1},u,\theta)$ and $m_{j}(y_{2},z_{2},u,\theta)$ are identical for all $(u,\theta) \in \mathcal{U} \times \Theta$. This means \eqref{eq_sdc_shattering} is impossible, since at least one of the vectors $(1,0)$ and $(0,1)$ cannot be picked out by $\mathcal{M}_{j}$.
	\item  $\mathbbm{1}\{z_{1,k}=z, y_{1,-k}=y_{-k}\}\neq \mathbbm{1}\{z_{1,k}=z, y_{1,-k}=y_{-k}\}$: In this case at least one of the functions $m_{j}(y_{1},z_{1},u,\theta)$ or $m_{j}(y_{2},z_{2},u,\theta)$ is the zero function. Again, this means \eqref{eq_sdc_shattering} is impossible. For example, if $m_{j}(y_{1},z_{1},u,\theta)$ is the zero function, then it is impossible for $\mathcal{M}_{j}$ to pick out both $(0,0)$ and $(1,0)$ or both $(0,1)$ and $(1,1)$.
\end{enumerate}
Since $(y_{1},z_{1})$ and $(y_{2},z_{2})$ were arbitrary, we conclude that there exists no set of size $2$ shattered by $\mathcal{M}_{j}$. This implies that $\mathcal{M}_{j}$ is a VC-subgraph class, and using, for example, Theorem 2.6.7 in \cite{van1996weak}, we can deduce:
\begin{align*}
\sup_{Q \in \mathcal{Q}_{n}} \log N \left( \varepsilon, \mathcal{M}_{j}, ||\cdot||_{Q,2} \right) = O(1).
\end{align*}
Thus, $\mathcal{M}_{j}$ easily satisfies the entropy growth condition. Finally, let $j'$ index a generic moment function:
\begin{align*}
m_{j}(Y_{-k},Z,U,\theta) = \left(0.5 - \mathbbm{1}\{U_{k}\leq \pi_{k}(z',y_{-k}';\theta)\}-\max\{-L_{0}\pi_{k}(z',y_{-k}';\theta),0\}\right) \mathbbm{1}\{Z_{k}=z, Y_{-k}=y_{-k}\},
\end{align*}
and let $\mathcal{M}_{j'}$ be the associated class of functions: 
\begin{align*}
\mathcal{M}_{j'}= \left\{ m_{j'}(\cdot,u,\theta):\mathcal{Y}\times\mathcal{Z} \to \mathbb{R} : (u,\theta) \in \mathcal{U}\times\Theta \right\}.
\end{align*}
A nearly identical argument as for $\mathcal{M}_{j}$ reveals that $\mathcal{M}_{j'}$ is a VC-subgraph class and thus trivially satisfies the entropy growth condition. We conclude using Theorem \ref{theorem_pampac_learnability}(ii) that our class of policies $\Gamma$ is PAMPAC learnable with a rate of convergence of $O(n^{-1/2})$.

\subsection{Example \ref{example_program_evaluation}: Program Evaluation}\label{appendix_additional_details_te}

\subsubsection{Verification of Assumptions \ref{assump_preliminary}, \ref{assumption_factual_domain} and \ref{assumption_counterfactual_domain}}\label{appendix_additional_details_te_A1_to_A3}

We will now proceed to verify Assumption \ref{assump_preliminary}, \ref{assumption_factual_domain} and \ref{assumption_counterfactual_domain}. First note that Assumption \ref{assump_preliminary} is trivially satisfied, since the probability space $(\Omega,\mathfrak{A},P)$ is complete, $\mathcal{U}$ is a compact subset of euclidean space, and $\Theta$ is a Polish space; in particular, since $\mathcal{Z}$ (and thus also $\mathcal{X}$) is finite, $\mathcal{G}$ can be considered as the set of all positive measurable functions  $g:\mathcal{Z}\to [0,1]$, in which case each $g \in \mathcal{G}$ has an equivalent representation as a vector in $[0,1]^{|\mathcal{Z}|}$. The same logic applies to each $t \in \mathcal{T}$. Next, let us recall the multifunction:
\begin{align}
\bm G^{-} \left( Y,D,Z,\theta \right):= \text{cl}\left\{ (U_{0},U_{1},U) \in \mathcal{U} : \begin{array}{l}
    Y = U_{0}(1-D) + U_{1}D,  \\
    D = \mathbbm{1}\{g(Z) \geq U\}
  \end{array} \right\}.
\end{align}
Close inspection of this multifunction shows that:
\begin{align}
\bm G^{-} \left( y,d,z,\theta \right)=\begin{cases}
	\{ y \} \times [\underline{Y},\overline{Y}] \times [g(z),1], &\text{ if } d=0,\\
	[\underline{Y},\overline{Y}] \times \{ y \} \times [0,g(z)], &\text{ if } d=1.
\end{cases}
\end{align}
Now for any $(u_{0},u_{1},u)\in \mathcal{U}$ we have:
\begin{align} 
&d((u_{0},u_{1},u),\bm G^{-} \left( Y,D,Z,\theta \right))\nonumber\\
&= D \max\{|u_{0} - Y|,g(Z)-u \} + (1-D)\max\{|u_{1} - Y|,Z-g(z) \}.
\end{align}
Since $g \in \mathcal{G}$ is measurable by definition, from here it is easily verified that the distance above is measurable with respect to $\mathfrak{B}(\mathcal{Y})\otimes\mathfrak{B}(\mathcal{D})\otimes\mathfrak{B}(\mathcal{Z})$. Since $(u_{0},u_{1},u) \in \mathcal{U}$ was arbitrary, by the result of \cite{himmelberg1975measurable} (see also Theorem 1.3.3 in \cite{molchanov2017theory}) this implies that $\bm G^{-}$ is an Effros-measurable multifunction, as desired. Modulo changes in notation, it is easily seen that the conditional distribution of the vector $(U_{0},U_{1},U)$ given $(Y,D,Z)$ satisfies \eqref{eq_puyz} in Assumption \ref{assumption_factual_domain} using the multifunction in \eqref{eq_POM_reverse_correspondence} with $g(\cdot)=g_{0}(\cdot)$. Finally, note that all of the moment functions in the moment conditions \eqref{eq_pom_mom1} - \eqref{eq_pom_mom6} are measurable and bounded by $1$, and the moment functions from the moment conditions in \eqref{eq_pom_mom7} and \eqref{eq_pom_mom8} are measurable and bounded by $\max\{|\underline{Y}|,|\overline{Y}|\}$.  

Turning to the counterfactual domain, recall the multifunction:
\begin{align}
\bm G^\star(Z,U_{0},U_{1},U,\theta,\gamma):=\left\{ (Y_{\gamma}^\star,D_{\gamma}^\star) \in \mathcal{Y} \times \{0,1\} :   \begin{array}{l}
    Y_{\gamma}^\star = U_{0}(1-D_{\gamma}^\star) + U_{1}D_{\gamma}^\star,  \\
    D_{\gamma}^\star = \mathbbm{1}\{g(\gamma(Z)) \geq U\}
  \end{array} \right\}.
\end{align}
Note here we take $\mathcal{Y}^\star = \mathcal{Y}$, although this is not necessary. Furthermore, close inspection of this multifunction shows that:
\begin{align}
\bm G^\star(z,u_{0},u_{1},u,\theta,\gamma)=\begin{cases}
	(u_{1},1), &\text{ if } u \leq g(\gamma(z)),\\
	(u_{0},0), &\text{ if } g(\gamma(z)) < u.
\end{cases}
\end{align} 
In this case, the counterfactual map in \eqref{eq_POM_counterfactual_correspondence} is single-valued. In this case, Effros measurability is equivalent to the usual notion of measurability for functions, and measurability of $\bm G^\star$ follows from familiar arguments after noting that both $g$ and $\gamma$ are measurable functions. Finally, modulo changes in notation, it is easily seen that the conditional distribution of the vector $(Y_{\gamma}^\star,D_{\gamma}^\star)$ given $(Y,D,Z,U_{0},U_{1},U)$ satisfies \eqref{eq_pygamma} in Assumption \ref{assumption_counterfactual_domain} using the multifunction in \eqref{eq_POM_counterfactual_correspondence} with $g(\cdot)=g_{0}(\cdot)$.

\subsubsection{Verification of Assumption \ref{assumption_error_bound}}\label{appendix_additional_details_te_A4_error_bound}

First we focus on \eqref{eq_pom_mom1} - \eqref{eq_pom_mom4}.\todolt{There is some question for this assumption when the moment conditions do not depend on $\theta$...} Since these moments do not depend on $t \in \mathcal{T}$, to verify Assumption \ref{assumption_error_bound} it suffices to focus on the parameter $g \in \mathcal{G}$. From the moment conditions \eqref{eq_pom_mom1} and \eqref{eq_pom_mom2} we have:
\begin{align}
g(z_{0},x) = P(D=1|Z=z_{0},X=x) \iff
\begin{cases}
	\E[\left(D - g_{0}(z_{0},x)\right)\mathbbm{1}\{Z_{0}=z_{0},X=x\}] \leq 0 \\
	\E[\left(g_{0}(z_{0},x)-D\right)\mathbbm{1}\{Z_{0}=z_{0},X=x\}] \leq 0
\end{cases},\label{eq_pom_mom1_violate}
\end{align}
and from \eqref{eq_pom_mom3} and \eqref{eq_pom_mom4} we have:
\begin{align}
g_{0}(z_{0},x)  = P(U \leq g_{0}(z_{0},x)| X=x) \iff
\begin{cases}
	\E[\left(\mathbbm{1}\{U \leq g_{0}(z_{0},x)\} - g_{0}(z_{0},x)\right)\mathbbm{1}\{X=x\}]\leq 0\\
	\E[\left(g_{0}(z_{0},x) - \mathbbm{1}\{U \leq g_{0}(z_{0},x)\}\right)\mathbbm{1}\{X=x\}]\leq 0
\end{cases},\label{eq_pom_mom3_violate}
\end{align}
For notational simplicity, let $g_{0}(z) := g_{0}(z_{0},x)$ for $z=(z_{0},x)$. From \eqref{eq_pom_mom1_violate} we see that $g_{0}(z)$ is point-identified. Define:
\begin{align*}
\mathcal{G}^{*} = \{ g : (g,t) \in \Theta^{*} \text{ for some $t\in \mathcal{T}$}\}. 
\end{align*}
Then point-identification of $g_{0}$ implies that $\mathcal{G}^{*}$ is a singleton, and that for any $g \in \mathcal{G}$:
\begin{align*}
d(g,\mathcal{G}^{*}) = \max_{z \in \mathcal{Z}} |g(z) - g_{0}(z)|.
\end{align*}
From here it is straightforward to use conditions \eqref{eq_pom_mom1_violate} and \eqref{eq_pom_mom3_violate} to argue that part (i) of Assumption \ref{assumption_error_bound} is satisfied with $C_{1}=1$ for any $\delta>0$. In particular, suppose $g \notin \mathcal{G}^{*}$, and that $z^{*} \in \mathcal{Z}$ satisfies:
\begin{align*}
d(g,\mathcal{G}^{*}) = \max_{z \in \mathcal{Z}} |g(z) - g_{0}(z)|=|g(z^{*}) - g_{0}(z^{*})|.
\end{align*}
Without loss of generality, suppose that $g(z^{*})>g_{0}(z^{*})$. Then from \eqref{eq_pom_mom1_violate} we have:
\begin{align*}
\E[\left(g_{0}(z^{*})-D\right)\mathbbm{1}\{Z=z^{*}\}]= 0<\E[\left(g(z^{*})-D\right)\mathbbm{1}\{Z=z^{*}\}].
\end{align*}
Thus:
\begin{align*}
\E[\left(g(z^{*})-D\right)\mathbbm{1}\{Z=z^{*}\}] &= \E[\left(g(z^{*})-D\right)\mathbbm{1}\{Z=z^{*}\}] - \E[\left(g_{0}(z^{*})-D\right)\mathbbm{1}\{Z=z^{*}\}]\\
&= g(z^{*}) - g_{0}(z^{*})\\
&= |g(z^{*}) - g_{0}(z^{*})|\\
&= d(g,\mathcal{G}^{*}).
\end{align*}
Now to complete the verification of part (i) of Assumption \ref{assumption_error_bound} we turn to \eqref{eq_pom_mom5} - \eqref{eq_pom_mom8}, which can be written as:
\begin{align}
\E\left[t(z_{0},x)-\mathbbm{1}\{Z=z_{0},X=x\}\right] &= 0,\qquad\forall z_{0} \in \mathcal{Z}_{0},\, x \in \mathcal{X},\label{eq_pom_mom56}
\end{align} 
and:
\begin{align}
\E\left[U_{d}\left(\mathbbm{1}\{Z=z_{0},X=x\}\sum_{z_{0} \in \mathcal{Z}_{0}} t(z_{0},x)- \mathbbm{1}\{X=x\}t(z_{0},x) \right)\right] &\leq 0, \,\,\forall z_{0} \in \mathcal{Z}_{0},\, x \in \mathcal{X},\, d\in \{0,1\}.\label{eq_pom_mom78}
\end{align}
Since these moments do not depend on $g \in \mathcal{G}$, to verify Assumption \ref{assumption_error_bound} for these moments it suffices to focus on the parameter $t \in \mathcal{T}$. Now define:
\begin{align*}
\mathcal{T}^{*} = \{ t : (g,t) \in \Theta^{*} \text{ for some $g\in \mathcal{G}$}\}. 
\end{align*}
From \eqref{eq_pom_mom56} it is clear that $t_{0}$ is also point identified. Since $g_{0}$ is also point identified we have $\Theta^{*} = \{g_{0}\}\times \{t_{0}\}$. Because of this, we claim that it suffices to focus on the conditions from \eqref{eq_pom_mom56}; indeed, $t \notin \mathcal{T}^{*} \iff t\neq t_{0}$ implies that $t \notin \mathcal{T}^{*}$ if and only if \eqref{eq_pom_mom56} is violated. Now consider any $t \notin \mathcal{T}^{*}$ and let $(z_{0}^{*},x^{*})$ satisfy:
\begin{align*}
(z_{0}^{*},x^{*}) = \argmax_{z_{0},x} |t(z_{0},x) - t_{0}(z_{0},x)|.
\end{align*}
Without loss of generality we can suppose $t(z_{0},x) > t_{0}(z_{0},x)$. Then:
\begin{align*}
\E\left[t(z_{0}^{*},x^{*})-\mathbbm{1}\{Z=z_{0}^{*},X=x^{*}\}\right] &= \E\left[t(z_{0}^{*},x^{*})-\mathbbm{1}\{Z=z_{0}^{*},X=x^{*}\}\right] - \E\left[t_{0}(z_{0}^{*},x^{*})-\mathbbm{1}\{Z=z_{0}^{*},X=x^{*}\}\right] \\
&=t(z_{0}^{*},x^{*}) -t_{0}(z_{0}^{*},x^{*})\\
&= |  t(z_{0}^{*},x^{*}) -t_{0}(z_{0}^{*},x^{*})|\\
&=d(t,\mathcal{T}^{*}).
\end{align*}
Combining everything, if $\mathcal{J}$ indexes all the moment constraints and if $\theta \notin \Theta^{*}$ with $\theta=(g,t)$, then we know:
\begin{align*}
\inf_{P_{U_{0},U_{1},U|Y,D,Z} \in \mathcal{P}_{U_{0},U_{1},U|Y,D,Z}(\theta) } \max_{j \in \mathcal{J}} | \E[m_{j}(y,d,z,u_{0},u_{1},u,\theta)]|_{+} \geq \max\{d(g,\mathcal{G}^{*}), d(t,\mathcal{T}^{*})\} \geq d(\theta,\Theta^{*}).
\end{align*}
Conclude that Assumption \ref{assumption_error_bound} is satisfied with $C_{1}=1$ for any $\delta>0$.\\

For part (ii) of Assumption \ref{assumption_error_bound}, we claim that we can set $C_{2} = 1$. To show why, we will apply Lemma \ref{lemma_lipschitz_condition} to our environment. First note that $\varphi$ is the identity function when we are interested in $\E_{P}[Y_{\gamma}^\star]$. Thus $L_{\varphi}=1$ in Lemma \ref{lemma_lipschitz_condition}. Next, note from the definition of our support restrictions $\bm G^{-}$ and $\bm G^\star$ we can deduce that: 
\begin{align}
d((u_{0},u_{1},u),\bm G^{-}(y,d,z,\theta)) &= 
\begin{cases}
	\max\{|u_{0}-y|,|g(z)-u|_{+}\}, &\text{ if }d=0,\\
	\max\{|u_{1}-y|,|u-g(z)|_{+}\}, &\text{ if }d=1.
\end{cases}\label{eq_distance_pom_1}\\\nonumber\\
d((y^\star,d^\star),\bm G^{\star}(y,d,z,u_{0},u_{1},u,\theta,\gamma)) &= 
\begin{cases}
	\max\{|u_{0}-y|,|g(z)-u|_{+}\}, &\text{ if }u> g(\gamma(z)),\\
	\max\{|u_{1}-y|,|u-g(z)|_{+}\}, &\text{ if }u\leq g(\gamma(z)).
\end{cases}\label{eq_distance_pom_2}
\end{align}
We now define the sets $\Theta^{-}$ and $\Theta^{\star}$ given in Lemma \ref{lemma_lipschitz_condition} in the context of this example:
\begin{align}
\Theta^{-}(y,d,z,u_{0},u_{1},u)\cap \Theta_{\delta}^{*} &:= \left\{\theta \in \Theta_{\delta}^{*} : (u_{0},u_{1},u) \in \bm G^{-}(y,d,z,\theta) \right\} &&\nonumber\\\nonumber\\
&=\begin{cases}
	\{\theta\in \Theta_{\delta}^{*} : g(z)\in [0,u]\}, &\text{ if }d=0 \text{ and } u_{0}=y,\\
	\{\theta\in \Theta_{\delta}^{*} : g(z)\in [u,1]\}, &\text{ if }d=1 \text{ and } u_{1}=y,\\
	\emptyset, &\text{ otherwise,}
\end{cases}&&\\\nonumber\\
\Theta^{\star}(v,\gamma)\cap \Theta_{\delta}^{*} &:= \left\{\Theta_{\delta}^{*} \in \Theta : (y^\star,d^\star) \in \bm G^{\star}(y,d,z,u_{0},u_{1},u,\theta,\gamma) \right\}\nonumber &&\\
\nonumber\\
&=\begin{cases}
	\{\theta \in \Theta_{\delta}^{*} : g(\gamma(z)) \in [0,u]\}, &\text{ if }d^{*}=0 \text{ and } y^\star=u_{0},\\
	\{\theta \in \Theta_{\delta}^{*} : g(\gamma(z)) \in [u,1]\}, &\text{ if }d^{*}=1 \text{ and } y^\star=u_{1},\\
	\emptyset, &\text{ otherwise.}
\end{cases}&&
\end{align}
With these definitions, we have for any $\theta \in \Theta_{\delta}^{*}$:
\begin{align}
d(\theta,\Theta^{-}(y,d,z,u_{0},u_{1},u)\cap \Theta_{\delta}^{*}) &=\begin{cases}
	|g(z)-u|_{+}, &\text{ if }d=0 \text{ and } u_{0}=y,\\
	|u-g(z)|_{+}, &\text{ if }d=1 \text{ and } u_{1}=y,\\
	+\infty, &\text{ otherwise,}
\end{cases}\label{eq_distance_pom_3}\\\nonumber\\
d(\theta,\Theta^{\star}(v,\gamma)\cap \Theta_{\delta}^{*})&=\begin{cases}
	|g(\gamma(z))-u|_{+}, &\text{ if }d^{*}=0 \text{ and } y^\star=u_{0},\\
	|u-g(\gamma(z))|_{+}, &\text{ if }d^{*}=1 \text{ and } y^\star=u_{1},\\
	\emptyset, &\text{ otherwise.}
\end{cases}\label{eq_distance_pom_4}
\end{align}
Combining \eqref{eq_distance_pom_1} with \eqref{eq_distance_pom_3} we can verify condition \eqref{eq_lipschitz_cond1} with $\ell_{1}=1$. Furthermore, by combining \eqref{eq_distance_pom_2} with \eqref{eq_distance_pom_4} we can verify condition \eqref{eq_lipschitz_cond2} with $\ell_{2}=1$. Applying Lemma \ref{lemma_lipschitz_condition} then yields the choice $C_{2} = L_{\varphi} \max\{ \ell_{1},\ell_{2}\} = 1$, as claimed above. Note also that this value of $C_{2}$ works for any $\delta>0$.

It thus suffices to set $\mu^{*}=1$ in Theorem \ref{thm_cortes}. Also, recall the moment functions for this example from equations \eqref{eq_pom_mom1} - \eqref{eq_pom_mom4}. The Theorem then states that the lower and upper bounds on the closed convex hull of the identified set for $\E[Y_{\gamma}^\star]$ can be computed as the solutions to the problems \eqref{eq_lb_varphi} and \eqref{eq_ub_varphi}. Intuitively, under the assumptions of the Theorem the infimum over $\theta \in \Theta$ and supremum over $\theta \in \Theta$ in problems \eqref{eq_lb_varphi} and \eqref{eq_ub_varphi} will be obtained at the value $\theta_{0} \in \Theta$.

\subsubsection{Verification of Learnability}\label{appendix_additional_details_te_verify_learnable}

We claim that $\Phi$ is a VC class with VC index of at most $|\mathcal{Z}|+1$. To prove this, we must show that there exists no set of points $\mathcal{Z}_{n}=\{z_{1},\ldots,z_{n}\}$ with $n=|\mathcal{Z}|+1$ shattered by $\Phi$. Let $t_{1},\ldots,t_{n}$ be arbitrary real numbers. Now define the set:
\begin{align*}
B:= \left\{\begin{bmatrix}\mathbbm{1}\left\{\,\, \mathbbm{1}\{g(\gamma(z_{1})) \geq u \}(u_{1}-u_{0}) + u_{0}\geq t_{1} \,\,\right\} \\ \mathbbm{1}\left\{\,\,\mathbbm{1}\{g(\gamma(z_{2})) \geq u \}(u_{1}-u_{0}) + u_{0}\geq t_{2} \,\,\right\} \\ \vdots \\ \mathbbm{1}\left\{\,\,\mathbbm{1}\{g(\gamma(z_{n})) \geq u \}(u_{1}-u_{0}) + u_{0}\geq t_{n}\,\,\right\} \end{bmatrix} : (u_{0},u_{1},u,\theta) \in \mathcal{U} \times \Theta \right\}.
\end{align*}
If $B$ contains the vector $b\in \{0,1\}^{n}$, then we say that $\Phi$ ``picks out'' $b$. It suffices to show that there always exists at least one vector $b \in \{0,1\}^{n}$ that $\Phi$ fails to pick out. Since $n>|\mathcal{Z}|$, there exists at least one $z \in \mathcal{Z}$ that appears twice in the set $\mathcal{Z}_{n}$. Thus there is some $i,j \in \{1,\ldots,n\}$ such that $z_{i}=z_{j}$. Then regardless of the values of $(u_{0},u_{1},u,\theta)$ we will always have:
\begin{align*}
\mathbbm{1}\{g(\gamma(z_{i})) \geq u \}(u_{1}-u_{0}) + u_{0} = \mathbbm{1}\{g(\gamma(z_{j})) \geq u \}(u_{1}-u_{0}) + u_{0}. 
\end{align*}
We then have:
\begin{enumerate}
	\item If $t_{i} = t_{j}$ then $\Phi$ fails to pick out any vector $b \in \{0,1\}^{n}$ with $b_{i}=0$ and $b_{j}=1$. 
	\item If $t_{i} < t_{j}$ then $\Phi$ fails to pick out any vector $b \in \{0,1\}^{n}$ with $b_{i}=0$ and $b_{j}=1$. 
	\item If $t_{j} < t_{i}$ then $\Phi$ fails to pick out any vector $b \in \{0,1\}^{n}$ with $b_{i}=1$ and $b_{j}=0$. 
\end{enumerate}
Since this covers all possibilities for $t_{i},t_{j} \in \mathbb{R}$, we conclude that there always exists at least one binary vector that $\Phi$ fails to pick out, and thus $\Phi$ shatters no set of size $n = |\mathcal{Z}|+1$. Now using, for example, Theorem 2.6.7 in \cite{van1996weak}, we can deduce:
\begin{align*}
\sup_{Q \in \mathcal{Q}_{n}} \log N \left( \varepsilon, \Phi, ||\cdot||_{Q,2} \right) = O(1),
\end{align*}
so that $\Phi$ easily satisfies the entropy growth condition. Now let $j$ index a generic moment function:
\begin{align*}
m_{j}(D,Z,\theta) = \left(D - g(z_{0},x)\right)\mathbbm{1}\{Z_{0}=z_{0},X=x\},
\end{align*}
and let $\mathcal{M}_{j}$ be the associated class of functions: 
\begin{align*}
\mathcal{M}_{j}= \left\{ m_{j}(\cdot,\theta):\{0,1\}\times \mathcal{Z} \to \mathbb{R} : \theta \in \Theta \right\}.
\end{align*}
Note this class indexes the moment functions from the moment conditions \eqref{eq_pom_mom1}. Also note that $(z_{0},x)$ are not arguments of the moment function, but are instead are associated with the index $j$.

We claim that there exists no set of size $3$ shattered by $\mathcal{M}_{j}$, implying $\mathcal{M}_{j}$ is a VC-subgraph class. We will prove this by way of contradiction. In particular, suppose that there exists three points $(d_{1},z_{1})$, $(d_{2},z_{2})$, and $(d_{3},z_{3})$ and values $t_{1},t_{2}, t_{3} \in \mathbb{R}$ such that:
\begin{align}
\left| \left\{\begin{bmatrix}\mathbbm{1}\{m_{j}(d_{1},z_{1},\theta) \geq t_{1}\}\\ \mathbbm{1}\{m_{j}(d_{2},z_{2},\theta) \geq t_{2}\}\\
 \mathbbm{1}\{m_{j}(d_{3},z_{3},\theta) \geq t_{3}\}\end{bmatrix}  : \theta \in \Theta \right\}\right| = 8.\label{eq_sdc_shattering2}
\end{align}
In other words, we suppose the set $\{(d_{1},z_{1}),(d_{2},z_{2}),(d_{3},z_{3})\}$ is shattered by $\mathcal{M}_{j}$, and that $t_{1},t_{2},t_{3} \in \mathbb{R}$ witness the shattering. We have:
\begin{align*}
m_{j}(d_{1},z_{1},\theta) &=\left(d_{1} - g(z_{0},x)\right)\mathbbm{1}\{z_{1,0}=z_{0},x_1=x\},\\
m_{j}(d_{2},z_{2},\theta) &=\left(d_{2} - g(z_{0},x)\right)\mathbbm{1}\{z_{2,0}=z_{0},x_{2}=x\}\\
m_{j}(d_{3},z_{3},\theta) &=\left(d_{3} - g(z_{0},x)\right)\mathbbm{1}\{z_{3,0}=z_{0},x_{3}=x\}.
\end{align*}
Now consider two cases:
\begin{enumerate}
	\item $\mathbbm{1}\{z_{1,0}=z_{0},x_1=x\}=\mathbbm{1}\{z_{2,0}=z_{0},x_{2}=x\}=\mathbbm{1}\{z_{3,0}=z_{0},x_{3}=x\}=1$: Note that since $d_{i} \in \{0,1\}$, at least two functions $m_{j}(d_{1},z_{1},\theta)$, $m_{j}(d_{2},z_{2},\theta)$ and $m_{j}(d_{3},z_{3},\theta)$ are identical for all $\theta \in \Theta$. This means \eqref{eq_sdc_shattering2} is impossible. For instance, suppose that $m_{j}(d_{1},z_{1},\theta)= m_{j}(d_{2},z_{2},\theta)$. Then at least one of the vectors $(1,0,0)$ or $(0,1,0)$ cannot be picked out by $\mathcal{M}_{j}$. 
	\item  Either $\mathbbm{1}\{z_{1,k}=z, y_{1,-k}=y_{-k}\}=0$ or $\mathbbm{1}\{z_{2,k}=z, y_{2,-k}=y_{-k}\}=0$ or $\mathbbm{1}\{z_{3,k}=z, y_{3,-k}=y_{-k}\}=0$: In this case at least one of the functions $m_{j}(d_{1},z_{1},\theta)$, $m_{j}(d_{2},z_{2},\theta)$ or $m_{j}(d_{3},z_{3},\theta)$ is equal to zero for all $\theta \in \Theta$. Again, this means \eqref{eq_sdc_shattering2} is impossible. For example, if $m_{j}(d_{1},z_{1},\theta)$ is the zero function, then it is impossible for $\mathcal{M}_{j}$ to pick out both $(0,0,0)$ and $(1,0,0)$.
\end{enumerate}
Since $(d_{1},z_{1})$, $(d_{2},z_{2})$, and $(d_{3},z_{3})$ were arbitrary, we conclude that there exists no set of size $3$ shattered by $\mathcal{M}_{j}$. This implies that $\mathcal{M}_{j}$ is a VC-subgraph class, and using, for example, Theorem 2.6.7 in \cite{van1996weak}, we can deduce:
\begin{align*}
\sup_{Q \in \mathcal{Q}_{n}} \log N \left( \varepsilon, \mathcal{M}_{j}, ||\cdot||_{Q,2} \right) = O(1).
\end{align*}
Thus, $\mathcal{M}_{j}$ easily satisfies the entropy growth condition. Given the relation between the moment functions from \eqref{eq_pom_mom1} and \eqref{eq_pom_mom2}, a nearly identical analysis holds for the moment functions from the moment conditions \eqref{eq_pom_mom2}.  

Now let $j'$ index a generic moment function:
\begin{align*}
m_{j'}(X,U,\theta) = \left(\mathbbm{1}\{U \leq g(z_{0},x)\} - g(z_{0},x)\right)\mathbbm{1}\{X=x\},
\end{align*}
and let $\mathcal{M}_{j'}$ be the associated class of functions: 
\begin{align*}
\mathcal{M}_{j'}= \left\{ m_{j'}(\cdot,u,\theta):\mathcal{X} \to \mathbb{R} : (u,\theta) \in \mathcal{U}\times\Theta \right\}.
\end{align*}
Note this class indexes the moment functions from the moment conditions \eqref{eq_pom_mom3}. Also note that $(z_{0},x)$ are not arguments of the moment function, but are instead are associated with the index $j'$.

We claim that there exists no set of size $2$ shattered by $\mathcal{M}_{j'}$, implying $\mathcal{M}_{j'}$ is a VC-subgraph class. We will prove this by way of contradiction. In particular, suppose that there exists two points $x_{1}$ and $x_{2}$, and values $t_{1},t_{2} \in \mathbb{R}$ such that:
\begin{align}
\left| \left\{\begin{bmatrix}\mathbbm{1}\{m_{j'}(x_{1},u,\theta) \geq t_{1}\}\\ \mathbbm{1}\{m_{j'}(x_{2},u,\theta) \geq t_{2}\}\end{bmatrix}  : (u,\theta) \in \mathcal{U} \times \Theta \right\}\right| = 4.\label{eq_sdc_shattering3}
\end{align}
In other words, we suppose the set $\{x_{1},x_{2}\}$ is shattered by $\mathcal{M}_{j'}$, and that $t_{1},t_{2} \in \mathbb{R}$ witness the shattering. We have:
\begin{align*}
m_{j'}(x_{1},u,\theta) =\left(\mathbbm{1}\{u \leq g(z_{0},x)\} - g(z_{0},x)\right)\mathbbm{1}\{x_{1}=x\},\\
m_{j'}(x_{2},u,\theta) =\left(\mathbbm{1}\{u \leq g(z_{0},x)\} - g(z_{0},x)\right)\mathbbm{1}\{x_{2}=x\}.
\end{align*}
Now consider two cases:
\begin{enumerate}
	\item $\mathbbm{1}\{x_{1}=x\}=\mathbbm{1}\{x_{2}=x\}=1$: Then the two functions $m_{j'}(x_{1},u,\theta)$ and $m_{j'}(x_{2},u,\theta)$ are identical for all $(u,\theta)\in \mathcal{U} \times \Theta$. This means \eqref{eq_sdc_shattering3} is impossible, since at least one of the vectors $(1,0)$ and $(0,1)$ cannot be picked out by $\mathcal{M}_{j'}$. 
	\item  Either $\mathbbm{1}\{x_{1}=x\}=0$ or $\mathbbm{1}\{x_{2}=x\}=0$: In this case at least one of the functions $m_{j'}(x_{1},u,\theta)$ or $m_{j'}(x_{2},u,\theta)$ is the zero function. Again, this means \eqref{eq_sdc_shattering2} is impossible. For example, if $m_{j'}(x_{1},u,\theta)$ is the zero function, then it is impossible for $\mathcal{M}_{j'}$ to pick out both $(0,0)$ and $(1,0)$.
\end{enumerate}
Since $x_{1}$ and $x_{2}$ were arbitrary, we conclude that there exists no set of size $2$ shattered by $\mathcal{M}_{j'}$. This implies that $\mathcal{M}_{j'}$ is a VC-subgraph class, and using, for example, Theorem 2.6.7 in \cite{van1996weak}, we can deduce:
\begin{align*}
\sup_{Q \in \mathcal{Q}_{n}} \log N \left( \varepsilon, \mathcal{M}_{j'}, ||\cdot||_{Q,2} \right) = O(1).
\end{align*}
Thus, $\mathcal{M}_{j'}$ easily satisfies the entropy growth condition. Given the relation between the moment functions from \eqref{eq_pom_mom3} and \eqref{eq_pom_mom4}, a nearly identical analysis holds for the moment functions from the moment conditions \eqref{eq_pom_mom4}. 

Now let $j''$ index a generic moment function:
\begin{align*}
m_{j''}(Z,\theta) = t(z_{0},x)-\mathbbm{1}\{Z_{0}=z_{0},X=x)\} ,
\end{align*}
and let $\mathcal{M}_{j''}$ be the associated class of functions: 
\begin{align*}
\mathcal{M}_{j''}= \left\{ m_{j''}(\cdot,\theta):\mathcal{Z} \to \mathbb{R} : \theta \in \Theta \right\}.
\end{align*}
Note this class indexes the moment functions from the moment conditions \eqref{eq_pom_mom5}. Also note that $(z_{0},x)$ are not arguments of the moment function, but are instead are associated with the index $j''$.

We claim that there exists no set of size $3$ shattered by $\mathcal{M}_{j''}$, implying $\mathcal{M}_{j''}$ is a VC-subgraph class. To see this, note that for any three points $\{z_{1},z_{2},z_{3}\}$ we have:
\begin{align*}
m_{j''}(z_{1},\theta) =t(z_{0},x) - \mathbbm{1}\{z_{1,0}=z_{0},x_{1}=x)\},\\
m_{j''}(z_{2},\theta) =t(z_{0},x) - \mathbbm{1}\{z_{2,0}=z_{0},x_{2}=x)\},\\
m_{j''}(z_{3},\theta) =t(z_{0},x) - \mathbbm{1}\{z_{2,0}=z_{0},x_{2}=x)\}.
\end{align*}
The conclusion follows from the fact that two of these moment functions must always be the same. This implies that $\mathcal{M}_{j''}$ is a VC-subgraph class, and using, for example, Theorem 2.6.7 in \cite{van1996weak}, we can deduce:
\begin{align*}
\sup_{Q \in \mathcal{Q}_{n}} \log N \left( \varepsilon, \mathcal{M}_{j''}, ||\cdot||_{Q,2} \right) = O(1).
\end{align*}
Thus, $\mathcal{M}_{j''}$ easily satisfies the entropy growth condition. Given the relation between the moment functions from \eqref{eq_pom_mom5} and \eqref{eq_pom_mom6}, a nearly identical analysis holds for the moment functions from the moment conditions \eqref{eq_pom_mom6}. 

Finally, let $j'''$ index a generic moment function:
\begin{align*}
m_{j'''}(Z,U_{d},\theta)=U_{d}\left(\mathbbm{1}\{Z=z_{0},X=x\}\sum_{z_{0} \in \mathcal{Z}_{0}} t(z_{0},x)- \mathbbm{1}\{X=x\}t(z_{0},x) \right).
\end{align*}
and let $\mathcal{M}_{j''}$ be the associated class of functions: 
\begin{align*}
\mathcal{M}_{j'''}= \left\{ m_{j'''}(\cdot,u_{d},\theta):\mathcal{Z} \to \mathbb{R} : (u_{d},\theta) \in [\underline{Y},\overline{Y}]\times \Theta \right\}.
\end{align*}
Note this class indexes the moment functions from the moment conditions \eqref{eq_pom_mom7}. Also note that $(z_{0},x)$ are not arguments of the moment function, but are instead are associated with the index $j'''$.

We claim that there exists no set of size $5$ shattered by $\mathcal{M}_{j'''}$, implying $\mathcal{M}_{j'''}$ is a VC-subgraph class. To see this, note that for any five points $\{z_{1},z_{2},z_{3},z_{4},z_{5}\}$ we have:
\begin{align*}
m_{j'''}(z_{1},u_{d},\theta) =u_{d}\left(\mathbbm{1}\{z_{1,0}=z_{0},x_{1}=x\}\sum_{z_{0} \in \mathcal{Z}_{0}} t(z_{0},x)- \mathbbm{1}\{x_{1}=x\}t(z_{0},x) \right),\\
m_{j'''}(z_{2},u_{d},\theta) =u_{d}\left(\mathbbm{1}\{z_{2,0}=z_{0},x_{2}=x\}\sum_{z_{0} \in \mathcal{Z}_{0}} t(z_{0},x)- \mathbbm{1}\{x_{2}=x\}t(z_{0},x) \right),\\
m_{j'''}(z_{3},u_{d},\theta) =u_{d}\left(\mathbbm{1}\{z_{3,0}=z_{0},x_{3}=x\}\sum_{z_{0} \in \mathcal{Z}_{0}} t(z_{0},x)- \mathbbm{1}\{x_{3}=x\}t(z_{0},x) \right),\\
m_{j'''}(z_{4},u_{d},\theta) =u_{d}\left(\mathbbm{1}\{z_{4,0}=z_{0},x_{4}=x\}\sum_{z_{0} \in \mathcal{Z}_{0}} t(z_{0},x)- \mathbbm{1}\{x_{4}=x\}t(z_{0},x) \right),\\
m_{j'''}(z_{5},u_{d},\theta) =u_{d}\left(\mathbbm{1}\{z_{5,0}=z_{0},x_{5}=x\}\sum_{z_{0} \in \mathcal{Z}_{0}} t(z_{0},x)- \mathbbm{1}\{x_{5}=x\}t(z_{0},x) \right).
\end{align*}
The conclusion follows from the fact that two of these moment functions must always be identical for all $\theta$. This implies that $\mathcal{M}_{j'''}$ is a VC-subgraph class, and using, for example, Theorem 2.6.7 in \cite{van1996weak}, we can deduce:
\begin{align*}
\sup_{Q \in \mathcal{Q}_{n}} \log N \left( \varepsilon, \mathcal{M}_{j'''}, ||\cdot||_{Q,2} \right) = O(1).
\end{align*}
Thus, $\mathcal{M}_{j'''}$ easily satisfies the entropy growth condition. Given the relation between the moment functions from \eqref{eq_pom_mom7} and \eqref{eq_pom_mom8}, a nearly identical analysis holds for the moment functions from the moment conditions \eqref{eq_pom_mom8}.

Combining everything and applying Theorem \ref{theorem_pampac_learnability}(ii), we thus have that the policy space $\Gamma$ for this problem is learnable.

\end{appendix}
\end{document}